\documentclass[aps,prl,notitlepage,twocolumn,superscriptaddress]{revtex4-2}

\makeatletter
\def\@hangfrom@section#1#2#3{\normalsize\@hangfrom{#1#2}#3}%\MakeTextUppercase{#3}}%
\def\@hangfroms@section#1#2{\normalsize#1#2}%\MakeTextUppercase{#2}}%
\makeatother

\usepackage{amsmath,amsfonts,amssymb,bm}
\usepackage{graphicx,color}
\usepackage{enumerate}
\usepackage[colorlinks=true, allcolors=blue]{hyperref}

\usepackage{theorem}
\newtheorem{definition}{Definition}
\newtheorem{proposition}{Proposition}
\newtheorem{lemma}[proposition]{Lemma}
\newtheorem{theorem}[proposition]{Theorem}
\newtheorem{corollary}[proposition]{Corollary}

\newenvironment{proof}{\noindent \textbf{{Proof~} }}{\hfill $\blacksquare$}
\newenvironment{sketchproof}{\noindent \textit{{Proof.} }}{\hfill $\square$}

\usepackage{tikz}
\newenvironment{mytikz}{\begin{tikzpicture}[x=0.3pt,y=0.3pt,yscale=-1,xscale=1,baseline={([yshift=+0ex]current bounding box.center)}]}{\end{tikzpicture}}
\newenvironment{mytikz2}{\begin{tikzpicture}[x=0.4pt,y=0.4pt,yscale=-1,xscale=1,baseline={([yshift=+0ex]current bounding box.center)}]}{\end{tikzpicture}}
\newenvironment{mytikz3}{\begin{tikzpicture}[x=0.7pt,y=0.7pt,yscale=-1,xscale=1,baseline={([yshift=+0ex]current bounding box.center)}]}{\end{tikzpicture}}
\newenvironment{mytikz4}{\begin{tikzpicture}[x=0.6pt,y=0.6pt,yscale=-1,xscale=1,baseline={([yshift=+0ex]current bounding box.center)}]}{\end{tikzpicture}}
\newcommand{\nc}{\newcommand}
\nc{\ket}[1]{|#1\rangle}
\nc{\bra}[1]{\langle#1|}
\nc{\ketbra}[2]{|#1\rangle\!\langle#2|}
\nc{\braket}[2]{\langle#1|#2\rangle}
\nc{\braoprket}[3]{\langle#1|#2|#3\rangle}
\nc{\op}[1]{\operatorname{#1}}
\nc{\avg}[1]{\langle#1\rangle}
\nc{\ketbrasame}[1]{|#1\rangle\!\langle#1|}
\nc{\tr}{\op{tr}}
\nc{\swap}{\op{SWAP}}
\nc{\E}{\mathbb{E}}
\nc{\var}{\op{Var}}

\nc{\sxz}[1]{\textcolor{blue}{\textbf{[sxz: #1]}}}
\nc{\hk}[1]{\textcolor{violet}{\textbf{[hk: #1]}}}
\nc{\ls}[1]{\textcolor{purple}{\textbf{[ls: #1]}}}
\usepackage[normalem]{ulem}
\nc{\hknew}[1]{\textcolor{violet}{#1}}

\begin{document}
\title{Absence of barren plateaus in finite local-depth circuits with long-range entanglement}

\author{Hao-Kai Zhang}
\email{zhk20@mails.tsinghua.edu.cn}
\thanks{These two authors contributed equally to this work.}
\affiliation{Institute for Advanced Study, Tsinghua University, Beijing 100084, China}

\author{Shuo Liu}
\thanks{These two authors contributed equally to this work.}
\affiliation{Institute for Advanced Study, Tsinghua University, Beijing 100084, China}

\author{Shi-Xin Zhang}
\affiliation{Tencent Quantum Laboratory, Tencent, Shenzhen, Guangdong 518057, China}

\date{\today}

\begin{abstract}
Ground state preparation is classically intractable for general Hamiltonians. On quantum devices, shallow parametrized circuits can be effectively trained to obtain short-range entangled states under the paradigm of variational quantum eigensolver, while deep circuits are generally untrainable due to the barren plateau phenomenon. In this Letter, we give a general lower bound on the variance of circuit gradients for arbitrary quantum circuits composed of local 2-designs. Based on our unified framework, we prove the absence of barren plateaus in training finite local-depth circuits (FLDC) for the ground states of local Hamiltonians. FLDCs are allowed to be deep in the conventional circuit depth to generate long-range entangled ground states, such as topologically ordered states, but their local depths are finite, i.e., there is only a finite number of gates acting on individual qubits. This characteristic sets FLDC apart from shallow circuits: FLDC in general cannot be classically simulated to estimate local observables efficiently by existing tensor network methods in two and higher dimensions. We validate our analytical results with extensive numerical simulations and demonstrate the effectiveness of variational training using the generalized toric code model.
\end{abstract}

\maketitle

\textit{Introduction.---} 
Predicting the ground state properties of a quantum many-body system, as a central task in modern quantum physics, generally requires exponential resources for classical computers due to the curse of dimensionality: the number of parameters needed to describe a quantum system scales exponentially with the system size. Although some successful classical algorithms have been developed in past decades~\cite{Ceperley1986, White1992, Verstraete2006, Vidal2007, Schollwock2011} such as tensor networks~\cite{Verstraete2006, Vidal2007, Schollwock2011}, their respective limitations restrict the performance on general systems~\cite{Schollwock2011, Loh1990, Schuch2007, Biamonte2015}. Quantum computers bring new hope for this problem of quantum nature~\cite{Feynman1982}. Despite the limitation posed by noisy intermediate-scale quantum (NISQ) devices~\cite{Preskill2018}, there are many tentative quantum algorithms proposed. One of the representatives is the variational quantum eigensolver (VQE)~\cite{Peruzzo2014, Kandala2017, Cerezo2021a, Tilly2021_z, Zhang2021b_z, Liu2021c_z, Tabares2023}, which trains a parametrized quantum circuit (PQC) using a classical optimizer to minimize the energy. This hybrid quantum-classical paradigm is expected as one of the most promising routes toward practical quantum advantage~\cite{Arute2019, Zhong2020} in the NISQ era.

\begin{figure}
    \centering
    \includegraphics[width=0.86\linewidth]{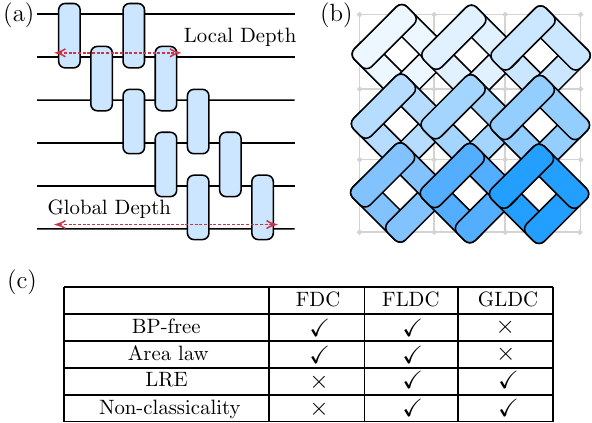}
    \caption{(a) and (b) Typical examples of finite local-depth circuits (FLDC) on 1D and 2D lattices, respectively. Darker colors in (b) indicate later action orders. (c) Compares the class of finite depth circuit (FDC), FLDC, and general linear depth circuit (GLDC) in terms of whether they are in general free from barren plateaus (BP), preserve entanglement area law, generate long-range entanglement (LRE), and can be simulated efficiently to compute local observable expectations by known classical methods (classicality). The inclusion relation is FDC $\subset$ FLDC $\subset$ GLDC.}
    \label{fig:fldc_table}
\end{figure}

However, these variational quantum algorithms including VQE still face great challenges for large-scale applications. One of the most notorious issues is the so-called barren plateau phenomenon~\cite{McClean2018}, which states that the circuit gradient vanishes exponentially with the system size under certain conditions, akin to the vanishing gradient issue in classical neural networks. The exponentially vanishing gradient will preclude the optimization progress and lead to the exponential measurement complexity. Extensive studies have been conducted to investigate barren plateau problems and possible remedies~\cite{Grimsley2019, Nakanishi2019, Ostaszewski2019, Grant2019, Arrasmith2020, Wang2021, Cerezo2021b, OrtizMarrero2021, Cerezo2021, Uvarov2021a, Pesah2021, Arrasmith2021, Holmes2021, Holmes2021a, Kim2021a, Liu2022a, Huembeli2021, Zhang2022, Grimsley2023, Liu2021a, Zhao2021a, CerveroMartin2023, Barthel2023, Miao2023, Miao2023a, Miao2023b, Liu2023c, Garcia2023, Liu2023a_z}. It is known that shallow circuits of finite or logarithmic depth are free from barren plateaus and can be trained efficiently for local Hamiltonians to obtain short-range entangled (SRE) states, while deep circuits of linear depth and beyond are in general untrainable~\cite{Cerezo2021, Uvarov2021a}. By contrast, many quantum states of physical interest exhibit long-range entanglement~\cite{Kitaev2003, Kitaev2006, Wen2003, Kitaev2006a, Gioia2022}, such as topologically ordered states, which cannot be prepared by circuits of less than linear depth~\cite{Bravyi2006, Satzinger2021, Chen2023a, Sun2023a}. Nevertheless, some evidence suggests that circuits corresponding to these long-range entangled (LRE) states possess characteristic architectures, such as sequential structures~\cite{Schon2005, Schon2007, Slattery2021, Wei2022, Chen2023a, Sun2023a, Zhang2023a, Liu2022d}, for the sake of the entanglement area law. This observation motivates us to rigorously explore the general relationship between barren plateaus, area law, and long-range entanglement.

In this Letter, we identify the critical role of the local depth as a key circuit feature that determines the trainability of PQCs. The local depth refers to the number of non-commuting gates acting on individual qubits, as illustrated in Fig.~\ref{fig:fldc_table}(a), in contrast to the conventional global depth defined by the minimum number of layers. This finding is based on our rigorously proved theorems, which establish a general lower bound on the gradient variance for arbitrary circuits composed of local 2-designs. The lower bound decays exponentially with the length and width of a certain set of paths on the circuit. For finite (or logarithmic) local-depth circuits and local Hamiltonians, the length and width can be upper bounded and hence give rise to the absence of barren plateaus. These finite local-depth circuits (FLDC) have strong expressibility to generate LRE states lacking in shallow circuits and are hard to simulate classically in two dimensions and above. This suggests that FLDC holds promise to serve as an appropriate class of ansatzes in VQE, as listed in Fig.~\ref{fig:fldc_table}(c). The absence of barren plateaus in FLDC is verified by numerical evaluations. Using the generalized 2D toric code model, we demonstrate that FLDC indeed has prominently better performance than both finite depth circuits and general linear depth circuits.

\textit{Basic setup.---} We start from the basic setup of VQE. The PQC can be written as $\mathbf{U}(\bm{\theta}) = \prod_{\mu=1}^{M} U_\mu(\theta_\mu)$, where $U_\mu(\theta_\mu)=e^{-i\Omega_\mu\theta_\mu}$ is a rotation gate, $\Omega_\mu$ is a Pauli-string generator and $\theta_\mu$ is a trainable parameter. The index $\mu$ follows the decreasing order from left to right in the product (the same below). For a given Hamiltonian $H$, the energy expectation $C(\bm{\theta}) = \avg{H} = \tr(\rho_0 \mathbf{U}^\dagger H \mathbf{U})$ is taken as the cost function, where $\rho_0=\ketbrasame{\bm{0}}$ and $\ket{\bm{0}}=\ket{0}^{\otimes N}$. $N$ is the number of qubits. We denote the Pauli decomposition of the Hamiltonian as $H=\sum_j\lambda_j h_j$ and assume the support of $H$ is within that of $\mathbf{U}$. The workflow of VQE involves running the PQC, measuring the cost function, and updating the trainable parameters iteratively using classical optimizers to minimize the cost function. In particular, the parameters are usually initialized randomly to thoroughly explore the parameter space in a probabilistic sense, rendering the PQC a random quantum circuit (RQC). A common assumption on RQCs is that the circuit is composed of blocks forming independent local $2$-designs. Here a block refers to a grouped continuous series of gates, which can be seen as the elementary unit when we construct a PQC. Grouping the $M$ gates into $M'$ blocks, the PQC can be rewritten as $\mathbf{U}=\prod_{k=1}^{M'} B_k$. The assumption of local $2$-designs will induce an ensemble of the entire circuit $\mathbf{U}$, which we denote as $\mathbb{U}$. Many statistical properties of RQCs can be analytically estimated based on $\mathbb{U}$, including the average and variance of the cost derivative $\partial_\mu C=\frac{\partial C}{\partial \theta_\mu}$. We provide preliminaries on unitary designs and the Weingarten calculus and a detailed introduction to our basic setup in Supplemental Material~\footnote{See Supplemental Material for preliminaries on unitary designs and the Weingarten calculus, basic setup, theoretical derivations including rigorous definitions and complete proofs, and additional numerical results with technical details, which includes Refs.~\cite{Collins2006, Hastings2007, Haferkamp2020, Bravyi2021, Zheng2017, Qin2020, Huang2017, Vidal2008, Malz2023, Landau2015, Chen2011, Chen2011a, Schuch2011, Zeng2019, Gray2021, Anshu2016, Schwarz2017, Ran2020a, Orus2019, Bluvstein2022, Bluvstein2023, Chu2023, Anshu2023, Bravyi2023, Xu2023, Pan2020, Pan2022, Paeckel2019, Raussendorf2005, Piroli2021, Verresen2021, Tantivasadakarn2021, Bravyi2022, Tantivasadakarn2023, Tantivasadakarn2023a}.}.

\begin{figure}
    \centering
    \includegraphics[width=0.98\linewidth]{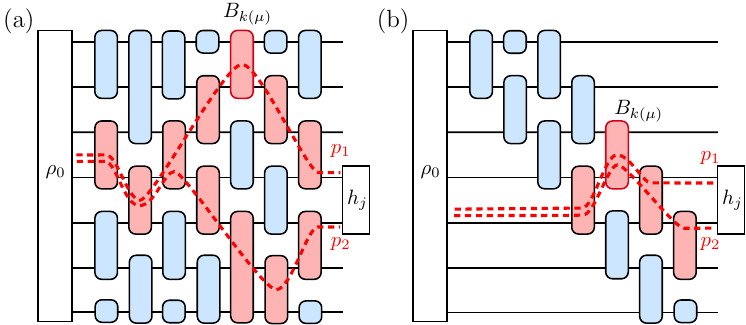}
    \caption{(a) A possible choice of path set $P_j=\{p_1,p_2\}$ on a general linear depth circuit. (b) Depicts a path set on a finite local-depth circuit correspondingly. The length of the path set in (a) grows linearly with the system size while that in (b) is bounded by the constant local depth.}
    \label{fig:path_set}
\end{figure}

\textit{General lower bound.---} We give an informal version of our general lower bound in Theorem~\ref{theorem:thm1} and leave the rigorous statement and proof to Supplemental Material~\cite{Note1}. The bound is closely related to a geometric concept of ``path,'' i.e., a time-ordered sequence of connected blocks on the circuit diagram as depicted in Fig.~\ref{fig:path_set}. For each Hamiltonian subterm $h_j$ in the causal cone of the differential block $B_{k(\mu)}$ (the block containing the differential parameter $\theta_\mu$), one can draw a collection of paths from $h_j$ to $\rho_0$, like $\{p_1, p_2\}$ in Fig.~\ref{fig:path_set}, with the right end covering $h_j$ and at least one of the paths passing through $B_{k(\mu)}$. We call it a chosen ``path set'' of $h_j$. We define two measures of the path set: length and head width. For common circuits composed of $2$-qubit blocks, the length is just the number of edges in the path set diagram and the head width is the number of blocks in the path set that are directly connected to $\rho_0$. Using these two measures, we can derive the following lower bound on the gradient variance.

\begin{theorem}\label{theorem:thm1}
{\rm \textbf{(informal version)}} The gradient variance $\var_\mathbb{U} [\partial_\mu C]$ can be lower bounded by a summation of contributions from each $h_j$ in the causal cone of $U_{\mu}$, where each contribution decays exponentially only with the length and head width of the chosen path set of $h_j$.
\end{theorem}

The lower bound in Theorem~\ref{theorem:thm1} holds for any possible choice of path sets. The path set with the minimum length and head width gives rise to the tightest bound. We remark that Theorem~\ref{theorem:thm1} holds for any RQCs composed of local $2$-designs regardless of circuit shapes, spatial dimensions, and gate locality. Further discussions on Theorem~\ref{theorem:thm1} including its consistency with previous literature, alternative initial states, gate generators, the location of the differential gate $U_\mu$ in $B_{k(\mu)}$, and its extension to a path-integral-like tighter form and application to other space-time correlators, are elaborated in Supplemental Material~\cite{Note1}. 

We provide an intuitive physical picture behind Theorem~\ref{theorem:thm1}. It is known that local quantum information will be scrambled~\cite{Sekino2008, Lashkari2013, Mi2021} when passing through the random gates in an RQC. The more gates it passes through, the more severe the scrambling becomes. If we consider a local term $h_j$ as a piece of information, finding a short path through the RQC will allow effective information transfer, so that adjusting parameters can make an effective difference in the expectation value, resulting in non-vanishing gradients. Conversely, if such a short path does not exist, local information will be scrambled globally, leaving no useful information for optimization.

\textit{Finite local-depth circuits.---} Before presenting Theorem~\ref{theorem:thm2}, we first clarify some relevant quantities. The maximum interaction range of a Hamiltonian is the maximum value of the support sizes of all $h_j$. An $r$-local Hamiltonian means that the maximum interaction range is fixed as $r$ that does not scale with $N$. The maximum block size $\beta$ is the maximum value of the support sizes of all blocks in the circuit. The local depth of a qubit is the number of blocks (or gates) acting on the qubit. We use $\chi$ to denote the maximum value of the local depths over all qubits, distinguished from the global depth $D$ which refers to the minimum number of layers where blocks within each layer commute with each other. An FLDC is defined as a circuit whose $\chi$ does not scale with $N$, without any other constraints such as circuit shapes, spatial dimensions, and gate locality. Based on Theorem~\ref{theorem:thm1}, we have the following theorem.

\begin{theorem}\label{theorem:thm2}
Suppose the maximum local depth of $\mathbf{U}$ is $\chi$ and the maximum block size is $\beta$. Then for any $r$-local Hamiltonian, the gradient variance is lower bounded by
\begin{equation}\label{eq:thm2}
    \var_\mathbb{U} [\partial_\mu C] ~\geq~ 4^{-r \chi \beta} \sum_j 2\lambda_j^2 ,
\end{equation}
where $j$ runs over $h_j$ that is nontrivial on the support of the differential block $B_{k(\mu)}$.
\end{theorem}
\begin{sketchproof}
The detailed proof is left in Supplemental Material~\cite{Note1}. The main idea is choosing the path sets in Theorem~\ref{theorem:thm1} to be the straight wires on the support of $h_j$, and hence the length and head width can be upper bounded in terms of $r,\chi$ and $\beta$. The contribution from $h_j$ that is trivial on the support of $B_{k(\mu)}$ is just neglected.
\end{sketchproof}

Theorem~\ref{theorem:thm2} elegantly integrates the factors related to barren plateaus in a concise manner, i.e., the block locality $\beta$~\cite{McClean2018}, the Hamiltonian locality $r$~\cite{Cerezo2021, Uvarov2021a} and the circuit deepness $\chi$. It is vitally important to note that the relevant quantity characterizing the circuit deepness is the local depth $\chi$, instead of the global depth $D$. These two depths may coincide~\cite{Cerezo2021, Uvarov2021a}, but they are distinct in general and can differ significantly as in Fig.~\ref{fig:fldc_table}. This implies that the circuit class free from barren plateaus can be enlarged to logarithmic local-depth circuits (Log-LDC), which is a superclass of circuit architectures proven previously, such as finite or logarithmic depth brickwall circuits~\cite{Cerezo2021, Uvarov2021a}, quantum convolutional neural networks (QCNN)~\cite{Pesah2021, Zhao2021a}, multiscale entanglement renormalization ansatzes (MERA)~\cite{Zhao2021a, Barthel2023, Miao2023, CerveroMartin2023}, tree tensor networks~\cite{Zhao2021a, Barthel2023, Miao2023, CerveroMartin2023}, matrix product states (MPS)~\cite{Liu2021a, Zhao2021a, Barthel2023, Miao2023, CerveroMartin2023}, and high-dimensional isometric tensor network states~\cite{Liu2023c}. We focus on FLDC in this work. Log-LDC class involves states beyond the area law, e.g., gapless topologically ordered states, which is also interesting to study in the future.

A significant feature of FLDCs composed of spatially local gates is that the generated quantum states satisfy the entanglement area law (or say boundary law) because the number of gates acting across any simple partition boundary entangling the two sides can be upper bounded by the local depth times the size of the boundary~\cite{Note1}. This feature makes them form a subclass of the projected entangled paired states (PEPS)~\cite{Verstraete2006, Schollwock2011} of the corresponding spatial dimension, where the local depth $\chi$ plays the role of bond dimension. Note that PEPS can represent LRE states because the non-unitary projectors in PEPS enable quantum teleportation, while FLDC relies on large global depth. Previously proposed circuits of tensor network states~\cite{Banuls2008, Ran2020, Zhou2021, Haghshenas2022, Zhao2021a, Barthel2023, Miao2023, CerveroMartin2023, Slattery2021, Wei2022, Liu2022d, Chen2023a} including sequential quantum circuits~\cite{Chen2023a}, isometric tensor network states~\cite{Zaletel2020, Soejima2020, Slattery2021, Liu2022d, Liu2023c, Haller2023}, and plaquette PEPS~\cite{Wei2022}, can all be seen as subclasses of FLDC. This implies that FLDC covers a wide range of physical ground states such as string-net states with anyons~\cite{Liu2022d, Liu2023} and fracton-ordered states~\cite{Chen2023a}.

\textit{Non-classicality of FLDC.---} A matter of recent concern is the classical simulability of the tasks with the provable absence of barren plateaus~\cite{Cerezo2023}. Previous results that are proven free from barren plateaus mainly focus on finite or logarithmic depth circuits~\cite{Cerezo2021, Uvarov2021a, Pesah2021, Zhao2021a, Barthel2023, Miao2023, CerveroMartin2023}, which can be efficiently simulated to compute local observable expectations due to the existence of small causal cones and small treewidths of the corresponding tensor networks~\cite{Markov2008}. Nevertheless, the causal cone in FLDC can be extensive due to the large global depth, and the loop structures in FLDC of two dimensions and above can lead to polynomially large treewidth, rendering FLDC in general hard to simulate classically for local observable expectations (the 1D case has constant treewidth and can be efficiently simulated via MPS methods). In fact, even for the subclasses of 2D FLDC such as isometric tensor network states~\cite{Zaletel2020, Slattery2021, Liu2022d} and plaquette PEPS~\cite{Wei2022}, there is no known efficient method to compute the expectation values of arbitrary local observables with controllable error, not to mention non-local observables of interest such as few-body long-distance correlators, non-local order parameters, dynamical correlations and so forth. In particular, a recent work~\cite{Malz2024} rigorously proves that computing local expectation values in isometric tensor networks is $\mathsf{BQP}$-complete, i.e., is hard to simulate classically unless $\mathsf{BQP}=\mathsf{BPP}$. The same naturally holds true for FLDC because the states generated by FLDC form a superclass of isometric tensor network states.

Therefore, FLDC (or Log-LDC) is a circuit class that is proven to be barren-plateau-free and at the same time generally cannot be efficiently simulated to estimate local observables by existing classical methods. On the contrary, it can be accomplished within polynomial time by running FLDCs on quantum devices and measuring corresponding observables. This suggests that FLDC is potentially relevant to quantum advantage in the ground state preparation task. A detailed discussion and a numerical demonstration of the computational overhead for contracting tensor networks of FDLC are provided in Supplemental Material~\cite{Note1}.

% Figure 3: direct verification of gradients
\begin{figure}
    \centering
    \includegraphics[width=0.98\linewidth]{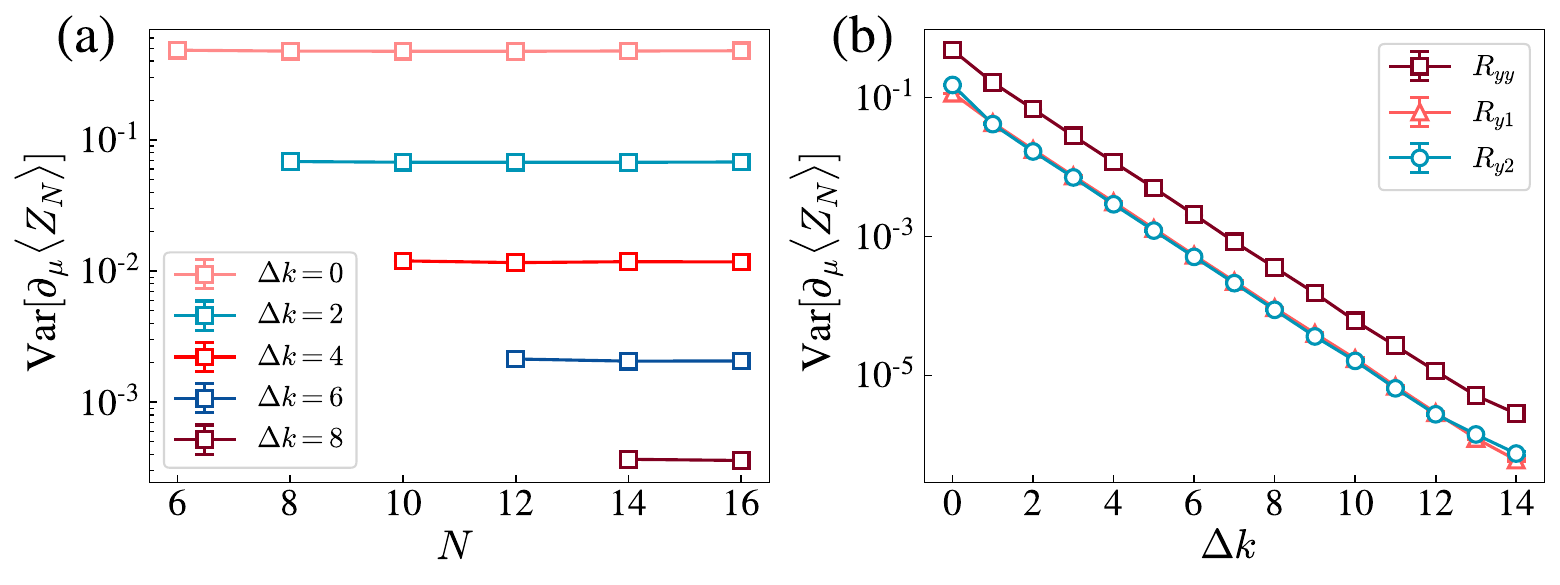}
    \caption{(a) The variance of the derivative vs the system size $N$ in a 1D FLDC instance. The observable is chosen as $Z_N$. $\Delta k$ is the distance between the differential block and the last block, proportional to the path length. (b) The variance vs $\Delta k$ by fixing $N=16$. $R_{yy}$, $R_{y1}$, and $R_{y2}$ represent the different choices of the differential gate~\cite{Note1}.}
    \label{fig:var_sample}
\end{figure}

\textit{Numerical experiments.---} FLDC has stronger expressibility than its subclass finite depth circuits (FDC), e.g., brickwall circuits of constant depth~\cite{Cerezo2021, Uvarov2021a}, because FDC can only generate SRE states such as symmetry-protected topological states~\cite{Azses2020}. FLDCs have less expressibility than its superclass general linear depth circuits (GLDC), e.g., brickwall circuits of linear depth, as typical GLDCs lead to entanglement volume law~\cite{OrtizMarrero2021}. But FLDC has better trainability than GLDC. We will compare the variational performance of the three circuit classes to see the advantages and the good trade off between trainability and expressibility brought by FLDC.

To demonstrate the absence of barren plateaus in FLDC, we estimate the cost gradient in a 1D FLDC ansatz with a ladder layout as in Fig.~\ref{fig:fldc_table}(a). The two-qubit block template is chosen as the Cartan decomposition~\cite{Note1}. The Hamiltonian is chosen as a single Pauli $Z$ operator on the last qubit. All the numerical experiments are implemented using TensorCircuit~\cite{Zhang2022_z}. As depicted in Fig.~\ref{fig:var_sample}, the gradient variance is almost constant with the system size, while it decays exponentially with the path length $\Delta k$. This resembles the phenomenon found in isometric tensor networks recently~\cite{Liu2021a, Barthel2023, Miao2023, Liu2023c}. However, we clarify that the exponential decay with $\Delta k$ does not indicate poor trainability in practice, because as long as the gradients of some circuit parameters do not vanish, the optimization could still proceed successfully. 

As an example of training FLDCs for LRE ground states, we use the generalized 2D toric code model under the external field $\bm{h}=(h^x, h^y, h^z)$ with open boundary conditions. The ground state near the zero-field limit is topologically ordered and then experiences a quantum phase transition to an SRE state with increasing $\bm{h}$~\cite{PhysRevLett.98.070602, PhysRevB.85.195104}. The ground state at $\bm{h}=0$ can be constructed by applying the Hadamard and CNOT gates sequentially~\cite{Satzinger2021, Chen2023a}, which belongs to the FLDC class. Possible generalization to $h^z\neq0$ has also been proposed~\cite{Sun2023a}. However, unlike in Ref.~\cite{Sun2023a}, we will not utilize any prior information about the exact ground state except the entanglement area law. Namely, we choose our ansatz to be an FLDC similar to Fig.~\ref{fig:fldc_table}(b), with each two-qubit block being the general Cartan decomposition. We also conducted the same simulation using typical ansatzes in FDC and GLDC for comparison. As shown in Fig.~\ref{fig:tc_xzfield}(a), the energies of FLDC almost coincide with the exact values from the exact diagonalization (ED). By contrast, the energies of GLDC severely deviate due to poor trainability. On the other hand, although FDC does not suffer from barren plateaus, it lacks the expressibility to represent LRE states faithfully, so FDC works well in the large field limit but deviates near the zero-field limit. We also show the results of the topological entanglement entropy $S_{\text{topo}}$ of these variational states in Fig.~\ref{fig:tc_xzfield}(b) correspondingly. The technical details and additional numerical results can be found in Supplemental Material~\cite{Note1}.

% Figure 4: generalized toric code model
\begin{figure}
    \centering
    \includegraphics[width=0.98\linewidth]{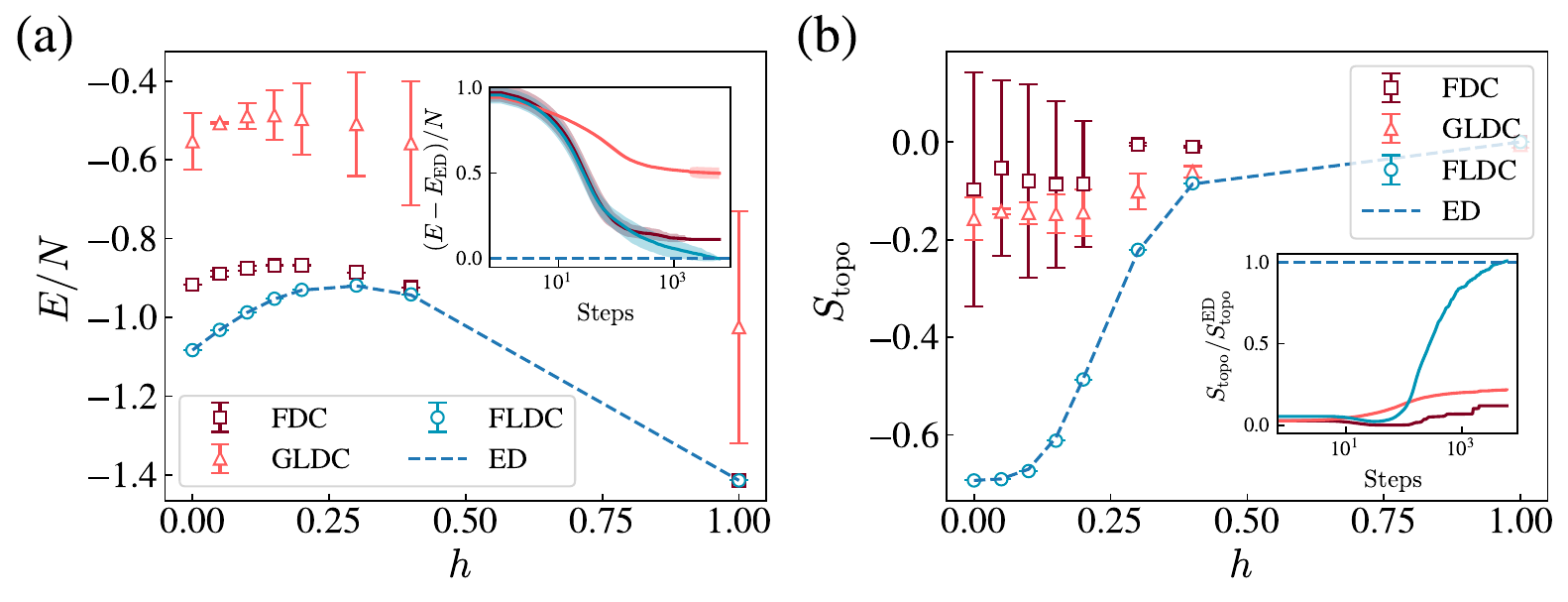}
    \caption{VQE performance comparison of the FDC, FLDC, and GLDC ansatzes using the generalized toric code model under the external field $h^z=h^x=h$ with $N=12$. The data are averaged over the best half of the $100$ training trajectories starting from different initializations. (a) The converged energy $E/N$ vs $h$. The inset depicts the energy training dynamics at $h=0.1$. The dashed lines represent the exact values obtained from ED. The (shaded) error bar represents the standard deviation. (b) The topological entanglement entropy $S_{\text{topo}}$ correspondingly.}
    \label{fig:tc_xzfield}
\end{figure}

\textit{Discussion.---} In this Letter, we prove a general lower bound on the gradient variance for arbitrary quantum circuits composed of local $2$-designs, which unifies the known gradient scaling behaviors of various architectures. An intuitive physical picture emerges that relates the non-vanishing gradients with the information scrambling in RQCs along certain path sets. We further prove the absence of barren plateaus for local Hamiltonians in a new circuit class---finite local-depth circuits, which can generate LRE states thanks to large global depths. FLDCs composed of spatially local gates preserve the entanglement area law, which makes it form a powerful and accessible subclass of PEPS that covers a wide range of physical ground states. Importantly, FLDC cannot be classically simulated efficiently in two and higher dimensions by the known tensor network methods. We remark that the indication of local depth is also instructive in developing quantum architecture search schemes~\cite{Zhang2020b_z, Du2020a_z, Lu2020_z, Zhang2021_z}. Finally, we point out that the absence of barren plateaus is a necessary but not sufficient condition for the effectiveness of training. There are other challenging issues such as the local minimum problem~\cite{Bittel2021, Zhang2023, Anschuetz2022}. Enhancing the VQE performance of FLDCs in more general systems requires further exploration in future studies. 

% As an independent interest, it is also intriguing to utilize FLDCs to investigate the circuit constructions of various LRE ground states.

\begin{acknowledgements}
\textit{Acknowledgments.---} We acknowledge stimulating discussions with Xie Chen, Jing-Yu Zhao, Zhi-Yuan Wei, Yu-Jie Liu, and Zhehao Dai. This work was supported in part by the Innovation Program for Quantum Science and Technology (grant No. 2021ZD0302502). 
\end{acknowledgements}

\bibliographystyle{apsreve}
\bibliography{autoref}

\clearpage
\newpage
% \appendix
\widetext

\begin{center}
% \textbf{\large Supplemental Material for \\``Training long-range entangled states with finite local-depth circuits''}
\textbf{\large Supplemental Material for \\``Absence of barren plateaus in finite local-depth circuits with long-range entanglement''}
\end{center}

\renewcommand{\theproposition}{S\arabic{proposition}}
\setcounter{proposition}{0}
\renewcommand{\thedefinition}{S\arabic{definition}}
\setcounter{definition}{0}

% \numberwithin{equation}{section}
\renewcommand{\thefigure}{S\arabic{figure}}
\setcounter{figure}{0}
\renewcommand{\theequation}{S\arabic{equation}}
\setcounter{equation}{0}
\renewcommand{\thesection}{\Roman{section}}
\setcounter{section}{0}
\setcounter{secnumdepth}{4}
% use ``secnumdepth'' to show up the section number

In this supplemental material, we first introduce the Weingarten integration techniques together with other mathematical preliminaries. Then, we provide rigorous formulations and complete proofs of the theorems mentioned in the main text. Next, we offer technical details and additional numerical results to further showcase our theorems and conclusions. Finally, we provide further discussions on the classical simulability of finite local-depth circuits and the relation with the measurement-assisted approach to prepare long-range entangled states.

\section{Preliminaries}

We start from the definition of unitary $t$-design. Consider a set $\mathbb{V}$ of unitaries $V$ on a $d$-dimensional Hilbert space. Denote $P_{t,t}(V)$ as a polynomial of degree at most $t$ in the entries of $V$ and $V^\dagger$. Then $\mathbb{V}$ is a unitary $t$-design if
\begin{equation}\label{eq:definiton_design}
    \frac{1}{|\mathbb{V}|}\sum_{V\in\mathbb{V}} P_{t,t}(V) = \int_{\mathcal{U}(d)} d\mu(V) P_{t,t}(V),
\end{equation}
for arbitrary $P_{t,t}(V)$, where $|\mathbb{V}|$ is the size of the set $\mathbb{V}$. $\mathcal{U}(d)$ is the unitary group of degree $d$ and $d\mu(V)$ is the Haar measure, or say the normalized uniform measure on $\mathcal{U}(d)$. That is to say, $P_{t,t}(V)$ averaging over $\mathbb{V}$ will yield exactly the same result as averaging over the entire unitary group $\mathcal{U}(d)$. In other words, $\mathbb{V}$ can mimic the Haar distribution up to the $t$-degree moment. We call a set an ensemble if there is a probability measure defined on it implicitly or explicitly. For example, $\mathbb{V}$ in Eq.~\eqref{eq:definiton_design} is an ensemble of unitaries since a uniform weight distribution is assigned implicitly. If $\mathbb{V}$ is continuous, the summation on the left-hand side of Eq.~\eqref{eq:definiton_design} is just replaced by the integral over $\mathbb{V}$ representing the average. Throughout this paper, all calculations only involve integrals over unitaries up to the second moment, e.g., the averages and variances of derivatives. Thus, as we will discuss below, the ensemble demanded by our theorems is just a series of local unitary $2$-designs, where ``local'' means that the unitaries just act on a few qubits instead of all of them. By contrast, a global $2$-design means that the unitaries act on all the qubits.

In order to calculate various integrals over unitaries, we introduce the notion of twirling channel. We focus on an $N$-qubit system with local Hilbert space $\mathcal{H}=\mathbb{C}^2$. The space of linear operators on $\mathcal{H}$ is denoted as $\mathcal{L}(\mathcal{H})$. The $N$ qubits are denoted as $\{q_1,q_2,\cdots,q_N\}$. A $t$-degree twirling channel, or simply a twirler, is a quantum channel $\mathcal{T}_{s}^{(t)}:\mathcal{L}(\mathcal{H}^{\otimes t})\rightarrow \mathcal{L}(\mathcal{H}^{\otimes t})$ of the form
\begin{equation}\label{eq:twirling_def}
    \mathcal{T}_{s}^{(t)}(\cdot) = \int d\mu_s(U)~ U^{\dagger\otimes t} (\cdot) U^{\otimes t},
\end{equation}
where $s$ denotes a given subset of qubits. $\mu_s$ is a given measure over the unitary group $\mathcal{U}(2^{|s|})$. $|s|$ denotes the number of qubits in the subset $s$. $U\in\mathcal{U}(2^{|s|})$ is a unitary operator which acts on $s$. Namely, $s$ is the support of $U$, which will be denoted as $s(U)$ for clarity if needed. In the following, we take $\mu_s$ as the Haar measure by default unless specially claimed. By tensor network diagrams, an instance of a twirling channel can be represented as
\begin{equation}
\begin{mytikz}

%Straight Lines [id:da685671816830796] 
\draw    (290,250) -- (420,250) ;
%Straight Lines [id:da9080497305734556] 
\draw    (110,250) -- (240,250) ;
%Straight Lines [id:da2077966271022651] 
\draw    (290,230) -- (420,230) ;
%Straight Lines [id:da3356861512442262] 
\draw    (110,230) -- (240,230) ;
%Straight Lines [id:da07060950380202802] 
\draw    (290,240) -- (420,240) ;
%Straight Lines [id:da7856047818817662] 
\draw    (110,240) -- (240,240) ;
%Straight Lines [id:da20355033086822205] 
\draw    (290,220) -- (420,220) ;
%Straight Lines [id:da842940946481108] 
\draw    (110,220) -- (240,220) ;
%Straight Lines [id:da16894149119039037] 
\draw    (290,170) -- (420,170) ;
%Straight Lines [id:da22658798660873392] 
\draw    (110,170) -- (240,170) ;
%Straight Lines [id:da8293194611223473] 
\draw    (290,150) -- (420,150) ;
%Straight Lines [id:da8394056485210961] 
\draw    (110,150) -- (240,150) ;
%Straight Lines [id:da2564388397641322] 
\draw    (290,160) -- (420,160) ;
%Straight Lines [id:da255098181830858] 
\draw    (110,160) -- (240,160) ;
%Straight Lines [id:da8711393597665376] 
\draw    (290,140) -- (420,140) ;
%Straight Lines [id:da20003120995408175] 
\draw    (110,140) -- (240,140) ;
%Shape: Polygon [id:ds7081991315343819] 
\draw  [fill={rgb, 255:red, 204; green, 230; blue, 255 }  ,fill opacity=1 ] (380,50) -- (380,280) -- (330,280) -- (330,100) -- (200,100) -- (200,280) -- (150,280) -- (150,50) -- cycle ;
%Shape: Rectangle [id:dp28611360639719474] 
\draw  [fill={rgb, 255:red, 230; green, 230; blue, 230 }, fill opacity=1] (240,120) -- (290,120) -- (290,280) -- (240,280) -- cycle ;
% Text Node
\draw (270,75) node [anchor=center][inner sep=0.75pt] [font=\normalsize] [align=left] {$\mathcal{T}_{s}^{(t)}$};
\end{mytikz}
\quad= \int d\mu_s(U) ~~
\begin{mytikz}
%Straight Lines [id:da685671816830796] 
\draw    (290,250) -- (420,250) ;
%Straight Lines [id:da9080497305734556] 
\draw    (110,250) -- (240,250) ;
%Straight Lines [id:da2077966271022651] 
\draw    (290,210) -- (420,210) ;
%Straight Lines [id:da3356861512442262] 
\draw    (110,210) -- (240,210) ;
%Straight Lines [id:da07060950380202802] 
\draw    (290,230) -- (420,230) ;
%Straight Lines [id:da7856047818817662] 
\draw    (110,230) -- (240,230) ;
%Straight Lines [id:da20355033086822205] 
\draw    (290,190) -- (420,190) ;
%Straight Lines [id:da842940946481108] 
\draw    (110,190) -- (240,190) ;
%Straight Lines [id:da16894149119039037] 
\draw    (290,140) -- (420,140) ;
%Straight Lines [id:da22658798660873392] 
\draw    (110,140) -- (240,140) ;
%Straight Lines [id:da8293194611223473] 
\draw    (290,100) -- (420,100) ;
%Straight Lines [id:da8394056485210961] 
\draw    (110,100) -- (240,100) ;
%Straight Lines [id:da2564388397641322] 
\draw    (290,120) -- (420,120) ;
%Straight Lines [id:da255098181830858] 
\draw    (110,120) -- (240,120) ;
%Straight Lines [id:da8711393597665376] 
\draw    (290,80) -- (420,80) ;
%Straight Lines [id:da20003120995408175] 
\draw    (110,80) -- (240,80) ;
%Shape: Rectangle [id:dp28611360639719474] 
\draw  [fill={rgb, 255:red, 230; green, 230; blue, 230 }  ,fill opacity=1 ] (240,50) -- (290,50) -- (290,280) -- (240,280) -- cycle ;
%Shape: Rectangle [id:dp36968294648720157] 
\draw  [fill={rgb, 255:red, 204; green, 230; blue, 255 }  ,fill opacity=1 ] (150,90) -- (190,90) -- (190,130) -- (150,130) -- cycle ;
%Shape: Rectangle [id:dp5200904166661877] 
\draw  [fill={rgb, 255:red, 204; green, 230; blue, 255 }  ,fill opacity=1 ] (150,200) -- (190,200) -- (190,240) -- (150,240) -- cycle ;
%Shape: Rectangle [id:dp06898001667318487] 
\draw  [fill={rgb, 255:red, 204; green, 230; blue, 255 }  ,fill opacity=1 ] (340,90) -- (379,90) -- (379,130) -- (340,130) -- cycle ;
%Shape: Rectangle [id:dp7244255712208405] 
\draw  [fill={rgb, 255:red, 204; green, 230; blue, 255 }  ,fill opacity=1 ] (340,200) -- (379,200) -- (379,240) -- (340,240) -- cycle ;

% Text Node
\draw (170,110) node [anchor=center][inner sep=0.75pt] [font=\small] [align=left] {$U^{\dagger }$};
% Text Node
\draw (170,220) node [anchor=center][inner sep=0.75pt] [font=\small] [align=left] {$U^{\dagger }$};
% Text Node
\draw (360,110) node [anchor=center][inner sep=0.75pt] [font=\small] [align=left] {$U$};
% Text Node
\draw (360,220) node [anchor=center][inner sep=0.75pt] [font=\small] [align=left] {$U$};
\end{mytikz}\quad,
\end{equation}
In this instance, we take $t=2$, $N=4$ and $s=\{q_2,q_3\}$. The grey block represents the input of $\mathcal{T}^{(t)}_s(\cdot)$. It is known that if $\mu_s$ is the Haar measure, the integral in the definition of twirling channel can be analytically calculated and expressed by the Weingarten function $\op{Wg}(\cdot)$~\cite{Collins2006}, i.e.,
\begin{equation}\label{eq:weingarten}
    \mathcal{T}_{s}^{(t)}(\cdot) = \sum_{\sigma,\tau\in \mathcal{S}_{t}} \op{Wg}^{(t)}(\sigma\tau^{-1}, 2^{|s|}) ~\tr_{ts}\left[S(\tau)\vert_{ts}(\cdot)\right] \otimes S(\sigma)\vert_{ts},
\end{equation}
where $\sigma$ and $\tau$ are elements of the $t$-degree permutation group $\mathcal{S}_t$. $\op{Wg}^{(t)}(\sigma\tau^{-1}, 2^{|s|})$ is the $t$-degree Weingarten function of the permutation element $\sigma\tau^{-1}$ and the $2^{|s|}$-degree unitary group $\mathcal{U}(2^{|s|})$. $ts$ denotes the $t$ copies of support $s$. $\tr_{ts}(\cdot)$ denotes the partial tracing operation on all the $t$ copies of support $s$. $S(\sigma)\vert_{ts}$ represents the generalized swap operator which permutes the indices of the $t$ copies of support $s$ according to the permutation element $\sigma$. We use a vertical bar ``$\vert$'' in front of the support notation $ts$ to avoid confusion with other common subscripts. The tensor product ``$\otimes$'' appears because $\tr_{ts}$ is generally a partial trace when $|s|\neq N$. Note that Eq.~\eqref{eq:weingarten} can also be viewed as a definition of unitary $t$-designs. That is to say, if the $t$-degree twirling channel of some ensemble $\mathbb{V}$ coincides with that of the corresponding Haar ensemble, then the given ensemble $\mathbb{V}$ is a $t$-design. For the cases of $t=1$ and $2$, Eq.~\eqref{eq:weingarten} becomes
\begin{equation}\label{eq:twirling_1st_order}
    \mathcal{T}_{s}^{(1)}(\cdot) = \frac{1}{2^{|s|}}\tr_s(\cdot)\otimes I\vert_s, 
\end{equation}
\begin{equation}\label{eq:twirling_2st_order}
\begin{aligned}
    \mathcal{T}_{s}^{(2)}(\cdot) =& \frac{1}{2^{|2s|}-1}\left[\tr_{2s}(\cdot)\otimes I\vert_{2s} + \tr_{2s}(S\vert_{2s}\cdot)\otimes S\vert_{2s} \right] \\
    &- \frac{1}{2^{|s|}(2^{|2s|}-1)}\left[\tr_{2s}(\cdot)\otimes S\vert_{2s} + \tr_{2s}(S\vert_{2s}\cdot)\otimes I\vert_{2s} \right],
\end{aligned}
\end{equation}
where $I\vert_{s}$ is the identity on the Hilbert space of the support $s$. $I\vert_{2s}$ and $S\vert_{2s}$ are the identity and swap operators on the doubled Hilbert space of the support $s$, respectively. By tensor network diagrams, Eqs.~\eqref{eq:twirling_1st_order} and \eqref{eq:twirling_2st_order} can be represented as
\begin{equation}
\begin{mytikz}
%Straight Lines [id:da16894149119039037] 
\draw    (290,170) -- (420,170) ;
%Straight Lines [id:da22658798660873392] 
\draw    (110,170) -- (240,170) ;
%Straight Lines [id:da8293194611223473] 
\draw    (291,150) -- (420,150) ;
%Straight Lines [id:da8394056485210961] 
\draw    (110,150) -- (240,150) ;
%Straight Lines [id:da2564388397641322] 
\draw    (290,160) -- (420,160) ;
%Straight Lines [id:da255098181830858] 
\draw    (110,160) -- (240,160) ;
%Straight Lines [id:da8711393597665376] 
\draw    (290,140) -- (420,140) ;
%Straight Lines [id:da20003120995408175] 
\draw    (110,140) -- (240,140) ;
%Shape: Polygon [id:ds7081991315343819] 
\draw  [fill={rgb, 255:red, 204; green, 230; blue, 255 }  ,fill opacity=1 ] (380,50) -- (380,190) -- (330,190) -- (330,100) -- (200,100) -- (200,190) -- (150,190) -- (150,50) -- cycle ;
%Shape: Rectangle [id:dp28611360639719474] 
\draw  [fill={rgb, 255:red, 230; green, 230; blue, 230 }  ,fill opacity=1 ] (240,120) -- (290,120) -- (290,190) -- (240,190) -- cycle ;

% Text Node
\draw (270,75) node [anchor=center][inner sep=0.75pt]  [font=\normalsize] [align=left] {$\mathcal{T}_{s}^{(1)}$};
\end{mytikz}
\quad=\frac{1}{2^{|s|}}\quad
\begin{mytikz}
%Straight Lines [id:da2531721179515183] 
\draw    (170,140) -- (240,140) ;
%Straight Lines [id:da10988044864831115] 
\draw    (170,80) -- (240,80) ;
%Straight Lines [id:da16894149119039037] 
\draw    (290,140) -- (360,140) ;
%Straight Lines [id:da8711393597665376] 
\draw    (290,80) -- (360,80) ;
%Shape: Rectangle [id:dp28611360639719474] 
\draw  [fill={rgb, 255:red, 230; green, 230; blue, 230 }  ,fill opacity=1 ] (240,50) -- (290,50) -- (290,170) -- (240,170) -- cycle ;
%Curve Lines [id:da17726730966326554] 
\draw [color={rgb, 255:red, 0; green, 0; blue, 0 }  ,draw opacity=1 ]   (290,100) .. controls (315,100) and (325,95) .. (265,95) .. controls (205,95) and (220,100) .. (240,100) ;
%Curve Lines [id:da24318764031777662] 
\draw [color={rgb, 255:red, 0; green, 130; blue, 200 }  ,draw opacity=1 ]   (360,100) .. controls (330,100) and (320,90) .. (265,90) .. controls (205,90) and (200,100) .. (170,100) ;
%Curve Lines [id:da985996013971671] 
\draw [color={rgb, 255:red, 0; green, 0; blue, 0 }  ,draw opacity=1 ]   (290,120) .. controls (317.16,120) and (320,115) .. (265,115) .. controls (205,115) and (220,120) .. (240,120) ;
%Curve Lines [id:da7739595450228074] 
\draw [color={rgb, 255:red, 0; green, 130; blue, 200 }  ,draw opacity=1 ]   (360,120) .. controls (330,120) and (320,110) .. (265,110) .. controls (205,110) and (200,120) .. (170,120) ;
\end{mytikz}\quad,
\end{equation}
\begin{equation}
\begin{aligned}
\begin{mytikz}
%Straight Lines [id:da685671816830796] 
\draw    (290,250) -- (420,250) ;
%Straight Lines [id:da9080497305734556] 
\draw    (110,250) -- (240,250) ;
%Straight Lines [id:da2077966271022651] 
\draw    (290,230) -- (420,230) ;
%Straight Lines [id:da3356861512442262] 
\draw    (110,230) -- (240,230) ;
%Straight Lines [id:da07060950380202802] 
\draw    (290,240) -- (420,240) ;
%Straight Lines [id:da7856047818817662] 
\draw    (110,240) -- (240,240) ;
%Straight Lines [id:da20355033086822205] 
\draw    (290,220) -- (420,220) ;
%Straight Lines [id:da842940946481108] 
\draw    (110,220) -- (240,220) ;
%Straight Lines [id:da16894149119039037] 
\draw    (290,170) -- (420,170) ;
%Straight Lines [id:da22658798660873392] 
\draw    (110,170) -- (240,170) ;
%Straight Lines [id:da8293194611223473] 
\draw    (290,150) -- (420,150) ;
%Straight Lines [id:da8394056485210961] 
\draw    (110,150) -- (240,150) ;
%Straight Lines [id:da2564388397641322] 
\draw    (290,160) -- (420,160) ;
%Straight Lines [id:da255098181830858] 
\draw    (110,160) -- (240,160) ;
%Straight Lines [id:da8711393597665376] 
\draw    (290,140) -- (420,140) ;
%Straight Lines [id:da20003120995408175] 
\draw    (110,140) -- (240,140) ;
%Shape: Polygon [id:ds7081991315343819] 
\draw  [fill={rgb, 255:red, 204; green, 230; blue, 255 }  ,fill opacity=1 ] (380,50) -- (380,280) -- (330,280) -- (330,100) -- (200,100) -- (200,280) -- (150,280) -- (150,50) -- cycle ;
%Shape: Rectangle [id:dp28611360639719474] 
\draw  [fill={rgb, 255:red, 230; green, 230; blue, 230 }, fill opacity=1] (240,120) -- (290,120) -- (290,280) -- (240,280) -- cycle ;
% Text Node
\draw (270,75) node [anchor=center][inner sep=0.75pt] [font=\normalsize] [align=left] {$\mathcal{T}_{s}^{(2)}$};
\end{mytikz}
\quad =& ~\frac{1}{2^{|2s|}-1} \left(
\begin{mytikz}
%Straight Lines [id:da4553034597499923] 
\draw    (170,250) -- (240,250) ;
%Straight Lines [id:da879995939230652] 
\draw    (170,190) -- (240,190) ;
%Straight Lines [id:da8334357940613453] 
\draw    (290,250) -- (360,250) ;
%Straight Lines [id:da875812750293314] 
\draw    (290,190) -- (360,190) ;
%Straight Lines [id:da16894149119039037] 
\draw    (290,140) -- (360,140) ;
%Straight Lines [id:da8711393597665376] 
\draw    (290,80) -- (360,80) ;
%Shape: Rectangle [id:dp28611360639719474] 
\draw  [fill={rgb, 255:red, 230; green, 230; blue, 230 }  ,fill opacity=1 ] (240,50) -- (290,50) -- (290,280) -- (240,280) -- cycle ;
%Curve Lines [id:da4915785647998141] 
\draw [color={rgb, 255:red, 0; green, 0; blue, 0 }  ,draw opacity=1 ]   (290,100) .. controls (317,100) and (323.89,95) .. (265,95) .. controls (205,95) and (217.88,100.2) .. (240,100) ;
%Curve Lines [id:da016058380764456315] 
\draw [color={rgb, 255:red, 0; green, 130; blue, 200 }  ,draw opacity=1 ]   (360,100) .. controls (330,100.2) and (322,90.24) .. (263.5,90.24) .. controls (204.5,90.24) and (200,100.2) .. (170,100) ;
%Curve Lines [id:da6281316098929215] 
    \draw [color={rgb, 255:red, 0; green, 0; blue, 0 }  ,draw opacity=1 ]   (290,120) .. controls (317.16,120.4) and (323.47,115) .. (265,115) .. controls (205,115) and (217.88,120) .. (240,120) ;
%Curve Lines [id:da9445339837656497] 
\draw [color={rgb, 255:red, 0; green, 130; blue, 200 }  ,draw opacity=1 ]   (360,120) .. controls (330,120) and (322,110.24) .. (263.5,110.24) .. controls (204.5,110.24) and (200,120) .. (170,120) ;
%Curve Lines [id:da012985407286528527] 
\draw [color={rgb, 255:red, 0; green, 0; blue, 0 }  ,draw opacity=1 ]   (290,210) .. controls (317,210.4) and (323.89,205.29) .. (265,205.29) .. controls (205,205.29) and (217.88,210.2) .. (240,210) ;
%Curve Lines [id:da8522348768695571] 
\draw [color={rgb, 255:red, 0; green, 130; blue, 200 }  ,draw opacity=1 ]   (360,210) .. controls (330,210.2) and (322,200.24) .. (263.5,200.24) .. controls (204.5,200.24) and (200,210.2) .. (170,210) ;
%Curve Lines [id:da23475407315647434] 
\draw [color={rgb, 255:red, 0; green, 0; blue, 0 }  ,draw opacity=1 ]   (290,230) .. controls (317.16,230.4) and (323.47,224.8) .. (265,224.8) .. controls (205,224.8) and (217.88,230.2) .. (240,230) ;
%Curve Lines [id:da5781821544352015] 
\draw [color={rgb, 255:red, 0; green, 130; blue, 200 }  ,draw opacity=1 ]   (360,230) .. controls (330,230.2) and (322,220.24) .. (263.5,220.24) .. controls (204.5,220.24) and (200,230.2) .. (170,230) ;
%Straight Lines [id:da0253696875260343] 
\draw    (170,140) -- (240,140) ;
%Straight Lines [id:da8849856590222602] 
\draw    (170,80) -- (240,80) ;
\end{mytikz}+
\begin{mytikz}
%Straight Lines [id:da4553034597499923] 
\draw    (170,250) -- (240,250) ;
%Straight Lines [id:da879995939230652] 
\draw    (170,190) -- (240,190) ;
%Straight Lines [id:da8334357940613453] 
\draw    (290,250) -- (360,250) ;
%Straight Lines [id:da875812750293314] 
\draw    (290,190) -- (360,190) ;
%Straight Lines [id:da16894149119039037] 
\draw    (290,140) -- (360,140) ;
%Straight Lines [id:da8711393597665376] 
\draw    (290,80) -- (360,80) ;
%Shape: Rectangle [id:dp28611360639719474] 
\draw  [fill={rgb, 255:red, 230; green, 230; blue, 230 }  ,fill opacity=1 ] (240,50) -- (290,50) -- (290,280) -- (240,280) -- cycle ;
%Curve Lines [id:da4915785647998141] 
\draw [color={rgb, 255:red, 0; green, 0; blue, 0 }  ,draw opacity=1 ]   (290,210) .. controls (317.2,210.4) and (307.2,190.04) .. (264.8,155.64) .. controls (222.4,121.24) and (217.88,100.2) .. (240,100) ;
%Curve Lines [id:da016058380764456315] 
\draw [color={rgb, 255:red, 0; green, 130; blue, 200 }  ,draw opacity=1 ]   (360,210) .. controls (330,210.2) and (330,210.44) .. (264.8,155.64) .. controls (199.6,100.84) and (200,100.2) .. (170,100) ;
%Curve Lines [id:da6281316098929215] 
\draw [color={rgb, 255:red, 0; green, 0; blue, 0 }  ,draw opacity=1 ]   (290,230) .. controls (317.2,230.44) and (308,210.84) .. (264.8,175.64) .. controls (221.6,140.44) and (217.88,120) .. (240,120) ;
%Curve Lines [id:da9445339837656497] 
\draw [color={rgb, 255:red, 0; green, 130; blue, 200 }  ,draw opacity=1 ]   (360,230.44) .. controls (330,230.64) and (330.8,230.84) .. (264.8,175.64) .. controls (198.8,120) and (200,120) .. (170,120) ;
%Curve Lines [id:da012985407286528527] 
\draw [color={rgb, 255:red, 0; green, 0; blue, 0 }  ,draw opacity=1 ]   (290.4,99.24) .. controls (317.6,99.64) and (305.6,121.24) .. (264.8,155.64) .. controls (224,190.04) and (217.88,210.2) .. (240,210) ;
%Curve Lines [id:da8522348768695571] 
\draw [color={rgb, 255:red, 0; green, 130; blue, 200 }  ,draw opacity=1 ]   (360,100.04) .. controls (330,100.24) and (330,100.04) .. (264.8,155.64) .. controls (199.6,211.24) and (200,210.2) .. (170,210) ;
%Curve Lines [id:da23475407315647434] 
\draw [color={rgb, 255:red, 0; green, 0; blue, 0 }  ,draw opacity=1 ]   (290,120) .. controls (317.2,120.04) and (305.2,141.64) .. (264.8,175.64) .. controls (224.4,209.64) and (217.88,230.2) .. (240,230) ;
%Curve Lines [id:da5781821544352015] 
\draw [color={rgb, 255:red, 0; green, 130; blue, 200 }  ,draw opacity=1 ]   (360,120) .. controls (330,120) and (330,120) .. (264.8,175.64) .. controls (199.6,230.84) and (200,230.2) .. (170,230) ;
%Straight Lines [id:da0253696875260343] 
\draw    (170,140) -- (240,140) ;
%Straight Lines [id:da8849856590222602] 
\draw    (170,80) -- (240,80) ;
\end{mytikz}\right)\\
&-\frac{1}{2^{|s|}(2^{|2s|}-1)}\left(
\begin{mytikz}
%Straight Lines [id:da4553034597499923] 
\draw    (170,250) -- (240,250) ;
%Straight Lines [id:da879995939230652] 
\draw    (170,190) -- (240,190) ;
%Straight Lines [id:da8334357940613453] 
\draw    (290,250) -- (360,250) ;
%Straight Lines [id:da875812750293314] 
\draw    (290,190) -- (360,190) ;
%Straight Lines [id:da16894149119039037] 
\draw    (290,140) -- (360,140) ;
%Straight Lines [id:da8711393597665376] 
\draw    (290,80) -- (360,80) ;
%Shape: Rectangle [id:dp28611360639719474] 
\draw  [fill={rgb, 255:red, 230; green, 230; blue, 230 }  ,fill opacity=1 ] (240,50) -- (290,50) -- (290,280) -- (240,280) -- cycle ;
%Curve Lines [id:da016058380764456315] 
\draw [color={rgb, 255:red, 0; green, 130; blue, 200 }  ,draw opacity=1 ]   (360,210) .. controls (330,210.2) and (330,210.44) .. (264.8,155.64) .. controls (199.6,100.84) and (200,100.2) .. (170,100) ;
%Curve Lines [id:da9445339837656497] 
\draw [color={rgb, 255:red, 0; green, 130; blue, 200 }  ,draw opacity=1 ]   (360,230.44) .. controls (330,230.64) and (330.8,230.84) .. (264.8,175.64) .. controls (198.8,120) and (200,120) .. (170,120) ;
%Curve Lines [id:da8522348768695571] 
\draw [color={rgb, 255:red, 0; green, 130; blue, 200 }  ,draw opacity=1 ]   (360,100.04) .. controls (330,100.24) and (330,100.04) .. (264.8,155.64) .. controls (199.6,211.24) and (200,210.2) .. (170,210) ;
%Curve Lines [id:da5781821544352015] 
\draw [color={rgb, 255:red, 0; green, 130; blue, 200 }  ,draw opacity=1 ]   (360,120) .. controls (330,120) and (330,120) .. (264.8,175.64) .. controls (199.6,230.84) and (200,230.2) .. (170,230) ;
%Straight Lines [id:da0253696875260343] 
\draw    (170,140) -- (240,140) ;
%Straight Lines [id:da8849856590222602] 
\draw    (170,80) -- (240,80) ;
%Curve Lines [id:da6185329426158466] 
\draw [color={rgb, 255:red, 0; green, 0; blue, 0 }  ,draw opacity=1 ]   (290,100) .. controls (312.24,100.12) and (323.89,95.29) .. (265.09,95.29) .. controls (206.29,95.29) and (217.88,100.2) .. (240,100) ;
%Curve Lines [id:da7550349367152329] 
\draw [color={rgb, 255:red, 0; green, 0; blue, 0 }  ,draw opacity=1 ]   (290,120) .. controls (310.24,120.12) and (323.89,115.29) .. (265.09,115.29) .. controls (206.29,115.29) and (217.88,120) .. (240,120) ;
%Curve Lines [id:da2677348601115841] 
\draw [color={rgb, 255:red, 0; green, 0; blue, 0 }  ,draw opacity=1 ]   (290,210) .. controls (311.84,210.12) and (323.89,205.29) .. (265.09,205.29) .. controls (206.29,205.29) and (217.88,210.2) .. (240,210) ;
%Curve Lines [id:da6590137919294887] 
\draw [color={rgb, 255:red, 0; green, 0; blue, 0 }  ,draw opacity=1 ]   (290,230) .. controls (312.24,230.12) and (323.89,225.29) .. (265.09,225.29) .. controls (206.29,225.29) and (217.88,230.2) .. (240,230) ;
\end{mytikz}
+
\begin{mytikz}
%Straight Lines [id:da4553034597499923] 
\draw    (170,250) -- (240,250) ;
%Straight Lines [id:da879995939230652] 
\draw    (170,190) -- (240,190) ;
%Straight Lines [id:da8334357940613453] 
\draw    (290,250) -- (360,250) ;
%Straight Lines [id:da875812750293314] 
\draw    (290,190) -- (360,190) ;
%Straight Lines [id:da16894149119039037] 
\draw    (290,140) -- (360,140) ;
%Straight Lines [id:da8711393597665376] 
\draw    (290,80) -- (360,80) ;
%Shape: Rectangle [id:dp28611360639719474] 
\draw  [fill={rgb, 255:red, 230; green, 230; blue, 230 }  ,fill opacity=1 ] (240,50) -- (290,50) -- (290,280) -- (240,280) -- cycle ;
%Curve Lines [id:da4915785647998141] 
\draw [color={rgb, 255:red, 0; green, 0; blue, 0 }  ,draw opacity=1 ]   (290,210) .. controls (317.2,210.4) and (307.2,190.04) .. (264.8,155.64) .. controls (222.4,121.24) and (217.88,100.2) .. (240,100) ;
%Curve Lines [id:da6281316098929215] 
\draw [color={rgb, 255:red, 0; green, 0; blue, 0 }  ,draw opacity=1 ]   (290,230) .. controls (317.2,230.44) and (308,210.84) .. (264.8,175.64) .. controls (221.6,140.44) and (217.88,120) .. (240,120) ;
%Curve Lines [id:da012985407286528527] 
\draw [color={rgb, 255:red, 0; green, 0; blue, 0 }  ,draw opacity=1 ]   (290,100) .. controls (317.2,100.4) and (305.6,121.24) .. (264.8,155.64) .. controls (224,190.04) and (217.88,210.2) .. (240,210) ;
%Curve Lines [id:da23475407315647434] 
\draw [color={rgb, 255:red, 0; green, 0; blue, 0 }  ,draw opacity=1 ]   (290,120) .. controls (317.2,120.04) and (305.2,141.64) .. (264.8,175.64) .. controls (224.4,209.64) and (217.88,230.2) .. (240,230) ;
%Straight Lines [id:da0253696875260343] 
\draw    (170,140) -- (240,140) ;
%Straight Lines [id:da8849856590222602] 
\draw    (170,80) -- (240,80) ;
%Curve Lines [id:da9702273737912543] 
\draw [color={rgb, 255:red, 0; green, 130; blue, 200 }  ,draw opacity=1 ]   (360,99.76) .. controls (330,100) and (322,90) .. (263.5,90) .. controls (204.5,90) and (200,100) .. (170,99.76) ;
%Curve Lines [id:da7530818309862197] 
\draw [color={rgb, 255:red, 0; green, 130; blue, 200 }  ,draw opacity=1 ]   (360,120) .. controls (330,120) and (322,110) .. (263.5,110) .. controls (204.5,110) and (200,120) .. (170,120) ;
%Curve Lines [id:da9803535938181618] 
\draw [color={rgb, 255:red, 0; green, 130; blue, 200 }  ,draw opacity=1 ]   (360,210) .. controls (330,210) and (322,200) .. (263.5,200) .. controls (204.5,200) and (200,210) .. (170,210) ;
%Curve Lines [id:da9389611907345525] 
\draw [color={rgb, 255:red, 0; green, 130; blue, 200 }  ,draw opacity=1 ]   (360,230) .. controls (330,230) and (322,220) .. (263.5,220) .. controls (204.5,220) and (200,230) .. (170,230) ;
\end{mytikz}
\right)\quad,
\end{aligned}
\end{equation}
where again we take $N=4$ and $s=\{q_2,q_3\}$ for instance. Note that these formulas can be directly generalized to qudit systems with local Hilbert space dimension $d$ by replacing the base $2$ with $d$. It is worth mentioning that Eqs.~\eqref{eq:twirling_1st_order} and \eqref{eq:twirling_2st_order} only hold for unitary $2$-designs while twirling channels in Eq.~\eqref{eq:twirling_def} can be defined for arbitrary ensembles. As we will demonstrate below, utilizing Eqs.~\eqref{eq:twirling_1st_order} and \eqref{eq:twirling_2st_order} repeatedly, one can estimate the average and variance of a given observable with respect to a random quantum circuit composed of independent unitary $2$-designs.

\section{Theorems and Proofs}
In this section, we present detailed statements and proofs of the theorems mentioned in the main text. We will first introduce some basic facts and lemmas, based on which we shall further develop our theorems and proofs.

A typical setting of variational quantum eigensolver (VQE) regarding an $N$-qubit system consists of an initial state $\rho_0=\ketbrasame{\bm{0}}$ with $\ket{\bm{0}}=\ket{0}^{\otimes N}$, a given Hamiltonian $H$ and a parametrized quantum circuit (PQC, also known as ansatz) $\mathbf{U}(\bm{\theta})$ with $M$ trainable parameters $\bm{\theta}=\{\theta_\mu\}_{\mu=1}^M$ of the form
\begin{equation}
    \mathbf{U}(\bm{\theta}) = \prod_{\mu=1}^{M} U_\mu(\theta_\mu) = U_M(\theta_M)\cdots U_2(\theta_2)U_1(\theta_1),
\end{equation}
where $U_\mu(\theta_\mu)=e^{-i\Omega_\mu\theta_\mu}$ and $\{\Omega_\mu\}_{\mu=1}^M$ is a set of Hermitian generators such as Pauli strings. Note that the quantum gates without trainable parameters such as the CNOT gate can be obtained simply by fixing the corresponding parameters in certain rotation gates. The energy expectation with respect to the output state $C(\bm{\theta})=\braoprket{\bm{0}}{\mathbf{U}^\dagger H \mathbf{U}}{\bm{0}}$ is usually taken as the cost function. To prepare the ground state of the given Hamiltonian, a standard workflow of VQE involves running the PQC on a quantum device, measuring the cost function, and updating the trainable parameters iteratively to minimize the cost function using classical optimizers such as gradient descent algorithms. In particular, at the beginning of this workflow, the trainable parameters are often randomly initialized so that the whole parameter space is explored probabilistically to search for the optimum. Equipped with the probability measure originated from random initialization, the PQC $\mathbf{U}(\bm{\theta})$ becomes a random quantum circuit, where some statistical quantities, such as the average and variance of the cost derivative, can be evaluated analytically for certain probability measures. Considering both computability and practicality, a common assumption on random quantum circuits is that the circuit is composed of blocks forming independent local $2$-designs. Here the term ``block'' refers to a grouped continuous series of gates close to each other in the PQC, which also can be seen as the elementary unit when we construct a PQC, as shown in the following example
\begin{equation}\label{eq:block_template}
\begin{mytikz2}
%Rounded Rect [id:dp08048917861594362] 
\draw  [fill={rgb, 255:red, 204; green, 230; blue, 255 }  ,fill opacity=0.35 ][dash pattern={on 3pt off 1.5pt}][line width=0.75]  (110,107.52) .. controls (110,97.84) and (117.84,90) .. (127.52,90) -- (322.48,90) .. controls (332.16,90) and (340,97.84) .. (340,107.52) -- (340,192.48) .. controls (340,202.16) and (332.16,210) .. (322.48,210) -- (127.52,210) .. controls (117.84,210) and (110,202.16) .. (110,192.48) -- cycle ;
%Straight Lines [id:da959319512549091] 
\draw [line width=0.75]    (100,180) -- (350,180) ;
%Straight Lines [id:da5589761164366458] 
\draw [line width=0.75]    (100,120) -- (350,120) ;
%Rounded Rect [id:dp37869621105417095] 
\draw  [fill={rgb, 255:red, 204; green, 230; blue, 255 }  ,fill opacity=1 ][line width=0.75]  (140,100.4) .. controls (140,100.4) and (140,100.4) .. (140,100.4) -- (180,100.4) .. controls (180,100.4) and (180,100.4) .. (180,100.4) -- (180,140.4) .. controls (180,140.4) and (180,140.4) .. (180,140.4) -- (140,140.4) .. controls (140,140.4) and (140,140.4) .. (140,140.4) -- cycle ;
%Rounded Rect [id:dp41813794368442037] 
\draw  [fill={rgb, 255:red, 204; green, 230; blue, 255 }  ,fill opacity=1 ][line width=0.75]  (140,160.4) .. controls (140,160.4) and (140,160.4) .. (140,160.4) -- (180,160.4) .. controls (180,160.4) and (180,160.4) .. (180,160.4) -- (180,200.4) .. controls (180,200.4) and (180,200.4) .. (180,200.4) -- (140,200.4) .. controls (140,200.4) and (140,200.4) .. (140,200.4) -- cycle ;
%Rounded Rect [id:dp13324445161386533] 
\draw  [fill={rgb, 255:red, 204; green, 230; blue, 255 }  ,fill opacity=1 ][line width=0.75]  (200,100) .. controls (200,100) and (200,100) .. (200,100) -- (240,100) .. controls (240,100) and (240,100) .. (240,100) -- (240,200) .. controls (240,200) and (240,200) .. (240,200) -- (200,200) .. controls (200,200) and (200,200) .. (200,200) -- cycle ;
%Rounded Rect [id:dp03188795811469913] 
\draw  [fill={rgb, 255:red, 204; green, 230; blue, 255 }  ,fill opacity=1 ][line width=0.75]  (260,100) .. controls (260,100) and (260,100) .. (260,100) -- (300,100) .. controls (300,100) and (300,100) .. (300,100) -- (300,140) .. controls (300,140) and (300,140) .. (300,140) -- (260,140) .. controls (260,140) and (260,140) .. (260,140) -- cycle ;
%Rounded Rect [id:dp17734087015399091] 
\draw  [fill={rgb, 255:red, 204; green, 230; blue, 255 }  ,fill opacity=1 ][line width=0.75]  (260,160) .. controls (260,160) and (260,160) .. (260,160) -- (300,160) .. controls (300,160) and (300,160) .. (300,160) -- (300,200) .. controls (300,200) and (300,200) .. (300,200) -- (260,200) .. controls (260,200) and (260,200) .. (260,200) -- cycle ;

% Text Node
\draw (145,110) node [anchor=north west][inner sep=0.75pt]  [font=\normalsize] [align=left] {$U_{1}$};
% Text Node
\draw (145,170) node [anchor=north west][inner sep=0.75pt]  [font=\normalsize] [align=left] {$U_{2}$};
% Text Node
\draw (205,140) node [anchor=north west][inner sep=0.75pt]  [font=\normalsize] [align=left] {$U_{3}$};
% Text Node
\draw (265,110) node [anchor=north west][inner sep=0.75pt]  [font=\normalsize] [align=left] {$U_{4}$};
% Text Node
\draw (265,170) node [anchor=north west][inner sep=0.75pt]  [font=\normalsize] [align=left] {$U_{5}$};
% Text Node
\draw (305,141) node [anchor=north west][inner sep=0.75pt]  [font=\normalsize] [align=left] {$B_{1}$};
\end{mytikz2}\quad.
\end{equation}
In this example, $B_1$ is a block composed of the gates $U_1, U_2, \cdots, U_5$. As thus, the PQC can be rewritten as 
\begin{equation}\label{eq:u_bbb}
    \mathbf{U}=\prod_{k=1}^{M'} B_k = B_{M'}\cdots B_{2} B_{1},
\end{equation}
where the dependence on $\bm{\theta}$ is omitted for simplicity. Namely, the $M$ gates are grouped into $M'$ blocks. For clarity, we will use $k(\mu)$ to denote the index of the block $B_{k(\mu)}$ that contains the gate $U_\mu$ when needed. Since each block $B_k$ consists of multiple gates, its randomness is larger than a single parametrized gate $U_\mu=e^{-i\Omega_\mu\theta_\mu}$ with a randomly assigned value of the parameter $\theta_\mu$. When the number of gates in one block is sufficiently large, the block together with the induced probability measure from the initialization can be regarded as a local $2$-design or approximate $2$-design on $s(B_k)$. Remember that $s(B_k)$ denotes the union set of qubits on which the gates in $B_k$ are applied. A typical template of two-qubit blocks is the Cartan decomposition of $\mathcal{SU}(4)$
\begin{equation}\label{eq:cartan}
\begin{mytikz2}
%Rounded Rect [id:dp08048917861594362] 
\draw  [fill={rgb, 255:red, 204; green, 230; blue, 255 }  ,fill opacity=0.35 ][dash pattern={on 3pt off 1.5pt}][line width=0.75]  (114.41,107.52) .. controls (114.41,97.84) and (122.25,90) .. (131.93,90) -- (428.27,90) .. controls (437.95,90) and (445.79,97.84) .. (445.79,107.52) -- (445.79,192.48) .. controls (445.79,202.16) and (437.95,210) .. (428.27,210) -- (131.93,210) .. controls (122.25,210) and (114.41,202.16) .. (114.41,192.48) -- cycle ;
%Straight Lines [id:da959319512549091] 
\draw [line width=0.75]    (100,180) -- (460.2,180) ;
%Straight Lines [id:da5589761164366458] 
\draw [line width=0.75]    (100,120) -- (460.2,120) ;
%Rounded Rect [id:dp37869621105417095] 
\draw  [fill={rgb, 255:red, 204; green, 230; blue, 255 }  ,fill opacity=1 ][line width=0.75]  (140,100.4) .. controls (140,100.4) and (140,100.4) .. (140,100.4) -- (180,100.4) .. controls (180,100.4) and (180,100.4) .. (180,100.4) -- (180,140.4) .. controls (180,140.4) and (180,140.4) .. (180,140.4) -- (140,140.4) .. controls (140,140.4) and (140,140.4) .. (140,140.4) -- cycle ;
%Rounded Rect [id:dp41813794368442037] 
\draw  [fill={rgb, 255:red, 204; green, 230; blue, 255 }  ,fill opacity=1 ][line width=0.75]  (140,160.4) .. controls (140,160.4) and (140,160.4) .. (140,160.4) -- (180,160.4) .. controls (180,160.4) and (180,160.4) .. (180,160.4) -- (180,200.4) .. controls (180,200.4) and (180,200.4) .. (180,200.4) -- (140,200.4) .. controls (140,200.4) and (140,200.4) .. (140,200.4) -- cycle ;
%Rounded Rect [id:dp13324445161386533] 
\draw  [fill={rgb, 255:red, 204; green, 230; blue, 255 }  ,fill opacity=1 ][line width=0.75]  (200,100) .. controls (200,100) and (200,100) .. (200,100) -- (240,100) .. controls (240,100) and (240,100) .. (240,100) -- (240,200) .. controls (240,200) and (240,200) .. (240,200) -- (200,200) .. controls (200,200) and (200,200) .. (200,200) -- cycle ;
%Rounded Rect [id:dp03188795811469913] 
\draw  [fill={rgb, 255:red, 204; green, 230; blue, 255 }  ,fill opacity=1 ][line width=0.75]  (380,100) .. controls (380,100) and (380,100) .. (380,100) -- (420,100) .. controls (420,100) and (420,100) .. (420,100) -- (420,140) .. controls (420,140) and (420,140) .. (420,140) -- (380,140) .. controls (380,140) and (380,140) .. (380,140) -- cycle ;
%Rounded Rect [id:dp17734087015399091] 
\draw  [fill={rgb, 255:red, 204; green, 230; blue, 255 }  ,fill opacity=1 ][line width=0.75]  (380,160) .. controls (380,160) and (380,160) .. (380,160) -- (420,160) .. controls (420,160) and (420,160) .. (420,160) -- (420,200) .. controls (420,200) and (420,200) .. (420,200) -- (380,200) .. controls (380,200) and (380,200) .. (380,200) -- cycle ;
%Rounded Rect [id:dp26146520500273596] 
\draw  [fill={rgb, 255:red, 204; green, 230; blue, 255 }  ,fill opacity=1 ][line width=0.75]  (260,100) .. controls (260,100) and (260,100) .. (260,100) -- (300,100) .. controls (300,100) and (300,100) .. (300,100) -- (300,200) .. controls (300,200) and (300,200) .. (300,200) -- (260,200) .. controls (260,200) and (260,200) .. (260,200) -- cycle ;
%Rounded Rect [id:dp9258529086066054] 
\draw  [fill={rgb, 255:red, 204; green, 230; blue, 255 }  ,fill opacity=1 ][line width=0.75]  (320,100) .. controls (320,100) and (320,100) .. (320,100) -- (360,100) .. controls (360,100) and (360,100) .. (360,100) -- (360,200) .. controls (360,200) and (360,200) .. (360,200) -- (320,200) .. controls (320,200) and (320,200) .. (320,200) -- cycle ;

% Text Node
\draw (145,110) node [anchor=north west][inner sep=0.75pt]  [font=\normalsize] [align=left] {$R_{3}$};
% Text Node
\draw (145,170) node [anchor=north west][inner sep=0.75pt]  [font=\normalsize] [align=left] {$R_{3}$};
% Text Node
\draw (197,140) node [anchor=north west][inner sep=0.75pt]  [font=\normalsize] [align=left] {$R_{xx}$};
% Text Node
\draw (385,110) node [anchor=north west][inner sep=0.75pt]  [font=\normalsize] [align=left] {$R_{3}$};
% Text Node
\draw (385,170) node [anchor=north west][inner sep=0.75pt]  [font=\normalsize] [align=left] {$R_{3}$};
% Text Node
\draw (257,140) node [anchor=north west][inner sep=0.75pt]  [font=\normalsize] [align=left] {$R_{yy}$};
% Text Node
\draw (317,140) node [anchor=north west][inner sep=0.75pt]  [font=\normalsize] [align=left] {$R_{zz}$};
\end{mytikz2}\quad,
\end{equation}
where $R_{xx}$, $R_{yy}$ and $R_{zz}$ are the two-qubit rotation gates with generators $X\otimes X$, $Y\otimes Y$ and $Z\otimes Z$, respectively. $R_3$ is a universal single-qubit gate, e.g., $R_zR_yR_z$ where $R_{y}$ and $R_{z}$ represent the single-qubit rotation gates with generators $Y$ and $Z$. This block template from the Cartan decomposition can represent any two-qubit gate and approximate a $2$-design when the parameters are randomized uniformly over $[0,2\pi)$~\cite{Uvarov2021a}. 

In the following, we suppose that each block $B_k$ in the circuit $\mathbf{U}$ forms a local $2$-design independently. This will induce an ensemble of the circuit $\mathbf{U}$ according to Eq.~\eqref{eq:u_bbb}, which we denote as $\mathbb{U}$. We will estimate the average and variance of the derivatives of the cost function with respect to this circuit ensemble $\mathbb{U}$. Note that we do not assume any configuration among these blocks, such as the block sizes and the relative locations. One of the most commonly used block configurations is the alternating layered ansatz, also known as the brickwall ansatz or checkerboard ansatz, which features a periodic layered structure and alternates the block pattern within one period. A typical 1D example is shown as follows
\begin{equation}
\begin{mytikz2}
%Straight Lines [id:da959319512549091] 
\draw [line width=0.75]    (100,141.93) -- (300.5,141.93) ;
%Straight Lines [id:da9966950528056495] 
\draw [line width=0.75]    (100,175.47) -- (300.5,175.46) ;
%Straight Lines [id:da5589761164366458] 
\draw [line width=0.75]    (100,108.39) -- (300.5,108.39) ;
%Rounded Rect [id:dp37869621105417095] 
\draw  [fill={rgb, 255:red, 204; green, 230; blue, 255 }  ,fill opacity=1 ][line width=0.75]  (116.72,105.02) .. controls (116.72,102.25) and (118.96,100) .. (121.73,100) -- (128.42,100) .. controls (131.19,100) and (133.44,102.25) .. (133.44,105.02) -- (133.44,145.29) .. controls (133.44,148.06) and (131.19,150.31) .. (128.42,150.31) -- (121.73,150.31) .. controls (118.96,150.31) and (116.72,148.06) .. (116.72,145.29) -- cycle ;
%Straight Lines [id:da3121417313882482] 
\draw [line width=0.75]    (100,242.54) -- (300.5,242.54) ;
%Straight Lines [id:da7403814283080539] 
\draw [line width=0.75]    (100,276.08) -- (300.5,276.07) ;
%Straight Lines [id:da056365607329372125] 
\draw [line width=0.75]    (100,209) -- (300.5,209) ;
%Rounded Rect [id:dp5114763011590588] 
\draw  [fill={rgb, 255:red, 204; green, 230; blue, 255 }  ,fill opacity=1 ][line width=0.75]  (141.93,138.56) .. controls (141.93,135.79) and (144.18,133.54) .. (146.95,133.54) -- (153.64,133.54) .. controls (156.41,133.54) and (158.65,135.79) .. (158.65,138.56) -- (158.65,178.83) .. controls (158.65,181.6) and (156.41,183.85) .. (153.64,183.85) -- (146.95,183.85) .. controls (144.18,183.85) and (141.93,181.6) .. (141.93,178.83) -- cycle ;
%Rounded Rect [id:dp8642104052146871] 
\draw  [fill={rgb, 255:red, 204; green, 230; blue, 255 }  ,fill opacity=1 ][line width=0.75]  (167.09,172.09) .. controls (167.09,169.32) and (169.34,167.08) .. (172.11,167.08) -- (178.8,167.08) .. controls (181.57,167.08) and (183.81,169.32) .. (183.81,172.09) -- (183.81,212.37) .. controls (183.81,215.14) and (181.57,217.38) .. (178.8,217.38) -- (172.11,217.38) .. controls (169.34,217.38) and (167.09,215.14) .. (167.09,212.37) -- cycle ;
%Rounded Rect [id:dp5970167848768184] 
\draw  [fill={rgb, 255:red, 204; green, 230; blue, 255 }  ,fill opacity=1 ][line width=0.75]  (192.25,205.63) .. controls (192.25,202.86) and (194.5,200.61) .. (197.27,200.61) -- (203.96,200.61) .. controls (206.73,200.61) and (208.97,202.86) .. (208.97,205.63) -- (208.97,245.9) .. controls (208.97,248.67) and (206.73,250.92) .. (203.96,250.92) -- (197.27,250.92) .. controls (194.5,250.92) and (192.25,248.67) .. (192.25,245.9) -- cycle ;
%Rounded Rect [id:dp890150743555258] 
\draw  [fill={rgb, 255:red, 204; green, 230; blue, 255 }  ,fill opacity=1 ][line width=0.75]  (217.41,239.17) .. controls (217.41,236.4) and (219.66,234.15) .. (222.43,234.15) -- (229.12,234.15) .. controls (231.89,234.15) and (234.13,236.4) .. (234.13,239.17) -- (234.13,279.44) .. controls (234.13,282.21) and (231.89,284.46) .. (229.12,284.46) -- (222.43,284.46) .. controls (219.66,284.46) and (217.41,282.21) .. (217.41,279.44) -- cycle ;
%Rounded Rect [id:dp36777104673273553] 
\draw  [fill={rgb, 255:red, 204; green, 230; blue, 255 }  ,fill opacity=1 ][line width=0.75]  (167.09,105.02) .. controls (167.09,102.25) and (169.34,100) .. (172.11,100) -- (178.8,100) .. controls (181.57,100) and (183.81,102.25) .. (183.81,105.02) -- (183.81,145.29) .. controls (183.81,148.06) and (181.57,150.31) .. (178.8,150.31) -- (172.11,150.31) .. controls (169.34,150.31) and (167.09,148.06) .. (167.09,145.29) -- cycle ;
%Rounded Rect [id:dp20295268958260881] 
\draw  [fill={rgb, 255:red, 204; green, 230; blue, 255 }  ,fill opacity=1 ][line width=0.75]  (192.25,138.56) .. controls (192.25,135.79) and (194.5,133.54) .. (197.27,133.54) -- (203.96,133.54) .. controls (206.73,133.54) and (208.97,135.79) .. (208.97,138.56) -- (208.97,178.83) .. controls (208.97,181.6) and (206.73,183.85) .. (203.96,183.85) -- (197.27,183.85) .. controls (194.5,183.85) and (192.25,181.6) .. (192.25,178.83) -- cycle ;
%Rounded Rect [id:dp4830860004171569] 
\draw  [fill={rgb, 255:red, 204; green, 230; blue, 255 }  ,fill opacity=1 ][line width=0.75]  (217.41,172.09) .. controls (217.41,169.32) and (219.66,167.08) .. (222.43,167.08) -- (229.12,167.08) .. controls (231.89,167.08) and (234.13,169.32) .. (234.13,172.09) -- (234.13,212.37) .. controls (234.13,215.14) and (231.89,217.38) .. (229.12,217.38) -- (222.43,217.38) .. controls (219.66,217.38) and (217.41,215.14) .. (217.41,212.37) -- cycle ;
%Rounded Rect [id:dp17288415325530893] 
\draw  [fill={rgb, 255:red, 204; green, 230; blue, 255 }  ,fill opacity=1 ][line width=0.75]  (242.57,205.63) .. controls (242.57,202.86) and (244.82,200.61) .. (247.59,200.61) -- (254.28,200.61) .. controls (257.05,200.61) and (259.29,202.86) .. (259.29,205.63) -- (259.29,245.9) .. controls (259.29,248.67) and (257.05,250.92) .. (254.28,250.92) -- (247.59,250.92) .. controls (244.82,250.92) and (242.57,248.67) .. (242.57,245.9) -- cycle ;
%Rounded Rect [id:dp7438051595882824] 
\draw  [fill={rgb, 255:red, 204; green, 230; blue, 255 }  ,fill opacity=1 ][line width=0.75]  (267.73,239.17) .. controls (267.73,236.4) and (269.98,234.15) .. (272.75,234.15) -- (279.44,234.15) .. controls (282.21,234.15) and (284.45,236.4) .. (284.45,239.17) -- (284.45,279.44) .. controls (284.45,282.21) and (282.21,284.46) .. (279.44,284.46) -- (272.75,284.46) .. controls (269.98,284.46) and (267.73,282.21) .. (267.73,279.44) -- cycle ;
%Rounded Rect [id:dp7175794669528333] 
\draw  [fill={rgb, 255:red, 204; green, 230; blue, 255 }  ,fill opacity=1 ][line width=0.75]  (116.77,172.11) .. controls (116.77,169.34) and (119.02,167.1) .. (121.79,167.1) -- (128.48,167.1) .. controls (131.25,167.1) and (133.49,169.34) .. (133.49,172.11) -- (133.49,212.39) .. controls (133.49,215.16) and (131.25,217.4) .. (128.48,217.4) -- (121.79,217.4) .. controls (119.02,217.4) and (116.77,215.16) .. (116.77,212.39) -- cycle ;
%Rounded Rect [id:dp13029356427012107] 
\draw  [fill={rgb, 255:red, 204; green, 230; blue, 255 }  ,fill opacity=1 ][line width=0.75]  (141.93,205.65) .. controls (141.93,202.88) and (144.18,200.63) .. (146.95,200.63) -- (153.64,200.63) .. controls (156.41,200.63) and (158.65,202.88) .. (158.65,205.65) -- (158.65,245.92) .. controls (158.65,248.69) and (156.41,250.94) .. (153.64,250.94) -- (146.95,250.94) .. controls (144.18,250.94) and (141.93,248.69) .. (141.93,245.92) -- cycle ;
%Rounded Rect [id:dp2909294591833491] 
\draw  [fill={rgb, 255:red, 204; green, 230; blue, 255 }  ,fill opacity=1 ][line width=0.75]  (116.77,239.18) .. controls (116.77,236.41) and (119.02,234.16) .. (121.79,234.16) -- (128.48,234.16) .. controls (131.25,234.16) and (133.49,236.41) .. (133.49,239.18) -- (133.49,279.45) .. controls (133.49,282.22) and (131.25,284.47) .. (128.48,284.47) -- (121.79,284.47) .. controls (119.02,284.47) and (116.77,282.22) .. (116.77,279.45) -- cycle ;
%Rounded Rect [id:dp6575792888168039] 
\draw  [fill={rgb, 255:red, 204; green, 230; blue, 255 }  ,fill opacity=1 ][line width=0.75]  (167.09,239.18) .. controls (167.09,236.41) and (169.34,234.16) .. (172.11,234.16) -- (178.8,234.16) .. controls (181.57,234.16) and (183.81,236.41) .. (183.81,239.18) -- (183.81,279.45) .. controls (183.81,282.22) and (181.57,284.47) .. (178.8,284.47) -- (172.11,284.47) .. controls (169.34,284.47) and (167.09,282.22) .. (167.09,279.45) -- cycle ;
%Rounded Rect [id:dp9910291718670561] 
\draw  [fill={rgb, 255:red, 204; green, 230; blue, 255 }  ,fill opacity=1 ][line width=0.75]  (217.41,105.01) .. controls (217.41,102.24) and (219.66,99.99) .. (222.43,99.99) -- (229.12,99.99) .. controls (231.89,99.99) and (234.13,102.24) .. (234.13,105.01) -- (234.13,145.28) .. controls (234.13,148.05) and (231.89,150.3) .. (229.12,150.3) -- (222.43,150.3) .. controls (219.66,150.3) and (217.41,148.05) .. (217.41,145.28) -- cycle ;
%Rounded Rect [id:dp29840669293909716] 
\draw  [fill={rgb, 255:red, 204; green, 230; blue, 255 }  ,fill opacity=1 ][line width=0.75]  (242.57,138.55) .. controls (242.57,135.78) and (244.82,133.53) .. (247.59,133.53) -- (254.28,133.53) .. controls (257.05,133.53) and (259.29,135.78) .. (259.29,138.55) -- (259.29,178.82) .. controls (259.29,181.59) and (257.05,183.84) .. (254.28,183.84) -- (247.59,183.84) .. controls (244.82,183.84) and (242.57,181.59) .. (242.57,178.82) -- cycle ;
%Rounded Rect [id:dp47053501459018743] 
\draw  [fill={rgb, 255:red, 204; green, 230; blue, 255 }  ,fill opacity=1 ][line width=0.75]  (267.73,172.08) .. controls (267.73,169.31) and (269.98,167.07) .. (272.75,167.07) -- (279.44,167.07) .. controls (282.21,167.07) and (284.45,169.31) .. (284.45,172.08) -- (284.45,212.36) .. controls (284.45,215.13) and (282.21,217.37) .. (279.44,217.37) -- (272.75,217.37) .. controls (269.98,217.37) and (267.73,215.13) .. (267.73,212.36) -- cycle ;
%Rounded Rect [id:dp09077192990722183] 
\draw  [fill={rgb, 255:red, 204; green, 230; blue, 255 }  ,fill opacity=1 ][line width=0.75]  (267.73,105.01) .. controls (267.73,102.24) and (269.98,99.99) .. (272.75,99.99) -- (279.44,99.99) .. controls (282.21,99.99) and (284.45,102.24) .. (284.45,105.01) -- (284.45,145.28) .. controls (284.45,148.05) and (282.21,150.3) .. (279.44,150.3) -- (272.75,150.3) .. controls (269.98,150.3) and (267.73,148.05) .. (267.73,145.28) -- cycle ;
\end{mytikz2}\quad.
\end{equation}
It is known that when the depth of this circuit reaches $\Omega(N)$, i.e., linear with the system size or beyond, the whole circuit approximates a global $2$-design and hence the training landscape would exhibit barren plateaus~\cite{McClean2018}, i.e., the variance of the cost derivative vanishing exponentially with the system size. Otherwise, if the circuit depth is only of order $\mathcal{O}(\log N)$ or even finite, i.e., does not scale with the qubit count $N$, the variance of the cost derivative would not vanish exponentially in the case of local cost functions~\cite{Cerezo2021, Uvarov2021a, Pesah2021}. However, these circuits of finite depth have limited expressibility and fail to capture long-range entanglement, which motivates us to explore other circuit structures, e.g., finite local-depth circuits as we will discuss below.

We start with the expression of the cost derivative. The derivative of the energy expectation with respect to the parameter $\theta_\mu$ can be expressed as
\begin{equation}\label{eq:cost_derivative_expression}
\begin{aligned}
    \partial_\mu C = \frac{\partial C}{\partial \theta_\mu} 
    & = \tr\left( U_{(\mu-1)\leftarrow 1} \rho_0 U_{(\mu-1)\leftarrow 1}^\dagger \left[ i\Omega_\mu, U_{M\leftarrow \mu}^\dagger H U_{M\leftarrow \mu} \right] \right)\\
    & = \tr\left( U_{\mu\leftarrow 1} \rho_0 U_{\mu\leftarrow 1}^\dagger \left[ i\Omega_\mu, U_{M\leftarrow (\mu+1)}^\dagger H U_{M\leftarrow (\mu+1)} \right] \right), \\
\end{aligned}
\end{equation}
where we have used the notation $U_{\nu'\leftarrow \nu}=\prod_{\nu''=\nu}^{\nu'} U_{\nu''}=U_{\nu'}\cdots U_{\nu+1}U_{\nu}$ if $\nu'\geq\nu$. Otherwise if $\nu'<\nu$, we just set $U_{\nu'\leftarrow \nu}$ as the identity. For clarity, we will use the term ``derivative'' to refer to a single component of the gradient and use ``gradient'' to refer to the entire multi-component vector throughout the paper. Using twirling channels corresponding to the $2$-design blocks, the average and variance of the cost derivative can be written as
\begin{align}
    \E_\mathbb{U}\left[\partial_\mu C\right] &= \tr\left[\rho_0 \mathcal{T}^{(1)}_{s(B_1)}\circ\mathcal{T}^{(1)}_{s(B_2)}\circ\cdots\circ \mathcal{D}_\mu \circ\cdots\circ \mathcal{T}^{(1)}_{s(B_{M'})}(H) \right], \label{eq:exp_expression_by_twirling}\\
    \var_\mathbb{U}\left[\partial_\mu C\right] &= \tr\left[\rho_0^{\otimes 2} \mathcal{T}^{(2)}_{s(B_1)}\circ\mathcal{T}^{(2)}_{s(B_2)}\circ\cdots\circ \mathcal{D}_\mu^{\otimes 2} \circ\cdots\circ \mathcal{T}^{(2)}_{s(B_{M'})}(H^{\otimes 2}) \right] - \left(\E_\mathbb{U}\left[\partial_\mu C\right]\right)^2,\label{eq:var_expression_by_twirling}
\end{align}
where $\mathcal{D}_\mu(\cdot)=[i\Omega_\mu,(\cdot)]$ is the quantum operation induced by the commutator. For convenience, we say that $\theta_\mu$ is the differential parameter, $U_\mu$ is the differential gate and the block $B_{k(\mu)}$ containing $U_\mu$ is the differential block. It is worth emphasizing that there is a detail on the location of $U_\mu$ within $B_{k(\mu)}$. Denote the decomposition of $B_{k(\mu)}$ regarding the gate $U_\mu$ as $B_{k(\mu)}=R_{k\mu}U_\mu L_{k\mu}$, i.e.,
\begin{equation}\label{eq:LUR}
\begin{mytikz2}
%Rounded Rect [id:dp08048917861594362] 
\draw  [fill={rgb, 255:red, 204; green, 230; blue, 255 }  ,fill opacity=0.35 ][dash pattern={on 3pt off 1.5pt}][line width=0.75]  (110,107.52) .. controls (110,97.84) and (117.84,90) .. (127.52,90) -- (372.48,90) .. controls (382.16,90) and (390,97.84) .. (390,107.52) -- (390,192.48) .. controls (390,202.16) and (382.16,210) .. (372.48,210) -- (127.52,210) .. controls (117.84,210) and (110,202.16) .. (110,192.48) -- cycle ;
%Straight Lines [id:da959319512549091] 
\draw [line width=0.75]    (100,180) -- (400,180) ;
%Straight Lines [id:da5589761164366458] 
\draw [line width=0.75]    (100,120) -- (400,120) ;
%Rounded Rect [id:dp13324445161386533] 
\draw  [fill={rgb, 255:red, 204; green, 230; blue, 255 }  ,fill opacity=1 ][line width=0.75]  (200,100) .. controls (200,100) and (200,100) .. (200,100) -- (240,100) .. controls (240,100) and (240,100) .. (240,100) -- (240,200) .. controls (240,200) and (240,200) .. (240,200) -- (200,200) .. controls (200,200) and (200,200) .. (200,200) -- cycle ;
%Rounded Rect [id:dp6918302271425316] 
\draw  [fill={rgb, 255:red, 204; green, 230; blue, 255 }  ,fill opacity=1 ][line width=0.75]  (140,100) .. controls (140,100) and (140,100) .. (140,100) -- (180,100) .. controls (180,100) and (180,100) .. (180,100) -- (180,200) .. controls (180,200) and (180,200) .. (180,200) -- (140,200) .. controls (140,200) and (140,200) .. (140,200) -- cycle ;
%Rounded Rect [id:dp9544430849481114] 
\draw  [fill={rgb, 255:red, 204; green, 230; blue, 255 }  ,fill opacity=1 ][line width=0.75]  (260,100) .. controls (260,100) and (260,100) .. (260,100) -- (300,100) .. controls (300,100) and (300,100) .. (300,100) -- (300,200) .. controls (300,200) and (300,200) .. (300,200) -- (260,200) .. controls (260,200) and (260,200) .. (260,200) -- cycle ;

% Text Node
\draw (204,140) node [anchor=north west][inner sep=0.75pt]  [font=\small] [align=left] {$U_{\mu }$};
% Text Node
\draw (320,140) node [anchor=north west][inner sep=0.75pt]  [font=\small] [align=left] {$B_{k(\mu)}$};
% Text Node
\draw (138,140) node [anchor=north west][inner sep=0.75pt]  [font=\small] [align=left] {$L_{k\mu }$};
% Text Node
\draw (258,140) node [anchor=north west][inner sep=0.75pt]  [font=\small] [align=left] {$R_{k\mu }$};
\end{mytikz2}\quad.
\end{equation}
Then the neighboring channels of $\mathcal{D}_\mu^{\otimes t}$ in Eqs.~\eqref{eq:exp_expression_by_twirling} and \eqref{eq:var_expression_by_twirling} can be written as 
\begin{equation}
    \cdots \circ\mathcal{T}_{s(L_{k\mu})}^{(t)} \circ \mathcal{D}_\mu^{\otimes t}\circ\mathcal{T}_{s(U_\mu R_{k\mu})}^{(t)} \circ\cdots \quad\text{or}\quad \cdots \circ\mathcal{T}_{s(L_{k\mu}U_\mu)}^{(t)} \circ \mathcal{D}_\mu^{\otimes t}\circ\mathcal{T}_{s(R_{k\mu})}^{(t)} \circ\cdots.
\end{equation}
depending on the two different expressions in Eq.~\eqref{eq:cost_derivative_expression}. Although we assume that $B_{k(\mu)}$ forms a $2$-design above, the sub-blocks $L_{k\mu}$ and $R_{k\mu}$ do not necessarily form $2$-designs, which prevents the use of Eqs.~\eqref{eq:twirling_1st_order} and \eqref{eq:twirling_2st_order}. To deal with this detail, we will first focus on the case where $L_{k\mu}$ and $R_{k\mu}$ indeed form $2$-designs on $s(B_{k(\mu)})$ so that our calculations could be performed clearly, and then discuss other cases on top of those calculations. Remember that $s(B)$ represents the support of the block $B$, i.e., the qubit subset which is acted non-trivially by $B$. We will call $|s(B)|$ as the block size of the block $B$.

To sum up, based on the assumptions that $\mathbb{U}$ is made up of local $2$-designs and the differential gate $U_\mu$ is sandwiched by two local $2$-designs, now we are going to estimate the average and variance of the cost derivative. We first calculate the average of the cost derivative. 

\begin{lemma}\label{lemma:vanishing_average}
If there exist a block or sub-block $B$ within either $U_{\mu\leftarrow 1}$ or $U_{M\leftarrow \mu}$ forming a local $1$-design with the support $s(B)$ covering $s(U_\mu)$, then the average of the cost derivative equals to zero, i.e., $\E_\mathbb{U}\left[\partial_\mu C\right] = 0$.
\end{lemma}

\begin{proof}
According to Eq.~\eqref{eq:twirling_1st_order}, if the $1$-design block $B$ is within $U_{\mu\leftarrow 1}$, the evolved state in Eq.~\eqref{eq:cost_derivative_expression} becomes a maximally mixed state on $s(U_\mu)$ after the integration with respect to $B$ so that the trace of a commutator will give rise to a zero. If $B$ is within $U_{M\leftarrow \mu}$, the backward evolved Hamiltonian in Eq.~\eqref{eq:cost_derivative_expression} becomes an identity after the integration with respect to $B$ so that the commutator with an identity will also give rise to a zero.
\end{proof}

Since every $2$-design is also a $1$-design by definition, the condition in Lemma~\ref{lemma:vanishing_average} can be easily satisfied by our assumptions above. Actually, the vanishing average of the cost derivative can also be obtained by considering the periodicity of the parameter space. For completeness, we point out that $1$-designs have the following property.

\begin{proposition}\label{proposition:global_1-design}
If each block $B_k$ in the circuit $\mathbf{U}$ forms a local $1$-design, then $\mathbf{U}$ forms a global $1$-design.
\end{proposition}
\begin{proof}
Be definition, the twirling channel of $\mathbf{U}$ with respect to the induced ensemble $\mathbb{U}$ is composed of twirling channels of each block $B_k$, e.g.,
\begin{equation}
    \mathcal{T}^{(1)}_{s(\mathbf{U})} = \mathcal{T}^{(1)}_{s(B_{M'})} \circ\cdots\circ \mathcal{T}^{(1)}_{s(B_2)}\circ \mathcal{T}^{(1)}_{s(B_1)},
\end{equation}
where $s(\mathbf{U})=\bigcup_{k} s(B_k)$. This means that the support of $\mathbf{U}$ is the union of the supports of all blocks. If $\mathbf{U}$ acts on all qubits, then $s(\mathbf{U})=N$. We emphasize that the measure in the definition of $\mathcal{T}^{(1)}_{s(\mathbf{U})}$ is determined by $\mathbb{U}$ instead of the Haar measure. According to Eq.~\eqref{eq:twirling_1st_order}, we have
\begin{equation}\label{eq:1-design_trace_composition}
    \mathcal{T}^{(1)}_{s(\mathbf{U})}(\cdot) = \tr_{s(B_{M'})}\left[ \cdots \tr_{s(B_{1})}\left[ \cdot\right]\otimes \frac{I\vert_{ s(B_{1})}}{2^{|s(B_{1})|}} \cdots\right] \otimes\frac{I\vert_{ s(B_{M'})}}{2^{|s(B_{M'})|}} = \tr_{s(\mathbf{U})}\left[ \cdot\right]\otimes \frac{I\vert_{ s(\mathbf{U})}}{2^{|s(\mathbf{U})|}}.
\end{equation}
This is exactly the definition of global $1$-design. Here ``global'' means the ensemble is defined on the support union $s(\mathbf{U})$. $I\vert_s$ is the identity on support $s$. The last equality means that the composition of a series of tracing and replacing with maximally mixed states is equivalent to tracing all the indices first and then replacing with the maximally mixed state on the support union.
\end{proof}

We remark that Proposition~\ref{proposition:global_1-design} holds especially for $1$-designs. That is to say, for example, the similar statement is not true for $2$-designs, i.e., the composition of local $2$-designs is not necessarily a global $2$-design. This is because there are multiple terms in Eq.~\eqref{eq:twirling_2st_order} which hinders a similar proof as in Eq.~\eqref{eq:1-design_trace_composition} for $2$-designs. As the integration progresses, the terms become exponentially numerous with both positive and negative contributions, making them difficult to estimate. Hence, the estimation regarding the composition of local $2$-designs, such as the variance of the cost derivative, is much more complicated than those of $1$-designs. This implies that we need an alternative method instead of brute-force computation. As we will show below, our method estimates the variance of the cost derivative in the Heisenberg picture, i.e., we track the backward evolution of Pauli strings instead of the forward evolution of the quantum state, which shares the same start point with that in Ref.~\cite{Uvarov2021a}. In the following lemma, we will first see how one Pauli string $h$ evolves by a single twirling channel. We use $h\vert_s$ to denote the sub-string of $h$ on the qubit subset $s$, e.g., $(Z\otimes I\otimes X)\vert_{\{q_1,q_2\}}=Z\otimes I$. The support of $h$, i.e., the qubit subset which is acted by $h$ non-trivially, is denoted by $s(h)$, e.g., $s(Z\otimes I\otimes X)=\{q_1,q_3\}$. Here ``non-trivial'' means that the corresponding sub-string is not the identity. We use $\mathcal{P}\vert_s$ to denote the set of all Pauli strings defined on $s$, and $\mathcal{P}'\vert_s$ to denote the set of all non-trivial Pauli strings on $s$, i.e., $\mathcal{P}'\vert_s=\mathcal{P}\vert_s- \{I\vert_s\}$.

\begin{lemma}\label{lemma:pauli_string_copies_one_step_evolution}
Suppose that $h$ is a Pauli string on an $N$-qubit system and $\mathcal{T}^{(2)}_{s}$ is the $2$-degree twirling channel defined on $s$. If the sub-string $h\vert_s$ is non-trivial, then the following equality holds
\begin{equation}\label{eq:twirling_h2}
    \mathcal{T}^{(2)}_{s} (h^{\otimes 2}) = \frac{1}{2^{|2s|}-1} \sum_{\sigma\in\mathcal{P}'\vert_s} \sigma^{\otimes 2} \otimes ( h\vert_{\bar{s}} )^{\otimes 2},
\end{equation}
where $\bar{s}$ is the complement of $s$ among the $N$ qubits. Otherwise if $h\vert_s$ is trivial, then $\mathcal{T}^{(2)}_{s} (h\otimes h)= h\otimes h$.
\end{lemma}
\begin{proof}
If $h\vert_s$ is non-trivial, then we have $\tr_{2s}(h\otimes h)=0$ and $\tr_{2s}[(h\otimes h)S\vert_{2s}]=2^{|s|}(h\vert_{\bar{s}})^{\otimes 2}$ in Eq.~\eqref{eq:twirling_2st_order} because Pauli matrices are traceless and the squares of them equal to the identity. Thus, Eq.~\eqref{eq:twirling_2st_order} gives rise to
\begin{equation}
    \mathcal{T}^{(2)}_s(h^{\otimes 2}) = \frac{2^{|s|} ( h\vert_{\bar{s}} )^{\otimes 2}}{2^{|2s|}-1} \otimes \left(S\vert_{2s} - \frac{I\vert_{2s}}{2^{|s|}}\right).
\end{equation}
Substituting the Pauli decomposition of the swap operator
\begin{equation}
    S\vert_{2s} = \frac{1}{2^{|s|}} \sum_{\sigma\in \mathcal{P}\vert_s} \sigma^{\otimes 2},
\end{equation}
we arrive at Eq.~\eqref{eq:twirling_h2}. If $h\vert_s$ is trivial, Eq.~\eqref{eq:twirling_2st_order} indicates that $\mathcal{T}^{(2)}_s$ just keeps the input unchanged.
\end{proof}

This lemma tells us that, a single step of backward evolution of two copies of a Pauli string by the twirling channel on the qubit subset $s$ will map the sub-string inside $s$ to a uniformly weighted sum of two copies of all possible non-trivial Pauli strings on $s$, while keep the part outside $s$ unchanged. Note that the coefficients are summed to one. We dub this transformation as the twirling channel ``smearing'' the sub-strings on $s$. From this point of view, one can regard a local $2$-design as a ``local scrambler'', which scrambles the information at $s$ uniformly. Besides identical copies of the form $h^{\otimes 2}$, we also need to pay attention to the evolution of cross-terms like $h_1\otimes h_2$ with $h_1\neq h_2$ from $H^{\otimes2}$ in Eq.~\eqref{eq:var_expression_by_twirling}. That is to say, if the Pauli decomposition of the Hamiltonian is $H=\sum_j \lambda_j h_j$ where $\lambda_j\in\mathbb{R}$ and $h_j\in\mathcal{P}'\vert_{s(H)}$ represents a non-trivial Pauli string on the support of $H$, the copied Hamiltonian $H^{\otimes2}$ can be expressed by $H^{\otimes2} = \sum_{ij} \lambda_i \lambda_j h_i\otimes h_j$, which contains the cross-terms $h_i\otimes h_j$ with $h_i\neq h_j$.

\begin{lemma}\label{lemma:cross_term_h1h2}
Suppose that $h_1$ and $h_2$ are two distinct Pauli strings and $s'$ is the qubit subset where $h_1$ and $h_2$ are distinct, i.e., $s'=\{q_i\mid h_1\vert_{q_i}\neq h_2\vert_{q_i} \}$. If $s'\cap s\neq\varnothing$, then $\mathcal{T}^{(2)}_s(h_1\otimes h_2)=0$. Otherwise if $s'\cap s=\varnothing$, i.e., $h_1$ and $h_2$ are identical on $s$, then Lemma~\ref{lemma:pauli_string_copies_one_step_evolution} applies.
\end{lemma}
\begin{proof}
Since $h_1\vert_s$ and $h_2\vert_s$ are distinct, we have $\tr_{2s}(h_1\otimes h_2)=0$ and $\tr_{2s}[(h_1\otimes h_2)S\vert_{2s}]=0$ in Eq.~\eqref{eq:twirling_2st_order} because Pauli matrices are traceless and the product of two distinct Pauli matrices is also traceless. According to Eq.~\eqref{eq:twirling_2st_order}, we have $\mathcal{T}^{(2)}_s(h_1\otimes h_2)=0$.
\end{proof}

That is to say, if two Pauli strings are not identical in the scope of the $2$-degree twirling channel, the output corresponding to the tensor product of these two Pauli strings is zero. Therefore, for the whole circuit $\mathbf{U}$ composed of local $2$-designs, we have the following corollary.

\begin{corollary}\label{corollary:var_no_cross_terms}
Suppose that $H$ is a Hamiltonian with the Pauli decomposition $H=\sum_j \lambda_j h_j$ and $\mathbf{U}$ is a random PQC composed of local $2$-designs. The differential gate $U_\mu$ within the block $B_{k(\mu)}$ is sandwiched by two local $2$-designs on $s(B_{k(\mu)})$. If the support of $\mathbf{U}$ covers the support of $H$, i.e., $s(\mathbf{U}) \supseteq s(H)$, then the following equality holds
\begin{equation}\label{eq:var_no_cross_terms}
    \var_\mathbb{U}[\partial_\mu \avg{H}] = \sum_j \lambda_j^2 \var_\mathbb{U}[\partial_\mu \avg{h_j}],
\end{equation}
where $\avg{\cdot}=\braoprket{\bm{0}}{\mathbf{U}^\dagger(\cdot)\mathbf{U}}{\bm{0}}$ denotes the expectation value with respect to the output state $\mathbf{U}\ket{\bm{0}}$.
\end{corollary}
\begin{proof}
Lemma~\ref{lemma:vanishing_average} ensures $\E_\mathbb{U}\left[\partial_\mu C\right]=0$. Thus, Eq.~\eqref{eq:var_expression_by_twirling} reduces to
\begin{equation}\label{eq:var_expression_no_squared_average}
    \var_\mathbb{U}\left[\partial_\mu C\right] = \tr\left[\rho_0^{\otimes 2} \mathcal{T}^{(2)}_{s(B_1)}\circ\mathcal{T}^{(2)}_{s(B_2)}\circ\cdots\circ \mathcal{D}_\mu^{\otimes 2} \circ\cdots\circ \mathcal{T}^{(2)}_{s(B_{M'})}(H^{\otimes 2}) \right].
\end{equation}
Since the maps above are all linear, the Pauli decomposition of $H$ gives
\begin{equation}
    \var_\mathbb{U}\left[\partial_\mu C\right] = \sum_{ij}\lambda_i\lambda_j\tr\left[\rho_0^{\otimes 2} \mathcal{T}^{(2)}_{s(B_1)}\circ\mathcal{T}^{(2)}_{s(B_2)}\circ\cdots\circ \mathcal{D}_\mu^{\otimes 2} \circ\cdots\circ \mathcal{T}^{(2)}_{s(B_{M'})}(h_i\otimes h_j) \right].
\end{equation}
If two Pauli strings $h_i$ and $h_j$ are distinct on $s(\mathbf{U})$, there must exist a block $B_{k'}$ that first acts on one of the qubits on which the two Pauli strings are distinct, resulting in $\mathcal{T}^{(2)}_{s(B_{k'})}(h_i\otimes h_j)=0$ by Lemma~\ref{lemma:cross_term_h1h2}. Lemma~\ref{lemma:pauli_string_copies_one_step_evolution} ensures that the distinction is preserved by preceding blocks acting elsewhere. Therefore, all the cross-terms are eliminated and hence we arrive at Eq.~\eqref{eq:var_no_cross_terms}. Note that the differential channel $\mathcal{D}_\mu^{\otimes 2}$ does affect the elimination since $U_\mu$ is sandwiched by two local $2$-designs so that the distinction will be acted by the right one before meeting $\mathcal{D}_\mu^{\otimes 2}$.
\end{proof}

In fact, even if the sub-block on the right of the differential gate $U_\mu$ does not form a local $2$-design, the cross-terms are still expected to be very small. This is because $\mathcal{D}_\mu$ can not map two distinct non-trivial Pauli strings to the same one, e.g., $[X,Y]=2iZ$ and $[X,Z]=-2iY$, so the distinction is preserved by $\mathcal{D}_\mu^{\otimes 2}$. The left and right sub-blocks combine together reforming a $2$-design. 

Corollary~\ref{corollary:var_no_cross_terms} tells us that in order to give a lower bound on the variance of the cost derivative, one only needs to bound the variance corresponding to each Pauli string in the Hamiltonian, because the coefficients in Eq.~\eqref{eq:var_no_cross_terms} are all non-negative, i.e., $\lambda_j^2\geq0$. To see the effect of the differential channel more clearly, we introduce the following lemma.

\begin{lemma}\label{lemma:commutator_channel}
Suppose that $h$ and $\Omega$ are two Pauli strings and $\mathcal{D}(\cdot)=[i\Omega,(\cdot)]$. $s$ is a qubit subset that covers $s(\Omega)$, i.e., $s(\Omega)\subseteq s$. If $h\vert_s$ is non-trivial, then the following equality holds
\begin{equation}\label{eq:commutator_channel_tdt}
    \mathcal{T}^{(2)}_s\circ\mathcal{D}^{\otimes 2}\circ\mathcal{T}^{(2)}_s (h^{\otimes2}) = \frac{2^{|2s|+1}}{2^{|2s|}-1} \mathcal{T}^{(2)}_s (h^{\otimes2}).
\end{equation}
Otherwise if $h\vert_s$ is trivial, $\mathcal{T}^{(2)}_s\circ\mathcal{D}^{\otimes 2}\circ\mathcal{T}^{(2)}_s (h^{\otimes2})=0$.
\end{lemma}
\begin{proof}
If $h\vert_s=I\vert_s$, Lemma~\ref{lemma:pauli_string_copies_one_step_evolution} gives $\mathcal{T}^{(2)}_s (h^{\otimes2})=h^{\otimes2}$ and hence $\mathcal{D}^{\otimes 2}\circ\mathcal{T}^{(2)}_s (h^{\otimes2})=\mathcal{D}^{\otimes 2}(h^{\otimes 2})=0$ because the identity commutes with any operator. If $h\vert_s\neq I\vert_s$, Lemma~\ref{lemma:pauli_string_copies_one_step_evolution} gives
\begin{equation}
    \mathcal{T}^{(2)}_s (h^{\otimes2}) =  \frac{1}{2^{|2s|}-1} \sum_{\sigma\in\mathcal{P}'\vert_s} \sigma^{\otimes 2} \otimes ( h\vert_{\bar{s}} )^{\otimes 2}.
\end{equation}
Since $\Omega$ is also a Pauli string, hence either $\Omega$ commutes with $\sigma$ or anti-commutes with $\sigma$. If they commute, $\mathcal{D}(\sigma)=0$. If they anti-commute, the commutation relation of Pauli matrices gives $\mathcal{D}(\sigma)=\pm2\sigma'$, where $\sigma'$ is some other non-trivial Pauli string. The number of Pauli strings in $\mathcal{P}\vert_s$ that anti-commute with $\Omega$ is just one-half of the total, i.e., $4^{|s|}/2$. This is because one can easily establish a one-to-one correspondence between the subsets of Pauli strings that commute and anti-commute with $\Omega$ by multiplying a new Pauli string $\Omega'$ that is only non-trivial on one of the qubits in $s(\Omega)$ and does not equal to $\Omega$. For example, if $\Omega=Z\otimes I$, one can choose $\Omega'=X\otimes I$ and map the commuting subset $\{I,Z\}\otimes \mathcal{P}\vert_{q_2}$ to the anti-commuting subset $\{X,Y\}\otimes \mathcal{P}\vert_{q_2}$ up to some phase factors. As a result, the final output of the three composed channels is
\begin{equation}
    \mathcal{T}^{(2)}_s\circ\mathcal{D}^{\otimes 2}\circ\mathcal{T}^{(2)}_s (h^{\otimes2}) = \frac{(\pm 2)^2\times 4^{|s|}/2}{(2^{|2s|}-1)^2} \sum_{\sigma\in\mathcal{P}'\vert_s} \sigma^{\otimes 2} \otimes ( h\vert_{\bar{s}} )^{\otimes 2},
\end{equation}
which is exactly Eq.~\eqref{eq:commutator_channel_tdt} after simplification.
\end{proof}

Lemma~\ref{lemma:commutator_channel} reveals that the effect of the differential channel is just to eliminate the Pauli strings that are trivial on the support of the differential gate $U_\mu$. For those Pauli strings that are non-trivial on $s(U_\mu)$, the differential channel just contributes a constant on top of the two adjacent local $2$-designs. Up to this point, we have introduced all the lemmas we need to prove the theorems in the main text. Before presenting the theorems and proofs, we would like to first clarify some concepts that are used in the formulation of the theorems.

\begin{definition}
The connecting support of two blocks $B_{k}$ and $B_{k'}$ with $k\leq k'$ in the circuit $\mathbf{U}$ is defined by $s_c(B_{k},B_{k'})=s(B_{k})\cap s(B_{k'}) - \bigcup_{k<k''<k'}s(B_{k''})$. If $k\geq k'$, the connecting support is defined by $s_c(B_{k},B_{k'})=s_c(B_{k'},B_{k})$. Non-empty connecting support means that there is at least a common edge between the two blocks in the circuit graph. In this case, we call the two blocks are ``connected'' within the circuit.
\end{definition}

\begin{definition}
The forward residual support $s_f(B_k)$ of a block $B_k$ in the circuit $\mathbf{U}$ is defined by the qubits in $s(B_k)$ that are not acted by the blocks in $\mathbf{U}$ that act earlier than $B_k$, i.e., $ s_f(B_k) = s(B_k) - \bigcup_{k'<k} s(B_{k'}) $. Correspondingly, the backward residual support $s_b(B_k)$ of $B_k$ is defined by the qubits in $s(B_k)$ that are not acted by the blocks in $\mathbf{U}$ that act later than $B_k$, i.e., $s_b(B_k) = s(B_k) - \bigcup_{k'>k} s(B_{k'})$. If the forward (backward) residual support is not empty, we say the block is a head (tail) block in the circuit, i.e., the block is the first (last) block in the entire circuit that acts on some qubit.
\end{definition}

\begin{definition}
A path $p$ on the circuit $\mathbf{U}$ is defined as a time-ordered sequence of blocks $p=\{B_{k_1}, B_{k_2},  B_{k_3}, \cdots\}$ with $k_1<k_2<k_3<\cdots$ where adjacent blocks are connected within the circuit. The head (tail) of a path $p$ refers to the first (last) block in the path, denoted as $\mathrm{Head}(p)$ and $\mathrm{Tail}(p)$. The forward (backward) residual support of a path $p$ is defined by the union of the forward (backward) residual supports of all the blocks in the path $p$, i.e., $s_f(p) = \bigcup_{B_k\in p} s_f(B_k)$ and $s_b(p) = \bigcup_{B_k\in p} s_b(B_k)$. Every path is demanded to traverse the entire circuit, i.e., $s_f(\mathrm{Head}(p))\neq\varnothing$ and $s_b(\mathrm{Tail}(p))\neq\varnothing$. In other words, the head (tail) of the path is also a head (tail) block in the circuit.
\end{definition}

The motivation for this path definition will be self-evident in the proof below. An example of a path is depicted as the following
\begin{equation}
\begin{mytikz2}
%Straight Lines [id:da959319512549091] 
\draw [line width=0.75]    (100,141.93) -- (300.5,141.93) ;
%Straight Lines [id:da9966950528056495] 
\draw [line width=0.75]    (100,175.47) -- (300.5,175.46) ;
%Straight Lines [id:da5589761164366458] 
\draw [line width=0.75]    (100,108.39) -- (300.5,108.39) ;
%Rounded Rect [id:dp37869621105417095] 
\draw  [fill={rgb, 255:red, 204; green, 230; blue, 255 }  ,fill opacity=1 ][line width=0.75]  (116.72,105.02) .. controls (116.72,102.25) and (118.96,100) .. (121.73,100) -- (128.42,100) .. controls (131.19,100) and (133.44,102.25) .. (133.44,105.02) -- (133.44,145.29) .. controls (133.44,148.06) and (131.19,150.31) .. (128.42,150.31) -- (121.73,150.31) .. controls (118.96,150.31) and (116.72,148.06) .. (116.72,145.29) -- cycle ;
%Straight Lines [id:da3121417313882482] 
\draw [line width=0.75]    (100,242.54) -- (300.5,242.54) ;
%Straight Lines [id:da7403814283080539] 
\draw [line width=0.75]    (100,276.08) -- (300.5,276.07) ;
%Straight Lines [id:da056365607329372125] 
\draw [line width=0.75]    (100,209) -- (300.5,209) ;
%Rounded Rect [id:dp5114763011590588] 
\draw  [fill={rgb, 255:red, 204; green, 230; blue, 255 }  ,fill opacity=1 ][line width=0.75]  (141.93,138.56) .. controls (141.93,135.79) and (144.18,133.54) .. (146.95,133.54) -- (153.64,133.54) .. controls (156.41,133.54) and (158.65,135.79) .. (158.65,138.56) -- (158.65,178.83) .. controls (158.65,181.6) and (156.41,183.85) .. (153.64,183.85) -- (146.95,183.85) .. controls (144.18,183.85) and (141.93,181.6) .. (141.93,178.83) -- cycle ;
%Rounded Rect [id:dp8642104052146871] 
\draw  [fill={rgb, 255:red, 255; green, 180; blue, 180 }  ,fill opacity=1 ][line width=0.75]  (167.09,172.09) .. controls (167.09,169.32) and (169.34,167.08) .. (172.11,167.08) -- (178.8,167.08) .. controls (181.57,167.08) and (183.81,169.32) .. (183.81,172.09) -- (183.81,212.37) .. controls (183.81,215.14) and (181.57,217.38) .. (178.8,217.38) -- (172.11,217.38) .. controls (169.34,217.38) and (167.09,215.14) .. (167.09,212.37) -- cycle ;
%Rounded Rect [id:dp5970167848768184] 
\draw  [fill={rgb, 255:red, 204; green, 230; blue, 255 }  ,fill opacity=1 ][line width=0.75]  (192.25,205.63) .. controls (192.25,202.86) and (194.5,200.61) .. (197.27,200.61) -- (203.96,200.61) .. controls (206.73,200.61) and (208.97,202.86) .. (208.97,205.63) -- (208.97,245.9) .. controls (208.97,248.67) and (206.73,250.92) .. (203.96,250.92) -- (197.27,250.92) .. controls (194.5,250.92) and (192.25,248.67) .. (192.25,245.9) -- cycle ;
%Rounded Rect [id:dp890150743555258] 
\draw  [fill={rgb, 255:red, 204; green, 230; blue, 255 }  ,fill opacity=1 ][line width=0.75]  (217.41,239.17) .. controls (217.41,236.4) and (219.66,234.15) .. (222.43,234.15) -- (229.12,234.15) .. controls (231.89,234.15) and (234.13,236.4) .. (234.13,239.17) -- (234.13,279.44) .. controls (234.13,282.21) and (231.89,284.46) .. (229.12,284.46) -- (222.43,284.46) .. controls (219.66,284.46) and (217.41,282.21) .. (217.41,279.44) -- cycle ;
%Rounded Rect [id:dp36777104673273553] 
\draw  [fill={rgb, 255:red, 204; green, 230; blue, 255 }  ,fill opacity=1 ][line width=0.75]  (167.09,105.02) .. controls (167.09,102.25) and (169.34,100) .. (172.11,100) -- (178.8,100) .. controls (181.57,100) and (183.81,102.25) .. (183.81,105.02) -- (183.81,145.29) .. controls (183.81,148.06) and (181.57,150.31) .. (178.8,150.31) -- (172.11,150.31) .. controls (169.34,150.31) and (167.09,148.06) .. (167.09,145.29) -- cycle ;
%Rounded Rect [id:dp20295268958260881] 
\draw  [fill={rgb, 255:red, 255; green, 180; blue, 180 }  ,fill opacity=1 ][line width=0.75]  (192.25,138.56) .. controls (192.25,135.79) and (194.5,133.54) .. (197.27,133.54) -- (203.96,133.54) .. controls (206.73,133.54) and (208.97,135.79) .. (208.97,138.56) -- (208.97,178.83) .. controls (208.97,181.6) and (206.73,183.85) .. (203.96,183.85) -- (197.27,183.85) .. controls (194.5,183.85) and (192.25,181.6) .. (192.25,178.83) -- cycle ;
%Rounded Rect [id:dp4830860004171569] 
\draw  [fill={rgb, 255:red, 204; green, 230; blue, 255 }  ,fill opacity=1 ][line width=0.75]  (217.41,172.09) .. controls (217.41,169.32) and (219.66,167.08) .. (222.43,167.08) -- (229.12,167.08) .. controls (231.89,167.08) and (234.13,169.32) .. (234.13,172.09) -- (234.13,212.37) .. controls (234.13,215.14) and (231.89,217.38) .. (229.12,217.38) -- (222.43,217.38) .. controls (219.66,217.38) and (217.41,215.14) .. (217.41,212.37) -- cycle ;
%Rounded Rect [id:dp17288415325530893] 
\draw  [fill={rgb, 255:red, 204; green, 230; blue, 255 }  ,fill opacity=1 ][line width=0.75]  (242.57,205.63) .. controls (242.57,202.86) and (244.82,200.61) .. (247.59,200.61) -- (254.28,200.61) .. controls (257.05,200.61) and (259.29,202.86) .. (259.29,205.63) -- (259.29,245.9) .. controls (259.29,248.67) and (257.05,250.92) .. (254.28,250.92) -- (247.59,250.92) .. controls (244.82,250.92) and (242.57,248.67) .. (242.57,245.9) -- cycle ;
%Rounded Rect [id:dp7438051595882824] 
\draw  [fill={rgb, 255:red, 204; green, 230; blue, 255 }  ,fill opacity=1 ][line width=0.75]  (267.73,239.17) .. controls (267.73,236.4) and (269.98,234.15) .. (272.75,234.15) -- (279.44,234.15) .. controls (282.21,234.15) and (284.45,236.4) .. (284.45,239.17) -- (284.45,279.44) .. controls (284.45,282.21) and (282.21,284.46) .. (279.44,284.46) -- (272.75,284.46) .. controls (269.98,284.46) and (267.73,282.21) .. (267.73,279.44) -- cycle ;
%Rounded Rect [id:dp7175794669528333] 
\draw  [fill={rgb, 255:red, 255; green, 180; blue, 180 }  ,fill opacity=1 ][line width=0.75]  (116.77,172.11) .. controls (116.77,169.34) and (119.02,167.1) .. (121.79,167.1) -- (128.48,167.1) .. controls (131.25,167.1) and (133.49,169.34) .. (133.49,172.11) -- (133.49,212.39) .. controls (133.49,215.16) and (131.25,217.4) .. (128.48,217.4) -- (121.79,217.4) .. controls (119.02,217.4) and (116.77,215.16) .. (116.77,212.39) -- cycle ;
%Rounded Rect [id:dp13029356427012107] 
\draw  [fill={rgb, 255:red, 255; green, 180; blue, 180 }  ,fill opacity=1 ][line width=0.75]  (141.93,205.65) .. controls (141.93,202.88) and (144.18,200.63) .. (146.95,200.63) -- (153.64,200.63) .. controls (156.41,200.63) and (158.65,202.88) .. (158.65,205.65) -- (158.65,245.92) .. controls (158.65,248.69) and (156.41,250.94) .. (153.64,250.94) -- (146.95,250.94) .. controls (144.18,250.94) and (141.93,248.69) .. (141.93,245.92) -- cycle ;
%Rounded Rect [id:dp2909294591833491] 
\draw  [fill={rgb, 255:red, 204; green, 230; blue, 255 }  ,fill opacity=1 ][line width=0.75]  (116.77,239.18) .. controls (116.77,236.41) and (119.02,234.16) .. (121.79,234.16) -- (128.48,234.16) .. controls (131.25,234.16) and (133.49,236.41) .. (133.49,239.18) -- (133.49,279.45) .. controls (133.49,282.22) and (131.25,284.47) .. (128.48,284.47) -- (121.79,284.47) .. controls (119.02,284.47) and (116.77,282.22) .. (116.77,279.45) -- cycle ;
%Rounded Rect [id:dp6575792888168039] 
\draw  [fill={rgb, 255:red, 204; green, 230; blue, 255 }  ,fill opacity=1 ][line width=0.75]  (167.09,239.18) .. controls (167.09,236.41) and (169.34,234.16) .. (172.11,234.16) -- (178.8,234.16) .. controls (181.57,234.16) and (183.81,236.41) .. (183.81,239.18) -- (183.81,279.45) .. controls (183.81,282.22) and (181.57,284.47) .. (178.8,284.47) -- (172.11,284.47) .. controls (169.34,284.47) and (167.09,282.22) .. (167.09,279.45) -- cycle ;
%Rounded Rect [id:dp9910291718670561] 
\draw  [fill={rgb, 255:red, 255; green, 180; blue, 180 }  ,fill opacity=1 ][line width=0.75]  (217.41,105.01) .. controls (217.41,102.24) and (219.66,99.99) .. (222.43,99.99) -- (229.12,99.99) .. controls (231.89,99.99) and (234.13,102.24) .. (234.13,105.01) -- (234.13,145.28) .. controls (234.13,148.05) and (231.89,150.3) .. (229.12,150.3) -- (222.43,150.3) .. controls (219.66,150.3) and (217.41,148.05) .. (217.41,145.28) -- cycle ;
%Rounded Rect [id:dp29840669293909716] 
\draw  [fill={rgb, 255:red, 255; green, 180; blue, 180 }  ,fill opacity=1 ][line width=0.75]  (242.57,138.55) .. controls (242.57,135.78) and (244.82,133.53) .. (247.59,133.53) -- (254.28,133.53) .. controls (257.05,133.53) and (259.29,135.78) .. (259.29,138.55) -- (259.29,178.82) .. controls (259.29,181.59) and (257.05,183.84) .. (254.28,183.84) -- (247.59,183.84) .. controls (244.82,183.84) and (242.57,181.59) .. (242.57,178.82) -- cycle ;
%Rounded Rect [id:dp47053501459018743] 
\draw  [fill={rgb, 255:red, 255; green, 180; blue, 180 }  ,fill opacity=1 ][line width=0.75]  (267.73,172.08) .. controls (267.73,169.31) and (269.98,167.07) .. (272.75,167.07) -- (279.44,167.07) .. controls (282.21,167.07) and (284.45,169.31) .. (284.45,172.08) -- (284.45,212.36) .. controls (284.45,215.13) and (282.21,217.37) .. (279.44,217.37) -- (272.75,217.37) .. controls (269.98,217.37) and (267.73,215.13) .. (267.73,212.36) -- cycle ;
%Rounded Rect [id:dp09077192990722183] 
\draw  [fill={rgb, 255:red, 204; green, 230; blue, 255 }  ,fill opacity=1 ][line width=0.75]  (267.73,105.01) .. controls (267.73,102.24) and (269.98,99.99) .. (272.75,99.99) -- (279.44,99.99) .. controls (282.21,99.99) and (284.45,102.24) .. (284.45,105.01) -- (284.45,145.28) .. controls (284.45,148.05) and (282.21,150.3) .. (279.44,150.3) -- (272.75,150.3) .. controls (269.98,150.3) and (267.73,148.05) .. (267.73,145.28) -- cycle ;
%Curve Lines [id:da6968554263301048] 
\draw [color={rgb, 255:red, 255; green, 0; blue, 0 }  ,draw opacity=1 ][line width=1.1]  [dash pattern={on 4pt off 2pt}]  (87.84,191.96) .. controls (91.99,192.01) and (116.93,192.3) .. (125.13,192.25) .. controls (133.33,192.2) and (135.93,225.97) .. (150.29,225.79) .. controls (164.65,225.61) and (210.28,125.33) .. (225.77,125.15) .. controls (241.27,124.97) and (267.6,192.3) .. (276.09,192.22) .. controls (284.59,192.14) and (315.84,191.96) .. (318.64,191.96) ;
\end{mytikz2}\quad.
\end{equation}
Some measures can be defined to quantify the properties of paths, such as the ``length'' below.

\begin{definition}
An edge in a given path $p$ refers to a pair of adjacent blocks $(B_{k}, B_{k'})$ in the path with $k<k'$. The ``length'' of the edge is defined by
\begin{equation}\label{eq:def_edge_length}
    l(B_{k},B_{k'}) = \log_4 \left[\frac{4^{|s(B_{k'})|}-1}{4^{|s_c(B_{k}, B_{k'})|}-1}\right].
\end{equation}
Specifically, we regard $(\rho_0, B_{k_1})$ also as an edge in the path. Remember that $\rho_0$ is the initial state. The length of this edge is defined by replacing the connecting support with the forward residual support, i.e.,
\begin{equation}\label{eq:def_edge_length_head}
    l(\rho_0, B_{k_1}) = \log_4 \left[\frac{4^{|s(B_{k_1})|}-1}{4^{|s_f(B_{k_1})|}-1}\right],
\end{equation}
In fact, the forward residual support of $B_{k_1}$ can be naturally seen as the connecting support of $\rho_0$ and $B_{k_1}$, namely $s_f(B_{k_1}) = s_c(\rho_0, B_{k_1})$, if $\rho_0$ is identified as the ``zeroth'' block $B_0$. The edge set corresponding to the path $p$ is defined by all edges in the path $\mathrm{Edge}(p)=\{(\rho_0, B_{k_1}), (B_{k_1}, B_{k_2}), (B_{k_2}, B_{k_3}),\cdots\}$. The ``length'' of the path is defined by the sum of the lengths of all edges, i.e.,
\begin{equation}\label{eq:path_length}
    l(p) = \sum_{(B_k,B_{k'})\in \mathrm{Edge}(p)} l(B_k,B_{k'}).
\end{equation}
\end{definition}
One can see that if the edge length is a constant for every edge in the path, like in the example above, then the length of the path naturally reduces to the geometrical length $|p|$ (the number of elements in $p$) times the constant. In the limit of large $|s(B_{k'})|$ and $|s_c(B_{k}, B_{k'})|$, the edge length can be approximated as $l(B_{k}, B_{k'}) \approx |s(B_{k'}) - s_c(B_{k}, B_{k'})|$, which only depends on the part of $B_{k'}$ that is not connected with $B_{k}$. Actually, even in the minimal case of staggered two-qubit gates with $|s(B_{k'})|=2$ and $|s_c(B_{k}, B_{k'})|=1$, this approximation is good enough, i.e., $l(B_{k}, B_{k'})=\log_4 5\approx1.16$ which is close to $|s(B_{k'}) - s_c(B_{k}, B_{k'})|=1$. As we will see, the edge length defined here signifies the information leakage by the staggered adjacent blocks in the process of backward evolution. Note that $l(B_{k}, B_{k'})=0$ just means that the two adjacent blocks are connected on the whole support and hence they can be combined into a single $2$-design.

We will use $s(p)$ to denote the qubit support union of all blocks in $p$, i.e., $s(p)=\bigcup_{B_k\in p}s(B_k)$. We will use the notation $p(B_{k(\mu)},h_j)$ to indicate that the path $p$ goes through the differential block $B_{k(\mu)}$ and connects to $h_j$, i.e., $B_{k(\mu)}\in p$ and $s_b(p)\cap s(h_j)\neq\varnothing$. Note that the existence of $p(B_{k(\mu)},h_j)$ is equivalent to the condition that $B_{k(\mu)}$ is within the causal cone of $h_j$, which means that $B_{k(\mu)}$ would not be eliminated directly by unitary conjugation when calculating the expectation value of $h_j$. Moreover, $p(B_{k(\mu)},h_j)$ may also be non-unique, i.e., there might exist multiple possible paths through $B_{k(\mu)}$ to $h_j$ on the PQC. If so, we just choose one of the paths and denote it as $p(B_{k(\mu)},h_j)$.

The concepts defined above will be useful when concerning the factor in the lower bound on the variance that is related to the deepness of the circuit. On the other hand, when we are also concerned about the factor in the lower bound that is related to the locality of the Hamiltonian, the following concepts will be helpful.

\begin{definition}
A path set $P$ on a given PQC is defined by a collection of paths $P=\{p_1,p_2,\cdots\}$ on the PQC. The node set of the path set $P$ is defined by the union of the paths in $P$, i.e., $\mathrm{Node}(P)=\bigcup_{p_i\in P} p_i$. The edge set of the path set $P$ is defined by the union of the edge sets corresponding to the paths in $P$, i.e., $\mathrm{Edge}(P)=\bigcup_{p_i\in P} \mathrm{Edge}(p_i)$. The length of the path set $P$ is defined by the sum of all edges in $P$, i.e.,
\begin{equation}\label{eq:length_path_set}
    l(P) = \sum_{(B_k, B_{k'})\in \mathrm{Edge}(P)} l(B_k, B_{k'}).
\end{equation}
\end{definition}

\begin{definition}
The head (tail) of a path in a path set $P$ is also called a head (tail) block of the path set. The head (tail) set of the path set refers to the set of all the head (tail) blocks in the path set, i.e., $\mathrm{Head}(P)=\{\mathrm{Head}(p_i)\mid p_i\in P\}$ and $\mathrm{Tail}(P)=\{\mathrm{Tail}(p_i)\mid p_i\in P\}$. The forward (backward) residual support of a path set $P$ is defined by the union of the forward (backward) residual supports of all the paths in $P$, i.e., $s_f(P)=\bigcup_{p_i\in P} s_f(p_i)$ and $s_b(P)=\bigcup_{p_i\in P} s_b(p_i)$. 
\end{definition}

\begin{definition}
The forward width of a block $B_k$ is defined by
\begin{equation}\label{eq:def_forward_width}
    w(B_k) = \log_2\left[\frac{4^{|s_f(B_k)|}-1}{2^{|s_f(B_k)|}-1}\right].
\end{equation}
The head width of a path $p$ is defined by the forward width of its head $w(p)=w(\mathrm{Head}(p))$. The head width of a path set $P$ is defined by the sum of the forward widths of all the head blocks in the path set
\begin{equation}\label{eq:width_path_set}
    w(P)=\sum_{B_k\in \mathrm{Head}(P)} w(B_k).
\end{equation}
\end{definition}
Similar to the case of length, one can see that if the head width is a constant for every head block in the path set, then the head width is just the geometrical width (the number of elements in the head set) times the constant. In the limit of large $|s_f(B_k)|$, the head width of the block $B_k$ can be approximated as $w(B_k)\approx|s_f(B_k)|$, which is just the size of the forward residual support, i.e., the part of $s(B_k)$ that is connected to $\rho_0$.

We will use the notation $P(B_{k(\mu)},h_j)$ to indicate that there exists a path $p_i$ in $P$ that goes through the block $B_{k(\mu)}$ and the backward residual support of $P$ covers the support of $h_j$, i.e., $\exists~ p_i\in P,~\text{s.t.}~B_{k(\mu)}\in p_i$ and $s_b(P)\supseteq s(h_j)$. We will call $P(B_{k(\mu)},h_j)$ a path set that goes through $B_{k(\mu)}$ and covers $h_j$ for short. An example of the path set $P(B_{k(\mu)},h_j)$ is depicted as the following
\begin{equation}
\begin{mytikz2}
%Straight Lines [id:da959319512549091] 
\draw [line width=0.75]    (100,141.93) -- (300.5,141.93) ;
%Straight Lines [id:da9966950528056495] 
\draw [line width=0.75]    (100,175.47) -- (300.5,175.46) ;
%Straight Lines [id:da5589761164366458] 
\draw [line width=0.75]    (100,108.39) -- (300.5,108.39) ;
%Rounded Rect [id:dp37869621105417095] 
\draw  [fill={rgb, 255:red, 204; green, 230; blue, 255 }  ,fill opacity=1 ][line width=0.75]  (116.72,105.02) .. controls (116.72,102.25) and (118.96,100) .. (121.73,100) -- (128.42,100) .. controls (131.19,100) and (133.44,102.25) .. (133.44,105.02) -- (133.44,145.29) .. controls (133.44,148.06) and (131.19,150.31) .. (128.42,150.31) -- (121.73,150.31) .. controls (118.96,150.31) and (116.72,148.06) .. (116.72,145.29) -- cycle ;
%Straight Lines [id:da3121417313882482] 
\draw [line width=0.75]    (100,242.54) -- (300.5,242.54) ;
%Straight Lines [id:da7403814283080539] 
\draw [line width=0.75]    (100,276.08) -- (300.5,276.07) ;
%Straight Lines [id:da056365607329372125] 
\draw [line width=0.75]    (100,209) -- (300.5,209) ;
%Rounded Rect [id:dp5114763011590588] 
\draw  [fill={rgb, 255:red, 204; green, 230; blue, 255 }  ,fill opacity=1 ][line width=0.75]  (141.93,138.56) .. controls (141.93,135.79) and (144.18,133.54) .. (146.95,133.54) -- (153.64,133.54) .. controls (156.41,133.54) and (158.65,135.79) .. (158.65,138.56) -- (158.65,178.83) .. controls (158.65,181.6) and (156.41,183.85) .. (153.64,183.85) -- (146.95,183.85) .. controls (144.18,183.85) and (141.93,181.6) .. (141.93,178.83) -- cycle ;
%Rounded Rect [id:dp8642104052146871] 
\draw  [fill={rgb, 255:red, 255; green, 180; blue, 180 }  ,fill opacity=1 ][line width=0.75]  (167.09,172.09) .. controls (167.09,169.32) and (169.34,167.08) .. (172.11,167.08) -- (178.8,167.08) .. controls (181.57,167.08) and (183.81,169.32) .. (183.81,172.09) -- (183.81,212.37) .. controls (183.81,215.14) and (181.57,217.38) .. (178.8,217.38) -- (172.11,217.38) .. controls (169.34,217.38) and (167.09,215.14) .. (167.09,212.37) -- cycle ;
%Rounded Rect [id:dp5970167848768184] 
\draw  [fill={rgb, 255:red, 255; green, 180; blue, 180 }  ,fill opacity=1 ][line width=0.75]  (192.25,205.63) .. controls (192.25,202.86) and (194.5,200.61) .. (197.27,200.61) -- (203.96,200.61) .. controls (206.73,200.61) and (208.97,202.86) .. (208.97,205.63) -- (208.97,245.9) .. controls (208.97,248.67) and (206.73,250.92) .. (203.96,250.92) -- (197.27,250.92) .. controls (194.5,250.92) and (192.25,248.67) .. (192.25,245.9) -- cycle ;
%Rounded Rect [id:dp890150743555258] 
\draw  [fill={rgb, 255:red, 255; green, 180; blue, 180 }  ,fill opacity=1 ][line width=0.75]  (217.41,239.17) .. controls (217.41,236.4) and (219.66,234.15) .. (222.43,234.15) -- (229.12,234.15) .. controls (231.89,234.15) and (234.13,236.4) .. (234.13,239.17) -- (234.13,279.44) .. controls (234.13,282.21) and (231.89,284.46) .. (229.12,284.46) -- (222.43,284.46) .. controls (219.66,284.46) and (217.41,282.21) .. (217.41,279.44) -- cycle ;
%Rounded Rect [id:dp36777104673273553] 
\draw  [fill={rgb, 255:red, 204; green, 230; blue, 255 }  ,fill opacity=1 ][line width=0.75]  (167.09,105.02) .. controls (167.09,102.25) and (169.34,100) .. (172.11,100) -- (178.8,100) .. controls (181.57,100) and (183.81,102.25) .. (183.81,105.02) -- (183.81,145.29) .. controls (183.81,148.06) and (181.57,150.31) .. (178.8,150.31) -- (172.11,150.31) .. controls (169.34,150.31) and (167.09,148.06) .. (167.09,145.29) -- cycle ;
%Rounded Rect [id:dp20295268958260881] 
\draw  [fill={rgb, 255:red, 255; green, 180; blue, 180 }  ,fill opacity=1 ][line width=0.75]  (192.25,138.56) .. controls (192.25,135.79) and (194.5,133.54) .. (197.27,133.54) -- (203.96,133.54) .. controls (206.73,133.54) and (208.97,135.79) .. (208.97,138.56) -- (208.97,178.83) .. controls (208.97,181.6) and (206.73,183.85) .. (203.96,183.85) -- (197.27,183.85) .. controls (194.5,183.85) and (192.25,181.6) .. (192.25,178.83) -- cycle ;
%Rounded Rect [id:dp4830860004171569] 
\draw  [fill={rgb, 255:red, 204; green, 230; blue, 255 }  ,fill opacity=1 ][line width=0.75]  (217.41,172.09) .. controls (217.41,169.32) and (219.66,167.08) .. (222.43,167.08) -- (229.12,167.08) .. controls (231.89,167.08) and (234.13,169.32) .. (234.13,172.09) -- (234.13,212.37) .. controls (234.13,215.14) and (231.89,217.38) .. (229.12,217.38) -- (222.43,217.38) .. controls (219.66,217.38) and (217.41,215.14) .. (217.41,212.37) -- cycle ;
%Rounded Rect [id:dp17288415325530893] 
\draw  [fill={rgb, 255:red, 204; green, 230; blue, 255 }  ,fill opacity=1 ][line width=0.75]  (242.57,205.63) .. controls (242.57,202.86) and (244.82,200.61) .. (247.59,200.61) -- (254.28,200.61) .. controls (257.05,200.61) and (259.29,202.86) .. (259.29,205.63) -- (259.29,245.9) .. controls (259.29,248.67) and (257.05,250.92) .. (254.28,250.92) -- (247.59,250.92) .. controls (244.82,250.92) and (242.57,248.67) .. (242.57,245.9) -- cycle ;
%Rounded Rect [id:dp7438051595882824] 
\draw  [fill={rgb, 255:red, 255; green, 180; blue, 180 }  ,fill opacity=1 ][line width=0.75]  (267.73,239.17) .. controls (267.73,236.4) and (269.98,234.15) .. (272.75,234.15) -- (279.44,234.15) .. controls (282.21,234.15) and (284.45,236.4) .. (284.45,239.17) -- (284.45,279.44) .. controls (284.45,282.21) and (282.21,284.46) .. (279.44,284.46) -- (272.75,284.46) .. controls (269.98,284.46) and (267.73,282.21) .. (267.73,279.44) -- cycle ;
%Rounded Rect [id:dp7175794669528333] 
\draw  [fill={rgb, 255:red, 255; green, 180; blue, 180 }  ,fill opacity=1 ][line width=0.75]  (116.77,172.11) .. controls (116.77,169.34) and (119.02,167.1) .. (121.79,167.1) -- (128.48,167.1) .. controls (131.25,167.1) and (133.49,169.34) .. (133.49,172.11) -- (133.49,212.39) .. controls (133.49,215.16) and (131.25,217.4) .. (128.48,217.4) -- (121.79,217.4) .. controls (119.02,217.4) and (116.77,215.16) .. (116.77,212.39) -- cycle ;
%Rounded Rect [id:dp13029356427012107] 
\draw  [fill={rgb, 255:red, 255; green, 180; blue, 180 }  ,fill opacity=1 ][line width=0.75]  (141.93,205.65) .. controls (141.93,202.88) and (144.18,200.63) .. (146.95,200.63) -- (153.64,200.63) .. controls (156.41,200.63) and (158.65,202.88) .. (158.65,205.65) -- (158.65,245.92) .. controls (158.65,248.69) and (156.41,250.94) .. (153.64,250.94) -- (146.95,250.94) .. controls (144.18,250.94) and (141.93,248.69) .. (141.93,245.92) -- cycle ;
%Rounded Rect [id:dp2909294591833491] 
\draw  [fill={rgb, 255:red, 204; green, 230; blue, 255 }  ,fill opacity=1 ][line width=0.75]  (116.77,239.18) .. controls (116.77,236.41) and (119.02,234.16) .. (121.79,234.16) -- (128.48,234.16) .. controls (131.25,234.16) and (133.49,236.41) .. (133.49,239.18) -- (133.49,279.45) .. controls (133.49,282.22) and (131.25,284.47) .. (128.48,284.47) -- (121.79,284.47) .. controls (119.02,284.47) and (116.77,282.22) .. (116.77,279.45) -- cycle ;
%Rounded Rect [id:dp6575792888168039] 
\draw  [fill={rgb, 255:red, 204; green, 230; blue, 255 }  ,fill opacity=1 ][line width=0.75]  (167.09,239.18) .. controls (167.09,236.41) and (169.34,234.16) .. (172.11,234.16) -- (178.8,234.16) .. controls (181.57,234.16) and (183.81,236.41) .. (183.81,239.18) -- (183.81,279.45) .. controls (183.81,282.22) and (181.57,284.47) .. (178.8,284.47) -- (172.11,284.47) .. controls (169.34,284.47) and (167.09,282.22) .. (167.09,279.45) -- cycle ;
%Rounded Rect [id:dp9910291718670561] 
\draw  [fill={rgb, 255:red, 255; green, 180; blue, 180 }  ,fill opacity=1 ][line width=0.75]  (217.41,105.01) .. controls (217.41,102.24) and (219.66,99.99) .. (222.43,99.99) -- (229.12,99.99) .. controls (231.89,99.99) and (234.13,102.24) .. (234.13,105.01) -- (234.13,145.28) .. controls (234.13,148.05) and (231.89,150.3) .. (229.12,150.3) -- (222.43,150.3) .. controls (219.66,150.3) and (217.41,148.05) .. (217.41,145.28) -- cycle ;
%Rounded Rect [id:dp29840669293909716] 
\draw  [fill={rgb, 255:red, 255; green, 180; blue, 180 }  ,fill opacity=1 ][line width=0.75]  (242.57,138.55) .. controls (242.57,135.78) and (244.82,133.53) .. (247.59,133.53) -- (254.28,133.53) .. controls (257.05,133.53) and (259.29,135.78) .. (259.29,138.55) -- (259.29,178.82) .. controls (259.29,181.59) and (257.05,183.84) .. (254.28,183.84) -- (247.59,183.84) .. controls (244.82,183.84) and (242.57,181.59) .. (242.57,178.82) -- cycle ;
%Rounded Rect [id:dp47053501459018743] 
\draw  [fill={rgb, 255:red, 255; green, 180; blue, 180 }  ,fill opacity=1 ][line width=0.75]  (267.73,172.08) .. controls (267.73,169.31) and (269.98,167.07) .. (272.75,167.07) -- (279.44,167.07) .. controls (282.21,167.07) and (284.45,169.31) .. (284.45,172.08) -- (284.45,212.36) .. controls (284.45,215.13) and (282.21,217.37) .. (279.44,217.37) -- (272.75,217.37) .. controls (269.98,217.37) and (267.73,215.13) .. (267.73,212.36) -- cycle ;
%Rounded Rect [id:dp09077192990722183] 
\draw  [fill={rgb, 255:red, 204; green, 230; blue, 255 }  ,fill opacity=1 ][line width=0.75]  (267.73,105.01) .. controls (267.73,102.24) and (269.98,99.99) .. (272.75,99.99) -- (279.44,99.99) .. controls (282.21,99.99) and (284.45,102.24) .. (284.45,105.01) -- (284.45,145.28) .. controls (284.45,148.05) and (282.21,150.3) .. (279.44,150.3) -- (272.75,150.3) .. controls (269.98,150.3) and (267.73,148.05) .. (267.73,145.28) -- cycle ;
%Curve Lines [id:da6968554263301048] 
\draw [color={rgb, 255:red, 255; green, 0; blue, 0 }  ,draw opacity=1 ][line width=1.2]  [dash pattern={on 3pt off 2pt}]  (103,192.23) .. controls (107.15,192.29) and (116.93,192.3) .. (125.13,192.25) .. controls (133.33,192.2) and (141,225.9) .. (150.29,225.79) .. controls (159.58,225.67) and (212.4,125) .. (225.77,125.15) .. controls (239.15,125.29) and (267.6,192.3) .. (276.09,192.22) .. controls (284.59,192.14) and (317.2,191.96) .. (320,191.96) ;
%Rounded Rect [id:dp7980083730000622] 
\draw  [fill={rgb, 255:red, 255; green, 255; blue, 255 }  ,fill opacity=1 ][line width=0.75]  (300.4,200.6) .. controls (300.4,200.6) and (300.4,200.6) .. (300.4,200.6) -- (325.4,200.6) .. controls (325.4,200.6) and (325.4,200.6) .. (325.4,200.6) -- (325.4,250.91) .. controls (325.4,250.91) and (325.4,250.91) .. (325.4,250.91) -- (300.4,250.91) .. controls (300.4,250.91) and (300.4,250.91) .. (300.4,250.91) -- cycle ;
%Rounded Rect [id:dp7376706164652396] 
\draw  [fill={rgb, 255:red, 255; green, 255; blue, 255 }  ,fill opacity=1 ][line width=0.75]  (74.67,100) .. controls (74.67,100) and (74.67,100) .. (74.67,100) -- (99.67,100) .. controls (99.67,100) and (99.67,100) .. (99.67,100) -- (99.67,283.9) .. controls (99.67,283.9) and (99.67,283.9) .. (99.67,283.9) -- (74.67,283.9) .. controls (74.67,283.9) and (74.67,283.9) .. (74.67,283.9) -- cycle ;
%Curve Lines [id:da03386621143437485] 
\draw [color={rgb, 255:red, 255; green, 0; blue, 0 }  ,draw opacity=1 ][line width=1.2]  [dash pattern={on 3pt off 2.25pt}]  (103,199.57) .. controls (110,199.57) and (116,199.57) .. (124.67,199.57) .. controls (133.33,199.57) and (138.33,229.9) .. (150,229.9) .. controls (161.67,229.9) and (168,204.57) .. (175.67,204.57) .. controls (183.33,204.57) and (212.4,259.4) .. (225.77,259.3) .. controls (239.15,259.21) and (268.09,260) .. (276.09,260) .. controls (284,260) and (316.67,260) .. (320,260) ;

% Text Node
\draw (300,214) node [anchor=north west][inner sep=0.75pt] [font=\small]  [align=left] {$h_{j}$};
% Text Node
\draw (210,75) node [anchor=north west][inner sep=0.75pt] [font=\small]  [align=left] {$B_{k(\mu)}$};
% Text Node
\draw (75,182) node [anchor=north west][inner sep=0.75pt] [font=\small]  [align=left] {$\rho_{0}$};
% Text Node
\draw (305,170) node [anchor=north west][inner sep=0.75pt]  [color={rgb, 255:red, 255; green, 0; blue, 0 }  ,opacity=1 ] [font=\small]  [align=left] {$p_1$};
% Text Node
\draw (305,267) node [anchor=north west][inner sep=0.75pt]  [color={rgb, 255:red, 255; green, 0; blue, 0 }  ,opacity=1 ] [font=\small]  [align=left] {$p_2$};
\end{mytikz2}\quad,\text{ where  } P=\{p_1,p_2\}.
\end{equation}
We emphasize that the examples shown above which use the alternating layered ansatz, are merely to illustrate these newly defined concepts. The applicability of the concepts together with the theorems below, are not limited to the scenarios in the examples. So far, we have introduced all the prerequisites related to Theorem~\ref{theorem:var_lower_bound_path_set} in the main text. Now we are prepared to prove Theorem~\ref{theorem:var_lower_bound_path_set}.

\renewcommand\theproposition{\textcolor{blue}{1}}
\begin{theorem}\label{theorem:var_lower_bound_path_set}
Suppose that $H$ is a Hamiltonian on an $N$-qubit system with the Pauli decomposition $H=\sum_j\lambda_j h_j$. $\mathbf{U}$ is a PQC composed of blocks forming independent local $2$-designs with $s(\mathbf{U})\supseteq s(H)$. The energy expectation with respect to the output state is taken as the cost function $C(\bm{\theta})=\tr(\rho_0\mathbf{U}^\dagger H \mathbf{U})$, where $\rho_0=\ketbrasame{\bm{0}}$. Suppose that the differential gate $U_\mu$ within the block $B_{k(\mu)}$ is sandwiched by two local $2$-designs on $s(B_{k(\mu)})$. The generator $\Omega_\mu$ of $U_\mu$ is a Pauli string. For each $h_j$ whose causal cone contains $B_{k(\mu)}$, there exists at least a path set that goes through $B_{k(\mu)}$ and covers $h_j$, denoted as $P_j$. Then, the variance of the cost derivative is lower bounded by
\begin{equation}\label{eq:var_bound_path_length_width}
    \var_\mathbb{U} [\partial_\mu C] \geq \sum_j \lambda_j^2 \cdot 2^{1 - 2 l(P_j) - w(P_j)},
\end{equation}
where $j$ runs over the indices of $h_j$ whose causal cone contains $B_{k(\mu)}$. $l(P_j)$ and $w(P_j)$ denote the length and the head width of the path set $P_j$, respectively. Note that this lower bound holds for arbitrary choice of $P_j$. The choice with the minimum path length and head width gives rise to the tightest bound.
\end{theorem}
\renewcommand{\theproposition}{S\arabic{proposition}}

\begin{proof}
According to Lemma~\ref{lemma:vanishing_average}, the average of the cost derivative equals to zero $\E_\mathbb{U}[\partial_\mu C]=0$. Thus we only need to estimate the first term in Eq.~\eqref{eq:var_expression_by_twirling}, which is exactly Eq.~\eqref{eq:var_expression_no_squared_average}. According to Corollary~\ref{corollary:var_no_cross_terms}, we only need to estimate the variances with respect to individual Pauli strings in the Hamiltonian
\begin{equation}\label{eq:var_hj}
    \var_\mathbb{U}\left[\partial_\mu \avg{h_j}\right] = \tr\left[\rho_0^{\otimes 2} \mathcal{T}^{(2)}_{s(B_1)}\circ\mathcal{T}^{(2)}_{s(B_2)}\circ\cdots\circ \mathcal{D}_\mu^{\otimes 2} \circ\cdots\circ \mathcal{T}^{(2)}_{s(B_{M'})}(h_j^{\otimes 2}) \right].
\end{equation}
and finally add them up via Eq.~\eqref{eq:var_no_cross_terms}. An important observation from Lemma~\ref{lemma:pauli_string_copies_one_step_evolution} is that a twirling channel just maps a doubled Pauli string into a sum of doubled Pauli strings, with the coefficients summed to one. Namely, the weight vector is normalized. Then, the subsequent twirling channels just map one sum of doubled Pauli strings to another sum of doubled Pauli strings, including the differential block. In the end, among the final sum of doubled Pauli strings, only those that contain only $Z$ and $I$ operators contribute to the expectation value with respect to the initial state $\ket{\bm{0}}$ because $\braoprket{0}{X}{0}=\braoprket{0}{Y}{0}=0$ while $\braoprket{0}{Z}{0}=\braoprket{0}{I}{0}=1$. So the variance equals the sum of the coefficients corresponding to all surviving doubled Pauli strings containing only $Z$ and $I$ in the output of the final channel. Our goal is to give a lower bound on the sum of these coefficients. We will omit the term ``doubled'' below without confusion.

For clarity, before introducing our method of bounding, we first analyze the action of the first two channels $\mathcal{T}^{(2)}_{s(B_{M'})}$ and $\mathcal{T}^{(2)}_{s(B_{M'-1})}$ in detail, which will be helpful for introducing our method of bounding later. We start from the action of $\mathcal{T}^{(2)}_{s(B_{M'})}$. Lemma~\ref{lemma:pauli_string_copies_one_step_evolution} tells us that
\begin{enumerate}[(i)]
    \item If $s(B_{M'})\cap s(h_j)=\varnothing$, i.e., the block $B_{M'}$ does not act on the support of the Pauli string $h_j$ at all, then the output of the corresponding twirling channel keeps the same as the input $\mathcal{T}^{(2)}_{s(B_{M'})}(h_j^{\otimes 2}) = h_j^{\otimes 2}$.
    \item If $s(B_{M'})\cap s(h_j)\neq\varnothing$, then $\mathcal{T}^{(2)}_{s(B_{M'})}$ keeps the sub-string $(h\vert_{\bar{s}(B_{M'})})^{\otimes 2}$ unchanged, but replaces the sub-string $(h\vert_{s(B_{M'})})^{\otimes 2}$ with a uniformly weighted sum of two copies of all non-trivial Pauli strings on $s(B_{M'})$, i.e., the elements in the set $\mathcal{P}'\vert_{s(B_{M'})}$. This kind of replacement is dubbed as ``smearing'' below. One can see that the qubit support union of the summed Pauli strings either grows or keeps the same after the action of $\mathcal{T}^{(2)}_{s(B_{M'})}$.
\end{enumerate}
Next we consider the action of $\mathcal{T}^{(2)}_{s(B_{M'-1})}$ under the condition $s(B_{M'})\cap s(h_j)\neq\varnothing$.
\begin{enumerate}[(i)]
    \item If $s(B_{M'-1})\cap s(h_j)=\varnothing$ and $s(B_{M'-1})\cap s(B_{M'})=\varnothing$, the output is the same as the input.
    \item If $s(B_{M'-1}) \supseteq s(B_{M'})$, the sub-strings on $s(B_{M'-1})$ are smeared, while the remaining part keeps unchanged. In this case, the effect of the block is $B_{M'}$ is invisible since it is fully covered by a $2$-design $B_{M'}$ with larger support.
    \item If $s(B_{M'-1})\cap [s(h_j)-s(B_{M'})]=\varnothing$ and $s(B_{M'-1})\cap s(B_{M'})\neq\varnothing$ and $s(B_{M'-1}) \not\supseteq s(B_{M'})$, the output might be not a uniformly weighted sum anymore. This is because $\mathcal{T}^{(2)}_{s(B_{M'-1})}$ only smears the Pauli strings that are non-trivial on $s(B_{M'-1})$, while not all summed Pauli strings in $\mathcal{T}^{(2)}_{s(B_{M'})}(h_j^{\otimes 2})$ satisfy this condition due to $s(B_{M'-1}) \not\supseteq s(B_{M'})$. For example, in the set of $2$-qubit non-trivial Pauli strings, the subset $\{X\otimes I, Y\otimes I, Z\otimes I\}$ would not be smeared by the twirling channel only covering the second qubit. The part of Pauli strings that are trivial on $s(B_{M'-1})$ keep their weights unchanged as $1/(4^{|s(B_{M'})|}-1)$, while the remaining part of strings that are non-trivial on $s(B_{M'-1})$ will be involved in the twirling channel. The latter forms the subset $\mathcal{P}\vert_{s(B_{M'})-s(B_{M'-1})} \otimes \mathcal{P}'\vert_{s(B_{M'})\cap s(B_{M'-1})} $, in which the number of elements is
    \begin{equation}
        \left(4^{|s(B_{M'}) \cap s(B_{M'-1})|} - 1\right) \times 4^{|s(B_{M'}) - s(B_{M'-1})|}.
    \end{equation}
    Each string will be smeared into $\mathcal{P}'\vert_{s(B_{M'-1})}$. After collecting the repeated terms that have the same sub-string on $s(B_{M'})-s(B_{M'-1})$, we obtain a summation over the set $\mathcal{P}\vert_{s(B_{M'})-s(B_{M'-1})} \otimes \mathcal{P}'\vert_{s(B_{M'-1})} $ with each term having a new uniform weight
    \begin{equation}\label{eq:new_weight}
        \frac{1}{4^{|s(B_{M'})|} - 1} \times \frac{1}{4^{|s(B_{M'-1})|} - 1} \times \left(4^{|s(B_{M'-1}) \cap s(B_{M'})|} - 1\right),
    \end{equation}
    where the last factor comes from collecting repeated terms. Plus those unchanged strings in $\mathcal{P}'\vert_{s(B_{M'})-s(B_{M'-1})} \otimes I\vert_{s(B_{M'-1})\cap s(B_{M'})} $ with the weight $1/(4^{|s(B_{M'})|}-1)$ for each term, we obtain the final output from $\mathcal{T}^{(2)}_{s(B_{M'-1})}$. One can check that these weights are still summed to one. Therefore, if $s(B_{M'-1})\subseteq s(B_{M'})$, then the weights are still uniform since $s(B_{M'-1})\cap s(B_{M'})=s(B_{M'-1})$. But if $s(B_{M'-1})\not\subseteq s(B_{M'})$, namely the two blocks have an overlap and also have their own exclusive supports, then the weights are indeed no more uniform.
    \item If $s(B_{M'-1})\cap [s(h_j)-s(B_{M'})]\neq\varnothing$ and $s(B_{M'-1})\cap s(B_{M'})=\varnothing$, the sub-strings are smeared uniformly with each summed Pauli string non-trivial on $s(B_{M'-1})$ and $s(B_{M'})$ simultaneously, i.e., the set $\mathcal{P}'\vert_{s(B_{M'-1})}\otimes \mathcal{P}'\vert_{s(B_{M'})}$.
    \item If $s(B_{M'-1})\cap [s(h_j)-s(B_{M'})]\neq\varnothing$ and $s(B_{M'-1})\cap s(B_{M'})\neq\varnothing$ and $s(B_{M'-1}) \not\supseteq s(B_{M'})$, the weights will also deviate from a uniform distribution. Although all summed strings in $\mathcal{T}^{(2)}_{s(B_{M'})}(h_j^{\otimes 2})$ will be involved in $\mathcal{T}^{(2)}_{s(B_{M'-1})}$ due to $s(B_{M'-1})\cap [s(h_j)-s(B_{M'})]\neq\varnothing$, the non-uniformity arises when collecting repeated terms. In the subset $\mathcal{P}'\vert_{s(B_{M'})-s(B_{M'-1})}\otimes \mathcal{P}\vert_{s(B_{M'})\cap s(B_{M'-1})}$, the strings that are the same on $s(B_{M'})-s(B_{M'-1})$ contribute the same, with a total number $4^{|s(B_{M'})\cap s(B_{M'-1})|}$, while in the subset $I\vert_{s(B_{M'})-s(B_{M'-1})}\otimes \mathcal{P}'\vert_{s(B_{M'})\cap s(B_{M'-1})}$, the strings that are the same on $s(B_{M'})-s(B_{M'-1})$ contribute the same, with a total number $\left(4^{|s(B_{M'})\cap s(B_{M'-1})|}-1\right)$. Thus the final output is a summation over $\mathcal{P}'\vert_{s(B_{M'})-s(B_{M'-1})} \otimes \mathcal{P}'\vert_{s(B_{M'-1})} $ with each term having a new weight
    \begin{equation}
        \frac{1}{4^{|s(B_{M'})|} - 1} \times \frac{1}{4^{|s(B_{M'-1})|} - 1} \times 4^{|s(B_{M'-1}) \cap s(B_{M'})|},
    \end{equation}
    plus a summation over $I\vert_{s(B_{M'})-s(B_{M'-1})} \otimes \mathcal{P}'\vert_{s(B_{M'-1})} $ with each term having a new weight as in Eq.~\eqref{eq:new_weight}.
\end{enumerate}
One can expect that the classified cases would proliferate exponentially with the number of twirling channels, and the corresponding weights are in general nonuniform. Thus it is hard to track the exact backward evolution of the Pauli string summation at each step in Eq.~\eqref{eq:var_hj}. Instead, noticing the fact that the weights are always non-negative, it is possible to track a part of these Pauli strings to give a lower bound in the end. 

According to Lemma~\ref{lemma:commutator_channel}, only those $h_j$ whose causal cone contains the differential block $B_{k(\mu)}$ can have a non-zero contribution to the variance $\var_{\mathbb{U}}[\partial_\mu C]$. So we only need to care about these Pauli strings in the Hamiltonian. For every such $h_j$, as supposed in the theorem, we have chosen a path set $P_j$ that goes through $B_{k(\mu)}$ and covers $h_j$. We first focus on the simplest scenario where $P_j$ only contains a single path $p_1$. By definition, we have $B_{k(\mu)}\in p_1$ and $s_b(p_1)\supseteq s(h_j)$. We first consider the case where $s_b(\mathrm{Tail}(p_1))\supseteq s(h_j)$. In this case, we aim to track the evolution of the Pauli strings that are non-trivial only along this path. That is to say, we will manually insert a selection channel $\mathcal{S}_{s(B_{k})}^{(2)}$ just after the twirling channel $\mathcal{T}_{s(B_{k})}^{(2)}$ of each block $B_k$ in the path
\begin{equation}
    \cdots \circ \mathcal{T}_{s(B_{k})}^{(2)} \circ \cdots \quad\rightarrow\quad \cdots \circ \mathcal{S}_{s(B_{k})}^{(2)} \circ \mathcal{T}_{s(B_{k})}^{(2)} \circ \cdots.
\end{equation}
A selection channel $\mathcal{S}_{s(B_{k})}^{(2)}$ on $s(B_{k})$ is defined as a linear projection channel with the following action. For two adjacent blocks $(B_{k'}, B_{k})$ in the path $p_1$, after acting $\mathcal{T}_{s(B_{k})}^{(2)}$, the Pauli strings whose sub-strings on $s(B_{k})$ are non-trivial only on $s_c(B_{k'}, B_{k})$ will be mapped to themselves by the selection channel while other strings are directly mapped to zero, i.e.,
\begin{equation}\label{eq:selection_channel}
    \mathcal{S}_{s(B_{k})}^{(2)} (\sigma^{\otimes 2}) = \begin{cases}
        \sigma^{\otimes 2} & ~\text{if}~ s(\sigma\vert_{s(B_{k})})\subseteq s_c(B_{k'}, B_{k}) \\
        0 & ~\text{else}~
    \end{cases}.
\end{equation}
where $\sigma$ is a Pauli string on the $N$ qubits. By doing so, the calculation result after inserting selection channels is always smaller than the original result since the contribution of the strings that are projected out is always non-negative. This can be utilized to establish a lower bound. Moreover, the key benefit of such a selection rule is that the tracked strings would not be influenced by the blocks outside the path since they are trivial outside, i.e., they are non-trivial only on
\begin{equation}\label{eq:track_single_path}
    s_c(p_1; B _{k}). 
\end{equation}
where we have used the notation $s_c(p_1; B_{k})=s_c(B_{k'}, B_{k})$ to represent the connecting support of $B_{k}$ and its predecessor $B_{k'}$ within the path $p_1$. This will be convenient below for the sake that we do not need to create a new symbol to represent the predecessor of $B_{k}$. As a consequence, the twirling channel corresponding to any block outside the path just reduces to the identity channel. Hence, the variance can be lower bounded by
\begin{equation}\label{eq:var_bound_selection_channel}
    \var_\mathbb{U}\left[\partial_\mu \avg{h_j}\right] \geq \tr\left[\rho_0^{\otimes 2} \prod_{B_k\in p_1}\left(\mathcal{S}_{s(B_{k})}^{(2)} \circ \mathcal{T}_{s(B_{k})}^{(2)} \right)(h_j^{\otimes 2}) \right],
\end{equation}
where the product of channels follows the increasing order of indices from left to right, and the differential channel in between is omitted for simplicity. At the first step of the backward evolution, $h_j^{\otimes 2}$ is transformed into a uniform weighted sum of two copies of Pauli strings that are non-trivial only on $s(\mathrm{Tail}(p_1))$. Then, for any two adjacent blocks $(B_k, B_{k'})$ in the path $p_1$, only a certain proportion of the summed strings in the output of $\mathcal{T}_{s(B_{k'})}^{(2)}$ satisfies our selection rule. This proportion is just the number of elements in $\mathcal{P}'\vert_{s_c(B_{k}, B_{k'})}$ over that in $\mathcal{P}'\vert_{s(B_{k'})}$, i.e.,
\begin{equation}\label{eq:proportion_along_path}
    \frac{4^{|s_c(B_{k}, B_{k'})|} - 1}{4^{|s(B_{k'})|} - 1}.
\end{equation}
Note that with our tracked strings as the input, the output of each twirling channel is still a uniform weighted summation because the contribution from other parts has been neglected. Therefore, this process can be carried out recursively. Normalizing the weights of the surviving Pauli strings will give rise to a factor not larger than $1$, which is just Eq.~\eqref{eq:proportion_along_path}. The normalized summation will enter the subsequent channels, while the normalization factor will contribute to our final lower bound. The normalization factor will be accumulated along the path until meeting the initial state, which contributes a total factor of
\begin{equation}
    \prod_{(B_{k},B_{k'})\in p_1}\frac{4^{|s_c(B_{k},B_{k'})|} - 1}{4^{|s(B_{k'})|} - 1} = 4^{-l(p_1)}.
\end{equation}
According to Lemma~\ref{lemma:commutator_channel}, the differential block just contributes a constant factor 
\begin{equation}
    \frac{2\cdot 4^{|s(B_{k(\mu)})|}}{4^{|s(B_{k(\mu)})|} - 1}.
\end{equation}
Finally, the tracked Pauli strings are non-trivial only on $s_f(\mathrm{Head}(p_1))$. Among these strings, only those that contain only the Pauli $Z$ operator and the identity operator $I$ will contribute to the variance by a base number of $1$, whose proportion is
\begin{equation}
    \frac{2^{|s_f(\mathrm{Head}(p_1))|}-1}{4^{|s_f(\mathrm{Head}(p_1))|}-1} = 2^{-w(p_1)}.
\end{equation}
To sum up, the final contribution of these selected Pauli strings to the variance is just
\begin{equation}\label{eq:lower_bound_single_path_single_hj}
    \frac{2\cdot 4^{|s(B_{k(\mu)})|}}{4^{|s(B_{k(\mu)})|} - 1} \times 4^{-l(p_1)} \times 2^{-w(p_1)} > 2\times 4^{-l(p_1)} \times 2^{-w(p_1)},
\end{equation}
where we have relaxed the first factor to $2$ for simplicity.

Next, we consider the case where $s_b(\mathrm{Tail}(p_1)) \not\supseteq s(h_j)$ but it still holds that $s_b(p_1) \supseteq s(h_j)$. In this case, the tracked strings are no longer only non-trivial along the path as in Eq.~\eqref{eq:track_single_path}, but also on the residual support of $h_j$, i.e.,
\begin{equation}
    s_c(p_1; B_{k}) \cup \left[ s(h_j) - \bigcup_{k'\geq k} s_b(B_{k'}) \right],
\end{equation}
because $h_j$ is not fully covered by the last block of $p_1$. However, this does not make a difference since we still have $s_b(p_1) \supseteq s(h_j)$, which means that sooner or later, the residual support of $h_j$ will be covered by subsequent blocks in $p_1$ before meeting the blocks outside the path. Thus, the inequality in Eq.~\eqref{eq:var_bound_selection_channel} still holds. The normalization factor of the tracked strings will remain the same as in Eq.~\eqref{eq:proportion_along_path} and hence the lower bound is the same as in Eq.~\eqref{eq:lower_bound_single_path_single_hj}.

Then, we consider the case where $s(h_j)$ is not fully covered by $s_b(p_1)$ so that more paths are needed to cover $s(h_j)$ jointly. We focus on the case of two paths first, i.e., $s_b(P_j)\supseteq s(h_j)$ with $P_j=\{p_1,p_2\}$. We first consider the following three irreducible configurations of the two paths. The front and back segment of a path $p$ regarding a block $B_k$ is denoted as $\mathrm{Seg}_{f}(p, B_k)=\{B_{k'}\in p\mid k'\leq k\}$ and $\mathrm{Seg}_{b}(p, B_k)=\{B_{k'}\in p\mid k'\geq k\}$, respectively.
\begin{enumerate}[(i)]
    \item If the two paths are separate from each other, i.e., $p_1\cap p_2=\varnothing$, we can generalize the selection channel in Eq.~\eqref{eq:selection_channel} to
    \begin{equation}\label{eq:selection_channel_separate}
        \mathcal{S}_{s(B_{k})}^{(2)} (\sigma^{\otimes 2}) = \begin{cases}
        \sigma^{\otimes 2} & ~\text{if}~ s(\sigma\vert_{s(B_{k})})\subseteq s_c(p_i; B_{k}) \\
        0 & ~\text{else}~
    \end{cases},
    \end{equation}
    where $p_i$ is the path that contains $B_{k}$, and hence the inequality in Eq.~\eqref{eq:var_bound_selection_channel} can be generalized to
    \begin{equation}\label{eq:var_bound_selection_channel_path_set}
        \var_\mathbb{U}\left[\partial_\mu \avg{h_j}\right] \geq \tr\left[\rho_0^{\otimes 2} \prod_{B_k\in \mathrm{Node}(P_j)}\left(\mathcal{S}_{s(B_{k})}^{(2)} \circ \mathcal{T}_{s(B_{k})}^{(2)} \right)(h_j^{\otimes 2}) \right].
    \end{equation}
    Based on Eq.~\eqref{eq:var_bound_selection_channel}, this inequality still holds because the selection regarding the two paths can proceed separately. That is to say, the selected strings for the block $B_{k}\in p_1$ by Eq.~\eqref{eq:selection_channel} are non-trivial only on
    \begin{equation}\label{eq:tracking_p1p2}
        s_c(p_1; B_{k}) \cup s_c(p_2; B_{k'}) \cup \left[ s(h_j) - \bigcup_{k''\geq k} s_b(B_{k''}) \right].
    \end{equation}
    where $B_{k'}=\mathrm{Head}(\mathrm{Seg}_b(p_2, B_k))$ represents the head block of the back segment of the path $p_2$ regarding the block $B_k$. The support of the selected strings for the block in $p_2$ can be written in a symmetric manner. Thus, the factor in Eq.~\eqref{eq:proportion_along_path} is accumulated along the two paths independently because the tracked Pauli string summation always keeps a tensor product structure. Therefore, the lower bound can be obtained just by counting the lengths and head widths of the two paths separately, i.e., generalizing $l(p_1)$ and $w(p_1)$ in Eq.~\eqref{eq:lower_bound_single_path_single_hj} to $l(P_j)$ and $w(P_j)$. 
    \item If the two paths merge together along the backward direction at some block $B_k$, i.e., $p_1\cap p_2=\mathrm{Seg}_f(p_1, B_k)=\mathrm{Seg}_f(p_2, B_k)$ and $B_k\in p_1\cap p_2$, then at the merging node where $(B_{k}, B_{k'})\in\mathrm{Edge}(p_1)$ and $(B_{k}, B_{k''})\in\mathrm{Edge}(p_2)$ with $k'\neq k''$, the input of $\mathcal{T}_{s(B_{k})}^{(2)}$ in Eq.~\eqref{eq:var_bound_selection_channel} is just the tensor product of two normalized uniform weighted summation of doubled Pauli strings over $\mathcal{P}'\vert_{s_c(B_k, B_{k'})}$ and $\mathcal{P}'\vert_{s_c(B_k, B_{k''})}$. The normalization factor is
    \begin{equation}
        \frac{4^{|s_c(B_{k}, B_{k'})|} - 1}{4^{|s(B_{k'})|} - 1} \times \frac{4^{|s_c(B_{k}, B_{k''})|} - 1}{4^{|s(B_{k''})|} - 1} = 4^{-l(B_{k}, B_{k'})-l(B_{k}, B_{k''})},
    \end{equation}
    which is just the exponentiated sum of the lengths of the two edges. Moreover, the accumulated factor does not need to be counted twice for the edges in $p_1\cap p_2$. This double-counting issue has already been considered in the definition of the length and head width of the path set in Eqs.~\eqref{eq:length_path_set} and \eqref{eq:width_path_set}. Therefore, the lower bound can be obtained again by generalizing $l(p_1)$ and $w(p_1)$ in Eq.~\eqref{eq:lower_bound_single_path_single_hj} to $l(P_j)$ and $w(P_j)$. 
    \item If the two paths first coincide and then split along the backward direction at some block $B_k$, i.e., $p_1\cap p_2=\mathrm{Seg}_b(p_1, B_k)=\mathrm{Seg}_b(p_2, B_k)$ and $B_k\in p_1\cap p_2$, then at the splitting node where $(B_{k'}, B_{k})\in\mathrm{Edge}(p_1)$ and $(B_{k''}, B_{k})\in\mathrm{Edge}(p_2)$ with $k'\neq k''$, the selection rule should be slightly modified since the Pauli strings in $\mathcal{P}'\vert_{s_c(B_{k'}, B_{k})}$ and $\mathcal{P}'\vert_{s_c(B_{k''}, B_{k})}$ should be selected simultaneously, i.e.,
    \begin{equation}\label{eq:selection_channel_split}
        \mathcal{S}_{s(B_{k})}^{(2)} (\sigma^{\otimes 2}) = \begin{cases}
        \sigma^{\otimes 2} & ~\text{if}~ s(\sigma\vert_{s(B_{k})})\subseteq s_c(p_1; B_{k}) \cup s_c(p_2; B_{k}) \\
        0 & ~\text{else}~
    \end{cases}.
    \end{equation}
    The selected result is just the tensor product of two normalized uniform weighted summations of doubled Pauli strings over $\mathcal{P}'\vert_{s_c(B_{k'}, B_{k})}$ and $\mathcal{P}'\vert_{s_c(B_{k''}, B_{k})}$, respectively. The normalization factor is
    \begin{equation}
        \frac{(4^{|s_c(B_{k'}, B_{k})|} - 1)\times (4^{|s_c(B_{k''}, B_{k})|} - 1)}{4^{|s(B_{k})|} - 1} \geq 4^{-l(B_{k'}, B_{k})-l(B_{k''}, B_{k})},
    \end{equation}
    which is not equal to but can be lower bounded by the exponentiated sum of the length of the two edges. Moreover, the accumulated factor again should be counted only once for the edges in $p_1\cap p_2$. Thus, the lower bound again can have the same form as in Eq.~\eqref{eq:lower_bound_single_path_single_hj} after generalizing $l(p_1)$ and $w(p_1)$ to $l(P_j)$ and $w(P_j)$.
\end{enumerate}
Other possible configurations of the two paths are just combinations of these three cases, i.e., separation, merger, and split. Therefore, the bound $2\times 4^{-l(P_j)}\times 2^{-w(P_j)}$ holds generally for arbitrary configurations of the two paths. Note that due to the tensor product structure inherited from $h_j$ and our selection rule, the two paths could start their backward evolution independently without causing additional factors.

If $P_j$ contains more than two paths, the lower bound can be obtained in a similar manner by updating the selection rule in Eq.~\eqref{eq:selection_channel} for the multi-split nodes as in Eq.~\eqref{eq:selection_channel_split}, i.e.,
\begin{equation}\label{eq:selection_channel_multi_split}
    \mathcal{S}_{s(B_{k})}^{(2)} (\sigma^{\otimes 2}) = \begin{cases}
    \sigma^{\otimes 2} & ~\text{if}~ s(\sigma\vert_{s(B_{k})})\subseteq \bigcup_i s_c(p_i; B_{k}) \\
    0 & ~\text{else}~
    \end{cases},
\end{equation}
where $p_i$ runs over the paths that contain the block $B_k$. This will give rise to the same lower bound that counts all the factors from the edges and head blocks together with the differential block, i.e., $2\times 4^{-l(P_j)}\times 2^{-w(P_j)}$. Adding the lower bounds corresponding to the individual Pauli strings $h_j$ in the Hamiltonian via Eq.~\eqref{eq:var_no_cross_terms}, we eventually arrive at the desired lower bound corresponding to the entire Hamiltonian.
\end{proof}

For better comprehension, we summarize the proof process above as follows. The assumption of local $2$-designs directly implies the vanishing average and the separable contribution form $\var_\mathbb{U} [\partial_\mu \avg{H}] = \sum_j \lambda_j^2 \var_\mathbb{U} [\partial_\mu \avg{h_j}]$. The contribution from $h_j$ can be expressed as the inner product between $\rho_0^{\otimes 2}$ and $h_j^{\otimes 2}$ after the backward evolution by the twirling channels corresponding to each block together with the differential channel. The twirling channel just maps a doubled Pauli string into a uniform weighted sum of strings. Hence, a lower bound can be obtained by only selecting the part of the summation that is non-trivial along the paths, which gives rise to the accumulated factor $4^{-l(P_j)}$. In the end, only the strings that contain only the Pauli $Z$ and the identity $I$ contribute, which leads to the fraction $2^{-w(P_j)}$. The differential channel restricts the path set must go through $B_{k(\mu)}$ and contributes the constant factor $2$.

We make several remarks on Theorem~\ref{theorem:var_lower_bound_path_set}. Firstly, we emphasize that Theorem~\ref{theorem:var_lower_bound_path_set} holds for any random quantum circuit composed of blocks forming independent local $2$-designs. No assumption is made regarding the number, sizes, or arrangement of these blocks, or the spatial dimension of the circuit. One can check that the lower bound in Theorem~\ref{theorem:var_lower_bound_path_set} is consistent with the previous literature on barren plateaus. For example, compared to the original work on barren plateaus~\cite{McClean2018} which considers a global $2$-design, the path set in Eq.~\eqref{eq:var_bound_path_length_width} is just this single global block and hence the head width grows linearly with the system size, i.e., $w(P_j)\sim N$, which leads to an exponential small lower bound, consistent with the exponential small upper bound proposed by Ref.~\cite{McClean2018}. For the alternating layered ansatz and local cost functions~\cite{Cerezo2021, Uvarov2021a}, the head width could be constant but the path length Eq.~\eqref{eq:var_bound_path_length_width} grows linearly with the system size, i.e., $l(P_j)\sim N$, which means that linear or even larger depth (deep) circuits encounter barren plateaus while finite or logarithmic depth (shallow) circuits do not, consistent with the lower bound in Ref.~\cite{Cerezo2021}. For global cost functions, even a single layer of blocks would require both the path length and head width to increase linearly with the system size, which leads to an exponentially small lower bound, again consistent with the exponentially small upper bound in Ref.~\cite{Cerezo2021}. Furthermore, in the quantum convolutional neural network (QCNN)~\cite{Pesah2021, Zhao2021a} and the multiscale entanglement renormalization ansatz (MERA)~\cite{Zhao2021a, Barthel2023, Miao2023, CerveroMartin2023}, the path length increases logarithmically with the system size, which causes a polynomially vanishing lower bound, consistent with the lower bound proved in Ref.~\cite{Pesah2021}. In summary, Theorem~\ref{theorem:var_lower_bound_path_set} unifies the known gradient scaling behaviors for all architectures of random quantum circuits composed of blocks forming $2$-designs. Two geometrical quantities play an important role in this lower bound, i.e., the path length and head width of the path set on the circuit. For a given setup of the variational quantum eigensolver with a certain circuit architecture, as long as there exists a term $h_j$ in $H$ with $\lambda_j\in\Omega(1)$ accompanied by a path set $P_j$ satisfying $l(P_j),w(P_j)\in\mathcal{O}(\log N)$, the variance of the cost derivative would vanish no faster than polynomially, implying that the corresponding setup does not exhibit barren plateaus.

In particular, we compare Theorem~\ref{theorem:var_lower_bound_path_set} with the lower bound proposed in Ref.~\cite{Uvarov2021a}, which also utilizes the backward evolution of Pauli strings in the proof. In the first place, Ref.~\cite{Uvarov2021a} focuses on the alternating layered ansatz, while our theorem can be applied to all ansatzes composed of local $2$-designs. Moreover, even the form of the causal-cone-width-dependent lower bound considered by Ref.~\cite{Uvarov2021a} has its own limitation. One can consider the case when the differential parameter is located at the last block in the finite local-depth circuit shown in Fig.~\textcolor{blue}{1} in the main text and the Hamiltonian is a single Pauli $Z$ operator on the support of the last block. In such a case, the width of the causal cone equals $N$, and hence the causal-cone-dependent bound like in Ref.~\cite{Uvarov2021a} would give a trivial exponential small lower bound, while in Theorem~\ref{theorem:var_lower_bound_path_set} we can just choose a short path with the head width not scaling with the system size, and provide a meaningful constant lower bound. The lower bound specific to finite local-depth circuits will be introduced in detail in Theorem~\textcolor{blue}{2} below.

Secondly, Theorem~\ref{theorem:var_lower_bound_path_set} defaults the initial state to $\ket{\bm{0}}$, so that the variance of the cost derivative equals the sum of weights of backward evolved doubled Pauli strings only containing $Z$ and $I$. In fact, the theorem can be generalized to the cases with the initial state being an arbitrary product state $\ket{\psi_0}=\bigotimes_{i=1}^{N} V_i \ket{\bm{0}}$ because $V_i$ can be absorbed into the block whose forward residual support covers $q_i$, based on the fact that a $t$-design times any unitary on the same support is still a $t$-design. More generally, as long as the initial state can be reduced back to $\ket{\bm{0}}$ by such an absorption procedure, Theorem~\ref{theorem:var_lower_bound_path_set} is applicable for the corresponding initial state.

Thirdly, Theorem~\ref{theorem:var_lower_bound_path_set} has assumed that the differential gate $U_\mu$ is sandwiched by two local $2$-designs on the support of the differential block $B_{k(\mu)}$, which in practice means that the theorem holds for the case where the differential parameter $\theta_\mu$ is located in the middle of the differential block, as shown in Eq.~\eqref{eq:LUR}. For the parameters located at the borders of blocks, it is indeed possible that the lower bound in Eq.~\eqref{eq:var_bound_path_length_width} fails. For example, if $U_\mu$ is a single-qubit rotation gate $R_z$ and is the first gate acting on some qubit initialized as $\ket{0}$, then the derivative will always equal to zero since $R_z$ only contributes a global phase to $\ket{0}$. Similar events occur also for the case where $U_\mu$ is the last gate acting on some qubit and $U_\mu$ commutes with $h_j$. However, taking these cases into consideration, the proof can still proceed if either of the following two conditions holds.
\begin{itemize}
    \item The left sub-block $L_{k\mu}$ forms a $2$-design and $R_{k\mu}$ is the identity. $s(U_\mu)$ intersects with the connecting support of $B_{k(\mu)}$ and its successor, or intersects with $s(h_j)\cap s_b(B_{k(\mu)})$ and $U_\mu$ does not commute with $h_j$, so that there must exist some selected Pauli strings surviving from the differential channel $\mathcal{D}_\mu^{\otimes 2}$ and hence can be recovered to a uniform weighted summation by the twirling channel of the left sub-block.
    \item The right sub-block $R_{k\mu}$ forms a $2$-design and $L_{k\mu}$ is the identity. $s(U_\mu)$ intersects with the connecting support of $B_{k(\mu)}$ and its predecessor, or intersects with $s_f(B_{k(\mu)})$ and $U_\mu$ does not commute with $\rho_0$, so that there must exist some surviving Pauli strings from $\mathcal{D}_\mu^{\otimes 2}$ still satisfying the selection rule.
\end{itemize}
Thus, we expect that the asymptotic behavior of the lower bound would not change drastically when the differential parameter is located at the borders of blocks. Taking a step back, this issue does not really matter in practice if we consider all trainable parameters simultaneously. That is to say, as long as some of the gradient components do not vanish, the optimization could still proceed successfully and no barren plateau really occurs.

In addition, Theorem~\ref{theorem:var_lower_bound_path_set} has assumed that the generator of the differential gate $\Omega_\mu$ is a Pauli string. This condition can also be relaxed to the case where $\Omega_\mu$ is a general Hermitian generator by considering the Pauli decomposition of $\Omega_\mu$. The new lower bound is just the linear combination of the lower bounds corresponding to the Pauli string components in the generator, after the cross terms caused by $\mathcal{D}_\mu^{\otimes 2}$ are eliminated by Lemma~\ref{lemma:cross_term_h1h2}. In addition, if the gate is commonly not parametrized as a rotation gate in the exponential form $e^{-i\Omega_\mu\theta_\mu}$, e.g., the controlled $R_y$ gate, then it can be seen as the composition of some basic rotation gates with correlating parameters, e.g., the controlled $R_y$ gate can be seen as the composition of two CNOT gates and two $R_y$ gates with opposite rotation angles, and hence the corresponding derivative can be obtained by the chain rule. This may cause the left and right sub-blocks in Eq.~\eqref{eq:LUR} to correlate with each other and be no longer independent $2$-designs. However, since the differential block only contributes a constant factor for non-trivial backward evolved Pauli strings in the proof, we expect the asymptotic behavior contributed from the path length and head width would not be distorted.

Finally, Theorem~\ref{theorem:var_lower_bound_path_set} requires $s(\mathbf{U})\supseteq s(H)$ so that Corollary~\ref{corollary:var_no_cross_terms} holds and one can only consider the support of the head of the path set when counting the Pauli strings that only contains $Z$ and $I$. As a counterexample, if $h_j$ contains a Pauli $X$ operator that is not covered by $\mathbf{U}$, then the variance of the cost derivative would directly vanish due to $\braoprket{0}{X}{0}=0$. But this issue could be easily fixed by multiplying the lower bound corresponding to $h_j$ by an extra factor, i.e., the expectation value of the sub-string of $h_j$ on $s(h_j)-s(\mathbf{U})$ with respect to the initial state.

It is worth mentioning that according to the strict equality Eq.~\eqref{eq:var_hj}, the lower bound in Eq.~\eqref{eq:var_bound_selection_channel_path_set} can be improved into a tighter path-integral-like form 
\begin{equation}
   \var_\mathbb{U}\left[\partial_\mu \avg{h_j}\right] \geq \sum_{P_j} \tr\left[\rho_0^{\otimes 2} \prod_{B_k\in \mathrm{Node}(P_j)}\left(\mathcal{S}_{s(B_{k})}^{(2)} \circ \mathcal{T}_{s(B_{k})}^{(2)} \right)(h_j^{\otimes 2}) \right].
\end{equation}
where the summation is taken over all legal path sets $P_j$ (i.e., with the right end covering $h_j$, and at least one of the paths passing through $B_{k(\mu)}$). Hence the final lower bound can be tightened as
\begin{equation}
\begin{aligned}
    \var_\mathbb{U}\left[\partial_\mu C\right] \geq &
    \sum_j \sum_{P_j} \lambda_j^2 \cdot 2^{1-2l(P_j)-w(P_j)}.
\end{aligned}
\end{equation}
In other words, this improved lower bound can be roughly written in plain text as
\begin{equation}\label{eq:contribution_paths}
\begin{aligned}
    \var_\mathbb{U}\left[\partial_\mu \avg{h_j}\right] \geq &
    \sum~[\text{contribution from all possible path sets}]\\
    =~& [\text{contribution from the chosen path set}] + [\text{contribution from other possible path sets}],
\end{aligned}
\end{equation}
where the first term is just the lower bound in Eq.~\eqref{eq:var_bound_selection_channel}. If we choose the shortest path set with the smallest head width, the first term will be the term of the leading order. The second term represents the sub-leading contributions from other possible path sets, which is directly discarded in our proof for clarity. The additivity of the contributions from different path sets without multiple counting is guaranteed by the orthogonality among different selection channels.

The proof method in Theorem~\ref{theorem:var_lower_bound_path_set} can be extended to bounding other quantities besides the variance of the cost derivative. For example, the variance of the primitive cost function could be lower bounded in a similar way by just removing the differential channel $\mathcal{D}_\mu^{\otimes 2}$ in Eqs.~\eqref{eq:var_expression_by_twirling} and \eqref{eq:var_bound_selection_channel}, which will lead to the same result as in Eq.~\eqref{eq:var_bound_path_length_width} except that the path set does not have to go through the differential block. In general, this path-set approach is promising to be used for estimating the average value of an arbitrary squared space-time correlator in random quantum circuits since the cost derivative in Eq.~\eqref{eq:cost_derivative_expression} itself can be seen as a two-point correlation function between the differential gate and the measured observable.

Up to this point, we have introduced all the contents on Theorem~\ref{theorem:var_lower_bound_path_set} concerning a general lower bound for arbitrary circuits composed of $2$-designs. In the following, we will first make clear some concepts in order to present Theorem~\ref{theorem:var_bound_local_depth}, which focuses on finite local-depth circuits.

\begin{definition}
The maximum block size of the PQC $\mathbf{U}$ is defined by the maximum value of block sizes of all blocks in the PQC, i.e., $\beta = \max\{|s(B_1)|, |s(B_2)|, \cdots, |s(B_{M'})|\}$.
\end{definition}

\begin{definition}
For a Hamiltonian $H$ with the Pauli decomposition $H = \sum_j \lambda_j h_j$, the maximum interaction range $r$ of the Hamiltonian is defined by the maximum value of the support sizes of all Pauli strings $h_j$, i.e., $r = \max\{|s(h_1)|, |s(h_2)|, |s(h_3)|, \cdots\}$. A $r$-local Hamiltonian refers to a Hamiltonian with the maximum interaction range $r$. Note that here the $r$ interacting qubits do not need to be close to each other spatially in the sense of the lattice geometry associated with the Hamiltonian.
\end{definition}

\begin{definition}
The supported block set of the qubit $q_i$ in the PQC contains the block $B_k$ acting on $q_i$, denoted as $b(q_i) = \{B_k \mid q_i\in s(B_k)\}$. The supported block set of a subset of qubits $s$ is the union of the supported sets of the qubits in $s$, i.e., $b(s) = \bigcup_{q_i\in s} b(q_i) = \{B_k \mid s\cap s(B_k)\neq \varnothing\}$. In particular, the supported block set of all the $N$ qubits is just the set of all the blocks in $\mathbf{U}$.
\end{definition}

\begin{definition}
The local depth (qubit-wise depth) of the PQC at a qubit $q_i$ can be defined by the number of elements in the supported block set $b(q_i)$, denoted as $\chi_i=|b(q_i)|$. The maximum local depth of the PQC is defined by the maximum value of local depths of all qubits, i.e., $\chi = \max\{\chi_1, \chi_2, \cdots, \chi_N\}$. Here we use the symbol $\chi$ to distinguish from the common circuit depth $D$, which is defined by the minimum number of layers where all blocks within each layer commute with each other.
\end{definition}

\begin{definition}
A finite local-depth circuit (FLDC) refers to a circuit whose maximum local depth $\chi$ does not scale with the system size $N$.
\end{definition}

The above definition of local depth regards blocks as the elementary units building the PQC. Alternatively, for general quantum circuits taking elementary gates as units, the local depth at the qubit $q_i$ can be defined by the number of non-commuting gates acting on $q_i$, which might be more complex to estimate because the detailed commutation relation of gates is involved. Given a block template like in Eq.~\eqref{eq:block_template}, these two definitions are supposed to differ only by a constant coefficient in general. In addition, we emphasize that the support sizes of the blocks constituting FLDCs are always finite, i.e., $\beta$ should not scale with the system size by the definition of ``circuit''. Correspondingly, a finite depth circuit (FDC) refers to a circuit whose global depth does not scale with the system size, while a general linear depth circuit (GLDC) refers to a circuit whose global depth scales linearly with the system size.

Note that as long as there is a finite local depth, regardless of how messy the gate configuration might be, the circuit can be considered as an FLDC. It is not necessary to arrange the gates uniformly and regularly as shown in Fig.~\textcolor{blue}{1} in the main text. Here we point out that FLDCs defined above are more general than the sequential quantum circuits (SQC) defined in Ref.~\cite{Chen2023a}, especially the instances listed in Ref.~\cite{Chen2023a} because we have not restricted the gates in the circuit as local gates, which is assumed by Ref.~\cite{Chen2023a}. In other words, we do not assume an underlying lattice or connectivity to define the spatial locality of gates. But when it comes to the content of the entanglement area law, a spatial lattice must be involved, and at that time, we will restrict our discussion on FLDCs to the circuits composed of spatially local gates as shown in Fig.~\textcolor{blue}{1} in the main text. With the concepts defined above, we are prepared to prove Theorem~\ref{theorem:var_bound_local_depth}.

\renewcommand\theproposition{\textcolor{blue}{2}}
\begin{theorem}\label{theorem:var_bound_local_depth}
Based on the assumptions in Theorem~\ref{theorem:var_lower_bound_path_set}, suppose that the maximum local depth of $\mathbf{U}$ is $\chi$ and the maximum block size of $\mathbf{U}$ is $\beta$. Then for any $r$-local Hamiltonian, the variance of the cost derivative is lower bounded by
\begin{equation}\label{eq:var_bound_local_depth}
    \var_\mathbb{U} [\partial_\mu C] ~\geq~ 4^{-r \chi \beta} \sum_j 2\lambda_j^2 ,
\end{equation}
where $j$ runs over the indices of the Pauli string $h_j$ that is non-trivial on $s(B_{k(\mu)})$. 
\end{theorem}
\renewcommand{\theproposition}{S\arabic{proposition}}
\begin{proof}
For each $h_j$, if $s(h_j)\cap s(B_{k(\mu)})= \varnothing$, we just neglect the contribution of $h_j$ in Eq.~\eqref{eq:var_bound_path_length_width}. Otherwise, we can choose the path set $P_j$ to be the straight wires on $s(h_j)$, i.e., $P_j = \{b(q_i) \mid q_i\in s(h_j) \}$. Note that the forward residual supports of the head blocks of the path set appear simultaneously in expressions for both length and head width in Eqs.~\eqref{eq:def_edge_length_head} and \eqref{eq:def_forward_width}, we will extract them and combine together before relaxation to obtain a tighter bound. We use $l_0$ to denote the length of the edges between $\rho_0$ and the head blocks of $P_j$. According to the definition in Eq.~\eqref{eq:def_edge_length}, the length of $P_j$ apart from $l_0$ is upper bounded by
\begin{equation}
    l(P_j)-l_0 \leq r \cdot (\chi-1) \cdot \log_4\left(\frac{4^\beta - 1}{4-1}\right) \leq r (\chi-1) \left(\beta - \frac{1}{2}\log_2 3\right) \leq r (\chi-1) \beta,
\end{equation}
where the sizes of the connecting supports are relaxed to one. Correspondingly, the head width of $P_j$ together with $2l_0$ is upper bounded by
\begin{equation}
    w(P_j) + 2l_0 = \sum_{B_k\in \mathrm{Head}(P_j)} \log_2\left( \frac{4^{|s_f(B_k)|}-1}{2^{|s_f(B_k)|}-1}\right) + \log_2\left(\frac{4^{|s(B_k)|}-1}{4^{|s_f(B_k)|}-1}\right) \leq r\times \log_2\left(\frac{4^\beta-1}{2-1}\right) \leq r\cdot 2\beta.
\end{equation}
According to Theorem~\ref{theorem:var_lower_bound_path_set}, the variance of the cost derivative can be lower bounded by
\begin{equation}
    \var_\mathbb{U} [\partial_\mu C] \geq \sum_j \lambda_j^2 \cdot 2^{1 - 2 l(P_j) - w(P_j)} \geq \sum_j \lambda_j^2 2^{1 - 2r(\chi-1)\beta - 2r\beta} =  4^{-r \chi \beta} \sum_j 2\lambda_j^2,
\end{equation}
where $j$ runs over the indices of $h_j$ that is non-trivial on $s(B_{k(\mu)})$. 
\end{proof}

Theorem~\ref{theorem:var_bound_local_depth} elegantly integrates the factors related to barren plateaus in an unexpectedly concise manner, i.e., the block locality $\beta$~\cite{McClean2018}, the Hamiltonian locality $r$~\cite{Cerezo2021, Uvarov2021a} and the circuit deepness $\chi$. It is vitally important to note that the relevant quantity characterizing the circuit deepness is the local depth $\chi$, instead of the common circuit depth, or say the global depth $D$. For the alternating layered ansatz, the local depth basically equals the global depth. But for general circuit structures, the local depth and the global depth are not equivalent and can differ significantly in specific scenarios, as illustrated in Fig.~\textcolor{blue}{1} in the main text. This implies that the class of circuit architectures that do not exhibit barren plateaus is larger than finite (or logarithmic) depth circuits (FDC)~\cite{Cerezo2021, Uvarov2021a}, which naturally give rise to the class of finite (or logarithmic) local-depth circuits. In this work, we mainly focus on finite local-depth circuits (FLDC), which have a clear physical correspondence to entanglement area-law states as we will introduce below. Logarithmic local-depth circuits may involve states beyond the area law such as gapless states with topological orders, which are also interesting to study in the future.

As a superclass, FLDCs have stronger expressibility than FDCs since they can generate long-range entanglement. To be specific, FDCs can only generate short-range entangled (SRE) states such as symmetry-protected topological (SPT) states like the Affleck-Kennedy-Lieb-Tasaki (AKLT) state. These SRE states are equivalent to the product states in the sense of entanglement renormalization. Those long-range entangled (LRE) area-law states such as the Greenberger–Horne–Zeilinger (GHZ) state and the topologically ordered states like the ground states of the toric code model, can not be generated from FDCs, while they can be generated from FLDCs with explicit forms~\cite{Satzinger2021, Chen2023a, Sun2023a}.

On the other hand, as a subclass, FLDCs have less expressibility than general linear depth circuits (GLDC), such as the alternating layered ansatz of linear depth. Given a spatial lattice of qubits, if we restrict that each block can act only on spatial adjacent qubits, then the output state of an FLDC should obey the entanglement area law, while GLDCs generate much more entanglement so that the output states typically exhibit the volume law~\cite{OrtizMarrero2021}. Here we give a simple demonstration to show why an FLDC composed of spatially local blocks generates area-law output states. Take $\beta=2$ for an example. The number of blocks acting on two adjacent qubits hence should be not larger than $\chi$. Suppose that the lattice is partitioned into two continuous regions $A$ and $\bar{A}$. The number of lattice edges cut by this bipartition is denoted by $|\partial A|$, which can also be seen as a measure of the boundary between $A$ and $\bar{A}$. Thus, the number of blocks acting across the boundary is not larger than $|\partial A|\cdot \chi$, so the Schmidt rank $r_A$ of the output state with respect to this bipartition is not larger than $2^{|\partial A|\cdot \chi}$. The entanglement entropy with respect to this bipartition could be upper bounded by $S_A = -\tr(\rho_A\log \rho_A) \leq \log r_A \leq \chi |\partial A|$, which means that for a fixed maximum local depth $\chi$, the entanglement entropy at most grows linearly with the size of the boundary, i.e., the output state satisfies the entanglement area law. Here the term ``area'', used by convention, refers to the boundary of a three-dimensional system.

Considering the two aspects mentioned above, FLDC constitutes an intermediate class of quantum circuits between FDC and GLDC, in the sense of both expressibility and trainability. FLDC possesses the favorable trainability of FDC and at the same time, partially exhibits the expressibility of GLDC in generating long-range entanglement. Under the condition that the target many-body Hamiltonian has a ground state with entanglement area law, FLDCs hold promise to serve as an appropriate ansatz in variation quantum eigensolvers to prepare ground states of gapped many-body systems on quantum devices. In other words, FLDCs composed of spatially local blocks can be regarded as the quantum circuit implementation of tensor network states, such as matrix product states (MPS) and projected entangled paired states (PEPS). Note that the ground states of gapped local Hamiltonians typically satisfy the entanglement area law, which is guaranteed by rigorous theorems in 1D systems~\cite{Hastings2007} and testified by practical experience in higher dimensions. Previously proposed quantum circuit implementation of tensor network states~\cite{Haghshenas2022, Zhou2021} can be seen as special cases of finite or logarithmic local-depth circuits.

The above discussion pertains to Hamiltonians with local interactions, i.e. $r\in\mathcal{O}(1)$, but in practice, it is possible to encounter Hamiltonians involving non-local interactions. For example, in condensed matter or quantum chemistry problems, fermionic statistics are usually encoded to qubit systems via the Jordan-Wigner transformation or the Bravyi-Kitaev transformation, where the support sizes of the resulting Pauli strings can scale with the interaction distance of fermions $R$ as $r\in\mathcal{O}(R)$ or $r\in\mathcal{O}(\log R)$, respectively. If $R\in\mathcal{O}(N)$, the Jordan-Wigner encoding method would lead to an exponentially small lower bound in Eq.~\eqref{eq:var_bound_local_depth} even with FLDCs due to $r\in\mathcal{O}(N)$, while the Bravyi-Kitaev encoding method may still persist with a polynomially vanishing lower bound with FLDCs due to $r\in\mathcal{O}(\log N)$.

In the end, we remark that one of the preconditions of the absence of barren plateaus from the lower bound in Eq.~\eqref{eq:var_bound_local_depth} is that, the sum of the coefficients $\sum_j\lambda_j^2$ does not vanish. That is to say, there is at least a term $h_j$ in the Hamiltonian that is non-trivial on the support of the differential block with $\lambda_j\sim\mathcal{O}(1)$. This precondition can be satisfied by common locally interacting quantum many-body systems, where the support of the Hamiltonian covers the whole system, i.e., $s(H) = s(\mathbf{U})$. However, if we do encounter cases where there is no term acting on $s(B_{k(\mu)})$, trivializing the lower bound in Theorem~\ref{theorem:var_bound_local_depth}, we need to revert back to Theorem~\ref{theorem:var_lower_bound_path_set} to consider the contribution from the terms acting elsewhere. For FLDCs composed of spatially local blocks, if there is a term in the Hamiltonian acting near $s(B_{k(\mu)})$, it is still possible to find a short path on the circuit connecting to $B_{k(\mu)}$, thereby obtaining a non-vanishing lower bound. However, if every term in the Hamiltonian acts far from $s(B_{k(\mu)})$, e.g., the distance grows linearly with the system size, then the bound in Theorem~\ref{theorem:var_lower_bound_path_set} would be exponentially small because the path length also grows linearly. To be specific, we can focus on the following ``ladder ansatz'' 
\begin{equation}
\begin{mytikz2}\label{eq:ladder_ansatz}
%Straight Lines [id:da959319512549091] 
\draw    (100,132.27) -- (269.92,132.27) ;
%Straight Lines [id:da9966950528056495] 
\draw    (100,157.98) -- (269.92,157.97) ;
%Straight Lines [id:da5589761164366458] 
\draw    (100,106.58) -- (269.92,106.58) ;
%Rounded Rect [id:dp37869621105417095] 
\draw  [fill={rgb, 255:red, 204; green, 230; blue, 255 }  ,fill opacity=1 ] (123.34,106.31) .. controls (123.34,103.19) and (125.86,100.67) .. (128.98,100.67) -- (136.51,100.67) .. controls (139.62,100.67) and (142.15,103.19) .. (142.15,106.31) -- (142.15,156.77) .. controls (142.15,159.89) and (139.62,162.42) .. (136.51,162.42) -- (128.98,162.42) .. controls (125.86,162.42) and (123.34,159.89) .. (123.34,156.77) -- cycle ;
%Straight Lines [id:da3121417313882482] 
\draw    (100,209.37) -- (269.92,209.37) ;
%Straight Lines [id:da7403814283080539] 
\draw    (100,235.07) -- (269.92,235.07) ;
%Straight Lines [id:da056365607329372125] 
\draw    (100,183.67) -- (269.92,183.67) ;
%Rounded Rect [id:dp5114763011590588] 
\draw  [fill={rgb, 255:red, 204; green, 230; blue, 255 }  ,fill opacity=1 ] (157.21,132.27) .. controls (157.21,129.15) and (159.74,126.62) .. (162.85,126.62) -- (170.38,126.62) .. controls (173.5,126.62) and (176.02,129.15) .. (176.02,132.27) -- (176.02,182.73) .. controls (176.02,185.85) and (173.5,188.37) .. (170.38,188.37) -- (162.85,188.37) .. controls (159.74,188.37) and (157.21,185.85) .. (157.21,182.73) -- cycle ;
%Rounded Rect [id:dp7966776855072697] 
\draw  [fill={rgb, 255:red, 204; green, 230; blue, 255 }  ,fill opacity=1 ] (191.12,158) .. controls (191.12,154.88) and (193.65,152.35) .. (196.76,152.35) -- (204.29,152.35) .. controls (207.41,152.35) and (209.93,154.88) .. (209.93,158) -- (209.93,208.46) .. controls (209.93,211.58) and (207.41,214.1) .. (204.29,214.1) -- (196.76,214.1) .. controls (193.65,214.1) and (191.12,211.58) .. (191.12,208.46) -- cycle ;
%Rounded Rect [id:dp8756393167051588] 
\draw  [fill={rgb, 255:red, 204; green, 230; blue, 255 }  ,fill opacity=1 ] (225.03,183.21) .. controls (225.03,180.09) and (227.55,177.57) .. (230.67,177.57) -- (238.2,177.57) .. controls (241.31,177.57) and (243.84,180.09) .. (243.84,183.21) -- (243.84,233.68) .. controls (243.84,236.79) and (241.31,239.32) .. (238.2,239.32) -- (230.67,239.32) .. controls (227.55,239.32) and (225.03,236.79) .. (225.03,233.68) -- cycle ;
\end{mytikz2}\quad.
\end{equation}
and derive a tighter lower bound than Theorem~\ref{theorem:var_bound_local_depth} by considering the contribution from the Hamiltonian terms acting elsewhere, which is elaborated in Corollary~\ref{corollary:var_ladder}.

\begin{corollary}\label{corollary:var_ladder}
Based on the assumptions in Theorem~\ref{theorem:var_lower_bound_path_set}, suppose that $H$ is a $r$-local $1$-dimensional Hamiltonian and $\mathbf{U}$ is a ladder ansatz as in Eq.~\eqref{eq:ladder_ansatz} with block size $\beta$. Denote the block which last acts on $s(h_j)$ as $B_{k'(j)}$. The variance of the cost derivative is lower bounded by
\begin{equation}\label{eq:var_ladder}
\begin{aligned}
    \var_\mathbb{U} [\partial_\mu C] ~\geq~ & \sum_j \lambda_j^2 \cdot 2^{1- 2 \beta (\Delta_{\mu j}+r)}.
\end{aligned}
\end{equation}
where $j$ runs over the indices of $h_j$ satisfying $k'(j)\geq k(\mu)$ and $\Delta_{\mu j} = k'(j) - k(\mu)$.
\end{corollary}
\begin{proof}
If $k'(j) < k(\mu)$, then the differential block is outside the causal cone of $h_j$ and hence the contribution to the variance equals zero. If $k'(j)\geq k(\mu)$, we can choose the path set containing a single path from $B_{k(\mu)}$ to $B_{k'(j)}$ and other $r-1$ path covering other qubits acted by $h_j$. The path length apart from the length $l_0$ of the edges between $\rho_0$ and the head blocks is upper bounded by
\begin{equation}
    l(P_j) - l_0 \leq \Delta_{\mu j} \cdot \log_4\left(\frac{4^\beta -1}{4-1}\right) \leq \Delta_{\mu j}\beta.
\end{equation}
The head width of $P_j$ together with $2l_0$ is upper bounded by
\begin{equation}
    w(P_j) + 2l_0 \leq r\cdot \log_2\left(\frac{4^\beta - 1}{2-1}\right) \leq r\beta.
\end{equation}
According to Theorem~\ref{theorem:var_lower_bound_path_set}, the variance of the cost derivative can be lower bounded by
\begin{equation}
    \var_\mathbb{U} [\partial_\mu C] \geq \sum_j \lambda_j^2 \cdot 2^{1 - 2 l(P_j) - w(P_j)} \geq \sum_j 2\lambda_j^2 \cdot 4^{- \beta(\Delta_{\mu j}+r)},
\end{equation}
where $j$ runs over the indices of $h_j$ that is non-trivial on $s(B_{k(\mu)})$. 
\end{proof}

In Corollary~\ref{corollary:var_ladder}, one can explicitly see that the contribution in the lower bound indeed decays exponentially with the ``distance'' $\Delta_{\mu j}$ between the differential block $B_{k(\mu)}$ and the observable $h_j$. Similar conclusions can be drawn with respect to the circuit structures depicted in Fig.~\textcolor{blue}{1} in the main text, which can be seen as the circuit implementation of MPS and PEPS. This coincides with the results found in unitary embedding tensor network states recently~\cite{Liu2021a}. Again, taking a step back, this issue does not really matter in practice if all trainable parameters are considered simultaneously. Since there are always blocks on $s(h_j)$ due to $s(\mathbf{U})\supseteq s(H)$, as long as some of the gradient components do not vanish, the optimization could still proceed successfully and no barren plateau really occurs.

It is worth mentioning that the information spreading in the ladder ansatz shown in Eq.~\eqref{eq:ladder_ansatz} is somewhat one-way because the blocks acting on the qubits with larger indices than $q_i$ are outside the causal cone of $q_i$. This disadvantage could be easily fixed by the following two-way ladder ansatz which still falls in FLDC class
\begin{equation}
\begin{mytikz2}
%Straight Lines [id:da959319512549091] 
\draw    (100,132.27) -- (370.2,132.27) ;
%Straight Lines [id:da9966950528056495] 
\draw    (100,157.98) -- (370.2,157.97) ;
%Straight Lines [id:da5589761164366458] 
\draw    (100,106.58) -- (370.2,106.58) ;
%Rounded Rect [id:dp37869621105417095] 
\draw  [fill={rgb, 255:red, 204; green, 230; blue, 255 }  ,fill opacity=1 ] (123.34,106.31) .. controls (123.34,103.19) and (125.86,100.67) .. (128.98,100.67) -- (136.51,100.67) .. controls (139.62,100.67) and (142.15,103.19) .. (142.15,106.31) -- (142.15,156.77) .. controls (142.15,159.89) and (139.62,162.42) .. (136.51,162.42) -- (128.98,162.42) .. controls (125.86,162.42) and (123.34,159.89) .. (123.34,156.77) -- cycle ;
%Straight Lines [id:da3121417313882482] 
\draw    (100,209.37) -- (370.2,209.37) ;
%Straight Lines [id:da7403814283080539] 
\draw    (100,235.07) -- (370.2,235.07) ;
%Straight Lines [id:da056365607329372125] 
\draw    (100,183.67) -- (370.2,183.67) ;
%Rounded Rect [id:dp5114763011590588] 
\draw  [fill={rgb, 255:red, 204; green, 230; blue, 255 }  ,fill opacity=1 ] (157.21,132.27) .. controls (157.21,129.15) and (159.74,126.62) .. (162.85,126.62) -- (170.38,126.62) .. controls (173.5,126.62) and (176.02,129.15) .. (176.02,132.27) -- (176.02,182.73) .. controls (176.02,185.85) and (173.5,188.37) .. (170.38,188.37) -- (162.85,188.37) .. controls (159.74,188.37) and (157.21,185.85) .. (157.21,182.73) -- cycle ;
%Rounded Rect [id:dp7966776855072697] 
\draw  [fill={rgb, 255:red, 204; green, 230; blue, 255 }  ,fill opacity=1 ] (191.12,158) .. controls (191.12,154.88) and (193.65,152.35) .. (196.76,152.35) -- (204.29,152.35) .. controls (207.41,152.35) and (209.93,154.88) .. (209.93,158) -- (209.93,208.46) .. controls (209.93,211.58) and (207.41,214.1) .. (204.29,214.1) -- (196.76,214.1) .. controls (193.65,214.1) and (191.12,211.58) .. (191.12,208.46) -- cycle ;
%Rounded Rect [id:dp8756393167051588] 
\draw  [fill={rgb, 255:red, 204; green, 230; blue, 255 }  ,fill opacity=1 ] (225.03,183.21) .. controls (225.03,180.09) and (227.55,177.57) .. (230.67,177.57) -- (238.2,177.57) .. controls (241.31,177.57) and (243.84,180.09) .. (243.84,183.21) -- (243.84,233.68) .. controls (243.84,236.79) and (241.31,239.32) .. (238.2,239.32) -- (230.67,239.32) .. controls (227.55,239.32) and (225.03,236.79) .. (225.03,233.68) -- cycle ;
%Rounded Rect [id:dp43642566820340334] 
\draw  [fill={rgb, 255:red, 204; green, 230; blue, 255 }  ,fill opacity=1 ] (258.9,158.04) .. controls (258.9,154.93) and (261.43,152.4) .. (264.54,152.4) -- (272.07,152.4) .. controls (275.19,152.4) and (277.71,154.93) .. (277.71,158.04) -- (277.71,208.51) .. controls (277.71,211.63) and (275.19,214.15) .. (272.07,214.15) -- (264.54,214.15) .. controls (261.43,214.15) and (258.9,211.63) .. (258.9,208.51) -- cycle ;
%Rounded Rect [id:dp6683963207159054] 
\draw  [fill={rgb, 255:red, 204; green, 230; blue, 255 }  ,fill opacity=1 ] (292.8,132.27) .. controls (292.8,129.15) and (295.33,126.62) .. (298.44,126.62) -- (305.97,126.62) .. controls (309.09,126.62) and (311.61,129.15) .. (311.61,132.27) -- (311.61,182.73) .. controls (311.61,185.85) and (309.09,188.37) .. (305.97,188.37) -- (298.44,188.37) .. controls (295.33,188.37) and (292.8,185.85) .. (292.8,182.73) -- cycle ;
%Rounded Rect [id:dp34180577027958226] 
\draw  [fill={rgb, 255:red, 204; green, 230; blue, 255 }  ,fill opacity=1 ] (326.7,106.34) .. controls (326.7,103.23) and (329.23,100.7) .. (332.34,100.7) -- (339.87,100.7) .. controls (342.99,100.7) and (345.51,103.23) .. (345.51,106.34) -- (345.51,156.81) .. controls (345.51,159.93) and (342.99,162.45) .. (339.87,162.45) -- (332.34,162.45) .. controls (329.23,162.45) and (326.7,159.93) .. (326.7,156.81) -- cycle ;
\end{mytikz2}\quad.
\end{equation}

\section{Additional Numerical Results and Technical Details}

In this section, we provide additional numerical results and technical details in the numerical experiments. In the experiment of gradient evaluation shown in Fig.~\textcolor{blue}{3} in the main text, we utilize the ladder ansatz like in Eq.~\eqref{eq:ladder_ansatz} with $\beta=2$. Each two-qubit block is parametrized by the Cartan decomposition as in Eq.~\eqref{eq:cartan}. The distance $\Delta k$ is just the difference between the indices of the differential block and the last block, i.e., $\Delta_{\mu j}$ in Eq.~\eqref{eq:var_ladder}. The legend in Fig.~\textcolor{blue}{3}(b) of $R_{y1}$, $R_{yy}$, and $R_{y2}$ refers to the following three red gates in the Cartan decomposition respectively.
\begin{equation}\label{eq:cartan_red}
\begin{mytikz2}
%Rounded Rect [id:dp08048917861594362] 
\draw  [fill={rgb, 255:red, 204; green, 230; blue, 255 }  ,fill opacity=0.35 ][dash pattern={on 3pt off 1.5pt}][line width=0.75]  (0.48,107.52) .. controls (0.48,97.84) and (8.32,90) .. (18,90) -- (542.56,90) .. controls (552.24,90) and (560.08,97.84) .. (560.08,107.52) -- (560.08,192.48) .. controls (560.08,202.16) and (552.24,210) .. (542.56,210) -- (18,210) .. controls (8.32,210) and (0.48,202.16) .. (0.48,192.48) -- cycle ;
%Straight Lines [id:da959319512549091] 
\draw [line width=0.75]    (-15,180) -- (570,180) ;
%Straight Lines [id:da5589761164366458] 
\draw [line width=0.75]    (-15,120) -- (570,120) ;
%Rounded Rect [id:dp37869621105417095] 
\draw  [fill={rgb, 255:red, 204; green, 230; blue, 255 }  ,fill opacity=1 ][line width=0.75]  (140,100.4) .. controls (140,100.4) and (140,100.4) .. (140,100.4) -- (180,100.4) .. controls (180,100.4) and (180,100.4) .. (180,100.4) -- (180,140.4) .. controls (180,140.4) and (180,140.4) .. (180,140.4) -- (140,140.4) .. controls (140,140.4) and (140,140.4) .. (140,140.4) -- cycle ;
%Rounded Rect [id:dp41813794368442037] 
\draw  [fill={rgb, 255:red, 204; green, 230; blue, 255 }  ,fill opacity=1 ][line width=0.75]  (140,160.4) .. controls (140,160.4) and (140,160.4) .. (140,160.4) -- (180,160.4) .. controls (180,160.4) and (180,160.4) .. (180,160.4) -- (180,200.4) .. controls (180,200.4) and (180,200.4) .. (180,200.4) -- (140,200.4) .. controls (140,200.4) and (140,200.4) .. (140,200.4) -- cycle ;
%Rounded Rect [id:dp13324445161386533] 
\draw  [fill={rgb, 255:red, 204; green, 230; blue, 255 }  ,fill opacity=1 ][line width=0.75]  (200,100) .. controls (200,100) and (200,100) .. (200,100) -- (240,100) .. controls (240,100) and (240,100) .. (240,100) -- (240,200) .. controls (240,200) and (240,200) .. (240,200) -- (200,200) .. controls (200,200) and (200,200) .. (200,200) -- cycle ;
%Rounded Rect [id:dp03188795811469913] 
\draw  [fill={rgb, 255:red, 204; green, 230; blue, 255 }  ,fill opacity=1 ][line width=0.75]  (380,100) .. controls (380,100) and (380,100) .. (380,100) -- (420,100) .. controls (420,100) and (420,100) .. (420,100) -- (420,140) .. controls (420,140) and (420,140) .. (420,140) -- (380,140) .. controls (380,140) and (380,140) .. (380,140) -- cycle ;
%Rounded Rect [id:dp17734087015399091] 
\draw  [fill={rgb, 255:red, 204; green, 230; blue, 255 }  ,fill opacity=1 ][line width=0.75]  (380,160) .. controls (380,160) and (380,160) .. (380,160) -- (420,160) .. controls (420,160) and (420,160) .. (420,160) -- (420,200) .. controls (420,200) and (420,200) .. (420,200) -- (380,200) .. controls (380,200) and (380,200) .. (380,200) -- cycle ;
%Rounded Rect [id:dp26146520500273596] 
\draw  [fill={rgb, 255:red, 255; green, 180; blue, 180 }  ,fill opacity=1 ][line width=0.75]  (260,100) .. controls (260,100) and (260,100) .. (260,100) -- (300,100) .. controls (300,100) and (300,100) .. (300,100) -- (300,200) .. controls (300,200) and (300,200) .. (300,200) -- (260,200) .. controls (260,200) and (260,200) .. (260,200) -- cycle ;
%Rounded Rect [id:dp9258529086066054] 
\draw  [fill={rgb, 255:red, 204; green, 230; blue, 255 }  ,fill opacity=1 ][line width=0.75]  (320,100) .. controls (320,100) and (320,100) .. (320,100) -- (360,100) .. controls (360,100) and (360,100) .. (360,100) -- (360,200) .. controls (360,200) and (360,200) .. (360,200) -- (320,200) .. controls (320,200) and (320,200) .. (320,200) -- cycle ;
%Rounded Rect [id:dp8187132833657995] 
\draw  [fill={rgb, 255:red, 255; green, 180; blue, 180 }  ,fill opacity=1 ][line width=0.75]  (80,160) .. controls (80,160) and (80,160) .. (80,160) -- (120,160) .. controls (120,160) and (120,160) .. (120,160) -- (120,200) .. controls (120,200) and (120,200) .. (120,200) -- (80,200) .. controls (80,200) and (80,200) .. (80,200) -- cycle ;
%Rounded Rect [id:dp035562334886791946] 
\draw  [fill={rgb, 255:red, 204; green, 230; blue, 255 }  ,fill opacity=1 ][line width=0.75]  (20,160) .. controls (20,160) and (20,160) .. (20,160) -- (60,160) .. controls (60,160) and (60,160) .. (60,160) -- (60,200) .. controls (60,200) and (60,200) .. (60,200) -- (20,200) .. controls (20,200) and (20,200) .. (20,200) -- cycle ;
%Rounded Rect [id:dp0865992091006087] 
\draw  [fill={rgb, 255:red, 204; green, 230; blue, 255 }  ,fill opacity=1 ][line width=0.75]  (80,100) .. controls (80,100) and (80,100) .. (80,100) -- (120,100) .. controls (120,100) and (120,100) .. (120,100) -- (120,140) .. controls (120,140) and (120,140) .. (120,140) -- (80,140) .. controls (80,140) and (80,140) .. (80,140) -- cycle ;
%Rounded Rect [id:dp560991869597459] 
\draw  [fill={rgb, 255:red, 204; green, 230; blue, 255 }  ,fill opacity=1 ][line width=0.75]  (20,100) .. controls (20,100) and (20,100) .. (20,100) -- (60,100) .. controls (60,100) and (60,100) .. (60,100) -- (60,140) .. controls (60,140) and (60,140) .. (60,140) -- (20,140) .. controls (20,140) and (20,140) .. (20,140) -- cycle ;
%Rounded Rect [id:dp6744513704866726] 
\draw  [fill={rgb, 255:red, 204; green, 230; blue, 255 }  ,fill opacity=1 ][line width=0.75]  (500,100.4) .. controls (500,100.4) and (500,100.4) .. (500,100.4) -- (540,100.4) .. controls (540,100.4) and (540,100.4) .. (540,100.4) -- (540,140.4) .. controls (540,140.4) and (540,140.4) .. (540,140.4) -- (500,140.4) .. controls (500,140.4) and (500,140.4) .. (500,140.4) -- cycle ;
%Rounded Rect [id:dp30560389994213466] 
\draw  [fill={rgb, 255:red, 204; green, 230; blue, 255 }  ,fill opacity=1 ][line width=0.75]  (440,100) .. controls (440,100) and (440,100) .. (440,100) -- (480,100) .. controls (480,100) and (480,100) .. (480,100) -- (480,140) .. controls (480,140) and (480,140) .. (480,140) -- (440,140) .. controls (440,140) and (440,140) .. (440,140) -- cycle ;
%Rounded Rect [id:dp8977278905089574] 
\draw  [fill={rgb, 255:red, 204; green, 230; blue, 255 }  ,fill opacity=1 ][line width=0.75]  (500,160.4) .. controls (500,160.4) and (500,160.4) .. (500,160.4) -- (540,160.4) .. controls (540,160.4) and (540,160.4) .. (540,160.4) -- (540,200.4) .. controls (540,200.4) and (540,200.4) .. (540,200.4) -- (500,200.4) .. controls (500,200.4) and (500,200.4) .. (500,200.4) -- cycle ;
%Rounded Rect [id:dp2799368492205201] 
\draw  [fill={rgb, 255:red, 255; green, 180; blue, 180 }  ,fill opacity=1 ][line width=0.75]  (440,160) .. controls (440,160) and (440,160) .. (440,160) -- (480,160) .. controls (480,160) and (480,160) .. (480,160) -- (480,200) .. controls (480,200) and (480,200) .. (480,200) -- (440,200) .. controls (440,200) and (440,200) .. (440,200) -- cycle ;

% Text Node
\draw (145,110) node [anchor=north west][inner sep=0.75pt]  [font=\normalsize] [align=left] {$R_{z}$};
% Text Node
\draw (145,170) node [anchor=north west][inner sep=0.75pt]  [font=\normalsize] [align=left] {$R_{z}$};
% Text Node
\draw (197,140) node [anchor=north west][inner sep=0.75pt]  [font=\normalsize] [align=left] {$R_{xx}$};
% Text Node
\draw (385,110) node [anchor=north west][inner sep=0.75pt]  [font=\normalsize] [align=left] {$R_{z}$};
% Text Node
\draw (385,170) node [anchor=north west][inner sep=0.75pt]  [font=\normalsize] [align=left] {$R_{z}$};
% Text Node
\draw (257,140) node [anchor=north west][inner sep=0.75pt]  [font=\normalsize] [align=left] {$R_{yy}$};
% Text Node
\draw (317,140) node [anchor=north west][inner sep=0.75pt]  [font=\normalsize] [align=left] {$R_{zz}$};
% Text Node
\draw (85,170) node [anchor=north west][inner sep=0.75pt]  [font=\normalsize] [align=left] {$R_{y}$};
% Text Node
\draw (25,170) node [anchor=north west][inner sep=0.75pt]  [font=\normalsize] [align=left] {$R_{z}$};
% Text Node
\draw (85,110) node [anchor=north west][inner sep=0.75pt]  [font=\normalsize] [align=left] {$R_{y}$};
% Text Node
\draw (25,110) node [anchor=north west][inner sep=0.75pt]  [font=\normalsize] [align=left] {$R_{z}$};
% Text Node
\draw (505,110) node [anchor=north west][inner sep=0.75pt]  [font=\normalsize] [align=left] {$R_{z}$};
% Text Node
\draw (445,110) node [anchor=north west][inner sep=0.75pt]  [font=\normalsize] [align=left] {$R_{y}$};
% Text Node
\draw (505,170) node [anchor=north west][inner sep=0.75pt]  [font=\normalsize] [align=left] {$R_{z}$};
% Text Node
\draw (445,170) node [anchor=north west][inner sep=0.75pt]  [font=\normalsize] [align=left] {$R_{y}$};
\end{mytikz2}\quad.
\end{equation}
In Fig.~\textcolor{blue}{3}(a), we just fix the location of the differential gate within the differential block as the red block of $R_{yy}$ above. Recall that $R_{xx}$, $R_{yy}$ and $R_{zz}$ denote the two-qubit rotation gates with the generators $X\otimes X$, $Y\otimes Y$ and $Z\otimes Z$ respectively as mentioned around Eq.~\eqref{eq:cartan}. $R_{y}$ and $R_{z}$ denote the single-qubit rotation gates with the generators $Y$ and $Z$ respectively.

In the experiment of VQE on the 2D toric code model shown in Fig.~\textcolor{blue}{4} in the main text, the model Hamiltonian reads
\begin{equation}
    H = (1 - h) \hat{H}_{0} - \sum_{j=1}^{N} \left(h^x X_{j} + h^y Y_{j} + h^z Z_{j}\right).
\end{equation}
The additional factor $(1-h)$ arises for the sake of properly approaching the infinite large field limit. In the two cases on which we focus below, we set either $(h^x, h^y, h^z)=(h,0,h)$ or $(h^x, h^y, h^z)=(0,0,h)$. Please do not confuse $(h^x, h^y, h^z)$ with the notation $h_j$ representing the Hamiltonian sub-terms above. The bare toric code Hamiltonian $H_0$ reads
\begin{equation}
    H_0 = - \sum_v A_v - \sum_p B_p,
\end{equation}
where $v$ and $p$ runs over all the vertices and plaquettes of the 2D square lattice. The operators $A_v$ and $B_p$ are defined as the products of four Pauli operators depicted below
\begin{equation}
\begin{mytikz2}
%Straight Lines [id:da5534412703042102] 
\draw [color={rgb, 255:red, 0; green, 0; blue, 0 }  ,draw opacity=1 ][line width=0.75]    (201.17,106.74) -- (100.3,106.74) ;
\draw [shift={(100.3,106.74)}, rotate = 180] [color={rgb, 255:red, 0; green, 0; blue, 0 }  ,draw opacity=1 ][fill={rgb, 255:red, 0; green, 0; blue, 0 }  ,fill opacity=1 ][line width=0.75]      (0, 0) circle [x radius= 1.74, y radius= 1.74]   ;
%Straight Lines [id:da9046781269612811] 
\draw [color={rgb, 255:red, 0; green, 0; blue, 0 }  ,draw opacity=1 ][line width=0.75]    (100.3,207.62) -- (100.3,106.74) ;
%Straight Lines [id:da955642535617528] 
\draw [color={rgb, 255:red, 0; green, 0; blue, 0 }  ,draw opacity=1 ][line width=0.75]    (302.05,106.74) -- (201.17,106.74) ;
\draw [shift={(201.17,106.74)}, rotate = 180] [color={rgb, 255:red, 0; green, 0; blue, 0 }  ,draw opacity=1 ][fill={rgb, 255:red, 0; green, 0; blue, 0 }  ,fill opacity=1 ][line width=0.75]      (0, 0) circle [x radius= 1.74, y radius= 1.74]   ;
%Straight Lines [id:da1667139798260897] 
\draw [color={rgb, 255:red, 0; green, 0; blue, 0 }  ,draw opacity=1 ][line width=0.75]    (201.17,207.62) -- (201.17,106.74) ;
%Straight Lines [id:da1194089355863337] 
\draw [color={rgb, 255:red, 0; green, 0; blue, 0 }  ,draw opacity=1 ][line width=0.75]    (398.21,106.74) -- (297.34,106.74) ;
\draw [shift={(297.34,106.74)}, rotate = 180] [color={rgb, 255:red, 0; green, 0; blue, 0 }  ,draw opacity=1 ][fill={rgb, 255:red, 0; green, 0; blue, 0 }  ,fill opacity=1 ][line width=0.75]      (0, 0) circle [x radius= 1.74, y radius= 1.74]   ;
%Straight Lines [id:da31938629131587337] 
\draw [color={rgb, 255:red, 0; green, 0; blue, 0 }  ,draw opacity=1 ][line width=0.75]    (297.34,207.62) -- (297.34,106.74) ;
%Straight Lines [id:da4400128423991456] 
\draw [color={rgb, 255:red, 0; green, 0; blue, 0 }  ,draw opacity=1 ][line width=0.75]    (399.56,207.62) -- (399.56,106.74) ;
\draw [shift={(399.56,106.74)}, rotate = 270] [color={rgb, 255:red, 0; green, 0; blue, 0 }  ,draw opacity=1 ][fill={rgb, 255:red, 0; green, 0; blue, 0 }  ,fill opacity=1 ][line width=0.75]      (0, 0) circle [x radius= 1.74, y radius= 1.74]   ;
%Straight Lines [id:da4024957746405886] 
\draw [color={rgb, 255:red, 0; green, 0; blue, 0 }  ,draw opacity=1 ][line width=0.75]    (201.17,207.62) -- (100.3,207.62) ;
\draw [shift={(100.3,207.62)}, rotate = 180] [color={rgb, 255:red, 0; green, 0; blue, 0 }  ,draw opacity=1 ][fill={rgb, 255:red, 0; green, 0; blue, 0 }  ,fill opacity=1 ][line width=0.75]      (0, 0) circle [x radius= 1.74, y radius= 1.74]   ;
%Straight Lines [id:da3974222475956555] 
\draw [color={rgb, 255:red, 0; green, 0; blue, 0 }  ,draw opacity=1 ][line width=0.75]    (100.3,308.49) -- (100.3,207.62) ;
%Straight Lines [id:da14399363398892273] 
\draw [color={rgb, 255:red, 0; green, 0; blue, 0 }  ,draw opacity=1 ][line width=0.75]    (302.05,207.62) -- (201.17,207.62) ;
\draw [shift={(201.17,207.62)}, rotate = 180] [color={rgb, 255:red, 0; green, 0; blue, 0 }  ,draw opacity=1 ][fill={rgb, 255:red, 0; green, 0; blue, 0 }  ,fill opacity=1 ][line width=0.75]      (0, 0) circle [x radius= 1.74, y radius= 1.74]   ;
%Straight Lines [id:da45033159207648366] 
\draw [color={rgb, 255:red, 0; green, 0; blue, 0 }  ,draw opacity=1 ][line width=0.75]    (201.17,308.49) -- (201.17,207.62) ;
%Straight Lines [id:da05371252928326897] 
\draw [color={rgb, 255:red, 0; green, 0; blue, 0 }  ,draw opacity=1 ][line width=0.75]    (398.21,207.62) -- (297.34,207.62) ;
\draw [shift={(297.34,207.62)}, rotate = 180] [color={rgb, 255:red, 0; green, 0; blue, 0 }  ,draw opacity=1 ][fill={rgb, 255:red, 0; green, 0; blue, 0 }  ,fill opacity=1 ][line width=0.75]      (0, 0) circle [x radius= 1.74, y radius= 1.74]   ;
%Straight Lines [id:da4844938528531313] 
\draw [color={rgb, 255:red, 0; green, 0; blue, 0 }  ,draw opacity=1 ][line width=0.75]    (297.34,308.49) -- (297.34,207.62) ;
%Straight Lines [id:da20042565583015048] 
\draw [color={rgb, 255:red, 0; green, 0; blue, 0 }  ,draw opacity=1 ][line width=0.75]    (399.56,308.49) -- (399.56,207.62) ;
\draw [shift={(399.56,207.62)}, rotate = 270] [color={rgb, 255:red, 0; green, 0; blue, 0 }  ,draw opacity=1 ][fill={rgb, 255:red, 0; green, 0; blue, 0 }  ,fill opacity=1 ][line width=0.75]      (0, 0) circle [x radius= 1.74, y radius= 1.74]   ;
%Straight Lines [id:da917193251439758] 
\draw [color={rgb, 255:red, 0; green, 0; blue, 0 }  ,draw opacity=1 ][line width=0.75]    (201.17,303.11) -- (100.3,303.11) ;
\draw [shift={(100.3,303.11)}, rotate = 180] [color={rgb, 255:red, 0; green, 0; blue, 0 }  ,draw opacity=1 ][fill={rgb, 255:red, 0; green, 0; blue, 0 }  ,fill opacity=1 ][line width=0.75]      (0, 0) circle [x radius= 1.74, y radius= 1.74]   ;
%Straight Lines [id:da44823606835734875] 
\draw [color={rgb, 255:red, 0; green, 0; blue, 0 }  ,draw opacity=1 ][line width=0.75]    (100.3,403.98) -- (100.3,303.11) ;
%Straight Lines [id:da019841982704100758] 
\draw [color={rgb, 255:red, 0; green, 0; blue, 0 }  ,draw opacity=1 ][line width=0.75]    (302.05,303.11) -- (201.17,303.11) ;
\draw [shift={(201.17,303.11)}, rotate = 180] [color={rgb, 255:red, 0; green, 0; blue, 0 }  ,draw opacity=1 ][fill={rgb, 255:red, 0; green, 0; blue, 0 }  ,fill opacity=1 ][line width=0.75]      (0, 0) circle [x radius= 1.74, y radius= 1.74]   ;
%Straight Lines [id:da7862777406744552] 
\draw [color={rgb, 255:red, 0; green, 0; blue, 0 }  ,draw opacity=1 ][line width=0.75]    (201.17,403.98) -- (201.17,303.11) ;
%Straight Lines [id:da5416541429960717] 
\draw [color={rgb, 255:red, 0; green, 0; blue, 0 }  ,draw opacity=1 ][line width=0.75]    (398.21,303.11) -- (297.34,303.11) ;
\draw [shift={(297.34,303.11)}, rotate = 180] [color={rgb, 255:red, 0; green, 0; blue, 0 }  ,draw opacity=1 ][fill={rgb, 255:red, 0; green, 0; blue, 0 }  ,fill opacity=1 ][line width=0.75]      (0, 0) circle [x radius= 1.74, y radius= 1.74]   ;
%Straight Lines [id:da543387029778271] 
\draw [color={rgb, 255:red, 0; green, 0; blue, 0 }  ,draw opacity=1 ][line width=0.75]    (297.34,403.98) -- (297.34,303.11) ;
%Straight Lines [id:da16365947217701438] 
\draw [color={rgb, 255:red, 0; green, 0; blue, 0 }  ,draw opacity=1 ][line width=0.75]    (399.56,403.98) -- (399.56,303.11) ;
\draw [shift={(399.56,303.11)}, rotate = 270] [color={rgb, 255:red, 0; green, 0; blue, 0 }  ,draw opacity=1 ][fill={rgb, 255:red, 0; green, 0; blue, 0 }  ,fill opacity=1 ][line width=0.75]      (0, 0) circle [x radius= 1.74, y radius= 1.74]   ;
%Straight Lines [id:da26504564448840884] 
\draw [color={rgb, 255:red, 0; green, 0; blue, 0 }  ,draw opacity=1 ][line width=0.75]    (201.17,402.29) -- (100.3,402.29) ;
\draw [shift={(100.3,402.29)}, rotate = 180] [color={rgb, 255:red, 0; green, 0; blue, 0 }  ,draw opacity=1 ][fill={rgb, 255:red, 0; green, 0; blue, 0 }  ,fill opacity=1 ][line width=0.75]      (0, 0) circle [x radius= 1.74, y radius= 1.74]   ;
%Straight Lines [id:da21029955313985793] 
\draw [color={rgb, 255:red, 0; green, 0; blue, 0 }  ,draw opacity=1 ][line width=0.75]    (302.05,402.29) -- (201.17,402.29) ;
\draw [shift={(201.17,402.29)}, rotate = 180] [color={rgb, 255:red, 0; green, 0; blue, 0 }  ,draw opacity=1 ][fill={rgb, 255:red, 0; green, 0; blue, 0 }  ,fill opacity=1 ][line width=0.75]      (0, 0) circle [x radius= 1.74, y radius= 1.74]   ;
%Straight Lines [id:da3505377991547953] 
\draw [color={rgb, 255:red, 0; green, 0; blue, 0 }  ,draw opacity=1 ][line width=0.75]    (399.84,402.29) -- (297.34,402.29) ;
\draw [shift={(297.34,402.29)}, rotate = 180] [color={rgb, 255:red, 0; green, 0; blue, 0 }  ,draw opacity=1 ][fill={rgb, 255:red, 0; green, 0; blue, 0 }  ,fill opacity=1 ][line width=0.75]      (0, 0) circle [x radius= 1.74, y radius= 1.74]   ;
\draw [shift={(399.84,402.29)}, rotate = 180] [color={rgb, 255:red, 0; green, 0; blue, 0 }  ,draw opacity=1 ][fill={rgb, 255:red, 0; green, 0; blue, 0 }  ,fill opacity=1 ][line width=0.75]      (0, 0) circle [x radius= 1.74, y radius= 1.74]   ;
%Shape: Ellipse [id:dp15145854803025705] 
\draw  [color={rgb, 255:red, 208; green, 2; blue, 27 }  ,draw opacity=1 ][fill={rgb, 255:red, 255; green, 255; blue, 255 }  ,fill opacity=1 ][line width=1.5]  (138.47,207.62) .. controls (138.47,201.06) and (143.78,195.75) .. (150.34,195.75) .. controls (156.89,195.75) and (162.21,201.06) .. (162.21,207.62) .. controls (162.21,214.17) and (156.89,219.48) .. (150.34,219.48) .. controls (143.78,219.48) and (138.47,214.17) .. (138.47,207.62) -- cycle ;
%Shape: Ellipse [id:dp5104332476064974] 
\draw  [color={rgb, 255:red, 208; green, 2; blue, 27 }  ,draw opacity=1 ][fill={rgb, 255:red, 255; green, 255; blue, 255 }  ,fill opacity=1 ][line width=1.5]  (189.31,157.18) .. controls (189.31,150.62) and (194.62,145.31) .. (201.17,145.31) .. controls (207.73,145.31) and (213.04,150.62) .. (213.04,157.18) .. controls (213.04,163.73) and (207.73,169.05) .. (201.17,169.05) .. controls (194.62,169.05) and (189.31,163.73) .. (189.31,157.18) -- cycle ;
%Shape: Ellipse [id:dp1404540351864505] 
\draw  [color={rgb, 255:red, 208; green, 2; blue, 27 }  ,draw opacity=1 ][fill={rgb, 255:red, 255; green, 255; blue, 255 }  ,fill opacity=1 ][line width=1.5]  (239.74,207.62) .. controls (239.74,201.06) and (245.06,195.75) .. (251.61,195.75) .. controls (258.17,195.75) and (263.48,201.06) .. (263.48,207.62) .. controls (263.48,214.17) and (258.17,219.48) .. (251.61,219.48) .. controls (245.06,219.48) and (239.74,214.17) .. (239.74,207.62) -- cycle ;
%Shape: Ellipse [id:dp9109891990261563] 
\draw  [color={rgb, 255:red, 208; green, 2; blue, 27 }  ,draw opacity=1 ][fill={rgb, 255:red, 255; green, 255; blue, 255 }  ,fill opacity=1 ][line width=1.5]  (189.31,258.05) .. controls (189.31,251.5) and (194.62,246.18) .. (201.17,246.18) .. controls (207.73,246.18) and (213.04,251.5) .. (213.04,258.05) .. controls (213.04,264.61) and (207.73,269.92) .. (201.17,269.92) .. controls (194.62,269.92) and (189.31,264.61) .. (189.31,258.05) -- cycle ;
%Shape: Ellipse [id:dp2555232698428074] 
\draw  [color={rgb, 255:red, 74; green, 144; blue, 226 }  ,draw opacity=1 ][fill={rgb, 255:red, 255; green, 255; blue, 255 }  ,fill opacity=1 ][line width=1.5]  (285.47,353.55) .. controls (285.47,346.99) and (290.79,341.68) .. (297.34,341.68) .. controls (303.9,341.68) and (309.21,346.99) .. (309.21,353.55) .. controls (309.21,360.1) and (303.9,365.41) .. (297.34,365.41) .. controls (290.79,365.41) and (285.47,360.1) .. (285.47,353.55) -- cycle ;
%Shape: Circle [id:dp9549392890946606] 
\draw  [color={rgb, 255:red, 74; green, 144; blue, 226 }  ,draw opacity=1 ][fill={rgb, 255:red, 255; green, 255; blue, 255 }  ,fill opacity=1 ][line width=1.5]  (335.91,401.29) .. controls (335.91,394.74) and (341.22,389.42) .. (347.78,389.42) .. controls (354.33,389.42) and (359.65,394.74) .. (359.65,401.29) .. controls (359.65,407.85) and (354.33,413.16) .. (347.78,413.16) .. controls (341.22,413.16) and (335.91,407.85) .. (335.91,401.29) -- cycle ;
%Shape: Ellipse [id:dp8696138167293255] 
\draw  [color={rgb, 255:red, 74; green, 144; blue, 226 }  ,draw opacity=1 ][fill={rgb, 255:red, 255; green, 255; blue, 255 }  ,fill opacity=1 ][line width=1.5]  (387.69,353.55) .. controls (387.69,346.99) and (393,341.68) .. (399.56,341.68) .. controls (406.11,341.68) and (411.43,346.99) .. (411.43,353.55) .. controls (411.43,360.1) and (406.11,365.41) .. (399.56,365.41) .. controls (393,365.41) and (387.69,360.1) .. (387.69,353.55) -- cycle ;
%Shape: Circle [id:dp4581032244890477] 
\draw  [color={rgb, 255:red, 74; green, 144; blue, 226 }  ,draw opacity=1 ][fill={rgb, 255:red, 255; green, 255; blue, 255 }  ,fill opacity=1 ][line width=1.5]  (335.91,303.11) .. controls (335.91,296.55) and (341.22,291.24) .. (347.78,291.24) .. controls (354.33,291.24) and (359.65,296.55) .. (359.65,303.11) .. controls (359.65,309.66) and (354.33,314.98) .. (347.78,314.98) .. controls (341.22,314.98) and (335.91,309.66) .. (335.91,303.11) -- cycle ;
%Shape: Ellipse [id:dp44012926307431255] 
\draw  [color={rgb, 255:red, 208; green, 2; blue, 27 }  ,draw opacity=1 ][fill={rgb, 255:red, 255; green, 255; blue, 255 }  ,fill opacity=1 ][line width=1.5]  (335.91,106.74) .. controls (335.91,100.19) and (341.22,94.87) .. (347.78,94.87) .. controls (354.33,94.87) and (359.65,100.19) .. (359.65,106.74) .. controls (359.65,113.3) and (354.33,118.61) .. (347.78,118.61) .. controls (341.22,118.61) and (335.91,113.3) .. (335.91,106.74) -- cycle ;
%Shape: Ellipse [id:dp7649014919199459] 
\draw  [color={rgb, 255:red, 208; green, 2; blue, 27 }  ,draw opacity=1 ][fill={rgb, 255:red, 255; green, 255; blue, 255 }  ,fill opacity=1 ][line width=1.5]  (387.69,157.18) .. controls (387.69,150.62) and (393,145.31) .. (399.56,145.31) .. controls (406.11,145.31) and (411.43,150.62) .. (411.43,157.18) .. controls (411.43,163.73) and (406.11,169.05) .. (399.56,169.05) .. controls (393,169.05) and (387.69,163.73) .. (387.69,157.18) -- cycle ;

% Text Node
\draw (192,149) node [anchor=north west][inner sep=0.75pt]  [font=\scriptsize] [align=left] {$Z$};
% Text Node
\draw (242,199) node [anchor=north west][inner sep=0.75pt]  [font=\scriptsize] [align=left] {$Z$};
% Text Node
\draw (192,250) node [anchor=north west][inner sep=0.75pt]  [font=\scriptsize] [align=left] {$Z$};
% Text Node
\draw (140,199) node [anchor=north west][inner sep=0.75pt]  [font=\scriptsize] [align=left] {$Z$};
% Text Node
\draw (287,345.6) node [anchor=north west][inner sep=0.75pt]  [font=\scriptsize] [align=left] {$X$};
% Text Node
\draw (337,393.37) node [anchor=north west][inner sep=0.75pt]  [font=\scriptsize] [align=left] {$X$};
% Text Node
\draw (389,345.9) node [anchor=north west][inner sep=0.75pt]  [font=\scriptsize] [align=left] {$X$};
% Text Node
\draw (337,295.23) node [anchor=north west][inner sep=0.75pt]  [font=\scriptsize] [align=left] {$X$};
% Text Node
\draw (335,343.44) node [anchor=north west][inner sep=0.75pt]  [font=\normalsize] [align=left] {$\textcolor[rgb]{0.16,0.48,0.85}{B_{p}}$};
% Text Node
\draw (167,175.46) node [anchor=north west][inner sep=0.75pt]  [font=\normalsize] [align=left] {$\textcolor[rgb]{0.82,0.01,0.11}{A_{v}}$};
% Text Node
\draw (338,98.64) node [anchor=north west][inner sep=0.75pt]  [font=\scriptsize] [align=left] {$Z$};
% Text Node
\draw (389,148.85) node [anchor=north west][inner sep=0.75pt]  [font=\scriptsize] [align=left] {$Z$};
% Text Node
\draw (358,118.57) node [anchor=north west][inner sep=0.75pt]  [font=\normalsize] [align=left] {$\textcolor[rgb]{0.82,0.01,0.11}{A_{v'}}$};
\end{mytikz2}\quad.
\end{equation}
Note that under open boundary conditions, the vertex operator $A_{v'}$ at the boundary should be slightly modified to the products of two or three Pauli operators as depicted above. The numerical experiments are conducted on a $3\times 3$ lattice with $N=12$ qubits defined on the edges. The ground state of $H_0$ can be exactly solved and is topologically ordered, which is signified by the non-vanishing topological entanglement entropy, or say tripartite mutual information $S_{\text{topo}}(A:B:C)$ defined by
\begin{equation}
    S_{\text{topo}}(A:B:C) = S_A + S_B + S_C - S_{AB} - S_{BC} - S_{CA} + S_{ABC}.
\end{equation}
where $S_A=-\tr(\rho_A\ln\rho_A)$ represents the entanglement entropy of the subsystem $A$ with $\rho_A$ being the reduced density matrix on $A$. The sub-regions $A$, $B$ and $C$ used in our experiment are chosen as the following colored qubits respectively
\begin{equation}
\begin{mytikz2}
%Straight Lines [id:da16997653324613604] 
\draw [color={rgb, 255:red, 0; green, 0; blue, 0 }  ,draw opacity=1 ][line width=0.75]    (302.05,207.62) -- (201.17,207.62) ;
\draw [shift={(201.17,207.62)}, rotate = 180] [color={rgb, 255:red, 0; green, 0; blue, 0 }  ,draw opacity=1 ][fill={rgb, 255:red, 0; green, 0; blue, 0 }  ,fill opacity=1 ][line width=0.75]      (0, 0) circle [x radius= 1.34, y radius= 1.34]   ;
%Straight Lines [id:da4883186411236118] 
\draw [color={rgb, 255:red, 0; green, 0; blue, 0 }  ,draw opacity=1 ][line width=0.75]    (201.17,308.49) -- (201.17,207.62) ;
%Straight Lines [id:da20217137283235154] 
\draw [color={rgb, 255:red, 0; green, 0; blue, 0 }  ,draw opacity=1 ][line width=0.75]    (398.21,207.62) -- (297.34,207.62) ;
\draw [shift={(297.34,207.62)}, rotate = 180] [color={rgb, 255:red, 0; green, 0; blue, 0 }  ,draw opacity=1 ][fill={rgb, 255:red, 0; green, 0; blue, 0 }  ,fill opacity=1 ][line width=0.75]      (0, 0) circle [x radius= 1.34, y radius= 1.34]   ;
%Straight Lines [id:da9129869525683583] 
\draw [color={rgb, 255:red, 0; green, 0; blue, 0 }  ,draw opacity=1 ][line width=0.75]    (297.34,308.49) -- (297.34,207.62) ;
%Straight Lines [id:da8348372056198812] 
\draw [color={rgb, 255:red, 0; green, 0; blue, 0 }  ,draw opacity=1 ][line width=0.75]    (399.56,308.49) -- (399.56,207.62) ;
\draw [shift={(399.56,207.62)}, rotate = 270] [color={rgb, 255:red, 0; green, 0; blue, 0 }  ,draw opacity=1 ][fill={rgb, 255:red, 0; green, 0; blue, 0 }  ,fill opacity=1 ][line width=0.75]      (0, 0) circle [x radius= 1.34, y radius= 1.34]   ;
%Straight Lines [id:da22474670667629515] 
\draw [color={rgb, 255:red, 0; green, 0; blue, 0 }  ,draw opacity=1 ][line width=0.75]    (302.05,303.11) -- (201.17,303.11) ;
\draw [shift={(201.17,303.11)}, rotate = 180] [color={rgb, 255:red, 0; green, 0; blue, 0 }  ,draw opacity=1 ][fill={rgb, 255:red, 0; green, 0; blue, 0 }  ,fill opacity=1 ][line width=0.75]      (0, 0) circle [x radius= 1.34, y radius= 1.34]   ;
%Straight Lines [id:da6397171392929024] 
\draw [color={rgb, 255:red, 0; green, 0; blue, 0 }  ,draw opacity=1 ][line width=0.75]    (201.17,403.98) -- (201.17,303.11) ;
%Straight Lines [id:da9779978245858085] 
\draw [color={rgb, 255:red, 0; green, 0; blue, 0 }  ,draw opacity=1 ][line width=0.75]    (398.21,303.11) -- (297.34,303.11) ;
\draw [shift={(297.34,303.11)}, rotate = 180] [color={rgb, 255:red, 0; green, 0; blue, 0 }  ,draw opacity=1 ][fill={rgb, 255:red, 0; green, 0; blue, 0 }  ,fill opacity=1 ][line width=0.75]      (0, 0) circle [x radius= 1.34, y radius= 1.34]   ;
%Straight Lines [id:da37741627964151436] 
\draw [color={rgb, 255:red, 0; green, 0; blue, 0 }  ,draw opacity=1 ][line width=0.75]    (297.34,403.98) -- (297.34,303.11) ;
%Straight Lines [id:da3585437243045979] 
\draw [color={rgb, 255:red, 0; green, 0; blue, 0 }  ,draw opacity=1 ][line width=0.75]    (399.56,403.98) -- (399.56,303.11) ;
\draw [shift={(399.56,303.11)}, rotate = 270] [color={rgb, 255:red, 0; green, 0; blue, 0 }  ,draw opacity=1 ][fill={rgb, 255:red, 0; green, 0; blue, 0 }  ,fill opacity=1 ][line width=0.75]      (0, 0) circle [x radius= 1.34, y radius= 1.34]   ;
%Straight Lines [id:da6866229272296165] 
\draw [color={rgb, 255:red, 0; green, 0; blue, 0 }  ,draw opacity=1 ][line width=0.75]    (302.05,402.29) -- (201.17,402.29) ;
\draw [shift={(201.17,402.29)}, rotate = 180] [color={rgb, 255:red, 0; green, 0; blue, 0 }  ,draw opacity=1 ][fill={rgb, 255:red, 0; green, 0; blue, 0 }  ,fill opacity=1 ][line width=0.75]      (0, 0) circle [x radius= 1.34, y radius= 1.34]   ;
%Straight Lines [id:da6116914044118027] 
\draw [color={rgb, 255:red, 0; green, 0; blue, 0 }  ,draw opacity=1 ][line width=0.75]    (399.84,402.29) -- (297.34,402.29) ;
\draw [shift={(297.34,402.29)}, rotate = 180] [color={rgb, 255:red, 0; green, 0; blue, 0 }  ,draw opacity=1 ][fill={rgb, 255:red, 0; green, 0; blue, 0 }  ,fill opacity=1 ][line width=0.75]      (0, 0) circle [x radius= 1.34, y radius= 1.34]   ;
\draw [shift={(399.84,402.29)}, rotate = 180] [color={rgb, 255:red, 0; green, 0; blue, 0 }  ,draw opacity=1 ][fill={rgb, 255:red, 0; green, 0; blue, 0 }  ,fill opacity=1 ][line width=0.75]      (0, 0) circle [x radius= 1.34, y radius= 1.34]   ;

%Shape: Ellipse [id:dp1404540351864505] 
\draw  [color={rgb, 255:red, 0; green, 0; blue, 0 }  ,draw opacity=1 ][fill={rgb, 255:red, 255; green, 255; blue, 255 }  ,fill opacity=1 ][line width=0.75]  (239.74,207.62) .. controls (239.74,201.06) and (245.06,195.75) .. (251.61,195.75) .. controls (258.17,195.75) and (263.48,201.06) .. (263.48,207.62) .. controls (263.48,214.17) and (258.17,219.48) .. (251.61,219.48) .. controls (245.06,219.48) and (239.74,214.17) .. (239.74,207.62) -- cycle ;
%Shape: Ellipse [id:dp9109891990261563] 
\draw  [color={rgb, 255:red, 0; green, 0; blue, 0 }  ,draw opacity=1 ][fill={rgb, 255:red, 255; green, 255; blue, 255 }  ,fill opacity=1 ][line width=0.75]  (189.31,258.05) .. controls (189.31,251.5) and (194.62,246.18) .. (201.17,246.18) .. controls (207.73,246.18) and (213.04,251.5) .. (213.04,258.05) .. controls (213.04,264.61) and (207.73,269.92) .. (201.17,269.92) .. controls (194.62,269.92) and (189.31,264.61) .. (189.31,258.05) -- cycle ;
%Shape: Ellipse [id:dp2555232698428074] 
\draw  [color={rgb, 255:red, 0; green, 0; blue, 0 }  ,draw opacity=1 ][fill={rgb, 255:red, 204; green, 230; blue, 255 }  ,fill opacity=1 ][line width=0.75]  (285.47,353.55) .. controls (285.47,346.99) and (290.79,341.68) .. (297.34,341.68) .. controls (303.9,341.68) and (309.21,346.99) .. (309.21,353.55) .. controls (309.21,360.1) and (303.9,365.41) .. (297.34,365.41) .. controls (290.79,365.41) and (285.47,360.1) .. (285.47,353.55) -- cycle ;
%Shape: Circle [id:dp9549392890946606] 
\draw  [color={rgb, 255:red, 0; green, 0; blue, 0 }  ,draw opacity=1 ][fill={rgb, 255:red, 255; green, 255; blue, 255 }  ,fill opacity=1 ][line width=0.75]  (335.91,401.29) .. controls (335.91,394.74) and (341.22,389.42) .. (347.78,389.42) .. controls (354.33,389.42) and (359.65,394.74) .. (359.65,401.29) .. controls (359.65,407.85) and (354.33,413.16) .. (347.78,413.16) .. controls (341.22,413.16) and (335.91,407.85) .. (335.91,401.29) -- cycle ;
%Shape: Ellipse [id:dp8696138167293255] 
\draw  [color={rgb, 255:red, 0; green, 0; blue, 0 }  ,draw opacity=1 ][fill={rgb, 255:red, 255; green, 255; blue, 255 }  ,fill opacity=1 ][line width=0.75]  (387.69,353.55) .. controls (387.69,346.99) and (393,341.68) .. (399.56,341.68) .. controls (406.11,341.68) and (411.43,346.99) .. (411.43,353.55) .. controls (411.43,360.1) and (406.11,365.41) .. (399.56,365.41) .. controls (393,365.41) and (387.69,360.1) .. (387.69,353.55) -- cycle ;
%Shape: Circle [id:dp4581032244890477] 
\draw  [color={rgb, 255:red, 0; green, 0; blue, 0 }  ,draw opacity=1 ][fill={rgb, 255:red, 207; green, 232; blue, 176 }  ,fill opacity=1 ][line width=0.75]  (335.91,303.11) .. controls (335.91,296.55) and (341.22,291.24) .. (347.78,291.24) .. controls (354.33,291.24) and (359.65,296.55) .. (359.65,303.11) .. controls (359.65,309.66) and (354.33,314.98) .. (347.78,314.98) .. controls (341.22,314.98) and (335.91,309.66) .. (335.91,303.11) -- cycle ;
%Shape: Circle [id:dp42370175293439183] 
\draw  [color={rgb, 255:red, 0; green, 0; blue, 0 }  ,draw opacity=1 ][fill={rgb, 255:red, 255; green, 255; blue, 255 }  ,fill opacity=1 ][line width=0.75]  (335.91,207.62) .. controls (335.91,201.06) and (341.22,195.75) .. (347.78,195.75) .. controls (354.33,195.75) and (359.65,201.06) .. (359.65,207.62) .. controls (359.65,214.17) and (354.33,219.48) .. (347.78,219.48) .. controls (341.22,219.48) and (335.91,214.17) .. (335.91,207.62) -- cycle ;
%Shape: Circle [id:dp03331677889576068] 
\draw  [color={rgb, 255:red, 0; green, 0; blue, 0 }  ,draw opacity=1 ][fill={rgb, 255:red, 255; green, 180; blue, 180 }  ,fill opacity=1 ][line width=0.75]  (285.47,258.05) .. controls (285.47,251.5) and (290.79,246.18) .. (297.34,246.18) .. controls (303.9,246.18) and (309.21,251.5) .. (309.21,258.05) .. controls (309.21,264.61) and (303.9,269.92) .. (297.34,269.92) .. controls (290.79,269.92) and (285.47,264.61) .. (285.47,258.05) -- cycle ;
%Shape: Circle [id:dp635261748427276] 
\draw  [color={rgb, 255:red, 0; green, 0; blue, 0 }  ,draw opacity=1 ][fill={rgb, 255:red, 255; green, 255; blue, 255 }  ,fill opacity=1 ][line width=0.75]  (387.69,258.05) .. controls (387.69,251.5) and (393,246.18) .. (399.56,246.18) .. controls (406.11,246.18) and (411.43,251.5) .. (411.43,258.05) .. controls (411.43,264.61) and (406.11,269.92) .. (399.56,269.92) .. controls (393,269.92) and (387.69,264.61) .. (387.69,258.05) -- cycle ;
%Shape: Circle [id:dp46613889221408344] 
\draw  [color={rgb, 255:red, 0; green, 0; blue, 0 }  ,draw opacity=1 ][fill={rgb, 255:red, 255; green, 180; blue, 180 }  ,fill opacity=1 ][line width=0.75]  (239.74,303.11) .. controls (239.74,296.55) and (245.06,291.24) .. (251.61,291.24) .. controls (258.17,291.24) and (263.48,296.55) .. (263.48,303.11) .. controls (263.48,309.66) and (258.17,314.98) .. (251.61,314.98) .. controls (245.06,314.98) and (239.74,309.66) .. (239.74,303.11) -- cycle ;
%Shape: Circle [id:dp8446430694291729] 
\draw  [color={rgb, 255:red, 0; green, 0; blue, 0 }  ,draw opacity=1 ][fill={rgb, 255:red, 255; green, 255; blue, 255 }  ,fill opacity=1 ][line width=0.75]  (189.31,353.55) .. controls (189.31,346.99) and (194.62,341.68) .. (201.17,341.68) .. controls (207.73,341.68) and (213.04,346.99) .. (213.04,353.55) .. controls (213.04,360.1) and (207.73,365.41) .. (201.17,365.41) .. controls (194.62,365.41) and (189.31,360.1) .. (189.31,353.55) -- cycle ;
%Shape: Circle [id:dp7247644904905148] 
\draw  [color={rgb, 255:red, 0; green, 0; blue, 0 }  ,draw opacity=1 ][fill={rgb, 255:red, 255; green, 255; blue, 255 }  ,fill opacity=1 ][line width=0.75]  (239.74,402.29) .. controls (239.74,395.74) and (245.06,390.42) .. (251.61,390.42) .. controls (258.17,390.42) and (263.48,395.74) .. (263.48,402.29) .. controls (263.48,408.85) and (258.17,414.16) .. (251.61,414.16) .. controls (245.06,414.16) and (239.74,408.85) .. (239.74,402.29) -- cycle ;

% Text Node
\draw (262,268.64) node [anchor=north west][inner sep=0.75pt]   [align=left] {$\displaystyle \textcolor[rgb]{0.82,0.01,0.11}{A}$};
% Text Node
\draw (258,342.24) node [anchor=north west][inner sep=0.75pt]   [align=left] {$\displaystyle \textcolor[rgb]{0.16,0.48,0.85}{B}$};
% Text Node
\draw (337,318.48) node [anchor=north west][inner sep=0.75pt]   [align=left] {$\displaystyle \textcolor[rgb]{0.25,0.46,0.02}{C}$};
\end{mytikz2}\quad.
\end{equation}
For the exact ground state at $h=0$, we have $S_{\text{topo}}=-\ln 2$. The FLDC ansatz used in our experiments is similar to the one shown in Fig.~\textcolor{blue}{1}(b). Namely, the blocks are applied plaquette-by-plaquette sequentially from left to right and then from top to bottom. Within each plaquette, the blocks are applied by the following ``claw-shape'' in our experiments corresponding to Fig.~\textcolor{blue}{4}
\begin{equation}
\begin{mytikz2}
%Straight Lines [id:da09159522868806036] 
\draw [color={rgb, 255:red, 211; green, 211; blue, 211 }  ,draw opacity=1 ]   (193.04,103.88) -- (106.83,103.88) ;
\draw [shift={(106.83,103.88)}, rotate = 180] [color={rgb, 255:red, 211; green, 211; blue, 211 }  ,draw opacity=1 ][fill={rgb, 255:red, 211; green, 211; blue, 211 }  ,fill opacity=1 ][line width=0.75]      (0, 0) circle [x radius= 1.34, y radius= 1.34]   ;
%Straight Lines [id:da32676348953047496] 
\draw [color={rgb, 255:red, 211; green, 211; blue, 211 }  ,draw opacity=1 ]   (106.83,190.1) -- (106.83,103.88) ;
%Straight Lines [id:da019350439368350347] 
\draw [color={rgb, 255:red, 211; green, 211; blue, 211 }  ,draw opacity=1 ]   (194.19,190.1) -- (194.19,103.88) ;
\draw [shift={(194.19,103.88)}, rotate = 270] [color={rgb, 255:red, 211; green, 211; blue, 211 }  ,draw opacity=1 ][fill={rgb, 255:red, 211; green, 211; blue, 211 }  ,fill opacity=1 ][line width=0.75]      (0, 0) circle [x radius= 1.34, y radius= 1.34]   ;
%Straight Lines [id:da8187166621344946] 
\draw [color={rgb, 255:red, 211; green, 211; blue, 211 }  ,draw opacity=1 ]   (193.04,187.8) -- (106.83,187.8) ;
\draw [shift={(106.83,187.8)}, rotate = 180] [color={rgb, 255:red, 211; green, 211; blue, 211 }  ,draw opacity=1 ][fill={rgb, 255:red, 211; green, 211; blue, 211 }  ,fill opacity=1 ][line width=0.75]      (0, 0) circle [x radius= 1.34, y radius= 1.34]   ;
\draw [shift={(193.04,187.8)}, rotate = 180] [color={rgb, 255:red, 211; green, 211; blue, 211 }  ,draw opacity=1 ][fill={rgb, 255:red, 211; green, 211; blue, 211 }  ,fill opacity=1 ][line width=0.75]      (0, 0) circle [x radius= 1.34, y radius= 1.34]   ;
%Rounded Rect [id:dp6399329820622208] 
\draw  [fill={rgb, 255:red, 215; green, 236; blue, 255 }  ,fill opacity=1 ] (164.06,177.99) .. controls (167.43,181.35) and (167.43,186.81) .. (164.06,190.18) -- (155.93,198.31) .. controls (152.57,201.67) and (147.11,201.67) .. (143.74,198.31) -- (94.97,149.54) .. controls (91.61,146.17) and (91.61,140.71) .. (94.97,137.35) -- (103.1,129.22) .. controls (106.47,125.85) and (111.93,125.85) .. (115.29,129.22) -- cycle ;
%Rounded Rect [id:dp592712766027802] 
\draw  [fill={rgb, 255:red, 147; green, 206; blue, 255 }  ,fill opacity=1 ] (135.61,103.79) .. controls (135.61,99.03) and (139.47,95.17) .. (144.24,95.17) -- (155.73,95.17) .. controls (160.49,95.17) and (164.35,99.03) .. (164.35,103.79) -- (164.35,190.18) .. controls (164.35,194.94) and (160.49,198.8) .. (155.73,198.8) -- (144.24,198.8) .. controls (139.47,198.8) and (135.61,194.94) .. (135.61,190.18) -- cycle ;
%Rounded Rect [id:dp3321565351634239] 
\draw  [fill={rgb, 255:red, 35; green, 159; blue, 255 }  ,fill opacity=1 ] (184.38,129.22) .. controls (187.75,125.85) and (193.21,125.85) .. (196.58,129.22) -- (204.7,137.35) .. controls (208.07,140.71) and (208.07,146.17) .. (204.7,149.54) -- (155.93,198.31) .. controls (152.57,201.67) and (147.11,201.67) .. (143.74,198.31) -- (135.61,190.18) .. controls (132.25,186.81) and (132.25,181.35) .. (135.61,177.99) -- cycle ;
\end{mytikz2}\quad,
\end{equation}
where darker colors indicate later action orders. Thus, the entire circuit can be depicted as
\begin{equation}\label{eq:claw}
\begin{mytikz4}
%Straight Lines [id:da7377078134991011] 
\draw [color={rgb, 255:red, 211; green, 211; blue, 211 }  ,draw opacity=1 ]   (107.21,52.21) -- (51.97,52.21) ;
\draw [shift={(51.97,52.21)}, rotate = 180] [color={rgb, 255:red, 211; green, 211; blue, 211 }  ,draw opacity=1 ][fill={rgb, 255:red, 211; green, 211; blue, 211 }  ,fill opacity=1 ][line width=0.75]      (0, 0) circle [x radius= 1.34, y radius= 1.34]   ;
%Straight Lines [id:da6242113050669598] 
\draw [color={rgb, 255:red, 211; green, 211; blue, 211 }  ,draw opacity=1 ]   (51.97,107.45) -- (51.97,52.21) ;
%Straight Lines [id:da2648607943481227] 
\draw [color={rgb, 255:red, 211; green, 211; blue, 211 }  ,draw opacity=1 ]   (162.46,52.21) -- (107.21,52.21) ;
\draw [shift={(107.21,52.21)}, rotate = 180] [color={rgb, 255:red, 211; green, 211; blue, 211 }  ,draw opacity=1 ][fill={rgb, 255:red, 211; green, 211; blue, 211 }  ,fill opacity=1 ][line width=0.75]      (0, 0) circle [x radius= 1.34, y radius= 1.34]   ;
%Straight Lines [id:da2387739276191705] 
\draw [color={rgb, 255:red, 211; green, 211; blue, 211 }  ,draw opacity=1 ]   (107.21,107.45) -- (107.21,52.21) ;
%Straight Lines [id:da017606196190701917] 
\draw [color={rgb, 255:red, 211; green, 211; blue, 211 }  ,draw opacity=1 ]   (215.13,52.21) -- (159.88,52.21) ;
\draw [shift={(159.88,52.21)}, rotate = 180] [color={rgb, 255:red, 211; green, 211; blue, 211 }  ,draw opacity=1 ][fill={rgb, 255:red, 211; green, 211; blue, 211 }  ,fill opacity=1 ][line width=0.75]      (0, 0) circle [x radius= 1.34, y radius= 1.34]   ;
%Straight Lines [id:da4616818798554865] 
\draw [color={rgb, 255:red, 211; green, 211; blue, 211 }  ,draw opacity=1 ]   (159.88,107.45) -- (159.88,52.21) ;
%Straight Lines [id:da2968209056658917] 
\draw [color={rgb, 255:red, 211; green, 211; blue, 211 }  ,draw opacity=1 ]   (215.86,107.45) -- (215.86,52.21) ;
\draw [shift={(215.86,52.21)}, rotate = 270] [color={rgb, 255:red, 211; green, 211; blue, 211 }  ,draw opacity=1 ][fill={rgb, 255:red, 211; green, 211; blue, 211 }  ,fill opacity=1 ][line width=0.75]      (0, 0) circle [x radius= 1.34, y radius= 1.34]   ;
%Straight Lines [id:da2762748577903331] 
\draw [color={rgb, 255:red, 211; green, 211; blue, 211 }  ,draw opacity=1 ]   (107.21,107.45) -- (51.97,107.45) ;
\draw [shift={(51.97,107.45)}, rotate = 180] [color={rgb, 255:red, 211; green, 211; blue, 211 }  ,draw opacity=1 ][fill={rgb, 255:red, 211; green, 211; blue, 211 }  ,fill opacity=1 ][line width=0.75]      (0, 0) circle [x radius= 1.34, y radius= 1.34]   ;
%Straight Lines [id:da014311797415645033] 
\draw [color={rgb, 255:red, 211; green, 211; blue, 211 }  ,draw opacity=1 ]   (51.97,162.7) -- (51.97,107.45) ;
%Straight Lines [id:da43245162612042654] 
\draw [color={rgb, 255:red, 211; green, 211; blue, 211 }  ,draw opacity=1 ]   (162.46,107.45) -- (107.21,107.45) ;
\draw [shift={(107.21,107.45)}, rotate = 180] [color={rgb, 255:red, 211; green, 211; blue, 211 }  ,draw opacity=1 ][fill={rgb, 255:red, 211; green, 211; blue, 211 }  ,fill opacity=1 ][line width=0.75]      (0, 0) circle [x radius= 1.34, y radius= 1.34]   ;
%Straight Lines [id:da5042194035056937] 
\draw [color={rgb, 255:red, 211; green, 211; blue, 211 }  ,draw opacity=1 ]   (107.21,162.7) -- (107.21,107.45) ;
%Straight Lines [id:da9162554765033266] 
\draw [color={rgb, 255:red, 211; green, 211; blue, 211 }  ,draw opacity=1 ]   (215.13,107.45) -- (159.88,107.45) ;
\draw [shift={(159.88,107.45)}, rotate = 180] [color={rgb, 255:red, 211; green, 211; blue, 211 }  ,draw opacity=1 ][fill={rgb, 255:red, 211; green, 211; blue, 211 }  ,fill opacity=1 ][line width=0.75]      (0, 0) circle [x radius= 1.34, y radius= 1.34]   ;
%Straight Lines [id:da9279479412779719] 
\draw [color={rgb, 255:red, 211; green, 211; blue, 211 }  ,draw opacity=1 ]   (159.88,162.7) -- (159.88,107.45) ;
%Straight Lines [id:da5065583617825644] 
\draw [color={rgb, 255:red, 211; green, 211; blue, 211 }  ,draw opacity=1 ]   (215.86,162.7) -- (215.86,107.45) ;
\draw [shift={(215.86,107.45)}, rotate = 270] [color={rgb, 255:red, 211; green, 211; blue, 211 }  ,draw opacity=1 ][fill={rgb, 255:red, 211; green, 211; blue, 211 }  ,fill opacity=1 ][line width=0.75]      (0, 0) circle [x radius= 1.34, y radius= 1.34]   ;
%Straight Lines [id:da9189195070035312] 
\draw [color={rgb, 255:red, 211; green, 211; blue, 211 }  ,draw opacity=1 ]   (107.21,159.75) -- (51.97,159.75) ;
\draw [shift={(51.97,159.75)}, rotate = 180] [color={rgb, 255:red, 211; green, 211; blue, 211 }  ,draw opacity=1 ][fill={rgb, 255:red, 211; green, 211; blue, 211 }  ,fill opacity=1 ][line width=0.75]      (0, 0) circle [x radius= 1.34, y radius= 1.34]   ;
%Straight Lines [id:da02049817166742729] 
\draw [color={rgb, 255:red, 211; green, 211; blue, 211 }  ,draw opacity=1 ]   (51.97,215) -- (51.97,159.75) ;
%Straight Lines [id:da2682458632580278] 
\draw [color={rgb, 255:red, 211; green, 211; blue, 211 }  ,draw opacity=1 ]   (162.46,159.75) -- (107.21,159.75) ;
\draw [shift={(107.21,159.75)}, rotate = 180] [color={rgb, 255:red, 211; green, 211; blue, 211 }  ,draw opacity=1 ][fill={rgb, 255:red, 211; green, 211; blue, 211 }  ,fill opacity=1 ][line width=0.75]      (0, 0) circle [x radius= 1.34, y radius= 1.34]   ;
%Straight Lines [id:da8109044497902238] 
\draw [color={rgb, 255:red, 211; green, 211; blue, 211 }  ,draw opacity=1 ]   (107.21,215) -- (107.21,159.75) ;
%Straight Lines [id:da13813074323184904] 
\draw [color={rgb, 255:red, 211; green, 211; blue, 211 }  ,draw opacity=1 ]   (215.13,159.75) -- (159.88,159.75) ;
\draw [shift={(159.88,159.75)}, rotate = 180] [color={rgb, 255:red, 211; green, 211; blue, 211 }  ,draw opacity=1 ][fill={rgb, 255:red, 211; green, 211; blue, 211 }  ,fill opacity=1 ][line width=0.75]      (0, 0) circle [x radius= 1.34, y radius= 1.34]   ;
%Straight Lines [id:da38006181153499496] 
\draw [color={rgb, 255:red, 211; green, 211; blue, 211 }  ,draw opacity=1 ]   (159.88,215) -- (159.88,159.75) ;
%Straight Lines [id:da3356165997646161] 
\draw [color={rgb, 255:red, 211; green, 211; blue, 211 }  ,draw opacity=1 ]   (215.86,215) -- (215.86,159.75) ;
\draw [shift={(215.86,159.75)}, rotate = 270] [color={rgb, 255:red, 211; green, 211; blue, 211 }  ,draw opacity=1 ][fill={rgb, 255:red, 211; green, 211; blue, 211 }  ,fill opacity=1 ][line width=0.75]      (0, 0) circle [x radius= 1.34, y radius= 1.34]   ;
%Straight Lines [id:da05427327438794549] 
\draw [color={rgb, 255:red, 211; green, 211; blue, 211 }  ,draw opacity=1 ]   (107.21,213.53) -- (51.97,213.53) ;
\draw [shift={(51.97,213.53)}, rotate = 180] [color={rgb, 255:red, 211; green, 211; blue, 211 }  ,draw opacity=1 ][fill={rgb, 255:red, 211; green, 211; blue, 211 }  ,fill opacity=1 ][line width=0.75]      (0, 0) circle [x radius= 1.34, y radius= 1.34]   ;
%Straight Lines [id:da7826525645030349] 
\draw [color={rgb, 255:red, 211; green, 211; blue, 211 }  ,draw opacity=1 ]   (162.46,213.53) -- (107.21,213.53) ;
\draw [shift={(107.21,213.53)}, rotate = 180] [color={rgb, 255:red, 211; green, 211; blue, 211 }  ,draw opacity=1 ][fill={rgb, 255:red, 211; green, 211; blue, 211 }  ,fill opacity=1 ][line width=0.75]      (0, 0) circle [x radius= 1.34, y radius= 1.34]   ;
%Straight Lines [id:da4004959628104314] 
\draw [color={rgb, 255:red, 211; green, 211; blue, 211 }  ,draw opacity=1 ]   (215.13,213.53) -- (159.88,213.53) ;
\draw [shift={(159.88,213.53)}, rotate = 180] [color={rgb, 255:red, 211; green, 211; blue, 211 }  ,draw opacity=1 ][fill={rgb, 255:red, 211; green, 211; blue, 211 }  ,fill opacity=1 ][line width=0.75]      (0, 0) circle [x radius= 1.34, y radius= 1.34]   ;
\draw [shift={(215.13,213.53)}, rotate = 180] [color={rgb, 255:red, 211; green, 211; blue, 211 }  ,draw opacity=1 ][fill={rgb, 255:red, 211; green, 211; blue, 211 }  ,fill opacity=1 ][line width=0.75]      (0, 0) circle [x radius= 1.34, y radius= 1.34]   ;

%Rounded Rect [id:dp6597682266652292] 
\draw  [fill={rgb, 255:red, 236; green, 247; blue, 255 }  ,fill opacity=1 ] (91.32,104.13) .. controls (93.48,106.28) and (93.48,109.78) .. (91.32,111.94) -- (86.11,117.15) .. controls (83.95,119.31) and (80.46,119.31) .. (78.3,117.15) -- (47.05,85.9) .. controls (44.89,83.74) and (44.89,80.24) .. (47.05,78.08) -- (52.26,72.88) .. controls (54.41,70.72) and (57.91,70.72) .. (60.07,72.88) -- cycle ;
%Rounded Rect [id:dp6577054920141538] 
\draw  [fill={rgb, 255:red, 182; green, 220; blue, 255 }  ,fill opacity=1 ] (90.79,155.69) .. controls (92.95,157.85) and (92.95,161.35) .. (90.79,163.5) -- (85.58,168.71) .. controls (83.43,170.87) and (79.93,170.87) .. (77.77,168.71) -- (46.52,137.46) .. controls (44.36,135.3) and (44.36,131.8) .. (46.52,129.65) -- (51.73,124.44) .. controls (53.89,122.28) and (57.38,122.28) .. (59.54,124.44) -- cycle ;
%Rounded Rect [id:dp3623108732887015] 
\draw  [fill={rgb, 255:red, 130; green, 195; blue, 255 }  ,fill opacity=1 ] (90.18,206.24) .. controls (92.34,208.4) and (92.34,211.89) .. (90.18,214.05) -- (84.97,219.26) .. controls (82.81,221.42) and (79.32,221.42) .. (77.16,219.26) -- (45.91,188.01) .. controls (43.75,185.85) and (43.75,182.35) .. (45.91,180.2) -- (51.12,174.99) .. controls (53.27,172.83) and (56.77,172.83) .. (58.93,174.99) -- cycle ;
%Rounded Rect [id:dp5007082734748711] 
\draw  [fill={rgb, 255:red, 236; green, 247; blue, 255 }  ,fill opacity=1 ] (72.36,54.49) .. controls (72.36,51.44) and (74.83,48.96) .. (77.89,48.96) -- (85.25,48.96) .. controls (88.3,48.96) and (90.78,51.44) .. (90.78,54.49) -- (90.78,108.71) .. controls (90.78,111.76) and (88.3,114.24) .. (85.25,114.24) -- (77.89,114.24) .. controls (74.83,114.24) and (72.36,111.76) .. (72.36,108.71) -- cycle ;
%Rounded Rect [id:dp19562399222283244] 
\draw  [fill={rgb, 255:red, 236; green, 247; blue, 255 }  ,fill opacity=1 ] (104.34,72.88) .. controls (106.5,70.72) and (110,70.72) .. (112.15,72.88) -- (117.36,78.08) .. controls (119.52,80.24) and (119.52,83.74) .. (117.36,85.9) -- (86.11,117.15) .. controls (83.95,119.31) and (80.46,119.31) .. (78.3,117.15) -- (73.09,111.94) .. controls (70.93,109.78) and (70.93,106.28) .. (73.09,104.13) -- cycle ;
%Rounded Rect [id:dp2158914421272573] 
\draw  [fill={rgb, 255:red, 225; green, 241; blue, 255 }  ,fill opacity=1 ] (143.95,104.2) .. controls (146.11,106.36) and (146.11,109.86) .. (143.95,112.01) -- (138.74,117.22) .. controls (136.58,119.38) and (133.09,119.38) .. (130.93,117.22) -- (99.68,85.97) .. controls (97.52,83.81) and (97.52,80.32) .. (99.68,78.16) -- (104.89,72.95) .. controls (107.04,70.79) and (110.54,70.79) .. (112.7,72.95) -- cycle ;
%Rounded Rect [id:dp7224922203471145] 
\draw  [fill={rgb, 255:red, 182; green, 220; blue, 255 }  ,fill opacity=1 ] (72.36,107.14) .. controls (72.36,104.09) and (74.83,101.62) .. (77.89,101.62) -- (85.25,101.62) .. controls (88.3,101.62) and (90.78,104.09) .. (90.78,107.14) -- (90.78,160.27) .. controls (90.78,163.33) and (88.3,165.8) .. (85.25,165.8) -- (77.89,165.8) .. controls (74.83,165.8) and (72.36,163.33) .. (72.36,160.27) -- cycle ;
%Rounded Rect [id:dp7338329879204237] 
\draw  [fill={rgb, 255:red, 182; green, 220; blue, 255 }  ,fill opacity=1 ] (103.82,124.44) .. controls (105.97,122.28) and (109.47,122.28) .. (111.63,124.44) -- (116.84,129.65) .. controls (118.99,131.8) and (118.99,135.3) .. (116.84,137.46) -- (85.58,168.71) .. controls (83.43,170.87) and (79.93,170.87) .. (77.77,168.71) -- (72.56,163.5) .. controls (70.41,161.35) and (70.41,157.85) .. (72.56,155.69) -- cycle ;
%Rounded Rect [id:dp9727988143172841] 
\draw  [fill={rgb, 255:red, 163; green, 213; blue, 255 }  ,fill opacity=1 ] (143.42,155.76) .. controls (145.58,157.92) and (145.58,161.42) .. (143.42,163.58) -- (138.22,168.79) .. controls (136.06,170.94) and (132.56,170.94) .. (130.4,168.79) -- (99.15,137.53) .. controls (96.99,135.38) and (96.99,131.88) .. (99.15,129.72) -- (104.36,124.51) .. controls (106.52,122.35) and (110.02,122.35) .. (112.17,124.51) -- cycle ;
%Rounded Rect [id:dp5824706930579204] 
\draw  [fill={rgb, 255:red, 225; green, 241; blue, 255 }  ,fill opacity=1 ] (126.25,54.16) .. controls (126.25,51.11) and (128.72,48.64) .. (131.77,48.64) -- (139.14,48.64) .. controls (142.19,48.64) and (144.66,51.11) .. (144.66,54.16) -- (144.66,107.98) .. controls (144.66,111.04) and (142.19,113.51) .. (139.14,113.51) -- (131.77,113.51) .. controls (128.72,113.51) and (126.25,111.04) .. (126.25,107.98) -- cycle ;
%Rounded Rect [id:dp422251869590474] 
\draw  [fill={rgb, 255:red, 225; green, 241; blue, 255 }  ,fill opacity=1 ] (156.97,72.95) .. controls (159.13,70.79) and (162.63,70.79) .. (164.79,72.95) -- (169.99,78.16) .. controls (172.15,80.32) and (172.15,83.81) .. (169.99,85.97) -- (138.74,117.22) .. controls (136.58,119.38) and (133.09,119.38) .. (130.93,117.22) -- (125.72,112.01) .. controls (123.56,109.86) and (123.56,106.36) .. (125.72,104.2) -- cycle ;
%Rounded Rect [id:dp45112230262195685] 
\draw  [fill={rgb, 255:red, 203; green, 230; blue, 255 }  ,fill opacity=1 ] (197.25,103.73) .. controls (199.41,105.88) and (199.41,109.38) .. (197.25,111.54) -- (192.04,116.75) .. controls (189.88,118.91) and (186.39,118.91) .. (184.23,116.75) -- (152.98,85.5) .. controls (150.82,83.34) and (150.82,79.84) .. (152.98,77.68) -- (158.19,72.48) .. controls (160.34,70.32) and (163.84,70.32) .. (166,72.48) -- cycle ;
%Rounded Rect [id:dp8907970394084659] 
\draw  [fill={rgb, 255:red, 203; green, 230; blue, 255 }  ,fill opacity=1 ] (178.68,53.04) .. controls (178.68,49.99) and (181.15,47.52) .. (184.2,47.52) -- (191.57,47.52) .. controls (194.62,47.52) and (197.09,49.99) .. (197.09,53.04) -- (197.09,107.51) .. controls (197.09,110.56) and (194.62,113.04) .. (191.57,113.04) -- (184.2,113.04) .. controls (181.15,113.04) and (178.68,110.56) .. (178.68,107.51) -- cycle ;
%Rounded Rect [id:dp3376374584824313] 
\draw  [fill={rgb, 255:red, 203; green, 230; blue, 255 }  ,fill opacity=1 ] (210.27,72.48) .. controls (212.43,70.32) and (215.93,70.32) .. (218.08,72.48) -- (223.29,77.68) .. controls (225.45,79.84) and (225.45,83.34) .. (223.29,85.5) -- (192.04,116.75) .. controls (189.88,118.91) and (186.39,118.91) .. (184.23,116.75) -- (179.02,111.54) .. controls (176.86,109.38) and (176.86,105.88) .. (179.02,103.73) -- cycle ;
%Rounded Rect [id:dp6251928707650753] 
\draw  [fill={rgb, 255:red, 163; green, 213; blue, 255 }  ,fill opacity=1 ] (125.32,105.92) .. controls (125.32,102.87) and (127.8,100.4) .. (130.85,100.4) -- (138.21,100.4) .. controls (141.27,100.4) and (143.74,102.87) .. (143.74,105.92) -- (143.74,159.95) .. controls (143.74,163) and (141.27,165.47) .. (138.21,165.47) -- (130.85,165.47) .. controls (127.8,165.47) and (125.32,163) .. (125.32,159.95) -- cycle ;
%Rounded Rect [id:dp7136295979824032] 
\draw  [fill={rgb, 255:red, 163; green, 213; blue, 255 }  ,fill opacity=1 ] (156.45,124.51) .. controls (158.6,122.35) and (162.1,122.35) .. (164.26,124.51) -- (169.47,129.72) .. controls (171.63,131.88) and (171.63,135.38) .. (169.47,137.53) -- (138.22,168.79) .. controls (136.06,170.94) and (132.56,170.94) .. (130.4,168.79) -- (125.19,163.58) .. controls (123.04,161.42) and (123.04,157.92) .. (125.19,155.76) -- cycle ;
%Rounded Rect [id:dp7975123541270419] 
\draw  [fill={rgb, 255:red, 147; green, 206; blue, 255 }  ,fill opacity=1 ] (197.72,154.29) .. controls (199.88,156.45) and (199.88,159.95) .. (197.72,162.1) -- (192.52,167.31) .. controls (190.36,169.47) and (186.86,169.47) .. (184.7,167.31) -- (153.45,136.06) .. controls (151.29,133.9) and (151.29,130.4) .. (153.45,128.25) -- (158.66,123.04) .. controls (160.82,120.88) and (164.32,120.88) .. (166.47,123.04) -- cycle ;
%Rounded Rect [id:dp9138578459054092] 
\draw  [fill={rgb, 255:red, 147; green, 206; blue, 255 }  ,fill opacity=1 ] (178.58,107.05) .. controls (178.58,104) and (181.05,101.52) .. (184.1,101.52) -- (191.47,101.52) .. controls (194.52,101.52) and (196.99,104) .. (196.99,107.05) -- (196.99,157.95) .. controls (196.99,161) and (194.52,163.47) .. (191.47,163.47) -- (184.1,163.47) .. controls (181.05,163.47) and (178.58,161) .. (178.58,157.95) -- cycle ;
%Rounded Rect [id:dp44935131518495197] 
\draw  [fill={rgb, 255:red, 147; green, 206; blue, 255 }  ,fill opacity=1 ] (210.75,123.04) .. controls (212.9,120.88) and (216.4,120.88) .. (218.56,123.04) -- (223.77,128.25) .. controls (225.92,130.4) and (225.92,133.9) .. (223.77,136.06) -- (192.52,167.31) .. controls (190.36,169.47) and (186.86,169.47) .. (184.7,167.31) -- (179.49,162.1) .. controls (177.34,159.95) and (177.34,156.45) .. (179.49,154.29) -- cycle ;
%Rounded Rect [id:dp9400928535115993] 
\draw  [fill={rgb, 255:red, 130; green, 195; blue, 255 }  ,fill opacity=1 ] (72.03,159.74) .. controls (72.02,156.69) and (74.49,154.21) .. (77.54,154.2) -- (84.91,154.19) .. controls (87.96,154.19) and (90.44,156.66) .. (90.44,159.71) -- (90.53,212.21) .. controls (90.54,215.26) and (88.07,217.73) .. (85.02,217.74) -- (77.65,217.75) .. controls (74.6,217.76) and (72.12,215.29) .. (72.12,212.24) -- cycle ;
%Rounded Rect [id:dp8446072302484551] 
\draw  [fill={rgb, 255:red, 130; green, 195; blue, 255 }  ,fill opacity=1 ] (103.2,175.39) .. controls (105.36,173.23) and (108.86,173.23) .. (111.01,175.39) -- (116.22,180.6) .. controls (118.38,182.75) and (118.38,186.25) .. (116.22,188.41) -- (84.97,219.66) .. controls (82.81,221.82) and (79.32,221.82) .. (77.16,219.66) -- (71.95,214.45) .. controls (69.79,212.29) and (69.79,208.8) .. (71.95,206.64) -- cycle ;
%Rounded Rect [id:dp3979257201544728] 
\draw  [fill={rgb, 255:red, 84; green, 174; blue, 255 }  ,fill opacity=1 ] (142.81,206.71) .. controls (144.97,208.87) and (144.97,212.37) .. (142.81,214.53) -- (137.6,219.73) .. controls (135.44,221.89) and (131.95,221.89) .. (129.79,219.73) -- (98.54,188.48) .. controls (96.38,186.33) and (96.38,182.83) .. (98.54,180.67) -- (103.75,175.46) .. controls (105.9,173.3) and (109.4,173.3) .. (111.56,175.46) -- cycle ;
%Rounded Rect [id:dp534570690302034] 
\draw  [fill={rgb, 255:red, 84; green, 174; blue, 255 }  ,fill opacity=1 ] (124.31,159.36) .. controls (124.31,156.31) and (126.78,153.84) .. (129.83,153.84) -- (137.2,153.84) .. controls (140.25,153.84) and (142.72,156.31) .. (142.72,159.36) -- (142.72,212.1) .. controls (142.72,215.15) and (140.25,217.62) .. (137.2,217.62) -- (129.83,217.62) .. controls (126.78,217.62) and (124.31,215.15) .. (124.31,212.1) -- cycle ;
%Rounded Rect [id:dp3010202576000376] 
\draw  [fill={rgb, 255:red, 84; green, 174; blue, 255 }  ,fill opacity=1 ] (155.83,175.46) .. controls (157.99,173.3) and (161.49,173.3) .. (163.65,175.46) -- (168.85,180.67) .. controls (171.01,182.83) and (171.01,186.33) .. (168.85,188.48) -- (137.6,219.73) .. controls (135.44,221.89) and (131.95,221.89) .. (129.79,219.73) -- (124.58,214.53) .. controls (122.42,212.37) and (122.42,208.87) .. (124.58,206.71) -- cycle ;
%Rounded Rect [id:dp8707767679496903] 
\draw  [fill={rgb, 255:red, 35; green, 159; blue, 255 }  ,fill opacity=1 ] (197.11,205.24) .. controls (199.27,207.4) and (199.27,210.89) .. (197.11,213.05) -- (191.9,218.26) .. controls (189.74,220.42) and (186.25,220.42) .. (184.09,218.26) -- (152.84,187.01) .. controls (150.68,184.85) and (150.68,181.35) .. (152.84,179.2) -- (158.05,173.99) .. controls (160.2,171.83) and (163.7,171.83) .. (165.86,173.99) -- cycle ;
%Rounded Rect [id:dp42629372187754] 
\draw  [fill={rgb, 255:red, 35; green, 159; blue, 255 }  ,fill opacity=1 ] (178.95,157.84) .. controls (178.95,154.79) and (181.43,152.32) .. (184.48,152.32) -- (191.84,152.32) .. controls (194.89,152.32) and (197.37,154.79) .. (197.37,157.84) -- (197.37,211) .. controls (197.37,214.05) and (194.89,216.52) .. (191.84,216.52) -- (184.48,216.52) .. controls (181.43,216.52) and (178.95,214.05) .. (178.95,211) -- cycle ;
%Rounded Rect [id:dp6803371001993503] 
\draw  [fill={rgb, 255:red, 35; green, 159; blue, 255 }  ,fill opacity=1 ] (210.13,173.99) .. controls (212.29,171.83) and (215.79,171.83) .. (217.94,173.99) -- (223.15,179.2) .. controls (225.31,181.35) and (225.31,184.85) .. (223.15,187.01) -- (191.9,218.26) .. controls (189.74,220.42) and (186.25,220.42) .. (184.09,218.26) -- (178.88,213.05) .. controls (176.72,210.89) and (176.72,207.4) .. (178.88,205.24) -- cycle ;
% %Rounded Rect [id:dp00850291727339636] 
% \draw  [fill={rgb, 255:red, 245; green, 166; blue, 35 }  ,fill opacity=1 ] (210.86,182.38) .. controls (210.86,182.38) and (210.86,182.38) .. (210.86,182.38) -- (220.86,182.38) .. controls (220.86,182.38) and (220.86,182.38) .. (220.86,182.38) -- (220.86,192.38) .. controls (220.86,192.38) and (220.86,192.38) .. (220.86,192.38) -- (210.86,192.38) .. controls (210.86,192.38) and (210.86,192.38) .. (210.86,192.38) -- cycle ;
% %Straight Lines [id:da16296576181874833] 
% \draw    (51.2,230.13) -- (215.2,230.13) ;
% \draw [shift={(215.2,230.13)}, rotate = 180] [color={rgb, 255:red, 0; green, 0; blue, 0 }  ][line width=0.75]    (0,5.59) -- (0,-5.59)(10.93,-4.9) .. controls (6.95,-2.3) and (3.31,-0.67) .. (0,0) .. controls (3.31,0.67) and (6.95,2.3) .. (10.93,4.9)   ;
% \draw [shift={(51.2,230.13)}, rotate = 0] [color={rgb, 255:red, 0; green, 0; blue, 0 }  ][line width=0.75]    (0,5.59) -- (0,-5.59)(10.93,-4.9) .. controls (6.95,-2.3) and (3.31,-0.67) .. (0,0) .. controls (3.31,0.67) and (6.95,2.3) .. (10.93,4.9)   ;

% % Text Node
% \draw (127.33,234) node [anchor=north west][inner sep=0.75pt]   [align=left] {$\displaystyle L$};

\end{mytikz4}\quad.
\end{equation}
The two-qubit block is parametrized as the Cartan decomposition as in Eq.~\eqref{eq:cartan}. As a comparison, we also conduct the same simulations with the instances from the FDC and GLDC classes. To compare fairly, we set the number of blocks in the FDC ansatz as the same as that in the FLDC ansatz above but arrange the blocks differently. To be specific, we apply the blocks located at the same position within the plaquette simultaneously for all plaquettes, as depicted below
\begin{equation}
\begin{mytikz2}
%Straight Lines [id:da17029882117663253] 
\draw [color={rgb, 255:red, 211; green, 211; blue, 211 }  ,draw opacity=1 ]   (194.84,140.51) -- (123.95,140.51) ;
\draw [shift={(123.95,140.51)}, rotate = 180] [color={rgb, 255:red, 211; green, 211; blue, 211 }  ,draw opacity=1 ][fill={rgb, 255:red, 211; green, 211; blue, 211 }  ,fill opacity=1 ][line width=0.75]      (0, 0) circle [x radius= 1.34, y radius= 1.34]   ;
%Straight Lines [id:da3322493523729324] 
\draw [color={rgb, 255:red, 211; green, 211; blue, 211 }  ,draw opacity=1 ]   (123.95,144.29) -- (123.95,73.4) ;
%Rounded Rect [id:dp5966866246892701] 
\draw  [fill={rgb, 255:red, 240; green, 247; blue, 255 }  ,fill opacity=1 ] (170.41,135.39) .. controls (173.18,138.16) and (173.18,142.65) .. (170.41,145.42) -- (163.73,152.1) .. controls (160.96,154.87) and (156.47,154.87) .. (153.7,152.1) -- (113.6,112) .. controls (110.83,109.23) and (110.83,104.74) .. (113.6,101.97) -- (120.28,95.29) .. controls (123.05,92.52) and (127.54,92.52) .. (130.31,95.29) -- cycle ;
%Straight Lines [id:da4627057342952041] 
\draw [color={rgb, 255:red, 211; green, 211; blue, 211 }  ,draw opacity=1 ]   (194.84,73.4) -- (123.95,73.4) ;
\draw [shift={(123.95,73.4)}, rotate = 180] [color={rgb, 255:red, 211; green, 211; blue, 211 }  ,draw opacity=1 ][fill={rgb, 255:red, 211; green, 211; blue, 211 }  ,fill opacity=1 ][line width=0.75]      (0, 0) circle [x radius= 1.34, y radius= 1.34]   ;
%Straight Lines [id:da9647713701833223] 
\draw [color={rgb, 255:red, 211; green, 211; blue, 211 }  ,draw opacity=1 ]   (262.42,73.4) -- (191.53,73.4) ;
\draw [shift={(191.53,73.4)}, rotate = 180] [color={rgb, 255:red, 211; green, 211; blue, 211 }  ,draw opacity=1 ][fill={rgb, 255:red, 211; green, 211; blue, 211 }  ,fill opacity=1 ][line width=0.75]      (0, 0) circle [x radius= 1.34, y radius= 1.34]   ;
%Straight Lines [id:da7178431843381743] 
\draw [color={rgb, 255:red, 211; green, 211; blue, 211 }  ,draw opacity=1 ]   (191.53,144.29) -- (191.53,73.4) ;
%Straight Lines [id:da9515779850288362] 
\draw [color={rgb, 255:red, 211; green, 211; blue, 211 }  ,draw opacity=1 ]   (263.37,144.29) -- (263.37,73.4) ;
\draw [shift={(263.37,73.4)}, rotate = 270] [color={rgb, 255:red, 211; green, 211; blue, 211 }  ,draw opacity=1 ][fill={rgb, 255:red, 211; green, 211; blue, 211 }  ,fill opacity=1 ][line width=0.75]      (0, 0) circle [x radius= 1.34, y radius= 1.34]   ;
%Straight Lines [id:da3522672236856563] 
\draw [color={rgb, 255:red, 211; green, 211; blue, 211 }  ,draw opacity=1 ]   (123.95,211.4) -- (123.95,140.51) ;
%Straight Lines [id:da9653853804736383] 
\draw [color={rgb, 255:red, 211; green, 211; blue, 211 }  ,draw opacity=1 ]   (262.42,140.51) -- (191.53,140.51) ;
\draw [shift={(191.53,140.51)}, rotate = 180] [color={rgb, 255:red, 211; green, 211; blue, 211 }  ,draw opacity=1 ][fill={rgb, 255:red, 211; green, 211; blue, 211 }  ,fill opacity=1 ][line width=0.75]      (0, 0) circle [x radius= 1.34, y radius= 1.34]   ;
%Straight Lines [id:da9000866764128148] 
\draw [color={rgb, 255:red, 211; green, 211; blue, 211 }  ,draw opacity=1 ]   (191.53,211.4) -- (191.53,140.51) ;
%Straight Lines [id:da12219894585550928] 
\draw [color={rgb, 255:red, 211; green, 211; blue, 211 }  ,draw opacity=1 ]   (263.37,211.4) -- (263.37,140.51) ;
\draw [shift={(263.37,140.51)}, rotate = 270] [color={rgb, 255:red, 211; green, 211; blue, 211 }  ,draw opacity=1 ][fill={rgb, 255:red, 211; green, 211; blue, 211 }  ,fill opacity=1 ][line width=0.75]      (0, 0) circle [x radius= 1.34, y radius= 1.34]   ;
%Straight Lines [id:da9652535499723098] 
\draw [color={rgb, 255:red, 211; green, 211; blue, 211 }  ,draw opacity=1 ]   (194.84,209.51) -- (123.95,209.51) ;
\draw [shift={(123.95,209.51)}, rotate = 180] [color={rgb, 255:red, 211; green, 211; blue, 211 }  ,draw opacity=1 ][fill={rgb, 255:red, 211; green, 211; blue, 211 }  ,fill opacity=1 ][line width=0.75]      (0, 0) circle [x radius= 1.34, y radius= 1.34]   ;
%Straight Lines [id:da03442653321958389] 
\draw [color={rgb, 255:red, 211; green, 211; blue, 211 }  ,draw opacity=1 ]   (262.42,209.51) -- (191.53,209.51) ;
\draw [shift={(191.53,209.51)}, rotate = 180] [color={rgb, 255:red, 211; green, 211; blue, 211 }  ,draw opacity=1 ][fill={rgb, 255:red, 211; green, 211; blue, 211 }  ,fill opacity=1 ][line width=0.75]      (0, 0) circle [x radius= 1.34, y radius= 1.34]   ;
\draw [shift={(262.42,209.51)}, rotate = 180] [color={rgb, 255:red, 211; green, 211; blue, 211 }  ,draw opacity=1 ][fill={rgb, 255:red, 211; green, 211; blue, 211 }  ,fill opacity=1 ][line width=0.75]      (0, 0) circle [x radius= 1.34, y radius= 1.34]   ;
%Rounded Rect [id:dp4822637655989628] 
\draw  [fill={rgb, 255:red, 240; green, 247; blue, 255 }  ,fill opacity=1 ] (237.25,135.39) .. controls (240.02,138.16) and (240.02,142.65) .. (237.25,145.42) -- (230.57,152.1) .. controls (227.8,154.87) and (223.31,154.87) .. (220.54,152.1) -- (180.44,112) .. controls (177.67,109.23) and (177.67,104.74) .. (180.44,101.97) -- (187.12,95.29) .. controls (189.89,92.52) and (194.38,92.52) .. (197.15,95.29) -- cycle ;
%Rounded Rect [id:dp7334940584609622] 
\draw  [fill={rgb, 255:red, 240; green, 247; blue, 255 }  ,fill opacity=1 ] (169.63,200.77) .. controls (172.39,203.54) and (172.39,208.03) .. (169.63,210.8) -- (162.94,217.48) .. controls (160.17,220.25) and (155.68,220.25) .. (152.92,217.48) -- (112.81,177.38) .. controls (110.04,174.61) and (110.04,170.12) .. (112.81,167.35) -- (119.5,160.67) .. controls (122.27,157.9) and (126.75,157.9) .. (129.52,160.67) -- cycle ;
%Rounded Rect [id:dp7766598358276802] 
\draw  [fill={rgb, 255:red, 240; green, 247; blue, 255 }  ,fill opacity=1 ] (236.45,200.44) .. controls (239.22,203.21) and (239.22,207.7) .. (236.45,210.47) -- (229.77,217.15) .. controls (227,219.92) and (222.51,219.92) .. (219.74,217.15) -- (179.64,177.05) .. controls (176.87,174.28) and (176.87,169.79) .. (179.64,167.03) -- (186.33,160.34) .. controls (189.09,157.57) and (193.58,157.57) .. (196.35,160.34) -- cycle ;
%Rounded Rect [id:dp4865587564517142] 
\draw  [fill={rgb, 255:red, 207; green, 228; blue, 255 }  ,fill opacity=1 ] (146,74.39) .. controls (146,70.47) and (149.17,67.3) .. (153.09,67.3) -- (162.54,67.3) .. controls (166.46,67.3) and (169.63,70.47) .. (169.63,74.39) -- (169.63,145.01) .. controls (169.63,148.93) and (166.46,152.1) .. (162.54,152.1) -- (153.09,152.1) .. controls (149.17,152.1) and (146,148.93) .. (146,145.01) -- cycle ;
%Rounded Rect [id:dp19648787635138998] 
\draw  [fill={rgb, 255:red, 207; green, 228; blue, 255 }  ,fill opacity=1 ] (212.62,74.39) .. controls (212.62,70.47) and (215.79,67.3) .. (219.71,67.3) -- (229.16,67.3) .. controls (233.08,67.3) and (236.25,70.47) .. (236.25,74.39) -- (236.25,145.01) .. controls (236.25,148.93) and (233.08,152.1) .. (229.16,152.1) -- (219.71,152.1) .. controls (215.79,152.1) and (212.62,148.93) .. (212.62,145.01) -- cycle ;
%Rounded Rect [id:dp7599278850106885] 
\draw  [fill={rgb, 255:red, 128; green, 185; blue, 255 }  ,fill opacity=1 ] (146,139.77) .. controls (146,135.85) and (149.17,132.68) .. (153.09,132.68) -- (162.54,132.68) .. controls (166.46,132.68) and (169.63,135.85) .. (169.63,139.77) -- (169.63,210.39) .. controls (169.63,214.3) and (166.46,217.48) .. (162.54,217.48) -- (153.09,217.48) .. controls (149.17,217.48) and (146,214.3) .. (146,210.39) -- cycle ;
%Rounded Rect [id:dp1919890796200303] 
\draw  [fill={rgb, 255:red, 52; green, 160; blue, 255 }  ,fill opacity=1 ] (186.33,160.67) .. controls (189.1,157.9) and (193.59,157.9) .. (196.36,160.67) -- (203.04,167.35) .. controls (205.81,170.12) and (205.81,174.61) .. (203.04,177.38) -- (162.94,217.48) .. controls (160.17,220.25) and (155.68,220.25) .. (152.92,217.48) -- (146.23,210.8) .. controls (143.46,208.03) and (143.46,203.54) .. (146.23,200.77) -- cycle ;
%Rounded Rect [id:dp09101817392213318] 
\draw  [fill={rgb, 255:red, 52; green, 160; blue, 255 }  ,fill opacity=1 ] (187.12,95.29) .. controls (189.89,92.52) and (194.38,92.52) .. (197.15,95.29) -- (203.83,101.97) .. controls (206.6,104.74) and (206.6,109.23) .. (203.83,112) -- (163.73,152.1) .. controls (160.96,154.87) and (156.47,154.87) .. (153.7,152.1) -- (147.02,145.42) .. controls (144.25,142.65) and (144.25,138.16) .. (147.02,135.39) -- cycle ;
%Rounded Rect [id:dp8255449421140859] 
\draw  [fill={rgb, 255:red, 128; green, 185; blue, 255 }  ,fill opacity=1 ] (211.66,140.44) .. controls (211.66,136.53) and (214.83,133.35) .. (218.74,133.35) -- (228.2,133.35) .. controls (232.11,133.35) and (235.29,136.53) .. (235.29,140.44) -- (235.29,211.06) .. controls (235.29,214.98) and (232.11,218.15) .. (228.2,218.15) -- (218.74,218.15) .. controls (214.83,218.15) and (211.66,214.98) .. (211.66,211.06) -- cycle ;
%Rounded Rect [id:dp3321565351634239] 
\draw  [fill={rgb, 255:red, 52; green, 160; blue, 255 }  ,fill opacity=1 ] (252.16,161.34) .. controls (254.93,158.57) and (259.42,158.57) .. (262.19,161.34) -- (268.87,168.03) .. controls (271.64,170.79) and (271.64,175.28) .. (268.87,178.05) -- (228.77,218.15) .. controls (226,220.92) and (221.51,220.92) .. (218.74,218.15) -- (212.06,211.47) .. controls (209.29,208.7) and (209.29,204.21) .. (212.06,201.44) -- cycle ;
%Rounded Rect [id:dp592712766027802] 
\draw  [fill={rgb, 255:red, 52; green, 160; blue, 255 }  ,fill opacity=1 ] (252.95,95.96) .. controls (255.72,93.19) and (260.21,93.19) .. (262.98,95.96) -- (269.66,102.65) .. controls (272.43,105.42) and (272.43,109.9) .. (269.66,112.67) -- (229.56,152.78) .. controls (226.79,155.54) and (222.3,155.54) .. (219.53,152.78) -- (212.85,146.09) .. controls (210.08,143.32) and (210.08,138.83) .. (212.85,136.07) -- cycle ;
\end{mytikz2}\quad.
\end{equation}
The GLDC ansatz is just repeating the FDC ansatz above by $N=12$ times. In the process of training, we use the Adam optimizer with a learning rate of $0.01$. For each ansatz, there are $100$ independent VQE trials with different initialized parameters with each parameter randomly chosen from the uniform distribution over $[0, 2\pi)$. 

\begin{figure}
    \centering
    \includegraphics[width=0.8\linewidth]{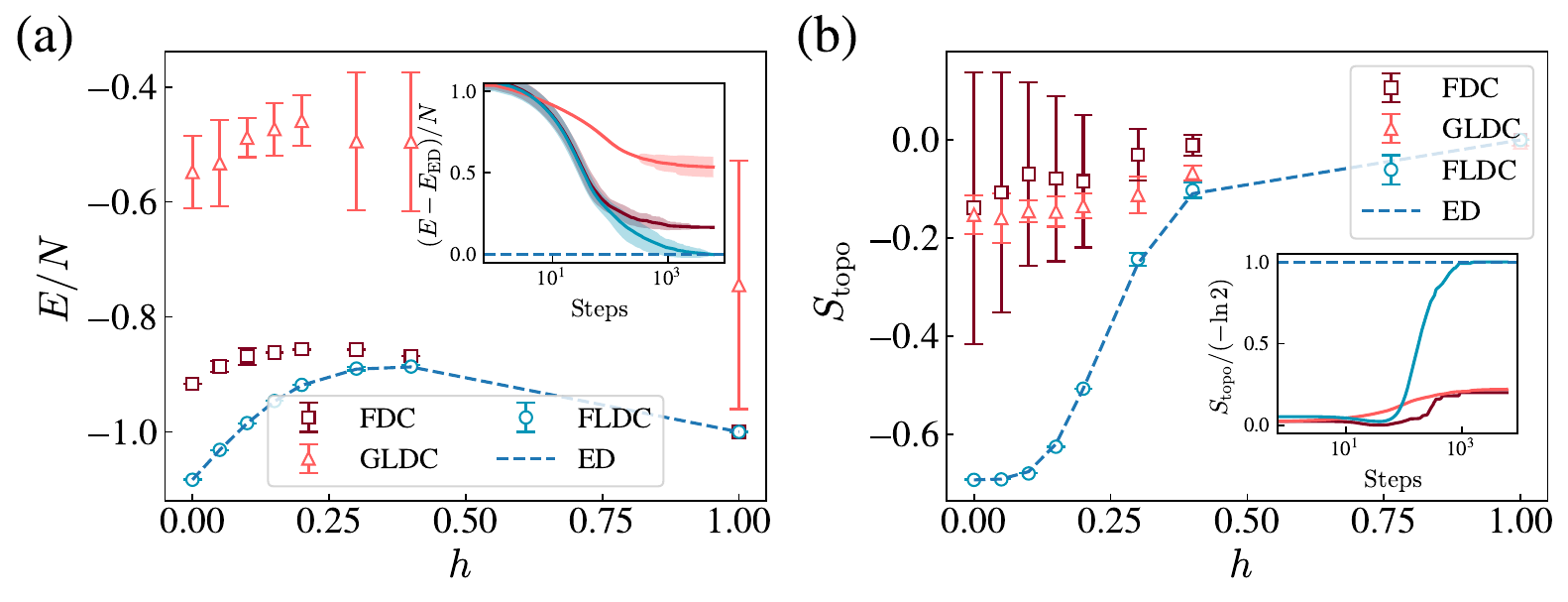}
    \caption{VQE performance comparison of the ``claw-shape'' FDC, FLDC, and GLDC ansatzes using the toric code model under the external field $(h^x, h^y, h^z)=(0,0,h)$ with $N=12$. The data is averaged over the best half of the $100$ training trajectories starting from different initializations. (a) shows the converged energy $E/N$ vs $h$. The inset depicts the energy training dynamics at $h=0$. The dashed lines represent the exact values obtained from ED. The (shaded) error bar represents the standard deviation. (b) shows the topological entanglement entropy $S_{\text{topo}}$ correspondingly.}
    \label{fig:tc_zfield}
\end{figure}

Besides the numerical results shown in the main text, we also investigate the case of $(h^x, h^y, h^z)=(0,0,h)$ and other FLDC ansatzes, where the conclusion remains the same as in the main text, i.e., the FLDCs offers the best performance. The numerical results for the case of $(h^x, h^y, h^z)=(0,0,h)$ are shown in Fig.~\ref{fig:tc_zfield}. The numerical results corresponding to a different FLDC ansatz where the blocks are applied by the following ``U-shape'' as follows
\begin{equation}\label{eq:u_shape}
\begin{mytikz2}
%Straight Lines [id:da09159522868806036] 
\draw [color={rgb, 255:red, 211; green, 211; blue, 211 }  ,draw opacity=1 ]   (193.04,103.88) -- (106.83,103.88) ;
\draw [shift={(106.83,103.88)}, rotate = 180] [color={rgb, 255:red, 211; green, 211; blue, 211 }  ,draw opacity=1 ][fill={rgb, 255:red, 211; green, 211; blue, 211 }  ,fill opacity=1 ][line width=0.75]      (0, 0) circle [x radius= 1.34, y radius= 1.34]   ;
%Straight Lines [id:da32676348953047496] 
\draw [color={rgb, 255:red, 211; green, 211; blue, 211 }  ,draw opacity=1 ]   (106.83,190.1) -- (106.83,103.88) ;
%Straight Lines [id:da019350439368350347] 
\draw [color={rgb, 255:red, 211; green, 211; blue, 211 }  ,draw opacity=1 ]   (194.19,190.1) -- (194.19,103.88) ;
\draw [shift={(194.19,103.88)}, rotate = 270] [color={rgb, 255:red, 211; green, 211; blue, 211 }  ,draw opacity=1 ][fill={rgb, 255:red, 211; green, 211; blue, 211 }  ,fill opacity=1 ][line width=0.75]      (0, 0) circle [x radius= 1.34, y radius= 1.34]   ;
%Straight Lines [id:da8187166621344946] 
\draw [color={rgb, 255:red, 211; green, 211; blue, 211 }  ,draw opacity=1 ]   (193.04,187.8) -- (106.83,187.8) ;
\draw [shift={(106.83,187.8)}, rotate = 180] [color={rgb, 255:red, 211; green, 211; blue, 211 }  ,draw opacity=1 ][fill={rgb, 255:red, 211; green, 211; blue, 211 }  ,fill opacity=1 ][line width=0.75]      (0, 0) circle [x radius= 1.34, y radius= 1.34]   ;
\draw [shift={(193.04,187.8)}, rotate = 180] [color={rgb, 255:red, 211; green, 211; blue, 211 }  ,draw opacity=1 ][fill={rgb, 255:red, 211; green, 211; blue, 211 }  ,fill opacity=1 ][line width=0.75]      (0, 0) circle [x radius= 1.34, y radius= 1.34]   ;
%Rounded Rect [id:dp6399329820622208] 
\draw  [fill={rgb, 255:red, 215; green, 236; blue, 255 }  ,fill opacity=1 ] (164.06,177.99) .. controls (167.43,181.35) and (167.43,186.81) .. (164.06,190.18) -- (155.93,198.31) .. controls (152.57,201.67) and (147.11,201.67) .. (143.74,198.31) -- (94.97,149.54) .. controls (91.61,146.17) and (91.61,140.71) .. (94.97,137.35) -- (103.1,129.22) .. controls (106.47,125.85) and (111.93,125.85) .. (115.29,129.22) -- cycle ;
%Rounded Rect [id:dp3321565351634239] 
\draw  [fill={rgb, 255:red, 145; green, 206; blue, 255 }  ,fill opacity=1 ] (184.38,129.22) .. controls (187.75,125.85) and (193.21,125.85) .. (196.58,129.22) -- (204.7,137.35) .. controls (208.07,140.71) and (208.07,146.17) .. (204.7,149.54) -- (155.93,198.31) .. controls (152.57,201.67) and (147.11,201.67) .. (143.74,198.31) -- (135.61,190.18) .. controls (132.25,186.81) and (132.25,181.35) .. (135.61,177.99) -- cycle ;
%Rounded Rect [id:dp5384823102905523] 
\draw  [fill={rgb, 255:red, 35; green, 159; blue, 255 }  ,fill opacity=1 ] (204.7,137.35) .. controls (208.07,140.71) and (208.07,146.17) .. (204.7,149.54) -- (196.58,157.67) .. controls (193.21,161.03) and (187.75,161.03) .. (184.38,157.67) -- (135.61,108.9) .. controls (132.25,105.53) and (132.25,100.07) .. (135.61,96.71) -- (143.74,88.58) .. controls (147.11,85.21) and (152.57,85.21) .. (155.93,88.58) -- cycle ;
\end{mytikz2}\quad.
\end{equation}
within each plaquatte are shown in Fig.~\ref{fig:tc_zfield_Ushape} and Fig.~\ref{fig:tc_xzfield_Ushape}. The notations are the same as in Fig.~\textcolor{blue}{4} in the main text. The error bar represents the standard deviation over all trials and the shaded error bar in the inset is the standard deviation over all trajectories from different random initializations. Finally, we point out that the practical performance of the FLDC ansatzes may depend on the specific sequence of the blocks in each plaquette. For example, if the two-qubit blocks are applied clockwise in the ``U-shape'' instead of anticlockwise as in Eq.~\eqref{eq:u_shape}, the energy would not converge well to the exact value because the exact ground state of the bare toric code model is not in the space of this FLDC ansatz at all. Therefore, in practice, it is also important to determine the specific structures of FLDC ansatzes appropriate for the target ground state.

\begin{figure}
    \centering
    \includegraphics[width=0.8\linewidth]{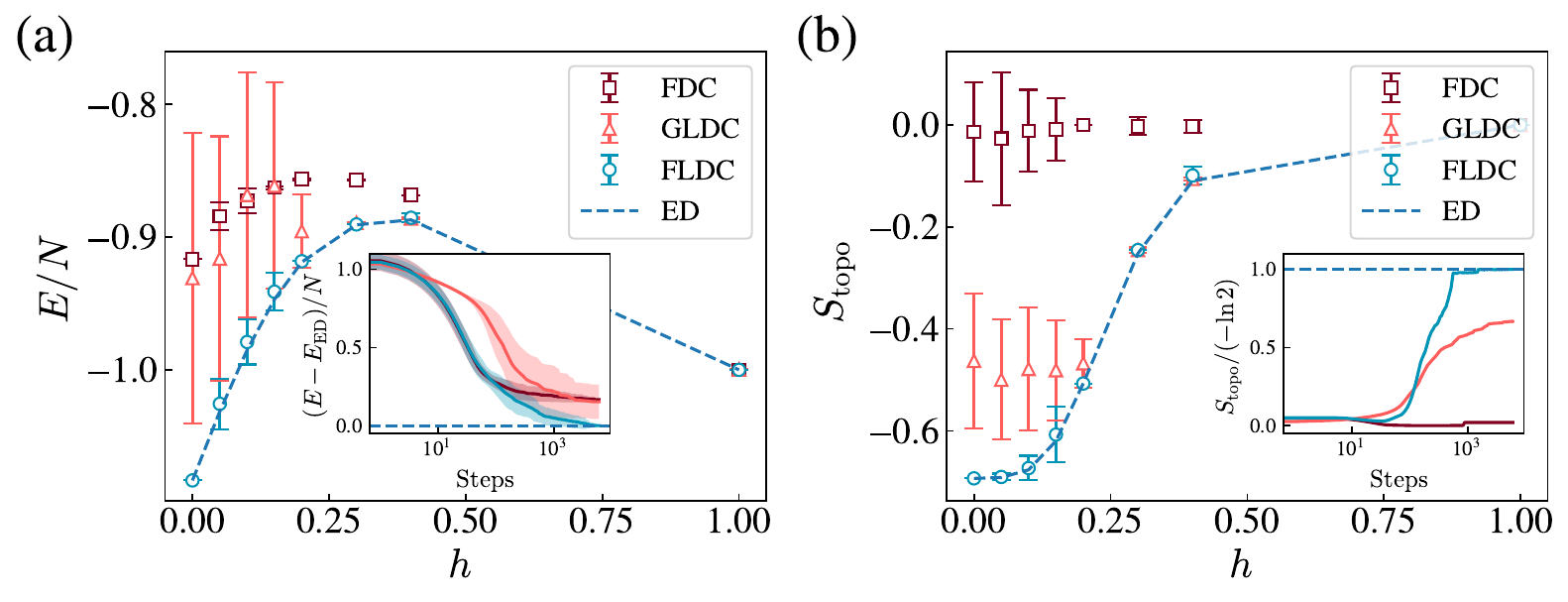}
    \caption{VQE performance comparison of the ``U-shape'' FDC, FLDC, and GLDC ansatzes using the toric code model under the external field $(h^x, h^y, h^z)=(0,0,h)$ with $N=12$. The data is averaged over the best half of the $100$ training trajectories starting from different initializations. (a) shows the converged energy $E/N$ vs $h$. The inset depicts the energy training dynamics at $h=0$. The dashed lines represent the exact values obtained from ED. The (shaded) error bar represents the standard deviation. (b) shows the topological entanglement entropy $S_{\text{topo}}$ correspondingly.}
    \label{fig:tc_zfield_Ushape}
\end{figure}

\begin{figure}
    \centering
    \includegraphics[width=0.8\linewidth]{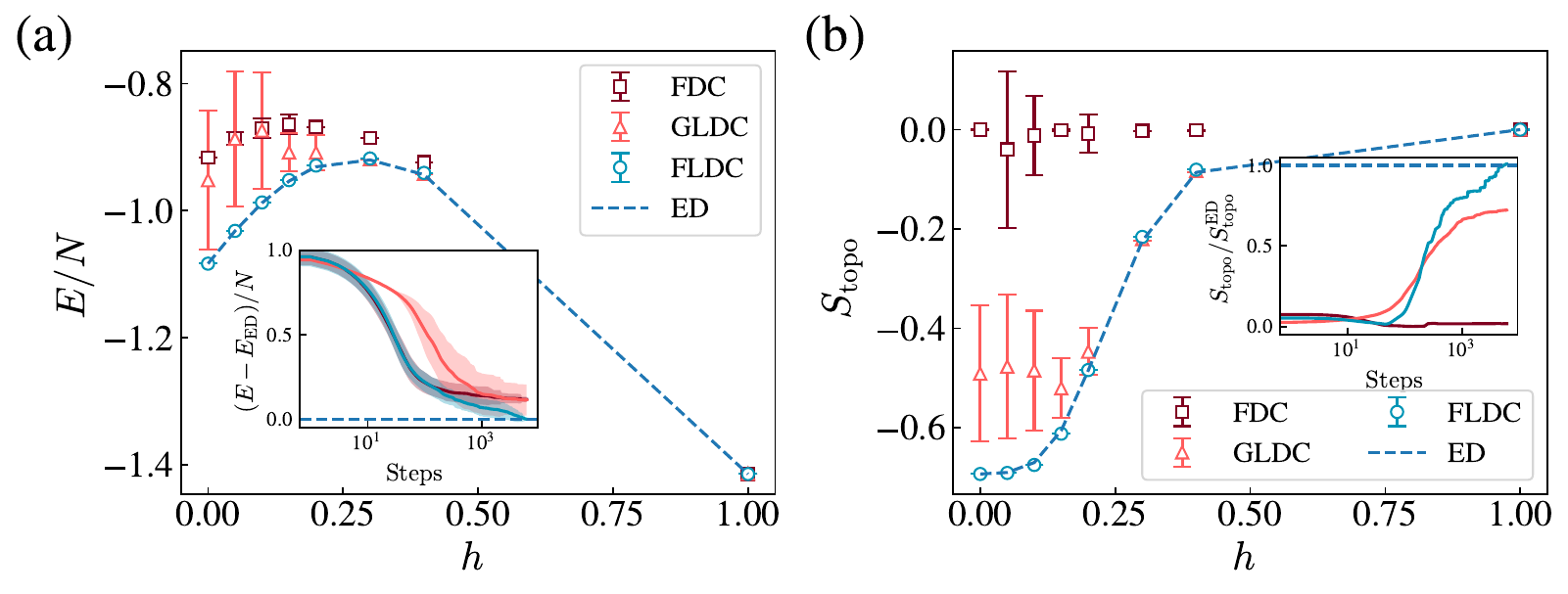}
    \caption{VQE performance comparison of the ``U-shape'' FDC, FLDC, and GLDC ansatzes using the toric code model under the external field $(h^x, h^y, h^z)=(h,0,h)$ with $N=12$. The data is averaged over the best half of the $100$ training trajectories starting from different initializations. (a) shows the converged energy $E/N$ vs $h$. The inset depicts the energy training dynamics at $h=0.1$. The dashed lines represent the exact values obtained from ED. The (shaded) error bar represents the standard deviation. (b) shows the topological entanglement entropy $S_{\text{topo}}$ correspondingly.}
    \label{fig:tc_xzfield_Ushape}
\end{figure}

\section{Classical Simulability of Finite Local-Depth Circuits}

Previous results on the absence of barren plateaus in certain circuit architectures~\cite{Cerezo2021, Uvarov2021a, Pesah2021, Zhao2021a, Barthel2023, Miao2023, CerveroMartin2023, Liu2021a, Liu2023c} mainly focus on constant or logarithmic depth circuits, which can be efficiently simulated classically to estimate local observables using the known methods due to the small causal cone and the small treewidths~\cite{Markov2008}. This is a matter of recent concern~\cite{Cerezo2023}, i.e., the classical simulability of the tasks with the provable absence of barren plateaus (BP). In the following, we will show that the known related classical methods in tensor networks~\cite{Banuls2008, Zaletel2020, Soejima2020, Haferkamp2020, Bravyi2021, Haller2023, Liu2023} cannot efficiently simulate finite local-depth circuits (FLDC) in general for the ground state preparation task (FLDC in 2 dimensions and above specifically). In other words, to the best of our knowledge, FLDC is the first circuit class that is proven to be BP-free and at the same time cannot be efficiently simulated by existing classical methods. We will provide a detailed discussion in this section and explain why the FDLC class contains a path to quantum advantage. 

First of all, we would like to clarify a fact to facilitate subsequent discussions. The task of ground state preparation is closely tied to the estimation of observables. Merely storing the ground state wavefunction in a classical or quantum memory cannot be considered as accomplishing a meaningful task completely. This is attributed to three reasons: (a) The classical and quantum outputs of ground state preparation take different forms and are not directly comparable. To establish comparability, it is necessary to designate the final output as a measurement result, such as the ground state energy. (b) In practical physical or chemical applications, the ground state wavefunction and even the ground state energy are often considered just as intermediate results. The ultimate concern is to know the properties of the ground state by measuring other observables on the ground state, such as order parameters, correlation functions, entanglement entropy, and so forth~\cite{Liu2023, Zheng2017, Qin2020, Huang2017, Liu2022d}. (c) In practice, the process of optimizing a variational ansatz to obtain the ground state inherently involves the estimation of observables. For instance, in VQE, the optimization is based on estimating the energy of the ansatz state, i.e., the expectation value of the Hamiltonian, which is usually obtained by measuring each summed term in the Hamiltonian. Therefore, to discuss the classical simulability clearly, we will investigate the following two cases respectively:
\begin{enumerate}[(i)]
    \item Preparing the ground state and estimating the ground state energy (or say the cost function in variational quantum algorithms). Here we focus on spatially local Hamiltonians so that only local observables need to be estimated.
    \item Preparing the ground state and estimating various quantities of interest, including the ground state energy, order parameters, correlation functions, and so on. Note that in this case, the observables that need to be estimated can be spatially local, or spatially non-local but still few-body, or even many-body, etc. These features will further increase the difficulty of classical simulation.
\end{enumerate}

We first consider the case (i). We will categorize the discussion in terms of circuit depth. In advance, we clarify that when mentioning the spatial dimension of a quantum circuit, we actually assume the existence of gate locality. In other words, we assume that there is a qubit connectivity graph such that the gates in the circuit only act on a few neighboring qubits. We refer to the circuits with gate locality as ``local circuits''~\cite{Chen2023a}. If the connectivity graph is a $D$-dimensional lattice (or say grid), we say the circuit has a spatial dimension of $D$. The local observables mentioned in case (i) mean that the support of the observables only involves a few neighboring qubits in the connectivity graph. Please note the distinction between the concepts of ``local'' and ``few-body''. Local naturally implies few-body (the support of the operator does not scale with the system size), together with the condition that the corresponding few qubits must be neighboring, while few-body does not necessarily imply local, which means that the involved few qubits can be far apart in the connectivity graph. In the subsequent discussions, we assume that quantum gates are all few-body, and when talking about circuits with spatial dimensions, gates are all local, unless otherwise stated.
\begin{itemize}
    \item For shallow circuits of constant or logarithmic depth~\cite{Cerezo2021, Uvarov2021a, Pesah2021, Zhao2021a, Barthel2023, Miao2023, CerveroMartin2023} such as the shallow brickwall circuits~\cite{Cerezo2021, Uvarov2021a}, the quantum convolutional neural networks (QCNN)~\cite{Pesah2021}, and the multi-scale entanglement renormalization ansatzes (MERA)~\cite{Vidal2008, Zhao2021a, Barthel2023, Miao2023, CerveroMartin2023}, local observables are easy to estimate~\cite{Cerezo2023} even regardless of gate locality. This can be easily understood by the fact that the backward causal cone of the local observable in the circuit can either be truncated by the constant depth like (the orange square indicate the location of the local observable and the orange shaded area indicate the causal cone)
    \begin{equation}\label{eq:causal_cone_fdc}
    \begin{mytikz2}
    %Straight Lines [id:da26034887043359056] 
    \draw [line width=0.75]    (124.97,83.82) -- (200.19,84.03) ;
    %Straight Lines [id:da7663859203170504] 
    \draw [line width=0.75]    (124.95,117.36) -- (200.17,117.57) ;
    %Straight Lines [id:da3247408669950853] 
    \draw [line width=0.75]    (124.98,50.28) -- (200.2,50.49) ;
    %Rounded Rect [id:dp18603953379337201] 
    \draw  [fill={rgb, 255:red, 204; green, 230; blue, 255 }  ,fill opacity=1 ][line width=0.75]  (141.81,46.93) .. controls (141.81,44.16) and (144.06,41.92) .. (146.83,41.92) -- (153.52,41.93) .. controls (156.29,41.93) and (158.53,44.18) .. (158.53,46.95) -- (158.48,87.22) .. controls (158.48,89.99) and (156.23,92.24) .. (153.46,92.23) -- (146.78,92.23) .. controls (144.01,92.22) and (141.76,89.98) .. (141.77,87.21) -- cycle ;
    %Straight Lines [id:da9305890289712528] 
    \draw [line width=0.75]    (124.93,184.83) -- (200.15,185.04) ;
    %Straight Lines [id:da9563336620512131] 
    \draw [line width=0.75]    (124.91,218.37) -- (200.13,218.58) ;
    %Straight Lines [id:da931472415391297] 
    \draw [line width=0.75]    (124.94,150.89) -- (200.16,151.1) ;
    %Rounded Rect [id:dp05941797307008834] 
    \draw  [fill={rgb, 255:red, 204; green, 230; blue, 255 }  ,fill opacity=1 ][line width=0.75]  (166.99,80.49) .. controls (166.99,77.72) and (169.24,75.48) .. (172.01,75.48) -- (178.7,75.49) .. controls (181.47,75.49) and (183.71,77.74) .. (183.71,80.51) -- (183.67,120.79) .. controls (183.66,123.56) and (181.41,125.8) .. (178.65,125.8) -- (171.96,125.79) .. controls (169.19,125.79) and (166.95,123.54) .. (166.95,120.77) -- cycle ;
    %Rounded Rect [id:dp5173764752779064] 
    \draw  [fill={rgb, 255:red, 204; green, 230; blue, 255 }  ,fill opacity=1 ][line width=0.75]  (141.79,114.03) .. controls (141.8,111.26) and (144.05,109.01) .. (146.82,109.02) -- (153.5,109.02) .. controls (156.27,109.03) and (158.52,111.27) .. (158.51,114.04) -- (158.47,154.32) .. controls (158.47,157.09) and (156.22,159.33) .. (153.45,159.33) -- (146.76,159.32) .. controls (143.99,159.32) and (141.75,157.07) .. (141.75,154.3) -- cycle ;
    %Rounded Rect [id:dp4917715848910389] 
    \draw  [fill={rgb, 255:red, 204; green, 230; blue, 255 }  ,fill opacity=1 ][line width=0.75]  (166.92,147.59) .. controls (166.92,144.82) and (169.17,142.58) .. (171.94,142.58) -- (178.63,142.59) .. controls (181.4,142.59) and (183.64,144.84) .. (183.64,147.61) -- (183.6,187.88) .. controls (183.59,190.65) and (181.34,192.89) .. (178.57,192.89) -- (171.89,192.88) .. controls (169.12,192.88) and (166.87,190.63) .. (166.88,187.86) -- cycle ;
    %Rounded Rect [id:dp5108488180284698] 
    \draw  [fill={rgb, 255:red, 204; green, 230; blue, 255 }  ,fill opacity=1 ][line width=0.75]  (141.72,181.09) .. controls (141.73,178.32) and (143.98,176.08) .. (146.75,176.08) -- (153.43,176.09) .. controls (156.2,176.09) and (158.45,178.34) .. (158.44,181.11) -- (158.4,221.38) .. controls (158.4,224.15) and (156.15,226.39) .. (153.38,226.39) -- (146.69,226.38) .. controls (143.92,226.38) and (141.68,224.13) .. (141.68,221.36) -- cycle ;
    %Straight Lines [id:da7202516304828359] 
    \draw [line width=0.75]    (124.96,252.26) -- (200.18,252.47) ;
    %Straight Lines [id:da8091987322093168] 
    \draw [line width=0.75]    (124.95,285.8) -- (200.17,286.01) ;
    %Rounded Rect [id:dp7906362004756282] 
    \draw  [fill={rgb, 255:red, 204; green, 230; blue, 255 }  ,fill opacity=1 ][line width=0.75]  (167.02,215.02) .. controls (167.02,212.25) and (169.27,210) .. (172.04,210.01) -- (178.72,210.01) .. controls (181.49,210.02) and (183.74,212.26) .. (183.73,215.03) -- (183.69,255.31) .. controls (183.69,258.08) and (181.44,260.32) .. (178.67,260.32) -- (171.98,260.31) .. controls (169.21,260.31) and (166.97,258.06) .. (166.97,255.29) -- cycle ;
    %Rounded Rect [id:dp578854246173367] 
    \draw  [fill={rgb, 255:red, 204; green, 230; blue, 255 }  ,fill opacity=1 ][line width=0.75]  (141.82,248.52) .. controls (141.82,245.75) and (144.07,243.5) .. (146.84,243.51) -- (153.53,243.51) .. controls (156.3,243.52) and (158.54,245.77) .. (158.54,248.54) -- (158.5,288.81) .. controls (158.49,291.58) and (156.24,293.82) .. (153.47,293.82) -- (146.79,293.81) .. controls (144.02,293.81) and (141.77,291.56) .. (141.78,288.79) -- cycle ;
    %Rounded Rect [id:dp6170842269955943] 
    \draw  [fill={rgb, 255:red, 245; green, 166; blue, 35 }  ,fill opacity=1 ][line width=0.75]  (200.02,244) .. controls (200.02,244) and (200.02,244) .. (200.02,244) -- (216.74,244.02) .. controls (216.74,244.02) and (216.74,244.02) .. (216.74,244.02) -- (216.72,261.44) .. controls (216.72,261.44) and (216.72,261.44) .. (216.72,261.44) -- (200,261.42) .. controls (200,261.42) and (200,261.42) .. (200,261.42) -- cycle ;
    %Shape: Trapezoid [id:dp11255734585767918] 
    \draw  [color={rgb, 255:red, 0; green, 0; blue, 0 }  ,draw opacity=0.16 ][fill={rgb, 255:red, 245; green, 166; blue, 35 }  ,fill opacity=0.2 ][dash pattern={on 4.5pt off 4.5pt}] (124.7,156.43) -- (194.61,216) -- (194.68,254.6) -- (125,314.43) -- cycle ;
    %Straight Lines [id:da6320979313658373] 
    \draw [line width=0.75]    (124.95,318.86) -- (200.17,319.07) ;
    %Straight Lines [id:da5318218364854821] 
    \draw [line width=0.75]    (124.94,352.39) -- (200.16,352.6) ;
    %Rounded Rect [id:dp0833361431173425] 
    \draw  [fill={rgb, 255:red, 204; green, 230; blue, 255 }  ,fill opacity=1 ][line width=0.75]  (166.99,281.99) .. controls (166.99,279.22) and (169.24,276.98) .. (172.01,276.98) -- (178.7,276.99) .. controls (181.47,276.99) and (183.71,279.24) .. (183.71,282.01) -- (183.67,322.29) .. controls (183.66,325.06) and (181.41,327.3) .. (178.65,327.3) -- (171.96,327.29) .. controls (169.19,327.29) and (166.95,325.04) .. (166.95,322.27) -- cycle ;
    %Rounded Rect [id:dp6079897912378207] 
    \draw  [fill={rgb, 255:red, 204; green, 230; blue, 255 }  ,fill opacity=1 ][line width=0.75]  (141.79,315.53) .. controls (141.8,312.76) and (144.05,310.51) .. (146.82,310.52) -- (153.5,310.52) .. controls (156.27,310.53) and (158.52,312.77) .. (158.51,315.54) -- (158.47,355.82) .. controls (158.47,358.59) and (156.22,360.83) .. (153.45,360.83) -- (146.76,360.82) .. controls (143.99,360.82) and (141.75,358.57) .. (141.75,355.8) -- cycle ;
    
    \end{mytikz2}\quad,
    \end{equation}
    or only have a small treewidth in log-depth circuits~\cite{Cerezo2023, Markov2008}. Note that the treewidth characterizes how close a graph is to a tree~\cite{Markov2008}, and tree tensor networks can be efficiently contracted because the degrees of the tensors would not increase when contracting from the leaves to the root. It is known that the computational overhead for contracting a tensor network (TN) grows exponentially with the treewidth in the worst case~\cite{Markov2008}. Thus, estimating local observables in shallow circuits can be simulated classically within polynomial time by determining the causal cone and contracting the corresponding tensor networks. However, we also point out that in some practical cases, due to potential large constant factors and power exponents in the complexity scaling, classical computers may still require very long times to contract the corresponding tensor networks, such as MERA in higher spatial dimensions.
    
    \item For FLDCs, the spatial dimension is of great importance and needs to be handled carefully. We will categorize and discuss them. We remark that a significant feature of local FLDCs is that the generated quantum states satisfy the entanglement area law (or say boundary law) because the number of gates acting across any simple partition boundary entangling the two sides can be upper bounded by the local depth times the size of the boundary. This feature makes them form a subclass of the projected entangled paired states (PEPS)~\cite{Verstraete2006, Schollwock2011} of the corresponding spatial dimension. 1D PEPS is just the class of matrix product states (MPS)~\cite{Schollwock2011}.
    \begin{itemize}
        \item[$\circ$] 1D FLDCs are classically easy to simulate using MPS. The technical reason is that 1D FLDCs have chain-like structures and hence have constant treewidths. To be specific, unlike in shallow circuits, the causal cone can be large in a 1D FLDC, especially for the local observable near the gates that act later in the sequence, as depicted below
        \begin{equation}
        \begin{mytikz2}
        %Straight Lines [id:da2586891076082667] 
        \draw [line width=0.75]    (269.65,84.25) -- (470.15,84.46) ;
        %Straight Lines [id:da20748869549307059] 
        \draw [line width=0.75]    (269.62,117.79) -- (470.12,117.99) ;
        %Straight Lines [id:da38096977334417703] 
        \draw [line width=0.75]    (269.69,50.71) -- (470.19,50.92) ;
        %Rounded Rect [id:dp8083712581638443] 
        \draw  [fill={rgb, 255:red, 204; green, 230; blue, 255 }  ,fill opacity=1 ][line width=0.75]  (286.41,47.36) .. controls (286.41,44.59) and (288.66,42.35) .. (291.43,42.35) -- (298.12,42.36) .. controls (300.89,42.36) and (303.13,44.61) .. (303.13,47.38) -- (303.08,87.65) .. controls (303.08,90.42) and (300.83,92.66) .. (298.06,92.66) -- (291.38,92.65) .. controls (288.61,92.65) and (286.36,90.4) .. (286.37,87.63) -- cycle ;
        %Straight Lines [id:da5420379814258982] 
        \draw [line width=0.75]    (269.55,184.86) -- (470.05,185.07) ;
        %Straight Lines [id:da970375029530645] 
        \draw [line width=0.75]    (269.51,218.4) -- (470.01,218.6) ;
        %Straight Lines [id:da9630113263532345] 
        \draw [line width=0.75]    (269.58,151.32) -- (470.08,151.53) ;
        %Rounded Rect [id:dp7732426873966081] 
        \draw  [fill={rgb, 255:red, 204; green, 230; blue, 255 }  ,fill opacity=1 ][line width=0.75]  (311.59,80.92) .. controls (311.59,78.15) and (313.84,75.91) .. (316.61,75.91) -- (323.3,75.92) .. controls (326.07,75.92) and (328.31,78.17) .. (328.31,80.94) -- (328.27,121.21) .. controls (328.26,123.98) and (326.01,126.23) .. (323.25,126.22) -- (316.56,126.22) .. controls (313.79,126.21) and (311.55,123.97) .. (311.55,121.2) -- cycle ;
        %Rounded Rect [id:dp31936149425274474] 
        \draw  [fill={rgb, 255:red, 204; green, 230; blue, 255 }  ,fill opacity=1 ][line width=0.75]  (336.72,114.49) .. controls (336.72,111.72) and (338.97,109.47) .. (341.74,109.48) -- (348.42,109.48) .. controls (351.19,109.48) and (353.44,111.73) .. (353.43,114.5) -- (353.39,154.78) .. controls (353.39,157.55) and (351.14,159.79) .. (348.37,159.79) -- (341.68,159.78) .. controls (338.91,159.78) and (336.67,157.53) .. (336.67,154.76) -- cycle ;
        %Rounded Rect [id:dp5029059744205102] 
        \draw  [fill={rgb, 255:red, 204; green, 230; blue, 255 }  ,fill opacity=1 ][line width=0.75]  (361.84,148.05) .. controls (361.84,145.28) and (364.09,143.04) .. (366.86,143.04) -- (373.55,143.05) .. controls (376.32,143.05) and (378.56,145.3) .. (378.56,148.07) -- (378.52,188.34) .. controls (378.51,191.11) and (376.27,193.35) .. (373.5,193.35) -- (366.81,193.34) .. controls (364.04,193.34) and (361.8,191.09) .. (361.8,188.32) -- cycle ;
        %Rounded Rect [id:dp9117877744526623] 
        \draw  [fill={rgb, 255:red, 204; green, 230; blue, 255 }  ,fill opacity=1 ][line width=0.75]  (386.97,181.61) .. controls (386.97,178.84) and (389.22,176.6) .. (391.99,176.6) -- (398.67,176.61) .. controls (401.44,176.61) and (403.69,178.86) .. (403.68,181.63) -- (403.64,221.9) .. controls (403.64,224.67) and (401.39,226.92) .. (398.62,226.91) -- (391.93,226.91) .. controls (389.16,226.9) and (386.92,224.66) .. (386.92,221.89) -- cycle ;
        %Straight Lines [id:da6164659805779229] 
        \draw [line width=0.75]    (269.64,252.29) -- (470.14,252.5) ;
        %Straight Lines [id:da6695305567352068] 
        \draw [line width=0.75]    (269.61,285.83) -- (470.11,286.03) ;
        %Rounded Rect [id:dp7616752937911562] 
        \draw  [fill={rgb, 255:red, 204; green, 230; blue, 255 }  ,fill opacity=1 ][line width=0.75]  (412.26,215.53) .. controls (412.26,212.76) and (414.51,210.52) .. (417.28,210.52) -- (423.96,210.53) .. controls (426.73,210.53) and (428.98,212.78) .. (428.97,215.55) -- (428.93,255.82) .. controls (428.93,258.59) and (426.68,260.83) .. (423.91,260.83) -- (417.22,260.82) .. controls (414.45,260.82) and (412.21,258.57) .. (412.21,255.8) -- cycle ;
        %Rounded Rect [id:dp6168210919183745] 
        \draw  [fill={rgb, 255:red, 204; green, 230; blue, 255 }  ,fill opacity=1 ][line width=0.75]  (437.38,249.09) .. controls (437.38,246.32) and (439.63,244.08) .. (442.4,244.08) -- (449.09,244.09) .. controls (451.86,244.09) and (454.1,246.34) .. (454.1,249.11) -- (454.06,289.38) .. controls (454.05,292.15) and (451.81,294.4) .. (449.04,294.39) -- (442.35,294.39) .. controls (439.58,294.38) and (437.34,292.14) .. (437.34,289.37) -- cycle ;
        %Rounded Rect [id:dp9555566614778821] 
        \draw  [fill={rgb, 255:red, 245; green, 166; blue, 35 }  ,fill opacity=1 ][line width=0.75]  (469.22,243.6) .. controls (469.22,243.6) and (469.22,243.6) .. (469.22,243.6) -- (485.94,243.62) .. controls (485.94,243.62) and (485.94,243.62) .. (485.94,243.62) -- (485.92,261.04) .. controls (485.92,261.04) and (485.92,261.04) .. (485.92,261.04) -- (469.2,261.02) .. controls (469.2,261.02) and (469.2,261.02) .. (469.2,261.02) -- cycle ;
        %Shape: Polygon [id:ds606396947421904] 
        \draw  [color={rgb, 255:red, 0; green, 0; blue, 0 }  ,draw opacity=0.16 ][fill={rgb, 255:red, 245; green, 166; blue, 35 }  ,fill opacity=0.2 ][dash pattern={on 4.5pt off 4.5pt}] (275.04,4.36) -- (460.2,247.6) -- (460.24,297.56) -- (274.64,296.76) -- cycle ;
        
        \end{mytikz2}\quad.
        \end{equation}
        Nonetheless, a chain-like circuit structure arises thanks to the 1D geometry and the constant local-depth. This chain-like circuit further gives rise to a chain-like TN structure for the expectation value $\braoprket{\psi}{O}{\psi}$ 
                \begin{equation}\label{eq:1d_fdlc_chain}
        \begin{mytikz2}
        %Straight Lines [id:da8280614218807647] 
        \draw [line width=0.75]    (596.94,49.86) -- (550.94,49.86) ;
        %Straight Lines [id:da6603883137664834] 
        \draw [line width=0.75]    (597.44,84.36) -- (551.44,84.36) ;
        %Straight Lines [id:da2774635717368257] 
        \draw [line width=0.75]    (596.94,116.36) -- (550.94,116.36) ;
        %Straight Lines [id:da7291566751397869] 
        \draw [line width=0.75]    (597.44,150.86) -- (551.44,150.86) ;
        %Straight Lines [id:da11956796439131301] 
        \draw [line width=0.75]    (597.44,184.86) -- (551.44,184.86) ;
        %Straight Lines [id:da07372799810088737] 
        \draw [line width=0.75]    (597.94,219.36) -- (551.94,219.36) ;
        %Straight Lines [id:da9419865186641809] 
        \draw [line width=0.75]    (598.44,49.86) -- (598.8,252.83) ;
        %Straight Lines [id:da7112267733341802] 
        \draw [line width=0.75]    (550.44,49.86) -- (550.2,253.15) ;
        %Straight Lines [id:da2586891076082667] 
        \draw [line width=0.75]    (20.99,83.58) -- (221.49,83.79) ;
        %Straight Lines [id:da20748869549307059] 
        \draw [line width=0.75]    (20.95,117.12) -- (221.45,117.33) ;
        %Straight Lines [id:da38096977334417703] 
        \draw [line width=0.75]    (21.02,50.04) -- (221.52,50.25) ;
        %Rounded Rect [id:dp8083712581638443] 
        \draw  [fill={rgb, 255:red, 204; green, 230; blue, 255 }  ,fill opacity=1 ][line width=0.75]  (37.74,46.69) .. controls (37.75,43.92) and (39.99,41.68) .. (42.76,41.68) -- (49.45,41.69) .. controls (52.22,41.69) and (54.46,43.94) .. (54.46,46.71) -- (54.42,86.98) .. controls (54.42,89.75) and (52.17,92) .. (49.4,91.99) -- (42.71,91.99) .. controls (39.94,91.98) and (37.7,89.74) .. (37.7,86.97) -- cycle ;
        %Straight Lines [id:da5420379814258982] 
        \draw [line width=0.75]    (20.88,184.19) -- (221.38,184.4) ;
        %Straight Lines [id:da970375029530645] 
        \draw [line width=0.75]    (20.85,217.74) -- (221.35,217.94) ;
        %Straight Lines [id:da9630113263532345] 
        \draw [line width=0.75]    (20.92,150.65) -- (221.42,150.86) ;
        %Rounded Rect [id:dp7732426873966081] 
        \draw  [fill={rgb, 255:red, 204; green, 230; blue, 255 }  ,fill opacity=1 ][line width=0.75]  (62.92,80.26) .. controls (62.93,77.49) and (65.17,75.24) .. (67.94,75.25) -- (74.63,75.25) .. controls (77.4,75.26) and (79.64,77.5) .. (79.64,80.27) -- (79.6,120.55) .. controls (79.6,123.32) and (77.35,125.56) .. (74.58,125.56) -- (67.89,125.55) .. controls (65.12,125.55) and (62.88,123.3) .. (62.88,120.53) -- cycle ;
        %Rounded Rect [id:dp31936149425274474] 
        \draw  [fill={rgb, 255:red, 204; green, 230; blue, 255 }  ,fill opacity=1 ][line width=0.75]  (88.05,113.82) .. controls (88.05,111.05) and (90.3,108.81) .. (93.07,108.81) -- (99.76,108.82) .. controls (102.53,108.82) and (104.77,111.07) .. (104.77,113.84) -- (104.72,154.11) .. controls (104.72,156.88) and (102.47,159.12) .. (99.7,159.12) -- (93.02,159.11) .. controls (90.25,159.11) and (88,156.86) .. (88.01,154.09) -- cycle ;
        %Rounded Rect [id:dp5029059744205102] 
        \draw  [fill={rgb, 255:red, 204; green, 230; blue, 255 }  ,fill opacity=1 ][line width=0.75]  (113.17,147.38) .. controls (113.18,144.61) and (115.42,142.37) .. (118.19,142.37) -- (124.88,142.38) .. controls (127.65,142.38) and (129.89,144.63) .. (129.89,147.4) -- (129.85,187.67) .. controls (129.85,190.44) and (127.6,192.69) .. (124.83,192.68) -- (118.14,192.68) .. controls (115.37,192.67) and (113.13,190.43) .. (113.13,187.66) -- cycle ;
        %Rounded Rect [id:dp9117877744526623] 
        \draw  [fill={rgb, 255:red, 204; green, 230; blue, 255 }  ,fill opacity=1 ][line width=0.75]  (138.3,180.94) .. controls (138.3,178.17) and (140.55,175.93) .. (143.32,175.93) -- (150.01,175.94) .. controls (152.78,175.94) and (155.02,178.19) .. (155.02,180.96) -- (154.97,221.24) .. controls (154.97,224.01) and (152.72,226.25) .. (149.95,226.25) -- (143.27,226.24) .. controls (140.5,226.24) and (138.25,223.99) .. (138.26,221.22) -- cycle ;
        %Straight Lines [id:da6164659805779229] 
        \draw [line width=0.75]    (19.98,251.62) -- (220.48,251.83) ;
        %Straight Lines [id:da6695305567352068] 
        \draw [line width=0.75]    (19.94,285.16) -- (220.44,285.37) ;
        %Rounded Rect [id:dp7616752937911562] 
        \draw  [fill={rgb, 255:red, 204; green, 230; blue, 255 }  ,fill opacity=1 ][line width=0.75]  (163.59,214.86) .. controls (163.59,212.09) and (165.84,209.85) .. (168.61,209.85) -- (175.3,209.86) .. controls (178.07,209.86) and (180.31,212.11) .. (180.31,214.88) -- (180.26,255.15) .. controls (180.26,257.92) and (178.01,260.17) .. (175.24,260.16) -- (168.56,260.16) .. controls (165.79,260.16) and (163.54,257.91) .. (163.55,255.14) -- cycle ;
        %Rounded Rect [id:dp6168210919183745] 
        \draw  [fill={rgb, 255:red, 204; green, 230; blue, 255 }  ,fill opacity=1 ][line width=0.75]  (188.71,248.43) .. controls (188.72,245.66) and (190.97,243.41) .. (193.74,243.42) -- (200.42,243.42) .. controls (203.19,243.43) and (205.43,245.67) .. (205.43,248.44) -- (205.39,288.72) .. controls (205.39,291.49) and (203.14,293.73) .. (200.37,293.73) -- (193.68,293.72) .. controls (190.91,293.72) and (188.67,291.47) .. (188.67,288.7) -- cycle ;
        %Shape: Triangle [id:dp2883780177000226] 
        \draw  [fill={rgb, 255:red, 222; green, 241; blue, 198 }  ,fill opacity=1 ] (7.67,49.87) -- (21.19,39.61) -- (21.05,60.33) -- cycle ;
        %Shape: Triangle [id:dp40151798968709174] 
        \draw  [fill={rgb, 255:red, 222; green, 241; blue, 198 }  ,fill opacity=1 ] (7.34,83.54) -- (20.86,73.28) -- (20.71,94) -- cycle ;
        %Shape: Triangle [id:dp23134685119742593] 
        \draw  [fill={rgb, 255:red, 222; green, 241; blue, 198 }  ,fill opacity=1 ] (7.34,117.54) -- (20.86,107.28) -- (20.71,128) -- cycle ;
        %Shape: Triangle [id:dp542950634258619] 
        \draw  [fill={rgb, 255:red, 222; green, 241; blue, 198 }  ,fill opacity=1 ] (7.01,151.21) -- (20.53,140.94) -- (20.38,161.66) -- cycle ;
        %Shape: Triangle [id:dp8163703189516667] 
        \draw  [fill={rgb, 255:red, 222; green, 241; blue, 198 }  ,fill opacity=1 ] (7.01,183.87) -- (20.53,173.61) -- (20.38,194.33) -- cycle ;
        %Shape: Triangle [id:dp31233389393939426] 
        \draw  [fill={rgb, 255:red, 222; green, 241; blue, 198 }  ,fill opacity=1 ] (6.67,217.54) -- (20.19,207.28) -- (20.05,228) -- cycle ;
        %Shape: Triangle [id:dp5098459858230888] 
        \draw  [fill={rgb, 255:red, 222; green, 241; blue, 198 }  ,fill opacity=1 ] (6.67,251.54) -- (20.19,241.28) -- (20.05,262) -- cycle ;
        %Shape: Triangle [id:dp4656976136753519] 
        \draw  [fill={rgb, 255:red, 222; green, 241; blue, 198 }  ,fill opacity=1 ] (6.34,285.21) -- (19.86,274.94) -- (19.71,295.66) -- cycle ;
        %Straight Lines [id:da29910224422056] 
        \draw [line width=0.75]    (419.05,83.58) -- (218.84,83.79) ;
        %Straight Lines [id:da7728897989490153] 
        \draw [line width=0.75]    (419.08,117.12) -- (218.87,117.33) ;
        %Straight Lines [id:da269217825332567] 
        \draw [line width=0.75]    (419.01,50.04) -- (218.8,50.25) ;
        %Rounded Rect [id:dp7505328958318567] 
        \draw  [fill={rgb, 255:red, 204; green, 230; blue, 255 }  ,fill opacity=1 ][line width=0.75]  (402.32,46.68) .. controls (402.31,43.92) and (400.07,41.68) .. (397.3,41.68) -- (390.63,41.69) .. controls (387.86,41.69) and (385.62,43.94) .. (385.62,46.7) -- (385.67,86.99) .. controls (385.67,89.76) and (387.91,92) .. (390.68,91.99) -- (397.36,91.99) .. controls (400.12,91.98) and (402.36,89.74) .. (402.36,86.97) -- cycle ;
        %Straight Lines [id:da43492197041417024] 
        \draw [line width=0.75]    (419.15,184.19) -- (218.94,184.4) ;
        %Straight Lines [id:da8874587175774764] 
        \draw [line width=0.75]    (419.19,217.74) -- (218.98,217.94) ;
        %Straight Lines [id:da12880249908994368] 
        \draw [line width=0.75]    (419.12,150.65) -- (218.91,150.86) ;
        %Rounded Rect [id:dp48257538315945214] 
        \draw  [fill={rgb, 255:red, 204; green, 230; blue, 255 }  ,fill opacity=1 ][line width=0.75]  (377.17,80.25) .. controls (377.17,77.48) and (374.92,75.24) .. (372.16,75.25) -- (365.48,75.25) .. controls (362.71,75.26) and (360.48,77.5) .. (360.48,80.27) -- (360.52,120.55) .. controls (360.52,123.32) and (362.77,125.56) .. (365.53,125.56) -- (372.21,125.55) .. controls (374.98,125.55) and (377.22,123.3) .. (377.21,120.54) -- cycle ;
        %Rounded Rect [id:dp06713966013013173] 
        \draw  [fill={rgb, 255:red, 204; green, 230; blue, 255 }  ,fill opacity=1 ][line width=0.75]  (352.08,113.81) .. controls (352.08,111.05) and (349.83,108.81) .. (347.07,108.81) -- (340.39,108.82) .. controls (337.63,108.82) and (335.39,111.06) .. (335.39,113.83) -- (335.43,154.12) .. controls (335.43,156.88) and (337.68,159.12) .. (340.44,159.12) -- (347.12,159.11) .. controls (349.89,159.11) and (352.13,156.87) .. (352.12,154.1) -- cycle ;
        %Rounded Rect [id:dp7093947972767793] 
        \draw  [fill={rgb, 255:red, 204; green, 230; blue, 255 }  ,fill opacity=1 ][line width=0.75]  (326.99,147.37) .. controls (326.99,144.61) and (324.75,142.37) .. (321.98,142.37) -- (315.3,142.38) .. controls (312.54,142.38) and (310.3,144.63) .. (310.3,147.39) -- (310.34,187.68) .. controls (310.34,190.45) and (312.59,192.69) .. (315.35,192.68) -- (322.03,192.68) .. controls (324.8,192.67) and (327.04,190.43) .. (327.04,187.66) -- cycle ;
        %Rounded Rect [id:dp01857192833908572] 
        \draw  [fill={rgb, 255:red, 204; green, 230; blue, 255 }  ,fill opacity=1 ][line width=0.75]  (301.9,180.94) .. controls (301.9,178.17) and (299.66,175.93) .. (296.89,175.93) -- (290.21,175.94) .. controls (287.45,175.94) and (285.21,178.19) .. (285.21,180.96) -- (285.25,221.24) .. controls (285.25,224.01) and (287.5,226.25) .. (290.27,226.25) -- (296.94,226.24) .. controls (299.71,226.24) and (301.95,223.99) .. (301.95,221.23) -- cycle ;
        %Straight Lines [id:da2886108126148885] 
        \draw [line width=0.75]    (420.06,251.62) -- (219.84,251.83) ;
        %Straight Lines [id:da6179560213663504] 
        \draw [line width=0.75]    (420.09,285.16) -- (219.88,285.37) ;
        %Rounded Rect [id:dp7694239705609485] 
        \draw  [fill={rgb, 255:red, 204; green, 230; blue, 255 }  ,fill opacity=1 ][line width=0.75]  (276.65,214.86) .. controls (276.65,212.09) and (274.4,209.85) .. (271.64,209.85) -- (264.96,209.86) .. controls (262.19,209.86) and (259.95,212.11) .. (259.96,214.87) -- (260,255.16) .. controls (260,257.93) and (262.24,260.17) .. (265.01,260.16) -- (271.69,260.16) .. controls (274.45,260.16) and (276.69,257.91) .. (276.69,255.14) -- cycle ;
        %Rounded Rect [id:dp15212958875689364] 
        \draw  [fill={rgb, 255:red, 204; green, 230; blue, 255 }  ,fill opacity=1 ][line width=0.75]  (251.56,248.42) .. controls (251.56,245.65) and (249.31,243.41) .. (246.55,243.42) -- (239.87,243.42) .. controls (237.1,243.43) and (234.86,245.67) .. (234.87,248.44) -- (234.91,288.73) .. controls (234.91,291.49) and (237.16,293.73) .. (239.92,293.73) -- (246.6,293.72) .. controls (249.36,293.72) and (251.6,291.47) .. (251.6,288.71) -- cycle ;
        %Shape: Triangle [id:dp7889232840382157] 
        \draw  [fill={rgb, 255:red, 222; green, 241; blue, 198 }  ,fill opacity=1 ] (432.34,49.87) -- (418.84,39.61) -- (418.99,60.33) -- cycle ;
        %Shape: Triangle [id:dp4749466881939539] 
        \draw  [fill={rgb, 255:red, 222; green, 241; blue, 198 }  ,fill opacity=1 ] (432.68,83.54) -- (419.17,73.28) -- (419.32,94) -- cycle ;
        %Shape: Triangle [id:dp930237813977697] 
        \draw  [fill={rgb, 255:red, 222; green, 241; blue, 198 }  ,fill opacity=1 ] (432.68,117.54) -- (419.17,107.28) -- (419.32,128) -- cycle ;
        %Shape: Triangle [id:dp15693804604114114] 
        \draw  [fill={rgb, 255:red, 222; green, 241; blue, 198 }  ,fill opacity=1 ] (433.01,151.21) -- (419.51,140.94) -- (419.66,161.66) -- cycle ;
        %Shape: Triangle [id:dp01780423639812967] 
        \draw  [fill={rgb, 255:red, 222; green, 241; blue, 198 }  ,fill opacity=1 ] (433.01,183.87) -- (419.51,173.61) -- (419.66,194.33) -- cycle ;
        %Shape: Triangle [id:dp11261381975704965] 
        \draw  [fill={rgb, 255:red, 222; green, 241; blue, 198 }  ,fill opacity=1 ] (433.34,217.54) -- (419.84,207.28) -- (419.99,228) -- cycle ;
        %Shape: Triangle [id:dp9209635200074626] 
        \draw  [fill={rgb, 255:red, 222; green, 241; blue, 198 }  ,fill opacity=1 ] (433.34,251.54) -- (419.84,241.28) -- (419.99,262) -- cycle ;
        %Shape: Triangle [id:dp005480881066062038] 
        \draw  [fill={rgb, 255:red, 222; green, 241; blue, 198 }  ,fill opacity=1 ] (433.67,285.21) -- (420.17,274.94) -- (420.32,295.66) -- cycle ;
        %Rounded Rect [id:dp9555566614778821] 
        \draw  [fill={rgb, 255:red, 245; green, 166; blue, 35 }  ,fill opacity=1 ][line width=0.75]  (212.55,242.94) .. controls (212.55,242.94) and (212.55,242.94) .. (212.55,242.94) -- (229.27,242.95) .. controls (229.27,242.95) and (229.27,242.95) .. (229.27,242.95) -- (229.25,260.37) .. controls (229.25,260.37) and (229.25,260.37) .. (229.25,260.37) -- (212.54,260.36) .. controls (212.54,260.36) and (212.54,260.36) .. (212.54,260.36) -- cycle ;
        %Straight Lines [id:da5877560157290977] 
        \draw [line width=0.75]    (551.8,252.83) -- (599.3,252.83) ;
        %Straight Lines [id:da8938814725419997] 
        \draw [line width=0.75]    (553,286.37) -- (600.3,286.37) ;
        %Rounded Rect [id:dp7819070492577196] 
        \draw  [fill={rgb, 255:red, 204; green, 230; blue, 255 }  ,fill opacity=1 ][line width=0.75]  (544.59,216.01) .. controls (544.59,213.24) and (546.84,211) .. (549.6,211.01) -- (556.3,211.01) .. controls (559.07,211.02) and (561.31,213.26) .. (561.31,216.03) -- (561.3,222.71) .. controls (561.3,225.47) and (559.05,227.72) .. (556.28,227.71) -- (549.59,227.71) .. controls (546.82,227.7) and (544.58,225.46) .. (544.58,222.69) -- cycle ;
        %Rounded Rect [id:dp4357254819292904] 
        \draw  [fill={rgb, 255:red, 204; green, 230; blue, 255 }  ,fill opacity=1 ][line width=0.75]  (544.71,249.43) .. controls (544.72,246.66) and (546.97,244.41) .. (549.74,244.42) -- (556.42,244.42) .. controls (559.19,244.43) and (561.43,246.67) .. (561.43,249.44) -- (561.39,289.72) .. controls (561.39,292.49) and (559.14,294.73) .. (556.37,294.73) -- (549.68,294.72) .. controls (546.91,294.72) and (544.67,292.47) .. (544.67,289.7) -- cycle ;
        %Rounded Rect [id:dp016504325957613863] 
        \draw  [fill={rgb, 255:red, 204; green, 230; blue, 255 }  ,fill opacity=1 ][line width=0.75]  (607.56,249.42) .. controls (607.56,246.65) and (605.31,244.41) .. (602.55,244.42) -- (595.87,244.42) .. controls (593.1,244.43) and (590.86,246.67) .. (590.87,249.44) -- (590.91,289.73) .. controls (590.91,292.49) and (593.16,294.73) .. (595.92,294.73) -- (602.6,294.72) .. controls (605.36,294.72) and (607.6,292.47) .. (607.6,289.71) -- cycle ;
        %Rounded Rect [id:dp16631921124057047] 
        \draw  [fill={rgb, 255:red, 245; green, 166; blue, 35 }  ,fill opacity=1 ][line width=0.75]  (568.55,243.94) .. controls (568.55,243.94) and (568.55,243.94) .. (568.55,243.94) -- (585.27,243.95) .. controls (585.27,243.95) and (585.27,243.95) .. (585.27,243.95) -- (585.25,261.37) .. controls (585.25,261.37) and (585.25,261.37) .. (585.25,261.37) -- (568.54,261.36) .. controls (568.54,261.36) and (568.54,261.36) .. (568.54,261.36) -- cycle ;
        %Rounded Rect [id:dp596194282488453] 
        \draw  [fill={rgb, 255:red, 204; green, 230; blue, 255 }  ,fill opacity=1 ][line width=0.75]  (590.59,216.01) .. controls (590.59,213.24) and (592.84,211) .. (595.6,211.01) -- (602.3,211.01) .. controls (605.07,211.02) and (607.31,213.26) .. (607.31,216.03) -- (607.3,222.71) .. controls (607.3,225.47) and (605.05,227.72) .. (602.28,227.71) -- (595.59,227.71) .. controls (592.82,227.7) and (590.58,225.46) .. (590.58,222.69) -- cycle ;
        %Rounded Rect [id:dp6669443936278598] 
        \draw  [fill={rgb, 255:red, 204; green, 230; blue, 255 }  ,fill opacity=1 ][line width=0.75]  (544.09,181.51) .. controls (544.09,178.74) and (546.34,176.5) .. (549.1,176.51) -- (555.8,176.51) .. controls (558.57,176.52) and (560.81,178.76) .. (560.81,181.53) -- (560.8,188.21) .. controls (560.8,190.97) and (558.55,193.22) .. (555.78,193.21) -- (549.09,193.21) .. controls (546.32,193.2) and (544.08,190.96) .. (544.08,188.19) -- cycle ;
        %Rounded Rect [id:dp3639717628167316] 
        \draw  [fill={rgb, 255:red, 204; green, 230; blue, 255 }  ,fill opacity=1 ][line width=0.75]  (590.09,181.51) .. controls (590.09,178.74) and (592.34,176.5) .. (595.1,176.51) -- (601.8,176.51) .. controls (604.57,176.52) and (606.81,178.76) .. (606.81,181.53) -- (606.8,188.21) .. controls (606.8,190.97) and (604.55,193.22) .. (601.78,193.21) -- (595.09,193.21) .. controls (592.32,193.2) and (590.08,190.96) .. (590.08,188.19) -- cycle ;
        %Rounded Rect [id:dp20870568349431773] 
        \draw  [fill={rgb, 255:red, 204; green, 230; blue, 255 }  ,fill opacity=1 ][line width=0.75]  (543.59,147.01) .. controls (543.59,144.24) and (545.84,142) .. (548.6,142.01) -- (555.3,142.01) .. controls (558.07,142.02) and (560.31,144.26) .. (560.31,147.03) -- (560.3,153.71) .. controls (560.3,156.47) and (558.05,158.72) .. (555.28,158.71) -- (548.59,158.71) .. controls (545.82,158.7) and (543.58,156.46) .. (543.58,153.69) -- cycle ;
        %Rounded Rect [id:dp15652983559165912] 
        \draw  [fill={rgb, 255:red, 204; green, 230; blue, 255 }  ,fill opacity=1 ][line width=0.75]  (589.59,147.01) .. controls (589.59,144.24) and (591.84,142) .. (594.6,142.01) -- (601.3,142.01) .. controls (604.07,142.02) and (606.31,144.26) .. (606.31,147.03) -- (606.3,153.71) .. controls (606.3,156.47) and (604.05,158.72) .. (601.28,158.71) -- (594.59,158.71) .. controls (591.82,158.7) and (589.58,156.46) .. (589.58,153.69) -- cycle ;
        %Rounded Rect [id:dp5415255538510952] 
        \draw  [fill={rgb, 255:red, 204; green, 230; blue, 255 }  ,fill opacity=1 ][line width=0.75]  (543.09,112.51) .. controls (543.09,109.74) and (545.34,107.5) .. (548.1,107.51) -- (554.8,107.51) .. controls (557.57,107.52) and (559.81,109.76) .. (559.81,112.53) -- (559.8,119.21) .. controls (559.8,121.97) and (557.55,124.22) .. (554.78,124.21) -- (548.09,124.21) .. controls (545.32,124.2) and (543.08,121.96) .. (543.08,119.19) -- cycle ;
        %Rounded Rect [id:dp6055697471419064] 
        \draw  [fill={rgb, 255:red, 204; green, 230; blue, 255 }  ,fill opacity=1 ][line width=0.75]  (589.09,112.51) .. controls (589.09,109.74) and (591.34,107.5) .. (594.1,107.51) -- (600.8,107.51) .. controls (603.57,107.52) and (605.81,109.76) .. (605.81,112.53) -- (605.8,119.21) .. controls (605.8,121.97) and (603.55,124.22) .. (600.78,124.21) -- (594.09,124.21) .. controls (591.32,124.2) and (589.08,121.96) .. (589.08,119.19) -- cycle ;
        %Rounded Rect [id:dp6782822926288614] 
        \draw  [fill={rgb, 255:red, 204; green, 230; blue, 255 }  ,fill opacity=1 ][line width=0.75]  (542.59,81.01) .. controls (542.59,78.24) and (544.84,76) .. (547.6,76.01) -- (554.3,76.01) .. controls (557.07,76.02) and (559.31,78.26) .. (559.31,81.03) -- (559.3,87.71) .. controls (559.3,90.47) and (557.05,92.72) .. (554.28,92.71) -- (547.59,92.71) .. controls (544.82,92.7) and (542.58,90.46) .. (542.58,87.69) -- cycle ;
        %Rounded Rect [id:dp2469980076878362] 
        \draw  [fill={rgb, 255:red, 204; green, 230; blue, 255 }  ,fill opacity=1 ][line width=0.75]  (588.59,81.01) .. controls (588.59,78.24) and (590.84,76) .. (593.6,76.01) -- (600.3,76.01) .. controls (603.07,76.02) and (605.31,78.26) .. (605.31,81.03) -- (605.3,87.71) .. controls (605.3,90.47) and (603.05,92.72) .. (600.28,92.71) -- (593.59,92.71) .. controls (590.82,92.7) and (588.58,90.46) .. (588.58,87.69) -- cycle ;
        %Rounded Rect [id:dp21735532825637738] 
        \draw  [fill={rgb, 255:red, 204; green, 230; blue, 255 }  ,fill opacity=1 ][line width=0.75]  (542.09,46.51) .. controls (542.09,43.74) and (544.34,41.5) .. (547.1,41.51) -- (553.8,41.51) .. controls (556.57,41.52) and (558.81,43.76) .. (558.81,46.53) -- (558.8,53.21) .. controls (558.8,55.97) and (556.55,58.22) .. (553.78,58.21) -- (547.09,58.21) .. controls (544.32,58.2) and (542.08,55.96) .. (542.08,53.19) -- cycle ;
        %Rounded Rect [id:dp9815280380911702] 
        \draw  [fill={rgb, 255:red, 204; green, 230; blue, 255 }  ,fill opacity=1 ][line width=0.75]  (588.09,46.51) .. controls (588.09,43.74) and (590.34,41.5) .. (593.1,41.51) -- (599.8,41.51) .. controls (602.57,41.52) and (604.81,43.76) .. (604.81,46.53) -- (604.8,53.21) .. controls (604.8,55.97) and (602.55,58.22) .. (599.78,58.21) -- (593.09,58.21) .. controls (590.32,58.2) and (588.08,55.96) .. (588.08,53.19) -- cycle ;
        %Shape: Polygon [id:ds2519677495695696] 
        \draw  [color={rgb, 255:red, 114; green, 114; blue, 114 }  ,draw opacity=0 ][fill={rgb, 255:red, 214; green, 214; blue, 214 }  ,fill opacity=1 ] (500,148.55) -- (518.05,169.58) -- (500,191.05) -- (500,177) -- (470.3,177) -- (470.3,162) -- (500,162) -- cycle ;
        
        \end{mytikz2}\quad,
        \end{equation}
        which has a constant treewidth and hence can be efficiently contracted~\cite{Markov2008}. The underlying physical reason is that the 1D entanglement boundary law is trivial, i.e., the entanglement entropy for any simple partition can be upper bounded by a constant (i.e., the boundaries of lines are just points). The variational optimization and observable estimation in MPS, also known as the famed density-matrix renormalization group (DMRG) algorithm~\cite{White1992}, can be accomplished within linear time. Even certain FLDCs of linear depth with ancilla qubits can be exactly seen as MPS of the canonical form~\cite{Schon2005, Liu2021a, Barthel2023, Miao2023, Malz2023}. 1D ground states that violate the entanglement area law, such as in gapless critical systems, can still be effectively simulated classically by MPS with polynomially large bond dimension or 1D MERA~\cite{Vidal2007, Markov2008}. To sum up, for 1D ground states, classical simulations have already provided efficient solutions~\cite{Landau2015}. In fact, it has been proven that there is no topological order in 1D qubit systems~\cite{Chen2011, Chen2011a, Schuch2011, Zeng2019} and log-depth circuits are enough to generate MPS~\cite{Malz2023}, which implies that 1D systems cannot be the suitable domain to showcase the advantage of FLDC in principle. This is also why the numerical experiments in this work use a 2D model. 
        
        \item[$\circ$] 2D FLDCs do not have the chain-like structures as in 1D anymore (i.e., strings can form loops in 2D and above). The TNs for the expectation values in 2D FLDCs may have very large treewidth that grows polynomially with the system size. A toy example of $2\times 2$ qubits with $2$-qubit blocks is depicted below (darker colors indicate later action orders)
        \begin{equation}\label{eq:2by2_loop}
        \begin{mytikz2}
        %Curve Lines [id:da2675912492971493] 
        \draw    (594.51,230.51) .. controls (580.07,274.83) and (539.53,256.23) .. (538.8,175.2) .. controls (538.07,94.17) and (580.07,74.17) .. (594.84,121.51) ;
        %Curve Lines [id:da5430189239132592] 
        \draw    (643.84,230.84) .. controls (660.4,274.83) and (698.73,254.5) .. (699.4,174.5) .. controls (700.07,94.5) and (660.4,73.83) .. (644.51,121.51) ;
        %Straight Lines [id:da09077043185826339] 
        \draw [line width=0.75]    (345.25,228.76) -- (345.3,127.43) ;
        %Straight Lines [id:da09086631091628017] 
        \draw [line width=0.75]    (296.58,228.76) -- (296.6,128.1) ;
        %Straight Lines [id:da7291566751397869] 
        \draw [line width=0.75]    (671.61,121.36) -- (570.3,121.36) ;
        %Straight Lines [id:da07372799810088737] 
        \draw [line width=0.75]    (671.61,160.36) -- (569.8,160.36) ;
        %Straight Lines [id:da9419865186641809] 
        \draw [line width=0.75]    (672.47,112.36) -- (672.47,197.83) ;
        %Straight Lines [id:da7112267733341802] 
        \draw [line width=0.75]    (565.87,112.36) -- (565.87,198.15) ;
        %Straight Lines [id:da5420379814258982] 
        \draw [line width=0.75]    (199.63,127.52) -- (320.71,127.73) ;
        %Straight Lines [id:da970375029530645] 
        \draw [line width=0.75]    (199.61,161.07) -- (320.69,161.27) ;
        %Rounded Rect [id:dp9117877744526623] 
        \draw  [fill={rgb, 255:red, 204; green, 230; blue, 255 }  ,fill opacity=1 ][line width=0.75]  (210.63,124.28) .. controls (210.64,121.51) and (212.88,119.27) .. (215.65,119.27) -- (222.34,119.27) .. controls (225.11,119.28) and (227.35,121.53) .. (227.35,124.3) -- (227.31,164.57) .. controls (227.3,167.34) and (225.06,169.58) .. (222.29,169.58) -- (215.6,169.57) .. controls (212.83,169.57) and (210.59,167.32) .. (210.59,164.55) -- cycle ;
        %Straight Lines [id:da6164659805779229] 
        \draw [line width=0.75]    (199.09,194.95) -- (320.17,195.16) ;
        %Straight Lines [id:da6695305567352068] 
        \draw [line width=0.75]    (199.07,228.5) -- (320.15,228.7) ;
        %Rounded Rect [id:dp7616752937911562] 
        \draw  [fill={rgb, 255:red, 204; green, 230; blue, 255 }  ,fill opacity=1 ][line width=0.75]  (235.92,158.2) .. controls (235.93,155.43) and (238.17,153.18) .. (240.94,153.19) -- (247.63,153.19) .. controls (250.4,153.2) and (252.64,155.44) .. (252.64,158.21) -- (252.6,198.49) .. controls (252.6,201.26) and (250.35,203.5) .. (247.58,203.5) -- (240.89,203.49) .. controls (238.12,203.49) and (235.88,201.24) .. (235.88,198.47) -- cycle ;
        %Rounded Rect [id:dp6168210919183745] 
        \draw  [fill={rgb, 255:red, 204; green, 230; blue, 255 }  ,fill opacity=1 ][line width=0.75]  (261.05,191.76) .. controls (261.05,188.99) and (263.3,186.75) .. (266.07,186.75) -- (272.76,186.76) .. controls (275.53,186.76) and (277.77,189.01) .. (277.77,191.78) -- (277.72,232.05) .. controls (277.72,234.82) and (275.47,237.06) .. (272.7,237.06) -- (266.02,237.05) .. controls (263.25,237.05) and (261,234.8) .. (261.01,232.03) -- cycle ;
        %Shape: Triangle [id:dp8163703189516667] 
        \draw  [fill={rgb, 255:red, 222; green, 241; blue, 198 }  ,fill opacity=1 ] (185.67,127.21) -- (199.19,116.94) -- (199.05,137.66) -- cycle ;
        %Shape: Triangle [id:dp31233389393939426] 
        \draw  [fill={rgb, 255:red, 222; green, 241; blue, 198 }  ,fill opacity=1 ] (185.34,160.87) -- (198.86,150.61) -- (198.71,171.33) -- cycle ;
        %Shape: Triangle [id:dp5098459858230888] 
        \draw  [fill={rgb, 255:red, 222; green, 241; blue, 198 }  ,fill opacity=1 ] (185.34,194.87) -- (198.86,184.61) -- (198.71,205.33) -- cycle ;
        %Shape: Triangle [id:dp4656976136753519] 
        \draw  [fill={rgb, 255:red, 222; green, 241; blue, 198 }  ,fill opacity=1 ] (185.01,228.54) -- (198.53,218.28) -- (198.38,239) -- cycle ;
        %Straight Lines [id:da43492197041417024] 
        \draw [line width=0.75]    (439.17,127.52) -- (318.26,127.73) ;
        %Straight Lines [id:da8874587175774764] 
        \draw [line width=0.75]    (439.19,161.07) -- (318.28,161.27) ;
        %Rounded Rect [id:dp01857192833908572] 
        \draw  [fill={rgb, 255:red, 204; green, 230; blue, 255 }  ,fill opacity=1 ][line width=0.75]  (429.24,124.27) .. controls (429.23,121.5) and (426.99,119.27) .. (424.22,119.27) -- (417.55,119.28) .. controls (414.78,119.28) and (412.54,121.52) .. (412.54,124.29) -- (412.59,164.58) .. controls (412.59,167.34) and (414.83,169.58) .. (417.6,169.58) -- (424.28,169.57) .. controls (427.04,169.57) and (429.28,167.33) .. (429.28,164.56) -- cycle ;
        %Straight Lines [id:da2886108126148885] 
        \draw [line width=0.75]    (439.71,194.95) -- (318.81,195.16) ;
        %Straight Lines [id:da6179560213663504] 
        \draw [line width=0.75]    (439.73,228.5) -- (318.83,228.7) ;
        %Rounded Rect [id:dp7694239705609485] 
        \draw  [fill={rgb, 255:red, 204; green, 230; blue, 255 }  ,fill opacity=1 ][line width=0.75]  (403.98,158.19) .. controls (403.98,155.42) and (401.73,153.18) .. (398.97,153.19) -- (392.29,153.19) .. controls (389.53,153.2) and (387.29,155.44) .. (387.29,158.21) -- (387.33,198.5) .. controls (387.33,201.26) and (389.58,203.5) .. (392.34,203.5) -- (399.02,203.49) .. controls (401.79,203.49) and (404.03,201.24) .. (404.02,198.48) -- cycle ;
        %Rounded Rect [id:dp15212958875689364] 
        \draw  [fill={rgb, 255:red, 204; green, 230; blue, 255 }  ,fill opacity=1 ][line width=0.75]  (378.89,191.75) .. controls (378.89,188.99) and (376.65,186.75) .. (373.88,186.75) -- (367.2,186.76) .. controls (364.44,186.76) and (362.2,189) .. (362.2,191.77) -- (362.24,232.06) .. controls (362.24,234.82) and (364.49,237.06) .. (367.25,237.06) -- (373.93,237.05) .. controls (376.7,237.05) and (378.94,234.81) .. (378.94,232.04) -- cycle ;
        %Shape: Triangle [id:dp01780423639812967] 
        \draw  [fill={rgb, 255:red, 222; green, 241; blue, 198 }  ,fill opacity=1 ] (452.68,127.21) -- (439.17,116.94) -- (439.32,137.66) -- cycle ;
        %Shape: Triangle [id:dp11261381975704965] 
        \draw  [fill={rgb, 255:red, 222; green, 241; blue, 198 }  ,fill opacity=1 ] (453.01,160.87) -- (439.51,150.61) -- (439.66,171.33) -- cycle ;
        %Shape: Triangle [id:dp9209635200074626] 
        \draw  [fill={rgb, 255:red, 222; green, 241; blue, 198 }  ,fill opacity=1 ] (453.01,194.87) -- (439.51,184.61) -- (439.66,205.33) -- cycle ;
        %Shape: Triangle [id:dp005480881066062038] 
        \draw  [fill={rgb, 255:red, 222; green, 241; blue, 198 }  ,fill opacity=1 ] (453.34,228.54) -- (439.84,218.28) -- (439.99,239) -- cycle ;
        %Rounded Rect [id:dp9555566614778821] 
        \draw  [fill={rgb, 255:red, 245; green, 166; blue, 35 }  ,fill opacity=1 ][line width=0.75]  (311.8,219.98) .. controls (311.8,219.98) and (311.8,219.98) .. (311.8,219.98) -- (328.51,220) .. controls (328.51,220) and (328.51,220) .. (328.51,220) -- (328.5,237.42) .. controls (328.5,237.42) and (328.5,237.42) .. (328.5,237.42) -- (311.78,237.4) .. controls (311.78,237.4) and (311.78,237.4) .. (311.78,237.4) -- cycle ;
        %Straight Lines [id:da5877560157290977] 
        \draw [line width=0.75]    (570.3,197.83) -- (672.97,197.83) ;
        %Straight Lines [id:da8938814725419997] 
        \draw [line width=0.75]    (571.8,231.37) -- (673.97,231.37) ;
        %Rounded Rect [id:dp7819070492577196] 
        \draw  [fill={rgb, 255:red, 204; green, 230; blue, 255 }  ,fill opacity=1 ][line width=0.75]  (558.26,157.01) .. controls (558.26,154.24) and (560.5,152) .. (563.27,152.01) -- (569.97,152.01) .. controls (572.74,152.02) and (574.98,154.26) .. (574.97,157.03) -- (574.97,163.71) .. controls (574.96,166.47) and (572.72,168.72) .. (569.95,168.71) -- (563.25,168.71) .. controls (560.49,168.7) and (558.25,166.46) .. (558.25,163.69) -- cycle ;
        %Rounded Rect [id:dp4357254819292904] 
        \draw  [fill={rgb, 255:red, 204; green, 230; blue, 255 }  ,fill opacity=1 ][line width=0.75]  (558.38,194.43) .. controls (558.38,191.66) and (560.63,189.41) .. (563.4,189.42) -- (570.09,189.42) .. controls (572.86,189.43) and (575.1,191.67) .. (575.1,194.44) -- (575.06,234.72) .. controls (575.05,237.49) and (572.81,239.73) .. (570.04,239.73) -- (563.35,239.72) .. controls (560.58,239.72) and (558.34,237.47) .. (558.34,234.7) -- cycle ;
        %Rounded Rect [id:dp016504325957613863] 
        \draw  [fill={rgb, 255:red, 204; green, 230; blue, 255 }  ,fill opacity=1 ][line width=0.75]  (681.23,194.42) .. controls (681.22,191.65) and (678.98,189.41) .. (676.21,189.42) -- (669.54,189.42) .. controls (666.77,189.43) and (664.53,191.67) .. (664.53,194.44) -- (664.57,234.73) .. controls (664.58,237.49) and (666.82,239.73) .. (669.59,239.73) -- (676.27,239.72) .. controls (679.03,239.72) and (681.27,237.47) .. (681.27,234.71) -- cycle ;
        %Rounded Rect [id:dp16631921124057047] 
        \draw  [fill={rgb, 255:red, 245; green, 166; blue, 35 }  ,fill opacity=1 ][line width=0.75]  (610.22,222.6) .. controls (610.22,222.6) and (610.22,222.6) .. (610.22,222.6) -- (626.94,222.62) .. controls (626.94,222.62) and (626.94,222.62) .. (626.94,222.62) -- (626.92,240.04) .. controls (626.92,240.04) and (626.92,240.04) .. (626.92,240.04) -- (610.2,240.02) .. controls (610.2,240.02) and (610.2,240.02) .. (610.2,240.02) -- cycle ;
        %Rounded Rect [id:dp596194282488453] 
        \draw  [fill={rgb, 255:red, 204; green, 230; blue, 255 }  ,fill opacity=1 ][line width=0.75]  (664.26,157.01) .. controls (664.26,154.24) and (666.5,152) .. (669.27,152.01) -- (675.97,152.01) .. controls (678.74,152.02) and (680.98,154.26) .. (680.97,157.03) -- (680.97,163.71) .. controls (680.96,166.47) and (678.72,168.72) .. (675.95,168.71) -- (669.25,168.71) .. controls (666.49,168.7) and (664.25,166.46) .. (664.25,163.69) -- cycle ;
        %Rounded Rect [id:dp20870568349431773] 
        \draw  [fill={rgb, 255:red, 204; green, 230; blue, 255 }  ,fill opacity=1 ][line width=0.75]  (557.26,118.01) .. controls (557.26,115.24) and (559.5,113) .. (562.27,113.01) -- (568.97,113.01) .. controls (571.74,113.02) and (573.98,115.26) .. (573.97,118.03) -- (573.97,124.71) .. controls (573.96,127.47) and (571.72,129.72) .. (568.95,129.71) -- (562.25,129.71) .. controls (559.49,129.7) and (557.25,127.46) .. (557.25,124.69) -- cycle ;
        %Rounded Rect [id:dp15652983559165912] 
        \draw  [fill={rgb, 255:red, 204; green, 230; blue, 255 }  ,fill opacity=1 ][line width=0.75]  (663.26,118.01) .. controls (663.26,115.24) and (665.5,113) .. (668.27,113.01) -- (674.97,113.01) .. controls (677.74,113.02) and (679.98,115.26) .. (679.97,118.03) -- (679.97,124.71) .. controls (679.96,127.47) and (677.72,129.72) .. (674.95,129.71) -- (668.25,129.71) .. controls (665.49,129.7) and (663.25,127.46) .. (663.25,124.69) -- cycle ;
        %Shape: Polygon [id:ds2519677495695696] 
        \draw  [color={rgb, 255:red, 114; green, 114; blue, 114 }  ,draw opacity=0 ][fill={rgb, 255:red, 214; green, 214; blue, 214 }  ,fill opacity=1 ] (500.33,145.88) -- (518.38,166.91) -- (500.33,188.38) -- (500.33,174.33) -- (470.63,174.33) -- (470.63,159.33) -- (500.33,159.33) -- cycle ;
        %Rounded Rect [id:dp9289309955251381] 
        \draw  [fill={rgb, 255:red, 204; green, 230; blue, 255 }  ,fill opacity=1 ][line width=0.75]  (353.59,225.44) .. controls (353.59,222.69) and (351.36,220.47) .. (348.61,220.47) -- (341.87,220.48) .. controls (339.12,220.48) and (336.9,222.71) .. (336.9,225.46) -- (336.91,232.09) .. controls (336.91,234.84) and (339.14,237.06) .. (341.89,237.06) -- (348.63,237.05) .. controls (351.38,237.05) and (353.6,234.82) .. (353.6,232.07) -- cycle ;
        %Rounded Rect [id:dp6954529072720121] 
        \draw  [fill={rgb, 255:red, 204; green, 230; blue, 255 }  ,fill opacity=1 ][line width=0.75]  (304.93,225.44) .. controls (304.93,222.69) and (302.7,220.47) .. (299.95,220.47) -- (293.2,220.48) .. controls (290.46,220.48) and (288.23,222.71) .. (288.23,225.46) -- (288.24,232.09) .. controls (288.24,234.84) and (290.47,237.06) .. (293.22,237.06) -- (299.96,237.05) .. controls (302.71,237.05) and (304.94,234.82) .. (304.94,232.07) -- cycle ;
        %Rounded Rect [id:dp8756535520503546] 
        \draw  [fill={rgb, 255:red, 204; green, 230; blue, 255 }  ,fill opacity=1 ][line width=0.75]  (304.94,124.77) .. controls (304.94,122.02) and (302.71,119.8) .. (299.96,119.8) -- (293.22,119.81) .. controls (290.47,119.81) and (288.24,122.04) .. (288.24,124.79) -- (288.25,131.42) .. controls (288.25,134.17) and (290.48,136.4) .. (293.23,136.39) -- (299.97,136.39) .. controls (302.72,136.38) and (304.95,134.15) .. (304.95,131.41) -- cycle ;
        %Rounded Rect [id:dp1764623833318688] 
        \draw  [fill={rgb, 255:red, 204; green, 230; blue, 255 }  ,fill opacity=1 ][line width=0.75]  (353.64,124.11) .. controls (353.64,121.36) and (351.41,119.13) .. (348.66,119.14) -- (341.92,119.14) .. controls (339.17,119.15) and (336.94,121.38) .. (336.94,124.12) -- (336.95,130.76) .. controls (336.95,133.51) and (339.18,135.73) .. (341.93,135.73) -- (348.67,135.72) .. controls (351.42,135.72) and (353.65,133.49) .. (353.65,130.74) -- cycle ;
        %Rounded Rect [id:dp4018694859877505] 
        \draw  [fill={rgb, 255:red, 204; green, 230; blue, 255 }  ,fill opacity=1 ][line width=0.75]  (586.48,118.16) .. controls (586.49,115.39) and (588.73,113.15) .. (591.5,113.16) -- (598.2,113.16) .. controls (600.96,113.17) and (603.2,115.41) .. (603.2,118.18) -- (603.19,124.86) .. controls (603.19,127.63) and (600.95,129.87) .. (598.18,129.86) -- (591.48,129.86) .. controls (588.71,129.85) and (586.47,127.61) .. (586.48,124.84) -- cycle ;
        %Rounded Rect [id:dp7082901040470568] 
        \draw  [fill={rgb, 255:red, 204; green, 230; blue, 255 }  ,fill opacity=1 ][line width=0.75]  (636.15,118.16) .. controls (636.15,115.39) and (638.4,113.15) .. (641.17,113.16) -- (647.86,113.16) .. controls (650.63,113.17) and (652.87,115.41) .. (652.87,118.18) -- (652.86,124.86) .. controls (652.86,127.63) and (650.61,129.87) .. (647.85,129.86) -- (641.15,129.86) .. controls (638.38,129.85) and (636.14,127.61) .. (636.14,124.84) -- cycle ;
        %Rounded Rect [id:dp811189936046921] 
        \draw  [fill={rgb, 255:red, 204; green, 230; blue, 255 }  ,fill opacity=1 ][line width=0.75]  (586.15,227.16) .. controls (586.15,224.39) and (588.4,222.15) .. (591.17,222.16) -- (597.86,222.16) .. controls (600.63,222.17) and (602.87,224.41) .. (602.87,227.18) -- (602.86,233.86) .. controls (602.86,236.63) and (600.61,238.87) .. (597.85,238.86) -- (591.15,238.86) .. controls (588.38,238.85) and (586.14,236.61) .. (586.14,233.84) -- cycle ;
        %Rounded Rect [id:dp029337247457201343] 
        \draw  [fill={rgb, 255:red, 204; green, 230; blue, 255 }  ,fill opacity=1 ][line width=0.75]  (635.48,227.49) .. controls (635.49,224.73) and (637.73,222.49) .. (640.5,222.49) -- (647.2,222.5) .. controls (649.96,222.5) and (652.2,224.75) .. (652.2,227.51) -- (652.19,234.19) .. controls (652.19,236.96) and (649.95,239.2) .. (647.18,239.2) -- (640.48,239.19) .. controls (637.71,239.19) and (635.47,236.94) .. (635.48,234.17) -- cycle ;
        %Shape: Polygon [id:ds5484234801624694] 
        \draw  [color={rgb, 255:red, 114; green, 114; blue, 114 }  ,draw opacity=0 ][fill={rgb, 255:red, 214; green, 214; blue, 214 }  ,fill opacity=1 ] (149.67,149.22) -- (167.72,170.25) -- (149.67,191.72) -- (149.67,177.67) -- (119.97,177.67) -- (119.97,162.67) -- (149.67,162.67) -- cycle ;
        %Rounded Rect [id:dp3375462123839019] 
        \draw  [fill={rgb, 255:red, 223; green, 244; blue, 254 }  ,fill opacity=1 ][line width=0.75]  (89.69,182.07) .. controls (93.14,182.07) and (95.93,184.87) .. (95.93,188.32) -- (95.93,196.65) .. controls (95.93,200.11) and (93.14,202.9) .. (89.69,202.9) -- (39.5,202.9) .. controls (36.05,202.9) and (33.25,200.11) .. (33.25,196.65) -- (33.25,188.32) .. controls (33.25,184.87) and (36.05,182.07) .. (39.5,182.07) -- cycle ;
        %Rounded Rect [id:dp829756466165392] 
        \draw  [fill={rgb, 255:red, 204; green, 230; blue, 255 }  ,fill opacity=1 ][line width=0.75]  (75.1,146.47) .. controls (75.1,143.02) and (77.9,140.22) .. (81.35,140.22) -- (89.69,140.22) .. controls (93.14,140.22) and (95.93,143.02) .. (95.93,146.47) -- (95.93,196.65) .. controls (95.93,200.11) and (93.14,202.9) .. (89.69,202.9) -- (81.35,202.9) .. controls (77.9,202.9) and (75.1,200.11) .. (75.1,196.65) -- cycle ;
        %Rounded Rect [id:dp12300946976951321] 
        \draw  [fill={rgb, 255:red, 166; green, 209; blue, 255 }  ,fill opacity=1 ][line width=0.75]  (89.75,140.22) .. controls (93.2,140.22) and (96,143.02) .. (96,146.47) -- (96,154.8) .. controls (96,158.25) and (93.2,161.05) .. (89.75,161.05) -- (39.57,161.05) .. controls (36.11,161.05) and (33.32,158.25) .. (33.32,154.8) -- (33.32,146.47) .. controls (33.32,143.02) and (36.11,140.22) .. (39.57,140.22) -- cycle ;
        %Rounded Rect [id:dp34730235559948164] 
        \draw  [fill={rgb, 255:red, 113; green, 186; blue, 255 }  ,fill opacity=1 ][line width=0.75]  (33.31,146.46) .. controls (33.31,143.01) and (36.11,140.22) .. (39.57,140.22) -- (47.9,140.23) .. controls (51.35,140.23) and (54.15,143.03) .. (54.14,146.49) -- (54.09,196.67) .. controls (54.09,200.12) and (51.28,202.92) .. (47.83,202.91) -- (39.5,202.9) .. controls (36.05,202.9) and (33.25,200.1) .. (33.26,196.65) -- cycle ;
        %Rounded Rect [id:dp7923185623153957] 
        \draw  [fill={rgb, 255:red, 245; green, 166; blue, 35 }  ,fill opacity=1 ][line width=0.75]  (37.62,143.22) .. controls (37.62,143.22) and (37.62,143.22) .. (37.62,143.22) -- (50.9,143.23) .. controls (50.9,143.23) and (50.9,143.23) .. (50.9,143.23) -- (50.88,157.06) .. controls (50.88,157.06) and (50.88,157.06) .. (50.88,157.06) -- (37.61,157.05) .. controls (37.61,157.05) and (37.61,157.05) .. (37.61,157.05) -- cycle ;
        
        \end{mytikz2}\quad,
        \end{equation}
        which has a treewidth larger than that of the strict 1D case in Eq.~\eqref{eq:1d_fdlc_chain} due to an additional loop structure. While this is just an example of minimal size (even topologically equivalent to a 1D ring), it reveals a simple fact: a loop on the block pattern within the causal cone [the left-hand side of Eq.~\eqref{eq:2by2_loop}] can correspond to a loop structure in the TN of the expectation value of the observable [the right-hand side of Eq.~\eqref{eq:2by2_loop}]. Hence, every time the $2$-qubit block pattern within the causal cone on the qubit connectivity graph forms a closed loop, the corresponding TN is expected to deviate further from being a tree, resulting in a larger treewidth. In addition, a $\beta$-qubit block with $\beta>2$ can be considered as containing loops of $2$-qubit blocks in itself. If there are extensive loop configurations on the block pattern within the causal cone like
        \begin{equation}
        \begin{mytikz4}
        %Rounded Rect [id:dp7047959815255691] 
        \draw  [fill={rgb, 255:red, 219; green, 234; blue, 255 }  ,fill opacity=1 ][line width=0.75]  (189.85,175.41) .. controls (193.3,175.41) and (196.1,178.2) .. (196.1,181.66) -- (196.1,189.99) .. controls (196.1,193.44) and (193.3,196.24) .. (189.85,196.24) -- (139.67,196.24) .. controls (136.22,196.24) and (133.42,193.44) .. (133.42,189.99) -- (133.42,181.66) .. controls (133.42,178.2) and (136.22,175.41) .. (139.67,175.41) -- cycle ;
        %Rounded Rect [id:dp1814873049316823] 
        \draw  [fill={rgb, 255:red, 219; green, 234; blue, 255 }  ,fill opacity=1 ][line width=0.75]  (176,140.47) .. controls (176,137.02) and (178.8,134.22) .. (182.25,134.22) -- (190.58,134.22) .. controls (194.03,134.22) and (196.83,137.02) .. (196.83,140.47) -- (196.83,190.65) .. controls (196.83,194.11) and (194.03,196.9) .. (190.58,196.9) -- (182.25,196.9) .. controls (178.8,196.9) and (176,194.11) .. (176,190.65) -- cycle ;
        %Rounded Rect [id:dp5800503559404628] 
        \draw  [fill={rgb, 255:red, 219; green, 234; blue, 255 }  ,fill opacity=1 ][line width=0.75]  (189.92,134.22) .. controls (193.37,134.22) and (196.17,137.02) .. (196.17,140.47) -- (196.17,148.8) .. controls (196.17,152.25) and (193.37,155.05) .. (189.92,155.05) -- (139.73,155.05) .. controls (136.28,155.05) and (133.48,152.25) .. (133.48,148.8) -- (133.48,140.47) .. controls (133.48,137.02) and (136.28,134.22) .. (139.73,134.22) -- cycle ;
        %Rounded Rect [id:dp3311803260653172] 
        \draw  [fill={rgb, 255:red, 219; green, 234; blue, 255 }  ,fill opacity=1 ][line width=0.75]  (133.48,140.46) .. controls (133.48,137.01) and (136.28,134.22) .. (139.73,134.22) -- (148.07,134.23) .. controls (151.52,134.23) and (154.31,137.03) .. (154.31,140.49) -- (154.26,190.67) .. controls (154.25,194.12) and (151.45,196.92) .. (148,196.91) -- (139.67,196.9) .. controls (136.22,196.9) and (133.42,194.1) .. (133.42,190.65) -- cycle ;
        %Rounded Rect [id:dp7076788040476094] 
        \draw  [fill={rgb, 255:red, 198; green, 221; blue, 255 }  ,fill opacity=1 ][line width=0.75]  (148.35,174.91) .. controls (151.8,174.91) and (154.6,177.7) .. (154.6,181.16) -- (154.6,189.49) .. controls (154.6,192.94) and (151.8,195.74) .. (148.35,195.74) -- (98.17,195.74) .. controls (94.72,195.74) and (91.92,192.94) .. (91.92,189.49) -- (91.92,181.16) .. controls (91.92,177.7) and (94.72,174.91) .. (98.17,174.91) -- cycle ;
        %Rounded Rect [id:dp1910935963007403] 
        \draw  [fill={rgb, 255:red, 198; green, 221; blue, 255 }  ,fill opacity=1 ][line width=0.75]  (133.77,139.97) .. controls (133.77,136.52) and (136.57,133.72) .. (140.02,133.72) -- (148.35,133.72) .. controls (151.8,133.72) and (154.6,136.52) .. (154.6,139.97) -- (154.6,190.15) .. controls (154.6,193.61) and (151.8,196.4) .. (148.35,196.4) -- (140.02,196.4) .. controls (136.57,196.4) and (133.77,193.61) .. (133.77,190.15) -- cycle ;
        %Rounded Rect [id:dp26996156613588873] 
        \draw  [fill={rgb, 255:red, 198; green, 221; blue, 255 }  ,fill opacity=1 ][line width=0.75]  (148.42,133.72) .. controls (151.87,133.72) and (154.67,136.52) .. (154.67,139.97) -- (154.67,148.3) .. controls (154.67,151.75) and (151.87,154.55) .. (148.42,154.55) -- (98.23,154.55) .. controls (94.78,154.55) and (91.98,151.75) .. (91.98,148.3) -- (91.98,139.97) .. controls (91.98,136.52) and (94.78,133.72) .. (98.23,133.72) -- cycle ;
        %Rounded Rect [id:dp948241177612352] 
        \draw  [fill={rgb, 255:red, 198; green, 221; blue, 255 }  ,fill opacity=1 ][line width=0.75]  (91.98,139.96) .. controls (91.98,136.51) and (94.78,133.72) .. (98.23,133.72) -- (106.57,133.73) .. controls (110.02,133.73) and (112.81,136.53) .. (112.81,139.99) -- (112.76,190.17) .. controls (112.75,193.62) and (109.95,196.42) .. (106.5,196.41) -- (98.17,196.4) .. controls (94.72,196.4) and (91.92,193.6) .. (91.92,190.15) -- cycle ;
        %Rounded Rect [id:dp7493904118225201] 
        \draw  [fill={rgb, 255:red, 178; green, 213; blue, 255 }  ,fill opacity=1 ][line width=0.75]  (106.35,174.91) .. controls (109.8,174.91) and (112.6,177.7) .. (112.6,181.16) -- (112.6,189.49) .. controls (112.6,192.94) and (109.8,195.74) .. (106.35,195.74) -- (56.17,195.74) .. controls (52.72,195.74) and (49.92,192.94) .. (49.92,189.49) -- (49.92,181.16) .. controls (49.92,177.7) and (52.72,174.91) .. (56.17,174.91) -- cycle ;
        %Rounded Rect [id:dp2071459170723675] 
        \draw  [fill={rgb, 255:red, 178; green, 213; blue, 255 }  ,fill opacity=1 ][line width=0.75]  (91.77,139.97) .. controls (91.77,136.52) and (94.57,133.72) .. (98.02,133.72) -- (106.35,133.72) .. controls (109.8,133.72) and (112.6,136.52) .. (112.6,139.97) -- (112.6,190.15) .. controls (112.6,193.61) and (109.8,196.4) .. (106.35,196.4) -- (98.02,196.4) .. controls (94.57,196.4) and (91.77,193.61) .. (91.77,190.15) -- cycle ;
        %Rounded Rect [id:dp3397719693102903] 
        \draw  [fill={rgb, 255:red, 178; green, 213; blue, 255 }  ,fill opacity=1 ][line width=0.75]  (106.42,133.72) .. controls (109.87,133.72) and (112.67,136.52) .. (112.67,139.97) -- (112.67,148.3) .. controls (112.67,151.75) and (109.87,154.55) .. (106.42,154.55) -- (56.23,154.55) .. controls (52.78,154.55) and (49.98,151.75) .. (49.98,148.3) -- (49.98,139.97) .. controls (49.98,136.52) and (52.78,133.72) .. (56.23,133.72) -- cycle ;
        %Rounded Rect [id:dp027719073409690598] 
        \draw  [fill={rgb, 255:red, 178; green, 213; blue, 255 }  ,fill opacity=1 ][line width=0.75]  (50.13,139.96) .. controls (50.14,136.51) and (52.94,133.72) .. (56.39,133.72) -- (64.72,133.73) .. controls (68.17,133.73) and (70.97,136.53) .. (70.97,139.99) -- (70.91,190.17) .. controls (70.91,193.62) and (68.11,196.42) .. (64.66,196.41) -- (56.32,196.4) .. controls (52.87,196.4) and (50.08,193.6) .. (50.08,190.15) -- cycle ;
        %Rounded Rect [id:dp2662933594212644] 
        \draw  [fill={rgb, 255:red, 151; green, 201; blue, 255 }  ,fill opacity=1 ][line width=0.75]  (190.85,133.91) .. controls (194.3,133.91) and (197.1,136.7) .. (197.1,140.16) -- (197.1,148.49) .. controls (197.1,151.94) and (194.3,154.74) .. (190.85,154.74) -- (140.67,154.74) .. controls (137.22,154.74) and (134.42,151.94) .. (134.42,148.49) -- (134.42,140.16) .. controls (134.42,136.7) and (137.22,133.91) .. (140.67,133.91) -- cycle ;
        %Rounded Rect [id:dp9831739566644744] 
        \draw  [fill={rgb, 255:red, 151; green, 201; blue, 255 }  ,fill opacity=1 ][line width=0.75]  (176,98.97) .. controls (176,95.52) and (178.8,92.72) .. (182.25,92.72) -- (190.58,92.72) .. controls (194.03,92.72) and (196.83,95.52) .. (196.83,98.97) -- (196.83,149.15) .. controls (196.83,152.61) and (194.03,155.4) .. (190.58,155.4) -- (182.25,155.4) .. controls (178.8,155.4) and (176,152.61) .. (176,149.15) -- cycle ;
        %Rounded Rect [id:dp042095030360360886] 
        \draw  [fill={rgb, 255:red, 151; green, 201; blue, 255 }  ,fill opacity=1 ][line width=0.75]  (190.92,92.72) .. controls (194.37,92.72) and (197.17,95.52) .. (197.17,98.97) -- (197.17,107.3) .. controls (197.17,110.75) and (194.37,113.55) .. (190.92,113.55) -- (140.73,113.55) .. controls (137.28,113.55) and (134.48,110.75) .. (134.48,107.3) -- (134.48,98.97) .. controls (134.48,95.52) and (137.28,92.72) .. (140.73,92.72) -- cycle ;
        %Rounded Rect [id:dp7219989642678724] 
        \draw  [fill={rgb, 255:red, 151; green, 201; blue, 255 }  ,fill opacity=1 ][line width=0.75]  (133.48,98.96) .. controls (133.48,95.51) and (136.28,92.72) .. (139.73,92.72) -- (148.07,92.73) .. controls (151.52,92.73) and (154.31,95.53) .. (154.31,98.99) -- (154.26,149.17) .. controls (154.25,152.62) and (151.45,155.42) .. (148,155.41) -- (139.67,155.4) .. controls (136.22,155.4) and (133.42,152.6) .. (133.42,149.15) -- cycle ;
        %Rounded Rect [id:dp735451697411637] 
        \draw  [fill={rgb, 255:red, 137; green, 195; blue, 255 }  ,fill opacity=1 ][line width=0.75]  (148.35,133.41) .. controls (151.8,133.41) and (154.6,136.2) .. (154.6,139.66) -- (154.6,147.99) .. controls (154.6,151.44) and (151.8,154.24) .. (148.35,154.24) -- (98.17,154.24) .. controls (94.72,154.24) and (91.92,151.44) .. (91.92,147.99) -- (91.92,139.66) .. controls (91.92,136.2) and (94.72,133.41) .. (98.17,133.41) -- cycle ;
        %Rounded Rect [id:dp3667392548153823] 
        \draw  [fill={rgb, 255:red, 137; green, 195; blue, 255 }  ,fill opacity=1 ][line width=0.75]  (133.77,98.47) .. controls (133.77,95.02) and (136.57,92.22) .. (140.02,92.22) -- (148.35,92.22) .. controls (151.8,92.22) and (154.6,95.02) .. (154.6,98.47) -- (154.6,148.65) .. controls (154.6,152.11) and (151.8,154.9) .. (148.35,154.9) -- (140.02,154.9) .. controls (136.57,154.9) and (133.77,152.11) .. (133.77,148.65) -- cycle ;
        %Rounded Rect [id:dp7232969219578267] 
        \draw  [fill={rgb, 255:red, 137; green, 195; blue, 255 }  ,fill opacity=1 ][line width=0.75]  (148.42,92.22) .. controls (151.87,92.22) and (154.67,95.02) .. (154.67,98.47) -- (154.67,106.8) .. controls (154.67,110.25) and (151.87,113.05) .. (148.42,113.05) -- (98.23,113.05) .. controls (94.78,113.05) and (91.98,110.25) .. (91.98,106.8) -- (91.98,98.47) .. controls (91.98,95.02) and (94.78,92.22) .. (98.23,92.22) -- cycle ;
        %Rounded Rect [id:dp23163734354515197] 
        \draw  [fill={rgb, 255:red, 137; green, 195; blue, 255 }  ,fill opacity=1 ][line width=0.75]  (91.98,98.46) .. controls (91.98,95.01) and (94.78,92.22) .. (98.23,92.22) -- (106.57,92.23) .. controls (110.02,92.23) and (112.81,95.03) .. (112.81,98.49) -- (112.76,148.67) .. controls (112.75,152.12) and (109.95,154.92) .. (106.5,154.91) -- (98.17,154.9) .. controls (94.72,154.9) and (91.92,152.1) .. (91.92,148.65) -- cycle ;
        %Rounded Rect [id:dp7313633387684038] 
        \draw  [fill={rgb, 255:red, 121; green, 186; blue, 253 }  ,fill opacity=1 ][line width=0.75]  (106.85,133.41) .. controls (110.3,133.41) and (113.1,136.2) .. (113.1,139.66) -- (113.1,147.99) .. controls (113.1,151.44) and (110.3,154.24) .. (106.85,154.24) -- (56.67,154.24) .. controls (53.22,154.24) and (50.42,151.44) .. (50.42,147.99) -- (50.42,139.66) .. controls (50.42,136.2) and (53.22,133.41) .. (56.67,133.41) -- cycle ;
        %Rounded Rect [id:dp6104997076755581] 
        \draw  [fill={rgb, 255:red, 121; green, 186; blue, 253 }  ,fill opacity=1 ][line width=0.75]  (92.27,98.47) .. controls (92.27,95.02) and (95.07,92.22) .. (98.52,92.22) -- (106.85,92.22) .. controls (110.3,92.22) and (113.1,95.02) .. (113.1,98.47) -- (113.1,148.65) .. controls (113.1,152.11) and (110.3,154.9) .. (106.85,154.9) -- (98.52,154.9) .. controls (95.07,154.9) and (92.27,152.11) .. (92.27,148.65) -- cycle ;
        %Rounded Rect [id:dp9368257264737407] 
        \draw  [fill={rgb, 255:red, 121; green, 186; blue, 253 }  ,fill opacity=1 ][line width=0.75]  (106.92,92.22) .. controls (110.37,92.22) and (113.17,95.02) .. (113.17,98.47) -- (113.17,106.8) .. controls (113.17,110.25) and (110.37,113.05) .. (106.92,113.05) -- (56.73,113.05) .. controls (53.28,113.05) and (50.48,110.25) .. (50.48,106.8) -- (50.48,98.47) .. controls (50.48,95.02) and (53.28,92.22) .. (56.73,92.22) -- cycle ;
        %Rounded Rect [id:dp7652153598599258] 
        \draw  [fill={rgb, 255:red, 121; green, 186; blue, 253 }  ,fill opacity=1 ][line width=0.75]  (50.06,98.16) .. controls (50.06,94.7) and (52.86,91.91) .. (56.31,91.91) -- (64.65,91.92) .. controls (68.1,91.92) and (70.89,94.73) .. (70.89,98.18) -- (70.84,148.36) .. controls (70.83,151.81) and (68.03,154.61) .. (64.58,154.6) -- (56.25,154.6) .. controls (52.8,154.59) and (50,151.79) .. (50.01,148.34) -- cycle ;
        %Rounded Rect [id:dp5344921282506214] 
        \draw  [fill={rgb, 255:red, 76; green, 163; blue, 255 }  ,fill opacity=1 ][line width=0.75]  (190.35,92.41) .. controls (193.8,92.41) and (196.6,95.2) .. (196.6,98.66) -- (196.6,106.99) .. controls (196.6,110.44) and (193.8,113.24) .. (190.35,113.24) -- (140.17,113.24) .. controls (136.72,113.24) and (133.92,110.44) .. (133.92,106.99) -- (133.92,98.66) .. controls (133.92,95.2) and (136.72,92.41) .. (140.17,92.41) -- cycle ;
        %Rounded Rect [id:dp6196900910516436] 
        \draw  [fill={rgb, 255:red, 76; green, 163; blue, 255 }  ,fill opacity=1 ][line width=0.75]  (176,57.47) .. controls (176,54.02) and (178.8,51.22) .. (182.25,51.22) -- (190.58,51.22) .. controls (194.03,51.22) and (196.83,54.02) .. (196.83,57.47) -- (196.83,107.65) .. controls (196.83,111.11) and (194.03,113.9) .. (190.58,113.9) -- (182.25,113.9) .. controls (178.8,113.9) and (176,111.11) .. (176,107.65) -- cycle ;
        %Rounded Rect [id:dp7300551741292702] 
        \draw  [fill={rgb, 255:red, 76; green, 163; blue, 255 }  ,fill opacity=1 ][line width=0.75]  (190.42,51.22) .. controls (193.87,51.22) and (196.67,54.02) .. (196.67,57.47) -- (196.67,65.8) .. controls (196.67,69.25) and (193.87,72.05) .. (190.42,72.05) -- (140.23,72.05) .. controls (136.78,72.05) and (133.98,69.25) .. (133.98,65.8) -- (133.98,57.47) .. controls (133.98,54.02) and (136.78,51.22) .. (140.23,51.22) -- cycle ;
        %Rounded Rect [id:dp8076347012571432] 
        \draw  [fill={rgb, 255:red, 76; green, 163; blue, 255 }  ,fill opacity=1 ][line width=0.75]  (133.98,57.46) .. controls (133.98,54.01) and (136.78,51.22) .. (140.23,51.22) -- (148.57,51.23) .. controls (152.02,51.23) and (154.81,54.03) .. (154.81,57.49) -- (154.76,107.67) .. controls (154.75,111.12) and (151.95,113.92) .. (148.5,113.91) -- (140.17,113.9) .. controls (136.72,113.9) and (133.92,111.1) .. (133.92,107.65) -- cycle ;
        %Rounded Rect [id:dp11793783447087214] 
        \draw  [fill={rgb, 255:red, 58; green, 153; blue, 255 }  ,fill opacity=1 ][line width=0.75]  (148.85,91.91) .. controls (152.3,91.91) and (155.1,94.7) .. (155.1,98.16) -- (155.1,106.49) .. controls (155.1,109.94) and (152.3,112.74) .. (148.85,112.74) -- (98.67,112.74) .. controls (95.22,112.74) and (92.42,109.94) .. (92.42,106.49) -- (92.42,98.16) .. controls (92.42,94.7) and (95.22,91.91) .. (98.67,91.91) -- cycle ;
        %Rounded Rect [id:dp31604307594908954] 
        \draw  [fill={rgb, 255:red, 58; green, 153; blue, 255 }  ,fill opacity=1 ][line width=0.75]  (134.27,56.97) .. controls (134.27,53.52) and (137.07,50.72) .. (140.52,50.72) -- (148.85,50.72) .. controls (152.3,50.72) and (155.1,53.52) .. (155.1,56.97) -- (155.1,107.15) .. controls (155.1,110.61) and (152.3,113.4) .. (148.85,113.4) -- (140.52,113.4) .. controls (137.07,113.4) and (134.27,110.61) .. (134.27,107.15) -- cycle ;
        %Rounded Rect [id:dp6257865349374669] 
        \draw  [fill={rgb, 255:red, 58; green, 153; blue, 255 }  ,fill opacity=1 ][line width=0.75]  (148.92,50.72) .. controls (152.37,50.72) and (155.17,53.52) .. (155.17,56.97) -- (155.17,65.3) .. controls (155.17,68.75) and (152.37,71.55) .. (148.92,71.55) -- (98.73,71.55) .. controls (95.28,71.55) and (92.48,68.75) .. (92.48,65.3) -- (92.48,56.97) .. controls (92.48,53.52) and (95.28,50.72) .. (98.73,50.72) -- cycle ;
        %Rounded Rect [id:dp476346437430198] 
        \draw  [fill={rgb, 255:red, 58; green, 153; blue, 255 }  ,fill opacity=1 ][line width=0.75]  (92.48,56.96) .. controls (92.48,53.51) and (95.28,50.72) .. (98.73,50.72) -- (107.07,50.73) .. controls (110.52,50.73) and (113.31,53.53) .. (113.31,56.99) -- (113.26,107.17) .. controls (113.25,110.62) and (110.45,113.42) .. (107,113.41) -- (98.67,113.4) .. controls (95.22,113.4) and (92.42,110.6) .. (92.42,107.15) -- cycle ;
        %Rounded Rect [id:dp5404523597945916] 
        \draw  [fill={rgb, 255:red, 50; green, 152; blue, 255 }  ,fill opacity=1 ][line width=0.75]  (107.35,91.91) .. controls (110.8,91.91) and (113.6,94.7) .. (113.6,98.16) -- (113.6,106.49) .. controls (113.6,109.94) and (110.8,112.74) .. (107.35,112.74) -- (57.17,112.74) .. controls (53.72,112.74) and (50.92,109.94) .. (50.92,106.49) -- (50.92,98.16) .. controls (50.92,94.7) and (53.72,91.91) .. (57.17,91.91) -- cycle ;
        %Rounded Rect [id:dp7931014820871032] 
        \draw  [fill={rgb, 255:red, 50; green, 152; blue, 255 }  ,fill opacity=1 ][line width=0.75]  (92.77,56.97) .. controls (92.77,53.52) and (95.57,50.72) .. (99.02,50.72) -- (107.35,50.72) .. controls (110.8,50.72) and (113.6,53.52) .. (113.6,56.97) -- (113.6,107.15) .. controls (113.6,110.61) and (110.8,113.4) .. (107.35,113.4) -- (99.02,113.4) .. controls (95.57,113.4) and (92.77,110.61) .. (92.77,107.15) -- cycle ;
        %Rounded Rect [id:dp61536012443635] 
        \draw  [fill={rgb, 255:red, 50; green, 152; blue, 255 }  ,fill opacity=1 ][line width=0.75]  (107.42,50.72) .. controls (110.87,50.72) and (113.67,53.52) .. (113.67,56.97) -- (113.67,65.3) .. controls (113.67,68.75) and (110.87,71.55) .. (107.42,71.55) -- (57.23,71.55) .. controls (53.78,71.55) and (50.98,68.75) .. (50.98,65.3) -- (50.98,56.97) .. controls (50.98,53.52) and (53.78,50.72) .. (57.23,50.72) -- cycle ;
        %Rounded Rect [id:dp17026585204836642] 
        \draw  [fill={rgb, 255:red, 50; green, 152; blue, 255 }  ,fill opacity=1 ][line width=0.75]  (50.12,56.96) .. controls (50.12,53.51) and (52.92,50.72) .. (56.37,50.72) -- (64.71,50.73) .. controls (68.16,50.73) and (70.95,53.53) .. (70.95,56.99) -- (70.9,107.17) .. controls (70.89,110.62) and (68.09,113.42) .. (64.64,113.41) -- (56.31,113.4) .. controls (52.86,113.4) and (50.06,110.6) .. (50.07,107.15) -- cycle ;
        %Rounded Rect [id:dp7923185623153957] 
        \draw  [fill={rgb, 255:red, 245; green, 166; blue, 35 }  ,fill opacity=1 ][line width=0.75]  (54,54.88) .. controls (54,54.88) and (54,54.88) .. (54,54.88) -- (67.28,54.9) .. controls (67.28,54.9) and (67.28,54.9) .. (67.28,54.9) -- (67.26,68.73) .. controls (67.26,68.73) and (67.26,68.73) .. (67.26,68.73) -- (53.99,68.72) .. controls (53.99,68.72) and (53.99,68.72) .. (53.99,68.72) -- cycle ;
        
        \end{mytikz4}\quad,
        \end{equation}
        or like in Fig.~\textcolor{blue}{1}(b) of the manuscript, such that the treewidth scales polynomially with the system size [A strict result in graph theory is that the treewidth of a 1D grid is $1$, while that of a 2D grid of size $L\times L$ is $L$, and that of a 3D grid of size $L\times L\times L$ is $\operatorname{\Omega}(L^2)$], the time complexity of contracting TNs of 2D FLDCs is expected to scale exponentially with the system size in the worst case~\cite{Markov2008}. We numerically verify this by conducting classical simulations for a 2D FLDC leveraging the optimized TN contraction library \href{https://cotengra.readthedocs.io/en/latest/advanced.html\#objective}{\sf{cotengra}}~\cite{Gray2021}. The circuit used here is exactly the same as that in our numerical experiment on the generalized toric code model, i.e., sequentially applying the ``claw-shape'' block on each plaquette of the $L\times L$ 2D lattice. The local observable is chosen as a single Pauli operator within the support of the last gate (so that the causal cone can cover the whole system). For clarity, we depict the circuit together with the location of the local observable below
        \begin{equation}\label{eq:claw_observable}
        \begin{mytikz4}
        %Straight Lines [id:da7377078134991011] 
        \draw [color={rgb, 255:red, 211; green, 211; blue, 211 }  ,draw opacity=1 ]   (107.21,52.21) -- (51.97,52.21) ;
        \draw [shift={(51.97,52.21)}, rotate = 180] [color={rgb, 255:red, 211; green, 211; blue, 211 }  ,draw opacity=1 ][fill={rgb, 255:red, 211; green, 211; blue, 211 }  ,fill opacity=1 ][line width=0.75]      (0, 0) circle [x radius= 1.34, y radius= 1.34]   ;
        %Straight Lines [id:da6242113050669598] 
        \draw [color={rgb, 255:red, 211; green, 211; blue, 211 }  ,draw opacity=1 ]   (51.97,107.45) -- (51.97,52.21) ;
        %Straight Lines [id:da2648607943481227] 
        \draw [color={rgb, 255:red, 211; green, 211; blue, 211 }  ,draw opacity=1 ]   (162.46,52.21) -- (107.21,52.21) ;
        \draw [shift={(107.21,52.21)}, rotate = 180] [color={rgb, 255:red, 211; green, 211; blue, 211 }  ,draw opacity=1 ][fill={rgb, 255:red, 211; green, 211; blue, 211 }  ,fill opacity=1 ][line width=0.75]      (0, 0) circle [x radius= 1.34, y radius= 1.34]   ;
        %Straight Lines [id:da2387739276191705] 
        \draw [color={rgb, 255:red, 211; green, 211; blue, 211 }  ,draw opacity=1 ]   (107.21,107.45) -- (107.21,52.21) ;
        %Straight Lines [id:da017606196190701917] 
        \draw [color={rgb, 255:red, 211; green, 211; blue, 211 }  ,draw opacity=1 ]   (215.13,52.21) -- (159.88,52.21) ;
        \draw [shift={(159.88,52.21)}, rotate = 180] [color={rgb, 255:red, 211; green, 211; blue, 211 }  ,draw opacity=1 ][fill={rgb, 255:red, 211; green, 211; blue, 211 }  ,fill opacity=1 ][line width=0.75]      (0, 0) circle [x radius= 1.34, y radius= 1.34]   ;
        %Straight Lines [id:da4616818798554865] 
        \draw [color={rgb, 255:red, 211; green, 211; blue, 211 }  ,draw opacity=1 ]   (159.88,107.45) -- (159.88,52.21) ;
        %Straight Lines [id:da2968209056658917] 
        \draw [color={rgb, 255:red, 211; green, 211; blue, 211 }  ,draw opacity=1 ]   (215.86,107.45) -- (215.86,52.21) ;
        \draw [shift={(215.86,52.21)}, rotate = 270] [color={rgb, 255:red, 211; green, 211; blue, 211 }  ,draw opacity=1 ][fill={rgb, 255:red, 211; green, 211; blue, 211 }  ,fill opacity=1 ][line width=0.75]      (0, 0) circle [x radius= 1.34, y radius= 1.34]   ;
        %Straight Lines [id:da2762748577903331] 
        \draw [color={rgb, 255:red, 211; green, 211; blue, 211 }  ,draw opacity=1 ]   (107.21,107.45) -- (51.97,107.45) ;
        \draw [shift={(51.97,107.45)}, rotate = 180] [color={rgb, 255:red, 211; green, 211; blue, 211 }  ,draw opacity=1 ][fill={rgb, 255:red, 211; green, 211; blue, 211 }  ,fill opacity=1 ][line width=0.75]      (0, 0) circle [x radius= 1.34, y radius= 1.34]   ;
        %Straight Lines [id:da014311797415645033] 
        \draw [color={rgb, 255:red, 211; green, 211; blue, 211 }  ,draw opacity=1 ]   (51.97,162.7) -- (51.97,107.45) ;
        %Straight Lines [id:da43245162612042654] 
        \draw [color={rgb, 255:red, 211; green, 211; blue, 211 }  ,draw opacity=1 ]   (162.46,107.45) -- (107.21,107.45) ;
        \draw [shift={(107.21,107.45)}, rotate = 180] [color={rgb, 255:red, 211; green, 211; blue, 211 }  ,draw opacity=1 ][fill={rgb, 255:red, 211; green, 211; blue, 211 }  ,fill opacity=1 ][line width=0.75]      (0, 0) circle [x radius= 1.34, y radius= 1.34]   ;
        %Straight Lines [id:da5042194035056937] 
        \draw [color={rgb, 255:red, 211; green, 211; blue, 211 }  ,draw opacity=1 ]   (107.21,162.7) -- (107.21,107.45) ;
        %Straight Lines [id:da9162554765033266] 
        \draw [color={rgb, 255:red, 211; green, 211; blue, 211 }  ,draw opacity=1 ]   (215.13,107.45) -- (159.88,107.45) ;
        \draw [shift={(159.88,107.45)}, rotate = 180] [color={rgb, 255:red, 211; green, 211; blue, 211 }  ,draw opacity=1 ][fill={rgb, 255:red, 211; green, 211; blue, 211 }  ,fill opacity=1 ][line width=0.75]      (0, 0) circle [x radius= 1.34, y radius= 1.34]   ;
        %Straight Lines [id:da9279479412779719] 
        \draw [color={rgb, 255:red, 211; green, 211; blue, 211 }  ,draw opacity=1 ]   (159.88,162.7) -- (159.88,107.45) ;
        %Straight Lines [id:da5065583617825644] 
        \draw [color={rgb, 255:red, 211; green, 211; blue, 211 }  ,draw opacity=1 ]   (215.86,162.7) -- (215.86,107.45) ;
        \draw [shift={(215.86,107.45)}, rotate = 270] [color={rgb, 255:red, 211; green, 211; blue, 211 }  ,draw opacity=1 ][fill={rgb, 255:red, 211; green, 211; blue, 211 }  ,fill opacity=1 ][line width=0.75]      (0, 0) circle [x radius= 1.34, y radius= 1.34]   ;
        %Straight Lines [id:da9189195070035312] 
        \draw [color={rgb, 255:red, 211; green, 211; blue, 211 }  ,draw opacity=1 ]   (107.21,159.75) -- (51.97,159.75) ;
        \draw [shift={(51.97,159.75)}, rotate = 180] [color={rgb, 255:red, 211; green, 211; blue, 211 }  ,draw opacity=1 ][fill={rgb, 255:red, 211; green, 211; blue, 211 }  ,fill opacity=1 ][line width=0.75]      (0, 0) circle [x radius= 1.34, y radius= 1.34]   ;
        %Straight Lines [id:da02049817166742729] 
        \draw [color={rgb, 255:red, 211; green, 211; blue, 211 }  ,draw opacity=1 ]   (51.97,215) -- (51.97,159.75) ;
        %Straight Lines [id:da2682458632580278] 
        \draw [color={rgb, 255:red, 211; green, 211; blue, 211 }  ,draw opacity=1 ]   (162.46,159.75) -- (107.21,159.75) ;
        \draw [shift={(107.21,159.75)}, rotate = 180] [color={rgb, 255:red, 211; green, 211; blue, 211 }  ,draw opacity=1 ][fill={rgb, 255:red, 211; green, 211; blue, 211 }  ,fill opacity=1 ][line width=0.75]      (0, 0) circle [x radius= 1.34, y radius= 1.34]   ;
        %Straight Lines [id:da8109044497902238] 
        \draw [color={rgb, 255:red, 211; green, 211; blue, 211 }  ,draw opacity=1 ]   (107.21,215) -- (107.21,159.75) ;
        %Straight Lines [id:da13813074323184904] 
        \draw [color={rgb, 255:red, 211; green, 211; blue, 211 }  ,draw opacity=1 ]   (215.13,159.75) -- (159.88,159.75) ;
        \draw [shift={(159.88,159.75)}, rotate = 180] [color={rgb, 255:red, 211; green, 211; blue, 211 }  ,draw opacity=1 ][fill={rgb, 255:red, 211; green, 211; blue, 211 }  ,fill opacity=1 ][line width=0.75]      (0, 0) circle [x radius= 1.34, y radius= 1.34]   ;
        %Straight Lines [id:da38006181153499496] 
        \draw [color={rgb, 255:red, 211; green, 211; blue, 211 }  ,draw opacity=1 ]   (159.88,215) -- (159.88,159.75) ;
        %Straight Lines [id:da3356165997646161] 
        \draw [color={rgb, 255:red, 211; green, 211; blue, 211 }  ,draw opacity=1 ]   (215.86,215) -- (215.86,159.75) ;
        \draw [shift={(215.86,159.75)}, rotate = 270] [color={rgb, 255:red, 211; green, 211; blue, 211 }  ,draw opacity=1 ][fill={rgb, 255:red, 211; green, 211; blue, 211 }  ,fill opacity=1 ][line width=0.75]      (0, 0) circle [x radius= 1.34, y radius= 1.34]   ;
        %Straight Lines [id:da05427327438794549] 
        \draw [color={rgb, 255:red, 211; green, 211; blue, 211 }  ,draw opacity=1 ]   (107.21,213.53) -- (51.97,213.53) ;
        \draw [shift={(51.97,213.53)}, rotate = 180] [color={rgb, 255:red, 211; green, 211; blue, 211 }  ,draw opacity=1 ][fill={rgb, 255:red, 211; green, 211; blue, 211 }  ,fill opacity=1 ][line width=0.75]      (0, 0) circle [x radius= 1.34, y radius= 1.34]   ;
        %Straight Lines [id:da7826525645030349] 
        \draw [color={rgb, 255:red, 211; green, 211; blue, 211 }  ,draw opacity=1 ]   (162.46,213.53) -- (107.21,213.53) ;
        \draw [shift={(107.21,213.53)}, rotate = 180] [color={rgb, 255:red, 211; green, 211; blue, 211 }  ,draw opacity=1 ][fill={rgb, 255:red, 211; green, 211; blue, 211 }  ,fill opacity=1 ][line width=0.75]      (0, 0) circle [x radius= 1.34, y radius= 1.34]   ;
        %Straight Lines [id:da4004959628104314] 
        \draw [color={rgb, 255:red, 211; green, 211; blue, 211 }  ,draw opacity=1 ]   (215.13,213.53) -- (159.88,213.53) ;
        \draw [shift={(159.88,213.53)}, rotate = 180] [color={rgb, 255:red, 211; green, 211; blue, 211 }  ,draw opacity=1 ][fill={rgb, 255:red, 211; green, 211; blue, 211 }  ,fill opacity=1 ][line width=0.75]      (0, 0) circle [x radius= 1.34, y radius= 1.34]   ;
        \draw [shift={(215.13,213.53)}, rotate = 180] [color={rgb, 255:red, 211; green, 211; blue, 211 }  ,draw opacity=1 ][fill={rgb, 255:red, 211; green, 211; blue, 211 }  ,fill opacity=1 ][line width=0.75]      (0, 0) circle [x radius= 1.34, y radius= 1.34]   ;
        
        %Rounded Rect [id:dp6597682266652292] 
        \draw  [fill={rgb, 255:red, 236; green, 247; blue, 255 }  ,fill opacity=1 ] (91.32,104.13) .. controls (93.48,106.28) and (93.48,109.78) .. (91.32,111.94) -- (86.11,117.15) .. controls (83.95,119.31) and (80.46,119.31) .. (78.3,117.15) -- (47.05,85.9) .. controls (44.89,83.74) and (44.89,80.24) .. (47.05,78.08) -- (52.26,72.88) .. controls (54.41,70.72) and (57.91,70.72) .. (60.07,72.88) -- cycle ;
        %Rounded Rect [id:dp6577054920141538] 
        \draw  [fill={rgb, 255:red, 182; green, 220; blue, 255 }  ,fill opacity=1 ] (90.79,155.69) .. controls (92.95,157.85) and (92.95,161.35) .. (90.79,163.5) -- (85.58,168.71) .. controls (83.43,170.87) and (79.93,170.87) .. (77.77,168.71) -- (46.52,137.46) .. controls (44.36,135.3) and (44.36,131.8) .. (46.52,129.65) -- (51.73,124.44) .. controls (53.89,122.28) and (57.38,122.28) .. (59.54,124.44) -- cycle ;
        %Rounded Rect [id:dp3623108732887015] 
        \draw  [fill={rgb, 255:red, 130; green, 195; blue, 255 }  ,fill opacity=1 ] (90.18,206.24) .. controls (92.34,208.4) and (92.34,211.89) .. (90.18,214.05) -- (84.97,219.26) .. controls (82.81,221.42) and (79.32,221.42) .. (77.16,219.26) -- (45.91,188.01) .. controls (43.75,185.85) and (43.75,182.35) .. (45.91,180.2) -- (51.12,174.99) .. controls (53.27,172.83) and (56.77,172.83) .. (58.93,174.99) -- cycle ;
        %Rounded Rect [id:dp5007082734748711] 
        \draw  [fill={rgb, 255:red, 236; green, 247; blue, 255 }  ,fill opacity=1 ] (72.36,54.49) .. controls (72.36,51.44) and (74.83,48.96) .. (77.89,48.96) -- (85.25,48.96) .. controls (88.3,48.96) and (90.78,51.44) .. (90.78,54.49) -- (90.78,108.71) .. controls (90.78,111.76) and (88.3,114.24) .. (85.25,114.24) -- (77.89,114.24) .. controls (74.83,114.24) and (72.36,111.76) .. (72.36,108.71) -- cycle ;
        %Rounded Rect [id:dp19562399222283244] 
        \draw  [fill={rgb, 255:red, 236; green, 247; blue, 255 }  ,fill opacity=1 ] (104.34,72.88) .. controls (106.5,70.72) and (110,70.72) .. (112.15,72.88) -- (117.36,78.08) .. controls (119.52,80.24) and (119.52,83.74) .. (117.36,85.9) -- (86.11,117.15) .. controls (83.95,119.31) and (80.46,119.31) .. (78.3,117.15) -- (73.09,111.94) .. controls (70.93,109.78) and (70.93,106.28) .. (73.09,104.13) -- cycle ;
        %Rounded Rect [id:dp2158914421272573] 
        \draw  [fill={rgb, 255:red, 225; green, 241; blue, 255 }  ,fill opacity=1 ] (143.95,104.2) .. controls (146.11,106.36) and (146.11,109.86) .. (143.95,112.01) -- (138.74,117.22) .. controls (136.58,119.38) and (133.09,119.38) .. (130.93,117.22) -- (99.68,85.97) .. controls (97.52,83.81) and (97.52,80.32) .. (99.68,78.16) -- (104.89,72.95) .. controls (107.04,70.79) and (110.54,70.79) .. (112.7,72.95) -- cycle ;
        %Rounded Rect [id:dp7224922203471145] 
        \draw  [fill={rgb, 255:red, 182; green, 220; blue, 255 }  ,fill opacity=1 ] (72.36,107.14) .. controls (72.36,104.09) and (74.83,101.62) .. (77.89,101.62) -- (85.25,101.62) .. controls (88.3,101.62) and (90.78,104.09) .. (90.78,107.14) -- (90.78,160.27) .. controls (90.78,163.33) and (88.3,165.8) .. (85.25,165.8) -- (77.89,165.8) .. controls (74.83,165.8) and (72.36,163.33) .. (72.36,160.27) -- cycle ;
        %Rounded Rect [id:dp7338329879204237] 
        \draw  [fill={rgb, 255:red, 182; green, 220; blue, 255 }  ,fill opacity=1 ] (103.82,124.44) .. controls (105.97,122.28) and (109.47,122.28) .. (111.63,124.44) -- (116.84,129.65) .. controls (118.99,131.8) and (118.99,135.3) .. (116.84,137.46) -- (85.58,168.71) .. controls (83.43,170.87) and (79.93,170.87) .. (77.77,168.71) -- (72.56,163.5) .. controls (70.41,161.35) and (70.41,157.85) .. (72.56,155.69) -- cycle ;
        %Rounded Rect [id:dp9727988143172841] 
        \draw  [fill={rgb, 255:red, 163; green, 213; blue, 255 }  ,fill opacity=1 ] (143.42,155.76) .. controls (145.58,157.92) and (145.58,161.42) .. (143.42,163.58) -- (138.22,168.79) .. controls (136.06,170.94) and (132.56,170.94) .. (130.4,168.79) -- (99.15,137.53) .. controls (96.99,135.38) and (96.99,131.88) .. (99.15,129.72) -- (104.36,124.51) .. controls (106.52,122.35) and (110.02,122.35) .. (112.17,124.51) -- cycle ;
        %Rounded Rect [id:dp5824706930579204] 
        \draw  [fill={rgb, 255:red, 225; green, 241; blue, 255 }  ,fill opacity=1 ] (126.25,54.16) .. controls (126.25,51.11) and (128.72,48.64) .. (131.77,48.64) -- (139.14,48.64) .. controls (142.19,48.64) and (144.66,51.11) .. (144.66,54.16) -- (144.66,107.98) .. controls (144.66,111.04) and (142.19,113.51) .. (139.14,113.51) -- (131.77,113.51) .. controls (128.72,113.51) and (126.25,111.04) .. (126.25,107.98) -- cycle ;
        %Rounded Rect [id:dp422251869590474] 
        \draw  [fill={rgb, 255:red, 225; green, 241; blue, 255 }  ,fill opacity=1 ] (156.97,72.95) .. controls (159.13,70.79) and (162.63,70.79) .. (164.79,72.95) -- (169.99,78.16) .. controls (172.15,80.32) and (172.15,83.81) .. (169.99,85.97) -- (138.74,117.22) .. controls (136.58,119.38) and (133.09,119.38) .. (130.93,117.22) -- (125.72,112.01) .. controls (123.56,109.86) and (123.56,106.36) .. (125.72,104.2) -- cycle ;
        %Rounded Rect [id:dp45112230262195685] 
        \draw  [fill={rgb, 255:red, 203; green, 230; blue, 255 }  ,fill opacity=1 ] (197.25,103.73) .. controls (199.41,105.88) and (199.41,109.38) .. (197.25,111.54) -- (192.04,116.75) .. controls (189.88,118.91) and (186.39,118.91) .. (184.23,116.75) -- (152.98,85.5) .. controls (150.82,83.34) and (150.82,79.84) .. (152.98,77.68) -- (158.19,72.48) .. controls (160.34,70.32) and (163.84,70.32) .. (166,72.48) -- cycle ;
        %Rounded Rect [id:dp8907970394084659] 
        \draw  [fill={rgb, 255:red, 203; green, 230; blue, 255 }  ,fill opacity=1 ] (178.68,53.04) .. controls (178.68,49.99) and (181.15,47.52) .. (184.2,47.52) -- (191.57,47.52) .. controls (194.62,47.52) and (197.09,49.99) .. (197.09,53.04) -- (197.09,107.51) .. controls (197.09,110.56) and (194.62,113.04) .. (191.57,113.04) -- (184.2,113.04) .. controls (181.15,113.04) and (178.68,110.56) .. (178.68,107.51) -- cycle ;
        %Rounded Rect [id:dp3376374584824313] 
        \draw  [fill={rgb, 255:red, 203; green, 230; blue, 255 }  ,fill opacity=1 ] (210.27,72.48) .. controls (212.43,70.32) and (215.93,70.32) .. (218.08,72.48) -- (223.29,77.68) .. controls (225.45,79.84) and (225.45,83.34) .. (223.29,85.5) -- (192.04,116.75) .. controls (189.88,118.91) and (186.39,118.91) .. (184.23,116.75) -- (179.02,111.54) .. controls (176.86,109.38) and (176.86,105.88) .. (179.02,103.73) -- cycle ;
        %Rounded Rect [id:dp6251928707650753] 
        \draw  [fill={rgb, 255:red, 163; green, 213; blue, 255 }  ,fill opacity=1 ] (125.32,105.92) .. controls (125.32,102.87) and (127.8,100.4) .. (130.85,100.4) -- (138.21,100.4) .. controls (141.27,100.4) and (143.74,102.87) .. (143.74,105.92) -- (143.74,159.95) .. controls (143.74,163) and (141.27,165.47) .. (138.21,165.47) -- (130.85,165.47) .. controls (127.8,165.47) and (125.32,163) .. (125.32,159.95) -- cycle ;
        %Rounded Rect [id:dp7136295979824032] 
        \draw  [fill={rgb, 255:red, 163; green, 213; blue, 255 }  ,fill opacity=1 ] (156.45,124.51) .. controls (158.6,122.35) and (162.1,122.35) .. (164.26,124.51) -- (169.47,129.72) .. controls (171.63,131.88) and (171.63,135.38) .. (169.47,137.53) -- (138.22,168.79) .. controls (136.06,170.94) and (132.56,170.94) .. (130.4,168.79) -- (125.19,163.58) .. controls (123.04,161.42) and (123.04,157.92) .. (125.19,155.76) -- cycle ;
        %Rounded Rect [id:dp7975123541270419] 
        \draw  [fill={rgb, 255:red, 147; green, 206; blue, 255 }  ,fill opacity=1 ] (197.72,154.29) .. controls (199.88,156.45) and (199.88,159.95) .. (197.72,162.1) -- (192.52,167.31) .. controls (190.36,169.47) and (186.86,169.47) .. (184.7,167.31) -- (153.45,136.06) .. controls (151.29,133.9) and (151.29,130.4) .. (153.45,128.25) -- (158.66,123.04) .. controls (160.82,120.88) and (164.32,120.88) .. (166.47,123.04) -- cycle ;
        %Rounded Rect [id:dp9138578459054092] 
        \draw  [fill={rgb, 255:red, 147; green, 206; blue, 255 }  ,fill opacity=1 ] (178.58,107.05) .. controls (178.58,104) and (181.05,101.52) .. (184.1,101.52) -- (191.47,101.52) .. controls (194.52,101.52) and (196.99,104) .. (196.99,107.05) -- (196.99,157.95) .. controls (196.99,161) and (194.52,163.47) .. (191.47,163.47) -- (184.1,163.47) .. controls (181.05,163.47) and (178.58,161) .. (178.58,157.95) -- cycle ;
        %Rounded Rect [id:dp44935131518495197] 
        \draw  [fill={rgb, 255:red, 147; green, 206; blue, 255 }  ,fill opacity=1 ] (210.75,123.04) .. controls (212.9,120.88) and (216.4,120.88) .. (218.56,123.04) -- (223.77,128.25) .. controls (225.92,130.4) and (225.92,133.9) .. (223.77,136.06) -- (192.52,167.31) .. controls (190.36,169.47) and (186.86,169.47) .. (184.7,167.31) -- (179.49,162.1) .. controls (177.34,159.95) and (177.34,156.45) .. (179.49,154.29) -- cycle ;
        %Rounded Rect [id:dp9400928535115993] 
        \draw  [fill={rgb, 255:red, 130; green, 195; blue, 255 }  ,fill opacity=1 ] (72.03,159.74) .. controls (72.02,156.69) and (74.49,154.21) .. (77.54,154.2) -- (84.91,154.19) .. controls (87.96,154.19) and (90.44,156.66) .. (90.44,159.71) -- (90.53,212.21) .. controls (90.54,215.26) and (88.07,217.73) .. (85.02,217.74) -- (77.65,217.75) .. controls (74.6,217.76) and (72.12,215.29) .. (72.12,212.24) -- cycle ;
        %Rounded Rect [id:dp8446072302484551] 
        \draw  [fill={rgb, 255:red, 130; green, 195; blue, 255 }  ,fill opacity=1 ] (103.2,175.39) .. controls (105.36,173.23) and (108.86,173.23) .. (111.01,175.39) -- (116.22,180.6) .. controls (118.38,182.75) and (118.38,186.25) .. (116.22,188.41) -- (84.97,219.66) .. controls (82.81,221.82) and (79.32,221.82) .. (77.16,219.66) -- (71.95,214.45) .. controls (69.79,212.29) and (69.79,208.8) .. (71.95,206.64) -- cycle ;
        %Rounded Rect [id:dp3979257201544728] 
        \draw  [fill={rgb, 255:red, 84; green, 174; blue, 255 }  ,fill opacity=1 ] (142.81,206.71) .. controls (144.97,208.87) and (144.97,212.37) .. (142.81,214.53) -- (137.6,219.73) .. controls (135.44,221.89) and (131.95,221.89) .. (129.79,219.73) -- (98.54,188.48) .. controls (96.38,186.33) and (96.38,182.83) .. (98.54,180.67) -- (103.75,175.46) .. controls (105.9,173.3) and (109.4,173.3) .. (111.56,175.46) -- cycle ;
        %Rounded Rect [id:dp534570690302034] 
        \draw  [fill={rgb, 255:red, 84; green, 174; blue, 255 }  ,fill opacity=1 ] (124.31,159.36) .. controls (124.31,156.31) and (126.78,153.84) .. (129.83,153.84) -- (137.2,153.84) .. controls (140.25,153.84) and (142.72,156.31) .. (142.72,159.36) -- (142.72,212.1) .. controls (142.72,215.15) and (140.25,217.62) .. (137.2,217.62) -- (129.83,217.62) .. controls (126.78,217.62) and (124.31,215.15) .. (124.31,212.1) -- cycle ;
        %Rounded Rect [id:dp3010202576000376] 
        \draw  [fill={rgb, 255:red, 84; green, 174; blue, 255 }  ,fill opacity=1 ] (155.83,175.46) .. controls (157.99,173.3) and (161.49,173.3) .. (163.65,175.46) -- (168.85,180.67) .. controls (171.01,182.83) and (171.01,186.33) .. (168.85,188.48) -- (137.6,219.73) .. controls (135.44,221.89) and (131.95,221.89) .. (129.79,219.73) -- (124.58,214.53) .. controls (122.42,212.37) and (122.42,208.87) .. (124.58,206.71) -- cycle ;
        %Rounded Rect [id:dp8707767679496903] 
        \draw  [fill={rgb, 255:red, 35; green, 159; blue, 255 }  ,fill opacity=1 ] (197.11,205.24) .. controls (199.27,207.4) and (199.27,210.89) .. (197.11,213.05) -- (191.9,218.26) .. controls (189.74,220.42) and (186.25,220.42) .. (184.09,218.26) -- (152.84,187.01) .. controls (150.68,184.85) and (150.68,181.35) .. (152.84,179.2) -- (158.05,173.99) .. controls (160.2,171.83) and (163.7,171.83) .. (165.86,173.99) -- cycle ;
        %Rounded Rect [id:dp42629372187754] 
        \draw  [fill={rgb, 255:red, 35; green, 159; blue, 255 }  ,fill opacity=1 ] (178.95,157.84) .. controls (178.95,154.79) and (181.43,152.32) .. (184.48,152.32) -- (191.84,152.32) .. controls (194.89,152.32) and (197.37,154.79) .. (197.37,157.84) -- (197.37,211) .. controls (197.37,214.05) and (194.89,216.52) .. (191.84,216.52) -- (184.48,216.52) .. controls (181.43,216.52) and (178.95,214.05) .. (178.95,211) -- cycle ;
        %Rounded Rect [id:dp6803371001993503] 
        \draw  [fill={rgb, 255:red, 35; green, 159; blue, 255 }  ,fill opacity=1 ] (210.13,173.99) .. controls (212.29,171.83) and (215.79,171.83) .. (217.94,173.99) -- (223.15,179.2) .. controls (225.31,181.35) and (225.31,184.85) .. (223.15,187.01) -- (191.9,218.26) .. controls (189.74,220.42) and (186.25,220.42) .. (184.09,218.26) -- (178.88,213.05) .. controls (176.72,210.89) and (176.72,207.4) .. (178.88,205.24) -- cycle ;
        %Rounded Rect [id:dp00850291727339636] 
        \draw  [fill={rgb, 255:red, 245; green, 166; blue, 35 }  ,fill opacity=1 ] (210.86,182.38) .. controls (210.86,182.38) and (210.86,182.38) .. (210.86,182.38) -- (220.86,182.38) .. controls (220.86,182.38) and (220.86,182.38) .. (220.86,182.38) -- (220.86,192.38) .. controls (220.86,192.38) and (220.86,192.38) .. (220.86,192.38) -- (210.86,192.38) .. controls (210.86,192.38) and (210.86,192.38) .. (210.86,192.38) -- cycle ;
        %Straight Lines [id:da16296576181874833] 
        \draw    (51.2,230.13) -- (215.2,230.13) ;
        \draw [shift={(215.2,230.13)}, rotate = 180] [color={rgb, 255:red, 0; green, 0; blue, 0 }  ][line width=0.75]    (0,5.59) -- (0,-5.59)(10.93,-4.9) .. controls (6.95,-2.3) and (3.31,-0.67) .. (0,0) .. controls (3.31,0.67) and (6.95,2.3) .. (10.93,4.9)   ;
        \draw [shift={(51.2,230.13)}, rotate = 0] [color={rgb, 255:red, 0; green, 0; blue, 0 }  ][line width=0.75]    (0,5.59) -- (0,-5.59)(10.93,-4.9) .. controls (6.95,-2.3) and (3.31,-0.67) .. (0,0) .. controls (3.31,0.67) and (6.95,2.3) .. (10.93,4.9)   ;
        
        % Text Node
        \draw (127.33,234) node [anchor=north west][inner sep=0.75pt]   [align=left] {$\displaystyle L$};
        
        \end{mytikz4}\quad.
        \end{equation}
        One can see there are loops of 2-qubit gates formed among plaquettes. The data in Figs.~\ref{fig:flops} shows an exponential growth in both time and space complexity with the side length. 
        \begin{figure}
            \centering
            \includegraphics[width=0.85\linewidth]{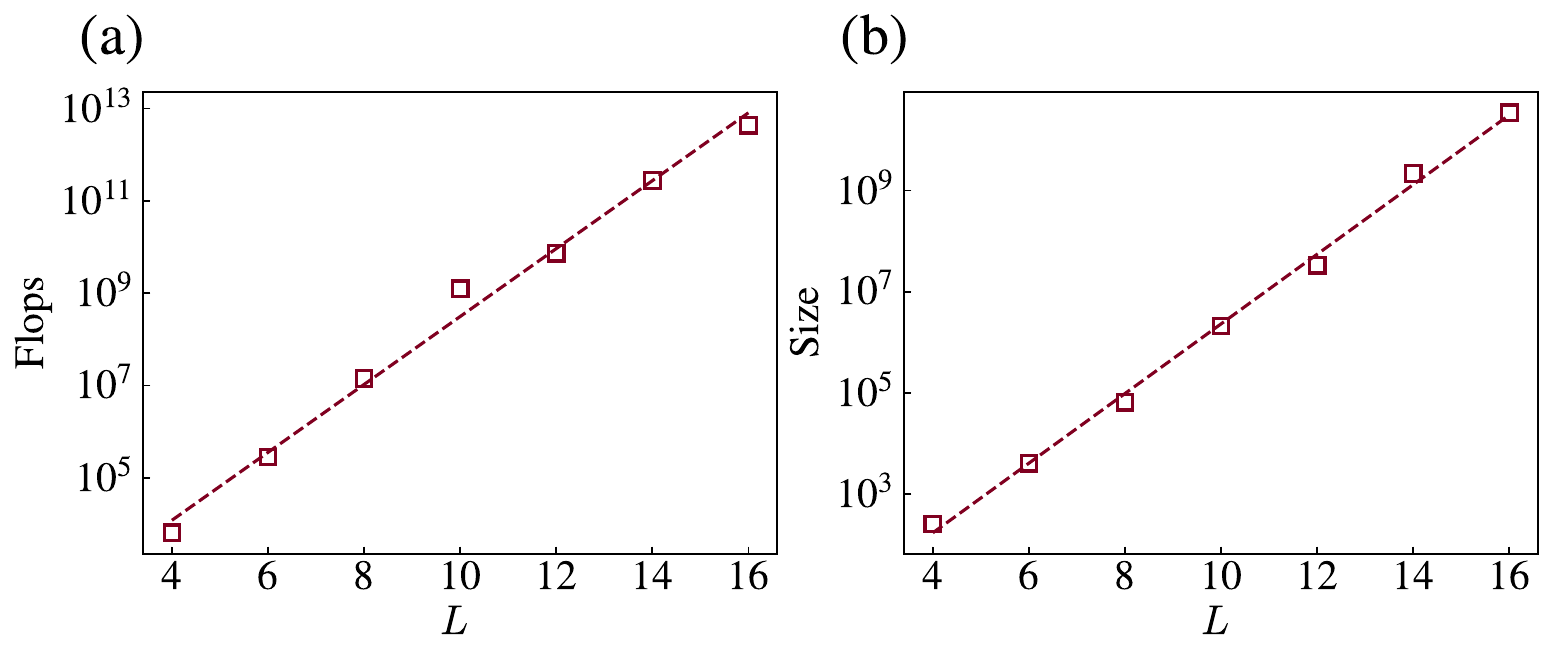}
            \caption{(a) and (b) depict the time and space cost, respectively, for contracting the TN of the expectation value of a single-qubit Pauli operator regarding the 2D FLDC in Eq.~\eqref{eq:claw} versus the side length of the $L\times L$ 2D grid. Here the term ``flops'' represents the total number of scalar operations during the tensor contraction, and ``size'' represents the size of the largest intermediate tensor during the tensor contraction. See the documentation of \href{https://cotengra.readthedocs.io/en/latest/advanced.html\#objective}{\sf{cotengra}}~\cite{Gray2021} for more information.}
            \label{fig:flops}
        \end{figure}
        This indicates that even with state-of-the-art optimized contraction strategies~\cite{Gray2021}, general 2D FLDCs still cannot be simulated efficiently via contracting TNs to estimate local observables. To further corroborate this point, we \textbf{review the subclasses of 2D PEPS} that have been studied previously~\cite{Banuls2008, Zaletel2020, Soejima2020, Haferkamp2020, Bravyi2021, Haller2023, Liu2023}, and explain why some of them can be efficiently simulated. Here, we start with cases that have theoretical guarantees and do not delve into ambiguous situations, such as heuristic algorithms without controllable errors~\cite{Anshu2016, Schwarz2017, Ran2020a, Orus2019}. 
        \begin{itemize}
            \item Sequential generated states (SGS)~\cite{Banuls2008} refer to those states generated by acting $2$-qubit gates sequentially on a series of parallel MPS along the direction perpendicular to the MPS. Alternatively, SGS can be seen as generated by FLDC under the condition that the $2$-qubit gates along the $\hat{y}$-direction are all applied first, and then those along the $\hat{x}$-direction start to be applied. Thus, local observables can be efficiently estimated classically~\cite{Banuls2008}. This can be understood by the fact that the causal cone only involves few $\hat{x}$-slices, which prevents from forming many loops mentioned above, making the corresponding TN tree-like.
            \item Isometric PEPS (isoPEPS)~\cite{Zaletel2020, Liu2023, Slattery2021} restricts each local tensor in PEPS as isometries. Only the local observables close to the so-called orthogonality hypersurface can be computed efficiently (where the ``causal cone'' is tree-like), and those not close can only be computed by uncontrollable approximation techniques~\cite{Zaletel2020}. 
            \item Plaquette PEPS (P-PEPS)~\cite{Wei2022} applies a linear number of unitaries acting on plaquettes of overlapping regions sequentially. Only the local observables at special locations (near the gates that act earlier in the sequence) whose causal cone is small can be efficiently estimated. No efficient method is known for estimating local observables in P-PEPS for general cases~\cite{Wei2022}.
        \end{itemize}
        One can see that the classical simulability of these subclasses is consistent with the ``loop argument'' mentioned below Eq.~\eqref{eq:2by2_loop}, i.e., the classically simulable cases always have a tree-like causal cone. That is to say, after canceling the trivial unitaries or isometries with their Hemitian conjugates, the remaining TN structure has a treewidth that does not scale with the system size. Moreover, it has been proved~\cite{Wei2022} that these circuit subclasses form a hierarchy $\text{SGS}\subset\text{isoPEPS} \subset\text{P-PEPS}$ with the increasing complexity of classical simulation. In addition, P-PEPS has a finite local depth by definition and hence we have $\text{P-PEPS}\subset\text{FLDC}$ (note that compared to P-PEPS, FLDC has the freedom to include non-local but few-qubit gates). Physical ground states in iso-PEPS include string-net states~\cite{Soejima2020, Liu2022d} and those in P-PEPS include certain fracton-ordered states~\cite{Chen2023a}. For the sake of the ability of FLDC to apply non-local gates (which could be experimentally relevant with the reconfigurable neutral atom platform \cite{Bluvstein2022, Bluvstein2023}), FLDC may include chiral topologically ordered states~\cite{Kitaev2006, Chu2023} and the ground states (code space) for the quantum low-density parity-check (qLDPC) code stabilizer Hamiltonians which contain few-body but non-local Pauli string terms \cite{Anshu2023, Bravyi2023, Xu2023}. However, determining precisely which physical states belong to which classes and which do not remains an open problem. Up to this point, we see that for iso-PEPS and especially P-PEPS, estimating local observables generally has no known efficient method with controllable error~\cite{Wei2022}. As a superset of P-PEPS, FLDC naturally does not either. Therefore, we conclude that to the best of our knowledge, in general, the existing classical TN methods~\cite{Banuls2008, Zaletel2020, Soejima2020, Haferkamp2020, Bravyi2021, Haller2023, Liu2023} cannot efficiently simulate 2D FLDCs even in case (i).
        
        \item[$\circ$] For 3D FLDCs and those of higher dimensions and even the scenarios without gate locality (i.e., arbitrary connectivity graphs such as in big chemical molecules), TN contraction can only be more challenging than 2D. Meanwhile, our proof of the absence of barren plateaus still holds for these higher-dimensional and non-local FLDCs.
        
    \end{itemize}
    To sum up, in case (i), 1D FLDCs can be efficiently simulated by TN methods classically while 2D FLDCs and beyond in general can not. On top of that, our theorem also indicates that a superclass of FLDC, i.e., logarithmic local-depth circuits (Log-LDC) which are beyond entanglement-area-law states, is also BP-free. Log-LDC can be included by generalized PEPS with polynomially increasing bond dimensions. Hence, the corresponding classical simulation can only be more challenging.
    
    \item For general linear depth circuits (GLDC) or polynomial depth circuits beyond FLDC and Log-LDC, barren plateaus are proven starting to occur~\cite{McClean2018}. At the same time, it comes to the regime of general quantum simulation~\cite{Feynman1982, Pan2020, Pan2022}, which is not solvable on a classical computer in polynomial time in general unless $\mathsf{BQP}=\mathsf{BPP}$. It is worth noting that the difficulty in simulating general quantum circuits through contraction TNs also arises from extensive loops~\cite{Markov2008}, albeit not in the pure spatial domains but rather in the space-time domains.
\end{itemize}

To conclude, among the hierarchy listed above, the critical point for BP lies between FDLC and GLDC, while the critical point of confirmed classical simulability lies between 1D FDLC and 2D FDLC. Hence, there exists a region that is BP-free and not efficiently classically simulable with known methods, i.e., FDLC in 2D and above, as depicted below.
\begin{equation}  
\begin{mytikz3}

%Shape: Rectangle [id:dp8557871860663901] 
\draw  [draw opacity=0][fill={rgb, 255:red, 237; green, 237; blue, 237 }  ,fill opacity=1 ] (6.8,58.67) -- (510.5,58.67) -- (510.5,162) -- (6.8,162) -- cycle ;
%Shape: Rectangle [id:dp9576655713144668] 
\draw  [draw opacity=0][fill={rgb, 255:red, 224; green, 238; blue, 255 }  ,fill opacity=1 ] (6.8,162) -- (186.13,162) -- (186.13,214.67) -- (6.8,214.67) -- cycle ;
%Shape: Rectangle [id:dp20475492111913352] 
\draw  [draw opacity=0][fill={rgb, 255:red, 254; green, 212; blue, 212 }  ,fill opacity=1 ] (186.13,162) -- (510.13,162) -- (510.13,214.67) -- (186.13,214.67) -- cycle ;
%Shape: Rectangle [id:dp6956952574681841] 
\draw  [draw opacity=0][fill={rgb, 255:red, 224; green, 238; blue, 255 }  ,fill opacity=1 ] (6.8,15.23) -- (335.47,15.23) -- (335.47,58.67) -- (6.8,58.67) -- cycle ;
%Shape: Rectangle [id:dp2648058358150249] 
\draw  [draw opacity=0][fill={rgb, 255:red, 254; green, 212; blue, 212 }  ,fill opacity=1 ] (335.47,15.23) -- (510.13,15.23) -- (510.13,58.67) -- (335.47,58.67) -- cycle ;
%Straight Lines [id:da8077684501224884] 
\draw [line width=1.5]    (6.8,108.9) -- (529.8,108.9) ;
\draw [shift={(533.8,108.9)}, rotate = 180] [fill={rgb, 255:red, 0; green, 0; blue, 0 }  ][line width=0.08]  [draw opacity=0] (13.4,-6.43) -- (0,0) -- (13.4,6.44) -- (8.9,0) -- cycle    ;
\draw [shift={(6.8,108.9)}, rotate = 0] [color={rgb, 255:red, 0; green, 0; blue, 0 }  ][fill={rgb, 255:red, 0; green, 0; blue, 0 }  ][line width=1.5]      (0, 0) circle [x radius= 4.36, y radius= 4.36]   ;
%Shape: Brace [id:dp1624143273196632] 
\draw   (332.27,127.33) .. controls (332.27,122.66) and (329.94,120.33) .. (325.27,120.33) -- (255.6,120.33) .. controls (248.93,120.33) and (245.6,118) .. (245.6,113.33) .. controls (245.6,118) and (242.27,120.33) .. (235.6,120.33)(238.6,120.33) -- (165.93,120.33) .. controls (161.26,120.33) and (158.93,122.66) .. (158.93,127.33) ;
%Straight Lines [id:da18295828791211766] 
\draw  [dash pattern={on 3.75pt off 2.25pt}]  (335.47,15.23) -- (336.5,108.05) ;
%Straight Lines [id:da2969797287605309] 
\draw  [dash pattern={on 3.75pt off 2.25pt}]  (186.13,108.55) -- (186.13,214.67) ;

% Text Node
\draw (22,84) node [anchor=north west][inner sep=0.75pt]   [align=left] {Finite/Log Depth};
% Text Node
\draw (164,84) node [anchor=north west][inner sep=0.75pt]   [align=left] {Finite/Log Local-Depth};
% Text Node
\draw (350,84) node [anchor=north west][inner sep=0.75pt]   [align=left] {General Linear Depth};
% Text Node
\draw (158,130) node [anchor=north west][inner sep=0.75pt]   [align=left] {\small 1D};
% Text Node
\draw (199,130) node [anchor=north west][inner sep=0.75pt]   [align=left] {\small 2D, 3D, ..., non-local};
% Text Node
\draw (192,27) node [anchor=north west][inner sep=0.75pt]   [align=left] {Barren-Plateau-free};
% Text Node
\draw (347,27) node [anchor=north west][inner sep=0.75pt]   [align=left] {Barren Plateaus};
% Text Node
\draw (25.6,170) node [anchor=north west][inner sep=0.75pt]   [align=left] {\begin{minipage}[lt]{98.52pt}\setlength\topsep{0pt}
\begin{center}
\small Classically Simulable\\(Local Observable)
\end{center}

\end{minipage}};
% Text Node
\draw (195.33,170) node [anchor=north west][inner sep=0.75pt]   [align=left] {\begin{minipage}[lt]{169.96pt}\setlength\topsep{0pt}
\begin{center}
\small Not Classically Simulable in General \\with Known Methods
\end{center}
\end{minipage}};

\end{mytikz3}
\end{equation}
In particular, it is worth mentioning that the fact that running FLDC is in $\mathsf{BQP}$ has no contradiction with the fact that contracting PEPS is $\#\mathsf{P}\text{-}\mathsf{complete}$~\cite{Schuch2007} which is believed not in $\mathsf{BQP}$, because the states generated by local FLDCs just form a subclass of all PEPS. 

After we complete this work, we notice that a very recent paper~\cite{Malz2024} rigorously proves that computing expectation values of local observables in isoPEPS (referred to as isometric tensor network states in Ref.~\cite{Malz2024}) is $\mathsf{BQP}$-complete, i.e., is hard to simulate classically unless $\mathsf{BQP}=\mathsf{BPP}$. Since isoPEPS is a subclass of FLDC, this result directly proves that FLDC is also hard to simulate classically in general unless $\mathsf{BQP}=\mathsf{BPP}$. That is to say, FLDC can be seen as a typical counterexample of the strong version of the conjecture in Ref.~\cite{Cerezo2023}. In other words, our results demonstrate that the provable absence of barren plateaus does not necessarily imply classical simulability in the full parameter space. Nevertheless, we also remark that demonstrating the quantum advantage of barren-plateau-free circuits in practical problems of interest still requires further research.

\vspace{10pt}

Next, we consider case (ii), which arises from the need to solve practical problems in quantum physics. In this case, we not only need to measure local observables in the Hamiltonian during the training process but also need to measure other observables of interest on the ground state after training, which can involve non-local and many-body observables. Similar to case (i), we categorize the discussion in terms of circuit depth. (One may think that the measurements after training can be separated from the training process, i.e., classically simulating the training and then loading the circuit on quantum devices, so that the quantum advantage in subsequent measurements has nothing to do with the training. However, we remark that the classical simulation of the training process can also be affected by barren plateaus~\cite{Cerezo2021, Liu2021a, Liu2023c}. Thus, whether the training process is classical or quantum, a ground state ansatz of circuit form without barren plateaus is always a necessary condition for achieving quantum advantage in subsequent measurements.)
\begin{itemize}
    \item For few-body but non-local observables, shallow circuits can be efficiently simulated due to the same reason for local observables mentioned above. For many-body observables, Ref.~\cite{Bravyi2021} has proved that 2D finite-depth circuits (FDC) can be efficiently simulated to estimate tensor-product observables, which is realized by leveraging the fact that the individual small causal cones allow for reducing the corresponding 2D TN into a series of quasi-1D structures, enabling efficient simulations. Of course, it is also true for 1D FDC due to the efficiency of MPS. However, no efficient method is known for 3D and beyond~\cite{Bravyi2021}. For logarithmic depth circuits, it has been known that the 1D case can be classically simulated within polynomial time even for many-body observables because the treewidth of the 1D local circuit can be upper bounded by the logarithmic depth~\cite{Markov2008}.
    \item For 1D FLDCs, many-body observables can be estimated efficiently again due to the efficiency of MPS. For the subclasses in 2D FLDCs, it is known that few-body observables can be estimated efficiently in SGS~\cite{Banuls2008}, but generally not in iso-PEPS~\cite{Zaletel2020} and P-PEPS~\cite{Wei2022}. When it comes to many-body observables, none of these subclasses is known to be classically simulable~\cite{Banuls2008, Zaletel2020, Wei2022}. For general 2D FLDCs and those of higher dimensions and the scenarios without gate locality, the simulation can only be more challenging.
    \item For general linear or polynomial depth circuits, classical simulation estimating non-local observables can only be more difficult than estimating local observables.
\end{itemize}

In particular, we consider an extreme task in case (ii): estimating the so-called ``dynamical correlation'' at zero temperature (i.e., on the ground state)~\cite{Paeckel2019, Huang2017}, which is particularly of physical interest as it encodes the important information of elementary excitations in the system. The expression of the dynamical correlation, or say propagator / non-equal-time correlation is of the form
\begin{equation}
    C(t) = \langle G | O_2(t) O_1(0) |G\rangle = \langle G | e^{iHt} O_2 e^{-iHt} O_1 |G\rangle,
\end{equation}
where $|G\rangle$ is the ground state of the Hamiltonian $H$. $O_1$ and $O_2$ are two given observables separated by a period of time evolution $e^{-iHt}$. $O_j(t)=e^{iHt} O_j e^{-iHt}$ represents the time-evolved operator under the Heisenberg picture. One can see that estimating $C(t)$ involves both the ground state preparation and the Hamiltonian simulation as its subroutines. The Hamiltonian simulation~\cite{Feynman1982} is hard not only for TN contractions but also for quantum Monte Carlo techniques, which in general cannot be classically simulated unless $\mathsf{BQP}=\mathsf{BPP}$. To enable the quantum advantage in estimating $C(t)$, a trainable and expressive ground state ansatz of quantum circuit form, especially those expressive enough to contain quantum states of physical interest such as topologically ordered states (i.e., FLDC), becomes a necessary condition. This further underscores the significance of the results in our manuscript.

To conclude, in case (ii), the scope of efficient classical simulation becomes even smaller than that in case (i). Especially, FLDCs of 2D and above are not classically simulable and the simulable special cases within them become fewer than those in case (i). In other words, if one considers measuring non-local observables of interest as a part of the task, the possibility that the FLDC class contains a path to quantum advantage becomes larger than merely preparing the ground states of spatially local Hamiltonians.

\section{Measurement-assisted Approach to Generate Long-Range Entanglement}

Apart from purely unitary quantum circuits, we notice that if assisted by intermediate measurement and non-local classical feedback, constant depth circuits can also generate long-range entanglement, which is one of the recent hot topics in the field of quantum many-body physics~\cite{Raussendorf2005, Piroli2021, Verresen2021, Tantivasadakarn2021, Bravyi2022, Tantivasadakarn2023, Tantivasadakarn2023a}. We will give a brief discussion on this measurement-assisted approach for comparison. 

To discuss more concretely in the following, we take the pure toric code model $H=-\sum_v A_v-\sum_p B_p$ for an example~\cite{Tantivasadakarn2023}. Starting from the initial state $\ket{0}^{\otimes N}$ which is already the eigenstate state of $A_v$, we can obtain one of the eigenstates of the Hamiltonian by just simply measuring all $B_p$ terms (note that reducing the measurements of $B_p$ to single-qubit measurements requires a constant depth circuit composed of CNOTs) so that the state will be projected into one of the eigenstates. The measurement results of $B_p$ can be $+1$ or $-1$ randomly, where the plaquettes with $-1$ can be regarded as the positions of anyon excitations. The ground state can be obtained by annihilating these anyon excitations in pairs via a single classical feedback layer of commuting string operators~\cite{Tantivasadakarn2023}, analogous to what one does in quantum error correction. Through this approach, one can see that certain quantum states with long-range entanglement can indeed be prepared by constant depth circuits assisted by constant depth measurements and classical feedback, which circumvents the linear bounds~\cite{Bravyi2006} thanks to the non-unitarity of measurement. Below, we further discuss the classical simulability of FLDC and long-range entangled (LRE) states regarding the context of this measurement-assisted approach.
\begin{itemize}
    \item Classically simulating constant depth circuits assisted with measurement and classical feedback, is not as easy as simulating only constant depth circuits. The efficiency in circuit depth of this measurement-assisted approach comes from the non-unitarity of measurement, but precisely because of the non-unitarity, the difficulty of classical simulation is significantly increased. To be specific, as discussed above, estimating local observables for constant depth circuits is easy to simulate classically for the sake of the small causal cone~\cite{Bravyi2021, Cerezo2023}, i.e., the unitaries outside the causal cone are eliminated by their Hermitian conjugations so that they do not need to be computed at all. Nevertheless, when there are non-unitary operators within the circuit (such as the projectors in projective measurements), especially when there are many of them covering the whole system (such as the pure toric code example above), the small causal cone does not exist anymore in general, so that one must contract all operators in the circuit to obtain the final result, which will make classical simulation extremely hard in 2D and above. In addition, it is worth mentioning that PEPS, short for ``Projected Entangled Paired State''~\cite{Verstraete2006}, itself can be viewed as obtained by projections onto states from constant depth circuits, which further highlights the difficulties of classical simulation introduced by the non-unitary factors like projections.
    \item There exist some long-range entangled ground states thought to be beyond the reach of this measurement-assisted approach of constant depth, such as the famous Fibonacci topologically ordered states~\cite{Tantivasadakarn2023a} which plays an important role in topological quantum computing, while these states are still within the reach of the subclasses of FLDC~\cite{Liu2022d, Chen2023a}.
    \item At the current stage, this measurement-assisted approach seems to strongly rely on our complete knowledge of the quantum state to be prepared, or say the exact solution of the given model, which makes the approach hard to incorporate into the framework of VQE. To be specific, in the common case of VQE with unitary circuits, in each iteration, we need to prepare the ansatz state many times to measure the energy expectation value for subsequent optimization. In terms of the measurement-assisted approach, for the exactly solvable models like the pure toric code model, we can also prepare the same state as many times as we want because we know exactly how to perform classical feedback for the different measurement outcomes (e.g., annihilating anyon excitations in pairs) to correct the resulting state. Nevertheless, for an arbitrary state, especially for the ground states of the non-exactly solvable models like the GTC model we used in the manuscript, one cannot even prepare the same state each time because one does not know what classical feedback to apply such that the different resulting states after the measurement can be corrected to the same state. Therefore, when considering it as a VQE ansatz, implementing this measurement-assisted approach poses extra challenges compared to directly running unitary quantum circuits like FLDC.
\end{itemize}

\end{document}